\documentclass[11pt]{article}

\usepackage{tabu}
\usepackage[colorlinks=true, linkcolor=red, urlcolor=blue, citecolor=gray]{hyperref}
\usepackage{enumerate}
\usepackage{color}
\usepackage[usenames,dvipsnames,svgnames,table]{xcolor}
\usepackage{amsmath,amsfonts,amsthm,amssymb}
\usepackage{url}
\usepackage{algorithm}
\usepackage{algpseudocode}
\usepackage{bbm}
\usepackage{epsfig}
\usepackage{setspace}
\usepackage{tablefootnote}
\usepackage{geometry}
\usepackage[most]{tcolorbox}
\tcbset{
    frame code={}
    center title,
    left=4pt,
    right=4pt,
    top=3pt,
    bottom=3pt,
    width=1.02\textwidth,
    colback=white,
    enlarge left by=-1mm,
    enlarge top by=-.5mm,
    enlarge bottom by=-.5mm,
    boxsep=2pt,
    arc=0pt,outer arc=0pt,
    }
    
\definecolor{verylight}{gray}{0.90}
\definecolor{lesslight}{gray}{0.85}
\definecolor{lightPeach}{RGB}{255,186,127}
\definecolor{lightCorn}{RGB}{118,200,245}
\DeclareMathOperator*{\E}{\mathbb{E}}
\let\Pr\relax
\DeclareMathOperator*{\Pr}{\mathbb{P}}

\newcommand{\R}{\mathbb{R}}
\newcommand{\Z}{\mathbb{Z}}

\newcommand{\wh}{\widehat}

\newcommand{\eqdef}{\mathbin{\stackrel{\rm def}{=}}}
\newcommand{\norm}[1]{\|#1\|}

\makeatletter

\newcommand{\Netg}{\mathcal{T}_{n,\gamma}}

\newcommand{\Netc}{\mathcal{L}_{n,\gamma}}

\newcommand{\Net}{\mathcal{N}}

\newcommand{\Emo}{\mathcal{E}_{mono}}
\newcommand{\ah}{\alpha_{hard,\mathcal{N}}}
\newcommand{\Xh}{X_{hard,\mathcal{N}}}

\newcommand{\Ef}{\mathcal{E}_{fail}}
\newcommand{\Ifo}{\mathcal{I}_{act}}

\newcommand{\Ess}{\mathcal{E}_{stab}}
\newcommand{\Ec}{\mathcal{E}_{term}}
\newcommand{\Ea}{\mathcal{E}_{act}}
\newcommand{\Ean}[1]{\mathcal{E}_{act,#1}}
\newcommand{\Ev}{\mathcal{E}_{val}}

\newcommand{\bEc}{\mathcal{\wh E}_{term}}
\newcommand{\term}{\mathrm{term}}
\newcommand{\Input}{X}
\newcommand{\Inh}{A}

\newcommand{\Output}{Y}
\newcommand{\All}{N}
\newcommand{\Ex}{E}
\newcommand{\Ih}{I}
\DeclareMathOperator{\pot}{pot}
\DeclareMathOperator{\tro}{output}
\DeclareMathOperator{\wta}{WTA}
\DeclareMathOperator{\ewta}{WTA-EXP}
\newcommand{\spara}[1]{\medskip\noindent{\bf #1}}
\newcommand{\kbound}{12(\log_2 n+2)}
\newcommand{\lognc}{\lceil \log_2 n \rceil}

\newtheorem*{rep@theorem}{\rep@title}
\newcommand{\newreptheorem}[2]{%
\newenvironment{rep#1}[1]{%
 \def\rep@title{#2 \ref{##1}}%
 \begin{rep@theorem}}%
 {\end{rep@theorem}}}
\makeatother
\newtheorem{theorem}{Theorem}
\newreptheorem{theorem}{Theorem}

\newtheorem{corollary}[theorem]{Corollary}
\newtheorem{lemma}[theorem]{Lemma}
\newtheorem*{lemma*}{Lemma}
\newreptheorem{lemma}{Lemma}

\newtheorem{claim}[theorem]{Claim}
\newtheorem{definition}[theorem]{Definition}

\newif\ifpodc
\podctrue
\podcfalse

\numberwithin{theorem}{section}

  \usepackage{nth}
  \usepackage{intcalc}

\title{Winner-Take-All Computation in Spiking Neural Networks}

\author{Nancy Lynch \\ MIT \\ \texttt{lynch@csail.mit.edu} \and Cameron Musco \\ Microsoft Research\\ \texttt{camusco@microsoft.com} \and Merav Parter \\ Weizmann Institute of Science\\ \texttt{merav.parter@weizmann.ac.il}}

\begin{document}

\maketitle

\begin{abstract}
In this work 
we study biological neural networks from an algorithmic perspective, focusing on understanding 
tradeoffs between computation time and network complexity. 
Our goal is to abstract real neural networks in a way that, while not capturing all interesting features, preserves high-level behavior and allows us to make biologically relevant conclusions.
Towards this goal, we consider the implementation of algorithmic primitives in a simple yet biologically plausible model of \emph{stochastic spiking neural networks}.
In particular, we show how the stochastic behavior of neurons in this model can be leveraged to solve a basic \emph{symmetry-breaking task} in which we are given neurons with identical firing rates and want to select a distinguished one. In computational neuroscience, this is known as the winner-take-all (WTA) problem, and it is believed to serve as a basic building block in many tasks, e.g., learning, pattern recognition, and clustering. We provide efficient constructions of WTA circuits in our stochastic spiking neural network model, as well as lower bounds in terms of the number of auxiliary  neurons required to drive convergence to WTA in a given number of steps. These lower bounds demonstrate that our constructions are near-optimal in some cases.

This work covers and gives more in-depth proofs of a subset of results originally published in \cite{lynch2016computational}. It is adapted from the last chapter of C. Musco's Ph.D. thesis \cite{cam18}.
\end{abstract}

\section{Background and Introduction to Results}\label{sec:background}

Neural networks are studied in a number of academic communities from a wide range of perspectives. Significant work in computational neuroscience focuses on developing somewhat realistic mathematical models for these networks and generally studying their capacity to process information \cite{izhikevich2004model,trappenberg2009fundamentals}. 
On the more theoretical side, a variety of artificial network models such as perceptron and sigmoidal networks, Hopfield networks, and Boltzmann machines have been developed. These models are tractable to theoretical analysis and studied in the context of their computational power, and applications to general function approximation, classification, and memory storage \cite{hornik1989multilayer,maass1991computational,siegelmann1992computational,maass1997networks}. 
In practical machine learning, biological fidelity and often theoretical tractability are put aside, and researchers study how neural-like networks and learning rules can be used to efficiently represent and learn complex concepts \cite{haykin2009neural,lecun2015deep}.

In contrast to the common approach in computational neuroscience and machine learning, in our work we focus not on general computation ability or broad learning tasks, but on specific algorithmic implementation and analysis. 
We define a model of neural computation along with algorithmic problems that seem to be an important building blocks for higher level processing and learning tasks. We then design neural networks in our model that solve these problems, rigorously analyzing the complexity of our solutions in terms of asymptotic runtime and network size bounds. 
We hope that this new paradigm will provide new insights about computational tradeoffs, the power of randomness, and the role of noise in biological systems.

While focusing on somewhat different questions, our line of work is inspired by (1) work on the computational power of spiking neural networks, most notably by Maass et al. \cite{maass1997networks,maass1999neural,maass2000computational} and (2) the work of Les Valiant \cite{valiant2000circuits,valiant2000neuroidal,valiant2005memorization}, who defined the
\emph{neuroidal model} of computation and investigated implementations of basic learning modules within this model.

\subsection{Spiking Neural Networks}
We consider a model of \emph{spiking neural networks} (SNNs) \cite{maass1996computational,maass1997networks,gerstner2002spiking,izhikevich2004model,habenschuss2013stochastic}, defined formally in Section \ref{sec:modelW}, in which neurons fire in discrete pulses, in response to a sufficiently high membrane potential. This potential is induced by spikes from neighboring neurons, which can have either an excitatory or inhibitory effect (increasing or decreasing the potential). Our model is \emph{stochastic} -- each neuron functions as a probabilistic threshold unit, spiking with probability given by applying a sigmoid function to the membrane potential. In this respect, our networks are similar to the popular Boltzmann machine \cite{ackley1985learning}, with the important distinction that synaptic weights are not required to be symmetric and, as observed in nature, neurons are either strictly inhibitory (all outgoing edge weights are negative) or excitatory. Additionally, in this work, we focus on networks with fixed edge weights. The literature on Boltzmann machines tends to focus on learning, in which edge weights are adjusted iteratively until the network converges to some desired distribution on firing patterns. 
While a rich literature focuses on deterministic threshold circuits \cite{minsky1969perceptrons,hopfield1986computing}
we employ a stochastic model as 
it is widely accepted that neural computation is inherently stochastic \cite{allen1994evaluation,shadlen1994noise,faisal2008noise}, and that while this can lead to a number of challenges, it also affords significant computational advantages \cite{maass2014noise}. 

\subsection{The Winner-Take-All Problem} 

In this work, we focus on the Winner-Take-All (WTA) problem, which is one of the most studied problems in computational neuroscience. 
A WTA network has $n$ input neurons, $n$ corresponding outputs, and a set of auxiliary neurons that facilitate computation. 
The goal is to pick a `winning' input -- that is, the network  should produce a single firing output which corresponds to a firing input. Often the winning input is  the one with the highest firing rate, in which case WTA serves as a neural max function. We focus on the case when all inputs have the same or similar firing rates, in which case WTA serves as a leader election unit. A formal definition of the WTA problem is given in Section \ref{sec:wtaDef}.

WTA is widely applicable, including in circuits that implement visual attention via WTA competition between groups of neurons that process different input classes \cite{koch1987shifts,lee1999attention,itti2001computational}. It is also the foundation of competitive learning \cite{nowlan1989maximum,kaski1994winner,gupta2009hebbian}, in which classifiers compete to respond to specific input types. More broadly, WTA is known to be a powerful computational primitive \cite{maass1999neural,maass2000computational} -- a network equipped with WTA units can perform some tasks significantly more efficiently than with just  linear threshold neurons (McCulloch-Pitts neurons or perceptrons). 

Due to its importance, there has been significant work on WTA, including in biologically plausible spiking networks \cite{lazzaro1988winner,yuille1989winner,thorpe1990spike,coultrip1992cortical,wang2003k,oster2006spiking,oster2009computation,al2015inherently}.
This work 
is extremely diverse -- while mathematical analysis is typically given, different papers show different guarantees and apply varying levels of rigor. To the best of our knowledge, prior to our work, no asymptotic time bounds (e.g., as a function of the number of inputs $n$) for solving WTA in spiking neural networks have been established.\footnote{Aside from immediate bounds for deterministic circuits using many ($\Omega(n)$) auxiliary neurons \cite{lazzaro1988winner,maass2000computational}.}
Additionally, previous analysis often requires a specific initial network state to show convergence and does not show that the network is self-stabilizing and converges from an arbitrary starting state, as is necessary in a biological system.

\subsection{Our Contributions}
We explore the tradeoff between the number of auxiliary neurons used in a WTA network (i.e., the complexity of the network) 
and the time required to select a winning output (to converge to a WTA state). 

\subsubsection{Network Constructions and Runtime Bounds}
One the upper bound side, in Section \ref{sec:wta2} we describe, for any  input size $n$ and failure probability $\delta  > 0$, a family  of networks using just two auxiliary  inhibitory neurons which solve the WTA problem in $O(\log n)$ steps in expectation, and  with probability  $\ge  1-\delta$ in $O(\log n \cdot  \log(1/\delta))$ steps (see Theorems \ref{thm:high} and \ref{thm:exp}). 

Our two-inhibitor construction is based on a simple random competition idea.
Outputs that fire in response to stimulation from their firing inputs excite two inhibitors, which, in turn, inhibit all the outputs. When more than one output fires, both inhibitors are excited. This leads to high levels of inhibition, causing firing outputs to stop firing and `drop out' of the WTA competition. When exactly one output fires, just one of the inhibitors, known as the \emph{stability inhibitor}, is excited. This inhibitor is responsible for maintaining a WTA steady-state: once a single output fires at a time step it becomes the winner of the network. It has a positive feedback self-loop that allows it to keep firing at subsequent times, while all other outputs do not fire due to inhibition from the stability inhibitor.

The basic network construction described above employs biologically plausible structures. In particular, convergence is driven using reciprocal excitatory-inhibitory connections, and stability is maintained via excitatory self-loops. Both these structures are used in many biological models of WTA computation \cite{yuille1989winner,coultrip1992cortical,roux2015tasks}.
It is widely believed that inhibition is crucial for solving WTA -- outputs compete for activation via \emph{lateral inhibition} or \emph{recurrent inhibition} \cite{coultrip1992cortical,roux2015tasks}. In our network, this inhibition is achieved through the two auxiliary inhibitors. Previous work has conjectured that widespread use of simple WTA implementations in the brain may explain how complex computation is possible even when inhibition is relatively limited and localized \cite{maass2000computational}. Our work shows that WTA can be achieved and maintained efficiently using very few inhibitors and with a very simple connectivity structure.

We also demonstrate that, with a larger number of auxiliary neurons, it is possible to obtain faster convergence. In particular, in Section \ref{sec:multi}, we describe, for any input size $n$ and failure probability $\delta  > 0$, a family  of networks using  $O(\log n)$ auxiliary inhibitory neurons which solve the WTA problem in $O(1)$ steps in expectation, and  with probability  $\ge  1-\delta$ in $O(\log(1/\delta))$ steps (see Theorems \ref{thm:logn} and \ref{thm:lognExp}). At a  high level, more auxiliary inhibitors allow for more fine-tuned levels of inhibition which drive faster convergence. In Section \ref{sec:tradeoffs} we sketch two constructions that allow for more general runtime-inhibitor tradeoffs, interpolating between our two-inhibitor and $O(\log n)$-inhibitor  constructions. See \cite{lynch2016computational} for more details.
\subsubsection{Lower Bounds}

Aside from the above network constructions and runtime analysis, we also prove lower bounds, showing that these constructions are optimal or near optimal. In Section \ref{sec:lb} we prove that no network can solve WTA (in a reasonable parameter regime) using just a single auxiliary  neuron (see Theorem \ref{thm:lb1}). Roughly, it is not possible for a single neuron to both drive fast convergence and maintain stability  of a valid WTA configuration once one has been reached. The dual role that inhibition plays in two-inhibitor construction of driving convergence and maintaining stability requires at least two inhibitors.
We also show that, considering a slightly restricted class of networks, our two-inhibitor construction is near-optimal. No network with just two-auxiliary neurons can solve WTA with constant probability in $o\left (\frac{\log n}{\log \log n} \right )$ steps (see Theorem \ref{thm:lb2}). This matches the runtime of our network up to a $O(\log \log n)$ factor.  See \cite{lynch2016computational} for additional lower bounds on the convergence time for networks using $\alpha > 2$ inhibitors, in a similar model to the one presented here.
\subsection{Road Map} 

In Section \ref{sec:modelW} we describe our spiking neural network model and specify the WTA problem. In Section \ref{sec:wta2} we describe and analyze the convergence of our simple family of two-inhibitor WTA networks. In Section \ref{sec:lb} we prove lower bounds that show the near optimality of our two-inhibitor construction. In Section \ref{sec:multi} we show how to obtain faster convergence using a network construction with $O(\log n)$ auxiliary inhibitors. This construction also requires generalizing our model to allow for a \emph{history period}, over which the firing of a neuron's neighbors can affect its membrane potential. Finally, in Section \ref{sec:wtadiscuss} we conclude by discussing open questions arising from our work and possible directions for future work.

\section{Spiking Neural Network Model}\label{sec:modelW}

In this section we describe our basic neural network model, which consists of a set of neurons connected by weighted synaptic connections. Each neuron fires (spikes) stochastically at each time step, with probability dependent on the firing of its neighbors in the previous time step. These neighbors may have either an excitatory  (inducing more firing) or inhibitory (suppressing firing) effect. In Section \ref{sec:multi} we describe a variation on this model in which the probability  that a neuron spikes depends not just on the spiking of its neighbors in the previous time step, but on the spikes during some history period preceding the current time.

\subsection{Network Structure}\label{sec:structure}

We first  describe the basic network structure and parameters.
A \emph{Spiking Neural Network} (SNN) $\Net =\langle \All, w,b,f\rangle$ consists of:
\begin{itemize}
\item $\All$, a set of neurons, partitioned into a set of input neurons $\Input$, a set of output neurons $\Output$, and a set of auxiliary  neurons $\Inh$. $\All$ is also partitioned into a set of \emph{excitatory} and \emph{inhibitory} neurons $\Ex$ and $\Ih$. All input and output neurons are excitatory. That is, $X \cup Y \subseteq E$.
\item $w: \All \times \All \rightarrow \R$, a weight function describing the weighted synaptic connections between the neurons in the network. $w$ is restricted in a few notable ways: 
\begin{itemize}
\item $w(u,x) = 0$ for all $u \in \All$, $x \in \Input$. 
\item Each excitatory neuron $v \in \Ex$ has $w(v,u)\ge 0$ for every $u$. Each inhibitory  neuron $v \in \Ih$ has $w(v,u)\le 0$ for every $u$.
\end{itemize}
\item $b: \All \rightarrow \mathbb{R}$, a bias function, assigning an activation bias to each neuron.
\item $f: \R \rightarrow [0,1]$, a spike probability function, satisfying a few restrictions:
\begin{itemize}
\item $f$ is continuous and monotonically increasing.
\item $\lim_{x\rightarrow \infty} f(x) = 1$ and $\lim_{x \rightarrow -\infty} f(x) = 0$.
\end{itemize}
\end{itemize}

\subsubsection{Remarks on Network Structure}
Before describing the dynamics of our neural network, we give a few remarks on, and explanations of, the above parameters determining the network structure.

\spara{Weight Function ($w$)}: The weight function $w$ describes the strength of the synaptic connections between neurons in $\All$.
The restriction that $w(u,x) = 0$ for every input neuron $x \in \Input$ 
 is motivated by the desire for networks to be composable. The input neurons in $\Input$ may be output neurons of another network,
and so incoming connections are avoided to simplify definitions and analysis when networks are composed in higher level modular designs.

The restriction that each neuron $v$ is either inhibitory or excitatory is motivated by the observation, known as \emph{Dale's principle}, that neurons typically employ the same neurotransmitter at each outgoing synapse, regardless of its target \cite{osborne2013dale}. Thus, all outgoing connections are either inhibitory or excitatory, depending on the transmitter  used.
For example,
inhibitory connections predominantly stem from inhibitory \emph{GABAergic neurons}, which employ the neurotransmitter \emph{gamma}-Aminobutyric acid (GABA) \cite{watanabe2002gaba,rudy2011three}. 

 We often view the weight function as defining the edge weights  of a directed graph, whose edges are the synaptic connections. Formally we can define:
\begin{definition}[Synaptic Connection Graph]  Given spiking neural network $\Net =\langle \All, w,b,f \rangle$, let $G(\Net)$ be the weighted directed graph with vertex set $N$ and a directed edge $(u,v)$ with weight $w(u,v)$ for  all $u,v$ with $w(u,v) \neq 0$.
\end{definition}
Note that the weight function $w(u,v)$ need not be symmetric, and typically will not be. Additionally, we allow $u \in \All$ with $w(u,u) \neq 0$. That is, $G(\Net)$ may have self-loops.

\spara{Bias Function ($b$)}:
The bias function, along with the spike probability  function, determines how large a neuron's membrane potential must be for the neuron to spike with good probability. The larger the bias, the more excited the neuron must be before it fires. We will see in Section \ref{sec:dynamics} exactly how the bias affects the spiking probability.

\spara{Spike Probability Function ($f$)}:
Common choices for the spike probability  function $f$ are symmetric functions with $f(0) = \frac{1}{2}$. For example, we will typically set $f$ to the sigmoid function $f(x) = \frac{1}{1+e^{-x/\lambda}}$ for some \emph{temperature parameter} $\lambda > 0$.

\subsection{Network Dynamics}\label{sec:dynamics}

We now describe in detail the dynamics of our neural network model.

A  \emph{configuration} $C: \All \rightarrow \{0,1\}$ of an SNN $\Net =\langle \All, w,b,f \rangle$ is a mapping from each neuron in the network to a firing state. $C(u) = 1$ indicates that $u$  fires (i.e., generates a spike). $C(u) = 0$ indicates that it does not fire. We similarly  define an \emph{input configuration} $C_\Input: \Input\rightarrow \{0,1\}$ to be a mapping from each input neuron to a firing state and an \emph{output configuration} $C_\Output: \Output\rightarrow \{0,1\}$ to be a mapping from each output neuron to a firing state. For any configuration $C$ and set of neurons $M \subseteq \All$, we let $C(M)$ be the restriction of $C$ to the domain $M$.

An SNN evolves in a sequence of discrete, synchronous times, which we label with integers $t = 0, 1,...$. We denote the configuration at time $t$ by  $\All^t$. Similarly, we denote the input and output configurations at time $t$ by $\Input^t \eqdef \All^t(\Input)$ and $\Output^t \eqdef \All^t(\Output)$ respectively. 

Formally, an \emph{execution} is a finite or infinite sequence of configurations. The \emph{length} of a finite execution $\All^0\All^1...\All^t$ is defined to be $t+1$. The length of an infinite execution $\All^0\All^1...$ is defined to be $\infty$. We analogously define an \emph{input execution} and an \emph{output execution} as a sequence of input configurations $\Input^0\Input^1...$ and output configurations $\Output^0\Output^1...$ respectively. 

For each neuron $u \in \All$, we use the notation $u^t \eqdef \All^t(u)$ to denote the firing state of neuron $u$ in the configuration $\All^t$. More generally, for any ordered set of neurons $M = \{m_1,...,m_n\}$ we let $M^t \in \{0,1\}^n$ denote the binary  vector with $m_j^t$ as its $j^{th}$ entry. For any configuration $C$, we let $\norm{C}_1 = | \{u \in N: C(u) = 1\}|$ denote the number of firing neurons in $C$.

We will typically use $\alpha$ to denote an execution, and $\alpha_\Input$, $\alpha_\Output$  to denote an input or output execution respectively. We will use a superscript to denote the length of a finite  execution.
For any execution $\alpha$ let $\tro(\alpha)$ be the output execution of the same length obtained by restricting each configuration in $\alpha$ to the output neurons $\Output$.

The behavior of an SNN is determined as follows:
\begin{itemize}
\item \textbf{Input Neurons}: For each problem we consider, we will specify how the infinite input execution $X^0X^1....$ is determined. In this work, we will typically fix the input so that for each $u \in \Input$, $u^t$ is constant for all $t \ge 0$. However, we may also specify a distribution from which $X^0X^1....$ is drawn. For example, this sequence may be generated by setting $u^t = 1$ with some probability  $p_u$ and $u^t = 0$ with probability $1-p_u$, independently at random for each $u \in \Input$ and each time $t$.
\item \textbf{Initial Firing States}: For each non-input $u \in \All \setminus \Input$, the firing state $u^0$ is arbitrary. In this work, we typically show convergence results that hold for all possible settings of these initial states, giving our networks a self-stabilizing property, since they will converge from any  arbitrary perturbation of the state (see e.g., Theorem \ref{thm:self}).

\item \textbf{Firing Dynamics}:
For each non-input neuron $u \in \All \setminus \Input $ and every time $t \ge 1$, let $\pot(u,t)$ denote the membrane potential at time $t$ and $p(u,t)$ denote
the corresponding firing probability. These values are calculated as:
\begin{align}
\label{eq:potentialOut}
\pot(u,t) = \left (\sum_{v \in \All}w(v,u) \cdot v^{t-1} \right) -b(u) 
\text{ and } p(u,t)= f(\pot(u,t))
\end{align}
where $f$  is the spike probability  function.
At time $t$, each non-input neuron $u$ fires independently with probability $p(u,t)$. 
\end{itemize}


Any SSN $\Net =\langle \All, w,b,f \rangle$, initial configuration $ \All^0$, and infinite input execution $\alpha_\Input$ define a probability distribution over infinite executions, $\mathcal{D}(\Net,\All^0,\alpha_\Input)$. This distribution is the natural distribution that follows from 
applying the stochastic firing dynamics of \eqref{eq:potentialOut}. Formally, for any  finite execution $\alpha$, we define the \emph{cone} of executions extending $\alpha$,  $A(\alpha)$, to be the set of all infinite executions that start with $\alpha$.
$\mathcal{D}(\Net,\All^0,\alpha_\Input): 
\mathcal{F} \rightarrow [0,1]$ is a probability measure where the $\sigma$-algebra $\mathcal{F}$ consists of all such cones, closed under complement, countable unions, and countable intersections.

Given $\mathcal{D}(\Net,\All^0,\alpha_\Input)$ we can also define a distribution $\mathcal{D}_\Output(\Net,\All^0,\alpha_\Input)$ on infinite output executions.
Given any finite output execution $\alpha_\Output$, we define the cone $A(\alpha_\Output)$ to be the set of all infinite output executions extending $\alpha_\Output$. We define the $\sigma$-algebra $\mathcal{F}_\Output$ to be the set of all such cones, closed under complement, countable union, and countable intersection.
Finally, for $F_\Output \in \mathcal{F}_\Output$,
 we define $\mathcal{D}_\Output(\Net,\All^0,\alpha_\Input): \mathcal{F}_\Output \rightarrow [0,1]$ by:
$$\mathcal{D}_\Output(\Net,\All^0,\alpha_\Input)[F_\Output ] = \mathcal{D}(\Net,\All^0,\alpha_\Input)[\{\alpha:\tro(\alpha) \in F_\Output\}].$$

\subsection{Problems and Solving Problems}\label{sec:problem}

A problem $P$ is a mapping from 
an infinite input execution $\alpha_\Input$ to a set of output distributions. A network $\Net$ is said to \emph{solve problem $P$ on input $\alpha_\Input$} if, for any initial configuration $\All^0$, the output distribution $\mathcal{D}_\Output(\Net,\All^0,\alpha_\Input)$  is an element of $P(\alpha_\Input)$. A network $\Net$ is said to \emph{solve problem P} if it solves $P$ on every infinite input execution $\alpha_\Input$. For an example of such a problem definition see Section \ref{sec:wtaDef}, where we formally define the winner-take-all problem.

\subsection{Basic Results and Properties of the Model}\label{sec:basicModel}

In this section we prove some basic properties of the spiking neural network model described in the preceding sections. The first property is a simple Markov independence claim: conditioned on the configuration at time $t-1$, a network's execution from time $t$ on is independent of all times before $t-1$. Formally:

\begin{lemma}[Markov Property]\label{lem:independence}
Let $\Net = \langle \All, w, b, f\rangle$ be an SNN. For any time $t \ge 1$, and finite execution $C^0...C^{t-1}$ of $\Net$, and any configuration $C$ of $\Net$:
\begin{align*}
\Pr[\All^t = C | \All^{t-1}\All^{t-2}...\All^0 = C^{t-1}C^{t-2}...C^0] = \Pr[\All^t = C | \All^{t-1}= C^{t-1}].
\end{align*}
\end{lemma}
\begin{proof}
The potential of every $u\in \All$ at time $t$ as computed in \eqref{eq:potentialOut} is determined by $\All^{t-1}$. Thus, the spike probability $p(u,t) = f(\pot(u,t))$ is fully determined by $\All^{t-1}$. 

So, conditioned on $\All^{t-1} = C^{t-1}$, $p(u,t)$ is a deterministic function of $C^{t-1}$.
We can compute:
\begin{align*}
\Pr[\All^t = C | \All^{t-1}= C^{t-1}] = \prod_{u\in \All} \left [C(u) \cdot p(u,t) + (1-C(u)) \cdot (1-p(u,t))\right ].
\end{align*}
So $\Pr[\All^t = C |\All^{t-1} = C^{t-1}] $ is a deterministic function of $C$ and all the $p(u,t)$ collectively and thus of $C$ and $C^{t-1}$. So for \emph{any} $C^0...C^{t-2}$,
\begin{align*}
\Pr[\All^t = C | \All^{t-1} = C^{t-1}] = \Pr[\All^t = C |\All^{t-1}\All^{t-2}...\All^{0} = C^{t-1}C^{t-2}...C^0],
\end{align*}
giving the lemma.
\end{proof}

In our proofs, we will often bound the probability of some event $\mathcal{E}_t$ occurring at time $t$, giving  a bound independent of the preceding network configuration $\All^{t-1}$. However, $\mathcal{E}_t$ itself  will depend on $\All^{t-1}$, and there may be correlations between $\mathcal{E}_t$ and $\mathcal{E}_{t'}$ for $t \neq t'$. Below, we give a useful lemma which allows us to bound the number of times that $\mathcal{E}_t$ occurs over a given time period by comparing  to  the number of  times that a coin tossed independently at each time would come up heads in the same time period.

\begin{lemma}\label{lem:coupling}
For every $t  \in \mathbb{Z}^{>0}$ let  $A_t \in \mathcal{A}$  be a  random variable in some domain $\mathcal{A}$, $f:\mathcal{A} \rightarrow \{0,1\}$ be any function, and $B_t = f(A_t)$. Let $Z_t \in \{0,1\}$ be a set of independent  random variables. Suppose:
\begin{enumerate}
\item $\Pr [B_1 =  1] \ge  \Pr[Z_1 = 1].$
\item For every $t \ge 2$, $\Pr[B_t  =1 | A_{t-1},...,A_1] \ge  \Pr[Z_t =1]$.
\end{enumerate}
Then for every $t$ and $d \in \mathbb{Z}^{\ge 0}$, 
\begin{align}\label{eq:coupleInduct}
\Pr \left [\sum_{i=1}^t B_i \ge d \right  ] \ge \Pr \left [\sum_{i=1}^t Z_i \ge d \right  ]. 
\end{align}
\end{lemma}
Lemma \ref{lem:coupling}  and its proof are similar to Lemma 2.2 of \cite{khabbazian2011decomposing}. However, we include a full proof for completeness  and since we are in a slightly different setting, where we condition on the  full past state, rather than just the preceding values of $B_t = f(X_t)$.
\begin{proof}
We prove the  result  via  induction on $t$.  The base  case  with $t=1$  is given by assumption (1). For  any $t  >1$, assuming that \eqref{eq:coupleInduct} holds for all  $t' <  t$, we have:
\small
\begin{align*}
\Pr \left [\sum_{i=1}^t B_i \ge d \right  ] &= \Pr \left [B_t =1 \mid \sum_{i=1}^{t-1} B_i = (d-1) \right  ]  \cdot \Pr \left [\sum_{i=1}^{t-1} B_i = (d-1) \right  ] + \Pr \left [\sum_{i=1}^{t-1} B_i \ge d \right  ]\\
&\ge \Pr \left [Z_t =1\right ] \cdot \Pr \left [\sum_{i=1}^{t-1} B_i = d-1 \right  ]  + \Pr \left [\sum_{i=1}^{t-1} B_i \ge d \right  ]\\
&= \Pr \left [Z_t =1\right ] \cdot \left (\Pr \left [\sum_{i=1}^{t-1} B_i \ge d-1 \right  ] -\Pr \left [\sum_{i=1}^{t-1} B_i \ge d \right  ] \right)  + \Pr \left [\sum_{i=1}^{t-1} B_i \ge d \right  ]\\
&= \Pr \left [Z_t =1\right ] \cdot \Pr \left [\sum_{i=1}^{t-1} B_i \ge d-1 \right  ] +  \Pr \left [Z_t = 0\right ] \cdot \Pr \left [\sum_{i=1}^{t-1} B_i \ge d \right  ].
\end{align*}
\normalsize
By the  inductive assumption we can then bound:
\begin{align*}
\Pr \left [\sum_{i=1}^t B_i \ge d \right  ]  &\ge \Pr \left [Z_t =1\right ] \cdot \Pr \left [\sum_{i=1}^{t-1} Z_i \ge d-1 \right  ] +  \Pr \left [Z_t = 0\right ] \cdot \Pr \left [\sum_{i=1}^{t-1} Z_i \ge d \right  ]\\
&= \Pr \left [\sum_{i=1}^{t} Z_i \ge d \right  ]
\end{align*}
which gives the lemma. 
\end{proof}


We next  prove a related theorem, but in a more specialized setting. We consider a set of neurons $\{u_1,...,u_s\}$ for which we can lower bound the probability of each $u_i$ spiking at time $t+1$ given that it spiked at time $t$ (i.e., given that $u_i^{t} = 1$). We show that, while the behavior of the neurons may be highly correlated, the number of neurons in the set that spike for $t$ consecutive times can be lower bounded by comparing these neurons to a set of independent  random variables with 
comparable spiking probabilities.

\begin{lemma}\label{lem:coupling2}
Let $\Net = \langle \All, w, b, f\rangle$ be an SNN, and let $\{u_1,...,u_s\} \subseteq \All$ be any set of neurons in the network. Let $Z_{i,t} \in \{0,1\}$ be a set of independent random variables. 
Suppose that:
\begin{enumerate}
\item The initial configuration $\All^0$ of $\Net$  has $\All^0(u_i) = 1$  for every $i \in \{1,...,s\}$.
\item For every $i \in \{1,...,s\}$, any configuration $C$ of $\Net$ with $C(u_i) = 1$, and any $t \ge 0$:
$$\Pr [u_i^{t+1} = 1 | \All^{t} = C] \ge  \Pr[Z_{i,t+1} = 1].$$
\end{enumerate}
Let $\mathcal{I}_i(t) \in \{0,1\}$ be an indicator variable for the event that $u_i^1 =... = u_i^t  = 1$. Let $\mathcal{\bar I}_i(t) \in \{0,1\}$ be an indicator variable for the event that $Z_{i,1} = ... = Z_{i,t} = 1$.
Then for every $t$ and $d \in \mathbb{Z}^{\ge 0}$, 
\begin{align}
\Pr \left [\sum_{i=1}^s \mathcal{I}_i(t) \ge d \right  ] \ge \Pr \left [\sum_{i=1}^s \mathcal{\bar I}_i(t)  \ge d \right  ]. 
\end{align}
\end{lemma}
\begin{proof}
We prove the lemma via a coupling argument. At a high level, we 
define a set of auxiliary random variables $\mathcal{\hat I}_i(t) $ for $i \in \{1,...,s\}$. We  construct these random variables such that their joint distribution is \emph{identical} to that of the random variables $\mathcal{ I}_i(t) $.  Additionally, we correlate $\mathcal{\hat I}_i(t)$ with the variables $\{Z_{i,t}\}$ in such a way that we always have $\mathcal{\hat I}_i(t) \ge \mathcal{\bar I}_i(t)$. We  thus have:
\begin{align}\label{highlevelD}
\Pr \left [\sum_{i=1}^s \mathcal{I}_i(t) \ge d \right  ] = \Pr \left [\sum_{i=1}^s \mathcal{\hat I}_i(t) \ge d \right  ] \ge \Pr \left [\sum_{i=1}^s \mathcal{\bar I}_i(t) \ge d \right  ],
\end{align}
which gives the lemma.

\medskip
\spara{Definition of Coupled Random Variables $\mathcal{\hat I}_i(t)$.}
\medskip

Given $\Net$, the distribution on executions of $\Net$ with   initial configuration $\All^0$ and input configuration $\alpha_\Input$, induced by the update rules described in Section \ref{sec:dynamics} is given by $\mathcal{D}(\Net, \All^0,\alpha_X)$. We define the distribution $\mathcal{\hat D}(\Net, \All^0,\alpha_X)$, which  is identical to $\mathcal{D}(\Net, \All^0,\alpha_X)$ except coupled to the auxiliary random variables $\{Z_{i,t}\}$ in the following way: 

For any $t \ge 0$, execution $\alpha^{t} = C^0...C^{t}$, and $i \in \{1,...,s\}$ with $C^0(u_i) = .... = C^{t}(u_i) = 1$ let 
\begin{align}\label{epsDef}
\epsilon_{i,\alpha^{t} } = \Pr_{\mathcal{D}(\Net, \All^0,\alpha_X)}[u_i^{t+1} =1| \All^0...\All^{t} = \alpha^{t}] - \Pr[Z_{i,t+1} = 1].
\end{align}
By assumption  (2) in the lemma statement and Lemma \ref{lem:independence} we have $\epsilon_{i,\alpha^{t} }  \ge 0$. Let $E_{i,\alpha^{t}} \in \{0,1\}$ be a random variable which is independently set to $1$ with probability $\frac{\epsilon_{i,\alpha^{t} } }{1-\Pr[Z_{i,t+1} = 1]}$ and $0$ otherwise. The distribution $\mathcal{\hat D}(\Net, \All^0,\alpha_X)$ is given by iteratively drawing a configuration $\All^{t+1}$ in the same way as in $\mathcal{D}(\Net, \All^0,\alpha_X)$, with spiking probabilities given by the potentials induced by $\All^{t}$. However, if $i \in \{1,...,s\}$ and $N^0...\All^{t} = \alpha^t$ with $N^0(u_i) = .... = N^{t}(u_i) = 1$, we set
\begin{align}\label{coupSet}
{u}_i^{t+1} = \max ( Z_{i,t+1}, E_{i,\alpha^{t}}).
\end{align}
Using \eqref{epsDef} and the definition of $E_{i,\alpha^{t}}$ we can see that 
\small
\begin{align*}
 \Pr_{\mathcal{\hat D}(\Net, \All^0,\alpha_X)}[u_i^{t+1} =1 | \All^0...\All^{t} = \alpha^{t}] &= 1-  \Pr_{\mathcal{\hat D}(\Net, \All^0,\alpha_X)}[u_i^{t+1} =0 | \All^0...\All^{t} = \alpha^{t}] \\
  &= 1 - (1-\Pr[Z_{i,t+1} = 1]) \cdot \left (1-\Pr[E_{i,\alpha^t} = 1]\right)\\
 &= 1 - (1-\Pr[Z_{i,t+1} = 1]) \cdot \left (1-\frac{\epsilon_{i,\alpha^{t} } }{1-\Pr[Z_{i,t+1} = 1]}\right)\\
 &= \Pr[Z_{i,t+1} = 1] + \epsilon_{i,\alpha^{t} }\\
 &= \Pr_{\mathcal{D}(\Net, \All^0,\alpha_X)}[u_i^{t+1} | \All^0...\All^{t} = \alpha^{t}].\tag{By  \eqref{epsDef}}
\end{align*}
\normalsize
Thus, the probability that any neuron spikes at time $t$ conditioned on the network configuration at all times before $t$ is identical in executions drawn from $\mathcal{\hat D}(\Net, \All^0,\alpha_X)$ and $\mathcal{ D}(\Net, \All^0,\alpha_X)$. So
 for any $t$, if $\All^0...\All^{t}$ is drawn from $\mathcal{D}(\Net, \All^0,\alpha_X)$ and $\hat{\All}^0...\hat{\All}^t$ from $\mathcal{\hat D}(\Net, \All^0,\alpha_X)$, these two executions are identically distributed. In particular, if $\mathcal{\hat  I}_i(t) \in \{0,1\}$ is an indicator variable for the event that $\hat{\All}^1(u_i) =... = \hat{\All}^t(u_i)  = 1$, then $\mathcal{\hat I}_i(t)$ and $\mathcal{I}_i(t)$ are identically distributed.
 
 \medskip
 \spara{Proof that Coupled Random Variables Upper Bound Independent Variables.}
\medskip

We can see that $\mathcal{\hat I}_i(t) \ge \mathcal{\bar I}_i(t)$ via an inductive argument. In the base case, since 
we assume $\All^0(u_i) = 1$ for all $i \in \{1,...,s\}$, we apply \eqref{coupSet} to generate $\hat{\All}^1(u_i)$. We
set $\hat{\All}^1(u_i) = \max ( Z_{i,1}, E_{i,\alpha^{0}}) \ge Z_{i,1}$ which gives $\mathcal{\hat I}_i(1) \ge \mathcal{\bar I}_i(1)$. For $t \ge 1$, if $\mathcal{\bar I}_i(t) = 0$ the claim holds trivially since $\mathcal{\hat I}_i(t), \mathcal{\bar I}_i(t) \in \{0,1\}$. Otherwise, we have $\mathcal{\bar I}_i(t) = 1$ which implies that $\mathcal{\bar I}_i(t-1) = 1$ and so $\mathcal{\hat I}_i(t-1) = 1$ by the inductive assumption. If $\mathcal{\hat I}_i(t-1) = 1$ then again we apply   \eqref{coupSet} to generate $\hat{\All}^t(u_i)$ and so have  $\hat{\All}^t(u_i) = \max ( Z_{i,t}, E_{i,\alpha^{t-1}}) \ge Z_{i,t}$, giving $\mathcal{\hat I}_i(t) \ge \mathcal{\bar I}_i(t)$.

Since $\mathcal{\hat I}_i(t)$ is identically distributed to $\mathcal{ I}_i(t)$ this completes the lemma by \eqref{highlevelD}.
\end{proof}


Our next lemma pertains specifically to networks with a sigmoid spike probability function, $f(x) = \frac{1}{1+e^{-x/\lambda}}$, which we use throughout this work. We show that given a network with temperature parameter $\lambda > 0$, we  can construct a network with an identical  execution distribution for any $\hat{\lambda} >  0$. Thus, we will always consider the case of $\lambda  = 1$, which implies the existence of networks satisfying all bounds given for  all $\lambda > 0$.

\begin{lemma}[Equivalence of Temperature Parameters]\label{lem:temp}
For $\lambda,\hat{\lambda}> 0$, let $f(x) = \frac{1}{1+e^{-x/\lambda}}$ and $\hat{f}(x) = \frac{1}{1+e^{-x/\hat{\lambda}}}$.
Given $\Net = \langle \All,w,b,f\rangle$, let  $\wh{\Net} = \langle \All,\hat{w},\hat{b},\hat{f}\rangle$ where  for all $u,v \in \All$, $\hat{w}(u,v) = w(u,v)  \cdot \frac{\lambda}{\hat{\lambda}}$ and $\hat{b}(u) = b(u) \cdot \frac{\lambda}{\hat{\lambda}}$. For any length initial configuration $\All^0$ and any infinite input execution $\alpha_\Input$:
$$\mathcal{D}(\Net,\All^0,\alpha_\Input) = \mathcal{D}(\wh{\Net},\All^0,\alpha_\Input).$$   
\end{lemma}
\begin{proof}
For  any $t \ge  1$ and any configuration $C$, we can compute the probability that $\Net$ is in this configuration at time $t$ conditioned on all past configurations as:
\small
\begin{align}\label{fullout}
\Pr_{\mathcal{D}(\Net,\All^0,\alpha_\Input)} [\All^t = C | \All^{t-1}...\All^0] = \prod_{u\in \All} \left [C(u) \cdot p(u,t) + (1-C(u)) \cdot (1-p(u,t))\right ]
\end{align}
\normalsize
where $p(u,t) = f(\pot(u,t))$. Fixing $\All^{t-1}...\All^0$, we can see that the potential computation \eqref{eq:potentialOut} is a linear function of $w(u,v)$ and $b(u)$ for all $u,v \in \All$. Thus, letting $\wh{\pot}(u,t)$ be the potential of $u$ at time $t$ in $\wh{\Net}$ given $\All^{t-1}...\All^0$, since  $\hat{w}(u,v) = w(u,v)  \cdot \frac{\lambda}{\hat{\lambda}}$ and $\hat{b}(u) = b(u) \cdot \frac{\lambda}{\hat{\lambda}}$, 
$$\wh{\pot}(u,t) = \pot(u,t) \cdot \frac{\lambda}{\hat\lambda}.$$
This gives that the probability of $u$ spiking at  time $t$ in $\wh{\Net}$ given $\All^{t-1}...\All^0$ equals :
\begin{align*}
\hat p(u,t) = \hat{f} (\wh \pot(u,t)) = \hat{f}  \left ( \pot(u,t) \cdot \frac{\lambda}{\hat\lambda}\right ) = f( \pot(u,t)) = p(u,t).
\end{align*}
And since $\hat p (u,t) = p(u,t)$ for all $u \in \All$ we have using \eqref{fullout}, for any $C$,
\begin{align}\label{singleStepBlah}
\Pr_{\mathcal{D}(\Net,\All^0,\alpha_\Input)} [\All^t = C | \All^{t-1}...\All^0] = \Pr_{\mathcal{D}(\wh{\Net},\All^0,\alpha_\Input)} [\All^t = C | \All^{t-1}...\All^0].
\end{align}
Finally inducting on $t$ we can show that for any finite execution $C^0...C^t$:
\begin{align}\label{tempMain}
\Pr_{\mathcal{D}(\Net,\All^0,\alpha_\Input)} [\All^0...\All^t = C^0...C^t]  = \Pr_{\mathcal{D}(\wh{\Net},\All^0,\alpha_\Input)} [\All^0...\All^t = C^0...C^t].
\end{align}
This holds trivially  for $t = 0$ since
$$\Pr_{\mathcal{D}(\Net,\All^{0},\alpha_\Input)} [\All^0 = C^0]  = \Pr_{\mathcal{D}(\wh{\Net},\All^{0},\alpha_\Input)} [\All^0 = C^0]  = 1$$
if $\All^0 = C^0$. Both probabilities are zero otherwise. For $t \ge 1$, assume that \eqref{tempMain} holds for all $t' < t$. Combined with \eqref{singleStepBlah} this gives:
\begin{align*}
\Pr_{\mathcal{D}(\Net,\All^{0},\alpha_\Input)} [\All^0...\All^t = C^0...C^t] &= \Pr_{\mathcal{D}(\Net,\All^{0},\alpha_\Input)} [\All^0...\All^{t-1} = C^0...C^{t-1}]\\ &\hspace{3em}\cdot \Pr_{\mathcal{D}(\Net,\All^{0},\alpha_\Input)} [\All^t = C^t | \All^{0}...\All^{t-1} =  C^0...C^{t-1}]\\
&= \Pr_{\mathcal{D}(\wh{\Net},\All^{0},\alpha_\Input)} [\All^0...\All^{t-1} = C^0...C^{t-1}]\\ &\hspace{3em} \cdot \Pr_{\mathcal{D}(\wh{\Net},\All^{0},\alpha_\Input)} [\All^t = C^t | \All^{t-1}...\All^0 =  C^0...C^{t-1}]\\
& = \Pr_{\mathcal{D}(\wh{\Net},\All^{0},\alpha_\Input)} [\All^0...\All^t = C^0...C^t].
\end{align*}
This completes the lemma.
\end{proof}

\subsection{Potential Modifications to the Basic Model}\label{sec:futureDirs}
There are many potential modifications to the basic network model described in Sections \ref{sec:structure} and \ref{sec:dynamics} which may be  interesting to consider in future work. We present some below.
 
\begin{itemize}
\item One interesting extension is to add a history period $h > 1$ to the network, so that the spiking probability of a neuron at time $t$ depends on the configuration of the network at times $t-1,...,t-h$. In Section \ref{sec:multi}, for example, we use a model with history period $h=2$ to design very fast WTA networks.
  \item We could consider a very general history model, defining $$\pot(u,t) = f(u,\All^{t-1},...,\All^{t-h})$$ where $f$ is any function. 
  \item A history period $h$ can be thought of as giving each neuron access to a length $h$ queue of firing patterns on which $\pot(u,t)$ depends. It may be interesting to model such a queue as residing within the neuron's state.  We could also consider neurons with other types of state, capturing various types of observed biological phenomena.  For example,
we could model a refractory period, in which a neuron cannot fire again for a certain number of time steps after firing \cite{andrew2003spiking,izhikevich2004model}.
\item A history period may  be used, for example, to model a universal decay in the influence of spikes over time. We
 may specify a non-increasing weight vector $ = (c_1,c_2,...,c_h) \in \mathbb{R}^{\ge 0}$ and modify the potential computation of \eqref{eq:potentialOut} such that for any $t \ge h$:
 \begin{align*}
 \pot(u,t) = \left (\sum_{i = 1}^h \sum_{v \in \All}c_i \cdot w(v,u) \cdot v^{t-i} \right) -b(u). 
 \end{align*} 

\item It may  be interesting to consider networks with neurons of multiple types, with different firing dynamics. The human brain contains as many as 10,000 distinct neuron types \cite{neuronBasic}. Understanding how important neuron specialization is and for what reasons it arises is a very interesting question.
\item Similarly, we note that there is evidence that Dale's principal can be violated and that some neurons do have both inhibitory and excitatory  outgoing connections \cite{osborne1979dale}. Modeling such neurons to better understand their role and importance is an interesting direction.
\end{itemize}

\subsection{The Winner-Take-All Problem}\label{sec:wtaDef}

We now define the main problem that we consider in this work, the binary winner-take-all (WTA) problem. In this problem, given $n$ input neurons, the goal is to converge to a configuration in which a single output neuron, corresponding to a firing input, fires. This neuron is referred to as the `winner' of the computation. We first define a valid WTA output configuration for a given input configuration:

\begin{definition}[Valid WTA Output Configuration]\label{def:output}
Consider any network $\Net$ with $n$ input neurons, labeled $x_1,...,x_n$, and $n$ output neurons, labeled $y_1,...,y_n$. For any input configuration $C_\Input$ of $\Net$, a \emph{valid WTA output configuration} for $C_\Input$  is any output configuration $C_\Output$ with $C_\Output(y_i) \le C_\Input(x_i)$ for all $i \in \{1,...,n\}$ and $\norm{C_\Output}_1 = \min(1,\norm{C_\Input}_1)$. 
\end{definition}

Interpreting the above definition, the restriction that $\norm{C_\Output}_1  = \min(1,\norm{C_\Input}_1)$ requires that if at least one input fires, exactly  one output fires. The condition $C_\Output(y_i) \le C_\Input(x_i)$ for all $i$ requires that this firing output corresponds to a firing input. If no inputs fire (i.e., if $\norm{C_\Input}_1 = 0$), then no outputs should fire.
With this definition, we can define the WTA problem (see Section \ref{sec:problem} for a description of how problems are defined in our SNN model):

\begin{definition}[Winner-Take-All Problem]\label{high:wta} Given input size $n \in \mathbb{Z}^{>0}$, convergence time $t_c \in \mathbb{Z}^{>0}$, stability time $t_s \in \mathbb{Z}^{>0}$, and failure probability $\delta > 0$, the winner-take-all problem $\wta(n,t_c,t_s,\delta)$ is defined as follows:
\begin{itemize}
\item If $\alpha_\Input$ is an input execution with $X^t$  fixed for all $t$, the output distribution $\mathcal{D}_\Output(\Net,\All^0,\alpha_\Input)$ can be any distribution on executions of $n$ output neurons satisfying:
\begin{itemize}
\item With probability $\ge 1- \delta$, there exists some $t \le t_c$ such that the output configuration is fixed at times $t,t+1,...,t+t_s$ and is a  valid WTA output  configuration for $X^t$ (Def. \ref{def:output}).
\end{itemize}
\item If $\alpha_\Input$ is any other input execution, the output distribution is unconstrained.
\end{itemize}
\end{definition}

Thus, to solve $\wta(n,t_c,t_s,\delta)$, with probability $\ge 1-\delta$, the network must converge to a valid output configuration within $t_c$ time steps and maintain this configuration for $t_s$  time steps. 

Due to the random firing behavior of our neurons, the network will eventually move to a different configuration with some probability. However, if the network solves $\wta(n,t_c,t_s,\delta)$, since convergence is required given any initial configuration $\All^0$, we can show that it must be self-stabilizing. That is, once it leaves a valid output configuration, it will converge again with probability $\ge1-\delta$ within $t_c$ steps, and maintain the new valid configuration again for $t_s$  steps. Formally:

\begin{theorem}[Self Stabilization of Winner-Take-All Networks]\label{thm:self}
If $\Net$ solves \\$\wta(n,t_c,t_s,\delta)$ for input execution $\alpha_\Input$ with $X^t$ fixed for all $t$, given any  finite execution $C^0...C^t$ of $\Net$, conditioned on $\All^0...\All^t = C^0...C^t$, with probability  $\ge 1-\delta$ there is some time $t' \le t + t_c$ such that  the output configuration for $\Net$ is fixed at times $t',t'+1,...,t'+t_s$ and is a  valid WTA output configuration for $X^t$.
\end{theorem}
\begin{proof}
Consider the distribution on infinite executions $\All^{t+1}\All^{t+2}...$ conditioned on $\All^0...\All^t = C^0...C^t$. Since the configuration at time $t' \ge t+1$ depends only on the configuration at time $t'-1$, this distribution is identical to $\mathcal{D}_\Output(\Net,C^{t},\alpha_\Input)$. 

Thus, if $\Net$ solves $\wta(n,t_c,t_s,\delta)$ on $\alpha_X$, conditioned on $\All^0...\All^t = C^0...C^t$, with probability  $1-\delta$, there is some time $t' \le t+t_c$ in which $\Net$ reaches a valid WTA output configuration for $X^t$ and remains there for $t_s$ steps, giving the lemma. 

\end{proof}
We can also define an expected-time version of the winner-take-all problem as follows:
\begin{definition}[Expected-Time Winner-Take-All Problem]\label{exp:wta} For any infinite input execution $\alpha_\Input = \Input^0\Input^1...$, stability time $t_s \in \mathbb{Z}^{> 0}$, and infinite output execution $\alpha_\Output = \Output^0\Output^1...$ define:
\small
\begin{align*}
 t(\alpha_\Input,t_s,\alpha_\Output) = \min \left \{t : \Output^t\text{ is a valid WTA output configuration and } \Output^t=...=\Output^{t+t_s}\right \}.
 \end{align*}
 \normalsize
  
Given input size $n \in \mathbb{Z}^{>0}$, convergence time $t_c \in \mathbb{Z}^{>0}$, stability time $t_s \in \mathbb{Z}^{>0}$, the expected-time winner-take-all problem $\ewta(n,t_c,t_s)$ is defined as follows:
\begin{itemize}
\item If $\alpha_\Input$ is an input execution with $X^t$  fixed for all $t$, the output distribution $\mathcal{D}_\Output(\Net,\All^0,\alpha_\Input)$ can be any distribution on executions of $n$ output neurons satisfying:
$$\E_{\mathcal{D}_\Output(\Net,\All^0,\alpha_\Input)} t(\alpha_\Input,t_s,\alpha_\Output) \le t_c.$$
\item If $\alpha_\Input$ is any other input execution, the output distribution is unconstrained.
\end{itemize}
\end{definition}

\section{A Two-Inhibitor Solution to the WTA Problem}\label{sec:wta2}
We now present  a simple solution to the WTA problems in Definitions \ref{high:wta} and \ref{exp:wta} in networks with spike probability given by a sigmoid function. 
We begin by defining a family of networks $\Netg$ for any input size $n$ and weight scaling parameter $\gamma \in \R^{+}$ that solve these problems.

\subsection{Network Definition}\label{sec:2inhDef}

We first  give a full definition of our family of two-inhibitor WTA networks, before  describing the intuition behind why these networks solve the WTA  problem (Definitions \ref{high:wta} and \ref{exp:wta}).

\begin{definition}[Two-Inhibitor WTA Network]\label{def:wtaNet} For any positive integer $n$ and $\gamma \in\R^{+}$, let $\Netg= \langle \All,w,b,f\rangle$ where the spike probability, weight, and bias functions are defined as follows:
\begin{itemize}
\item The spike probability  function $f$ is defined to be the basic sigmoid function:
\begin{align}\label{eq:sigmoid}
f(x) \eqdef \frac{1}{1+e^{-x}}.
\end{align}
\item The set  of neurons $\All$ consists of a set of $n$ input neurons $\Input$, labeled $x_1,...,x_n$, a set  of $n$ corresponding outputs $\Output$, labeled $y_1,...,y_n$, and two auxiliary inhibitor neurons labeled $a_s,a_c$.
\item The weight function $w$ is given by:
\begin{itemize}
\item $w(x_i,y_i) = 3\gamma$, for all $i$.
\item $w(y_i,y_i) = 2\gamma$, for all $i$.
\item $w(a_s,y_i) = w(a_c,y_i) = -\gamma$, for all $i$.
\item $w(y_i,a_s) = w(y_i,a_c) = \gamma$, for all $i$.
\item $w(u,v) = 0$ for any $u,v$ whose connection is not specified above.  
\end{itemize}
\item The bias function $b$ is given by:
\begin{itemize}
\item $b(y_i) = 3\gamma$ for all $i$.
\item $b(a_s) = \gamma/2$.
\item $b(a_c) = 3\gamma/2$.
\end{itemize}
\end{itemize}
\end{definition}

A diagram of $\Netg$ is shown in Figure \ref{fig:two}. Note that the two inhibitors $a_s$ and $a_c$ have identical outgoing connections, and differ just  in their bias.

\begin{figure}[h]
\centering
\includegraphics[width=0.6\textwidth]{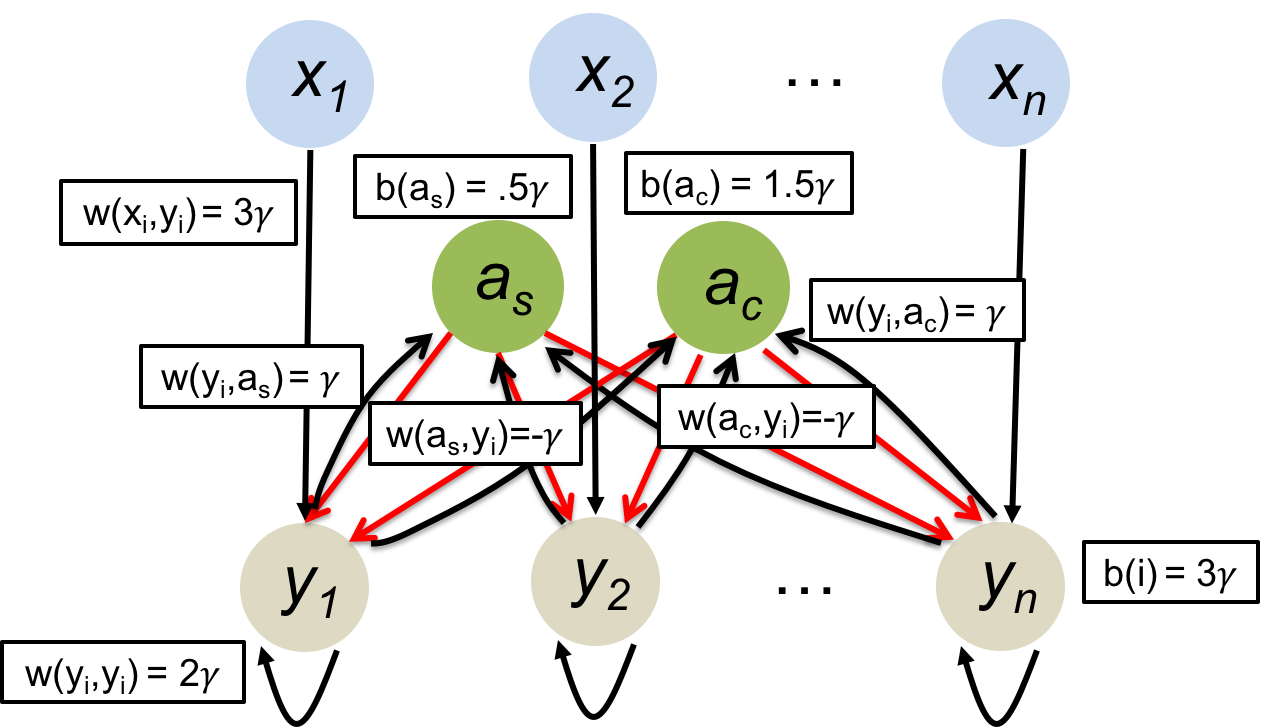}
\caption{Our two-inhibitor WTA network $\Netg$ as described in Definition \ref{def:wtaNet}.}
\label{fig:two}
\end{figure}

\subsubsection{Intuition Behind the Two-Inhibitor Network}
Before giving a formal analysis of the behavior of $\Netg$, we give some  intuition behind why this family of two-inhibitor networks solves the WTA problem.
In the description below, we informally refer to events that occur `with high probability'. We will quantify the meaning of such statements in our full analysis.

In $\Netg$, each input is connected to its corresponding output with a positive weight. Thus, firing inputs will initially  cause their corresponding outputs to fire with high probability. For the network to solve WTA, it must converge to a state in  which just a single one of these outputs fires.

To ensure this convergence,
$\Netg$ has two inhibitors $\Inh  = \{a_s, a_c\}$. The neuron $a_s$ is a \emph{stability} inhibitor that maintains a valid WTA output configuration once it has been reached. It fires with high probability at time $t$ whenever at least one output fires at time $t-1$. The neuron $a_c$ is a \emph{convergence} inhibitor that fires with high probability whenever at least two outputs fire at time $t-1$.

The weights connecting $a_s$ and $a_c$ to the outputs are set such that when both fire at time $t$, any output that fired at time $t$ will fire with probability $1/2$
at time $t+1$. Any output that did not fire at time $t$ will not fire at time $t+1$ with high probability. This distinguished behavior between previously firing and non-firing outputs is due to the self-loops on each output neuron, which allow firing outputs to partially  overcome the strong inhibition from $a_s$ and $a_c$.

In this way, if two or more outputs fire at time $t$, both inhibitors fire with high probability and the high level of inhibition causes outputs to `drop out of contention' for the winning position with probability $1/2$. After $O(\log n)$ time steps, nearly all the outputs stop firing and, with constant probability, there is a time step in which exactly one output fires.
Once this step occurs, with high probability, $a_c$ ceases firing and just $a_s$ fires. This decreased level of inhibition allows the winner to keep firing with high probability, as the inhibition is fully offset by the winner's excitatory self-loop. However, with high probability, the inhibition prevents any other output whose excitatory self-loop is inactive from firing. Thus the network remains in the valid WTA output configuration for a large number of time steps with high probability. 

In the event that a time in which a single output fires does not occur, then the network `resets'. No outputs fire at some time, causing the inhibitors to both cease firing. Thus, all outputs with firing inputs are able to fire, and convergence starts again. Within $O(\log 1/\delta)$  of these resets each reaching a valid WTA output state  with constant probability, the network reaches a valid WTA  output state  with probability  $\ge 1-\delta$ and so solves the WTA problem of Definition \ref{high:wta}. Similarly, the network requires $O(1)$ resets in expectation to reach a valid WTA output state, giving a solution to the expected-time version of the problem in Definition \ref{exp:wta}.
Formally, we will prove the following:

\tcbset{colback=white}
\begin{tcolorbox}
\begin{theorem}[Two-Inhibitor WTA]\label{thm:high}
For $\gamma \ge 4\ln((n+2)t_s/\delta)+10$,
$\Netg$ solves $\wta(n,t_c,t_s,\delta)$ for any  
$
t_c \ge 72(\log_2 n+1)\cdot (\log_2(1/\delta)+1).
$
\end{theorem}
\end{tcolorbox}

\begin{tcolorbox}
\begin{theorem}[Two-Inhibitor Expected-Time WTA]\label{thm:exp}
For $\gamma \ge 4\ln((n+2)t_s)+10$,
$\Netg$ solves $\ewta(n,t_c,t_s)$ for any  $
t_c \ge 108(\log_2 n + 3).$
\end{theorem}
\end{tcolorbox}

\spara{Proof Roadmap.}
We prove Theorems \ref{thm:high} and \ref{thm:exp} in
in Sections  \ref{sec:basic}-\ref{sec:completing}. The analysis is broken down as follows:
\begin{itemize}
\item[] Section \ref{sec:basic}: Prove basic \emph{one-step lemmas} which characterize  single time step transitions of $\Netg$, showing that  the neurons behave as described in the above high-level description.
\item [] Section  \ref{sec:stability}: Prove that, once in a valid WTA configuration, $\Netg$ stays in this configuration with high probability (that is, valid WTA configurations are stable).
\item [] Section  \ref{sec:bad2good}: Show that all configurations of $\Netg$ transition with high probability within two time steps to a small set of \emph{good configurations}, from which we will prove fast convergence.
\item [] Section \ref{sec:good2WTA} Show basic transition lemmas for this set of good configurations, characterizing the network's behavior at the times immediately following a good configuration.
\item [] Section \ref{sec:convergence} Use the above transition lemmas to show that the network converges, with constant probability, from any good configuration (and hence any  configuration by Section \ref{sec:bad2good}) to a valid WTA configuration within $O(\log n)$ time steps.
\item [] Section \ref{sec:completingA} Complete the analysis, demonstrating with what parameters $\Netg$ solves the winner-take-all problem (Definitions \ref{high:wta} and \ref{exp:wta}).
\end{itemize}

\subsection{Basic Results and One-step Lemmas}\label{sec:basic}
We begin with some basic results that will be important throughout our analysis, including a few `one-step' lemmas, which characterize the transition probabilities from a set  of configurations at time $t$ to a set of configurations at time $t+1$.


We first show that, unless a neuron has potential $0$, either it fires with high probability (i.e., except with probability that is inverse exponential in $\gamma$) or it does not fire with high probability.
\begin{lemma}[Characterization of Firing Probabilities]\label{lem:gap}
For any time $t\ge 1$ and any $u \in \All$:
\begin{align*}
&\text{If }\pot(u,t) = 0, \text{ then }  p(u,t) = 1/2.\\
&\text{If }\pot(u,t) < 0, \text{ then } p(u,t) \le e^{-\gamma/2}.\\
&\text{If }\pot(u,t) > 0, \text{ then } p(u,t) \ge 1-e^{-\gamma/2}.
 \end{align*}
\end{lemma}
\begin{proof}
If $\pot(u,t) = 0$, then by \eqref{eq:potentialOut} $p(u,t) = f(\pot(u,t)) = \frac{1}{1+e^0} = 1/2$. Otherwise consider the potential calculation of \eqref{eq:potentialOut} in the case when $h=1$:
\begin{align*}
\pot(u,t) = \sum_{v \in \All}w(v,u) \cdot v^{t-1} -b(u).
\end{align*}
By  Definition \ref{def:wtaNet}, for all $u,v$, $w(v,u)$ and $b(u)$ are integer multiples of $\gamma/2$. Thus, since $v^{t-1} \in \{0,1\}$, $\pot(u,t)$ is also an integer multiple of $\gamma/2$. So, if $\pot(u,t) < 0$, then $\pot(u,t) \le -\gamma/2$ and:
\begin{align*}
p(u,t) = f(\pot(u,t)) \le f(-\gamma/2) = \frac{1}{1+e^{\gamma/2}} \le e^{-\gamma/2}.
\end{align*}
Similarly, if $\pot(u,t) > 0$, then $\pot(u,t) \ge \gamma/2$ and so:
\begin{align*}
p(u,t+1) = f(\pot(u,t)) \ge f(\gamma/2) = \frac{1}{1+e^{-\gamma/2}} \ge 1- e^{-\gamma/2}.
\end{align*}
\end{proof}

We next show that if output $y_i$ does not correspond to a firing input (i.e., $x_i^t  = 0$), then starting from any configuration of $\Netg$ at time $t$, with high probability $y_i$ does not fire at time $t+1$.  That  is, with high probability, outputs that are not valid winners of the WTA computation do not fire.

\begin{tcolorbox}
\begin{lemma}[Correct Output Behavior]\label{lem:mono} For any time $t$, any configuration $C$ of $\Netg$, and any $i$ with $C(x_i) = 0$, 
$$\Pr[y_i^{t+1} = 1 | \All^t = C] \le  e^{-\gamma/2}.$$
\end{lemma}
\end{tcolorbox}
\begin{proof}
If $N^t = C$ then $x_i^t = C(x_i)$.
We can compute $y_i$'s potential at time $t+1$, assuming $x_i^{t} = 0$:
\begin{align*}
\pot(y_i,t+1) &= w(x_i,y_i) x_i^{t} + w(y_i,y_i) y_i^{t} + w(a_s,y_i) a_s^{t}+ w(a_c,y_i) a_c^{t} - b(y_i) \\ &\le 0 + 2\gamma + 0 + 0 - 3\gamma = -\gamma.
\end{align*}
Thus, by Lemma \ref{lem:gap}, since $\pot(y_i,t+1) < 0$, $p(y_i,t+1) \le e^{-\gamma/2}$.
\end{proof}


Applying Lemma \ref{lem:mono} and a simple union bound over all $n$ outputs yields the following corollary:
\begin{tcolorbox}
\begin{corollary}[Correct Output Behavior, All Neurons]\label{cor:mono} For any time $t$ and configuration $C$ of $\Netg$,  
$$\Pr[y_i^{t+1} \le x_i^t\text{ for all }i | \All^t = C ] \ge 1-ne^{-\gamma/2}.$$
\end{corollary}
\end{tcolorbox}
\begin{proof}
If $C(x_i) = 1$ then conditioned on $\All^t  = C$, $x_i^t = 1$ and so $y_i^{t+1} \le x_i^t$ always. Otherwise, by Lemma \ref{lem:mono}, if $C(x_i) = 0$, then $\Pr[y_i^{t+1} = 0 | \All^t = C] \ge  1-e^{-\gamma/2}$. Union bounding over  all such inputs (of which there are at most $n$) gives the corollary. 
\end{proof}

We next show that the inhibitors $a_s$ and $a_c$ behave as expected with high probability.
\begin{tcolorbox}
\begin{lemma}[Correct Inhibitor Behavior]\label{lem:reset}
For any time $t$ and configuration $C$  of $\Netg$,
\begin{enumerate}
\item If $\norm{C(\Output)}_1 = 0$, then $\Pr[a_s^{t+1} = a_c^{t+1} = 0 | \All^t = C] \ge 1-2e^{-\gamma/2}$.
\item If $\norm{C(\Output)}_1 = 1$, then $\Pr[a_s^{t+1} = 1\text{ and }a_c^{t+1} = 0 | \All^t = C] \ge 1-2e^{-\gamma/2}$.
\item If $\norm{C(\Output)}_1 \ge 2$, then $\Pr[a_s^{t+1} = 1=a_c^{t+1} = 1 | \All^t = C] \ge 1-2e^{-\gamma/2}$.
\end{enumerate}
\end{lemma}
\end{tcolorbox}
\begin{proof}
We prove each case above separately. Note that, conditioned on $\All^t  = C$, $\Output^t = C(Y)$.

\medskip
\spara{Case  1: $\norm{C(\Output)}_1 = 0$.}
\medskip

In this case, the inhibitors receive no excitatory signal from the outputs so, 
$$\pot(a_s,t+1) = -b(a_s) < 0\hspace{1em}\text{ and }\hspace{1em}\pot(a_c,t+1) = -b(a_c) < 0.$$
Thus by Lemma \ref{lem:gap} and a union bound over the two inhibitors, 
$$\Pr[a_s^{t+1} = a_c^{t+1} = 0 | \All^t  = C] \ge 1-2e^{-\gamma/2}.$$

\medskip
\spara{Case  2: $\norm{C(\Output)}_1 = 1$.}
\medskip

In this case we have:
\begin{align*}
\pot(a_s,t+1) &= \sum_{j=1}^n w(y_j,a_s) y_j^t - b(a_s) \\ &= \gamma - \gamma/2 = \gamma/2.\\
\pot(a_c,t+1) &= \sum_{j=1}^n w(y_j,a_c) y_j^t - b(a_c) \\ &= \gamma - 3\gamma/2 = -\gamma/2.
\end{align*}
Again by  Lemma \ref{lem:gap} and a union bound, $\Pr[a_s^{t+1} = 1\text{ and }a_c^{t+1} = 0 | \All^t  = C] \ge 1-2e^{-\gamma/2}$.

\medskip
\spara{Case  3: $\norm{C(\Output)}_1 \ge 2$}
\medskip

Finally, in this case:
\begin{align*}
\pot(a_s,t+1) &= \sum_{j=1}^n w(y_j,a_s) y_j^t - b(a_s) \\ &\ge 2\gamma - \gamma/2 = 3\gamma/2.\\
\pot(a_c,t+1) &= \sum_{j=1}^n w(y_j,a_c) y_j^t - b(a_c) \\ &\ge 2\gamma - 3\gamma/2 = \gamma/2.
\end{align*}
So by  Lemma \ref{lem:gap} and a union bound, $\Pr[a_s^{t+1} = a_c^{t+1} = 1 | \All^t  = C] \ge 1-2e^{-\gamma/2}$, completing the lemma.
\end{proof}

Combined with Corollary  \ref{cor:mono}, Lemma \ref{lem:reset} conclusion (1) gives:
\begin{tcolorbox}
\begin{lemma}[Quiescent Behavior]\label{lem:quiet}
Assume the input execution $\alpha_\Input$ of $\Netg$ has $\Input^t$ fixed for all $t$ and $\norm{\Input^t}_1 = 0$. For any time $t$ and configuration $C$ with $C(X) = X^t$,
\begin{align*}
\Pr[\norm{\All^{t+2}}_1 = 0 | \All^t = C] \ge 1-2(n+1)e^{-\gamma/2}.
\end{align*}
\end{lemma}
\end{tcolorbox}
\begin{proof}
Let $\mathcal{E}_{10}$ be the event that $\All^t = C$ and $\norm{Y^{t+1}}_1 = 0$. Let $\mathcal{E}_{20}$ be the event that $\norm{\All^{t+2}}_1 = 0$. That is, that no neurons fire  at time $t+2$. Conditioned on $\mathcal{E}_{10}$, by Lemma \ref{lem:reset} conclusion (1), with probability $\ge 1-2e^{-\gamma/2}$, $a_s^{t+2} = a_c^{t+2} = 0$. Again by Corollary \ref{cor:mono}, conditioned on $\mathcal{E}_{10}$, with probability $\ge 1-n e^{-\gamma/2}$, $\norm{Y^{t+2}}_1 = 0$. So, overall by a  union bound, 
\begin{align*}
\Pr[\mathcal{E}_{20}|\mathcal{E}_{10}] \ge 1-(n+2)e^{-\gamma/2}.
\end{align*}
By Corollary \ref{cor:mono}, since $\norm{X^t}_1 = 0$, $\Pr[\mathcal{E}_{10}|\All^t = C] \ge 1-n e^{-\gamma/2}$. 
We can thus bound:
\begin{align*}
\Pr[\mathcal{E}_{20} | \All^t = C] &\ge \Pr[\mathcal{E}_{10} | \All^t = C] \cdot \Pr[\mathcal{E}_{20}|\mathcal{E}_{10},\All^t = C]\\
&= \Pr[\mathcal{E}_{10} | \All^t = C] \cdot \Pr[\mathcal{E}_{20}|\mathcal{E}_{10}]\tag{Since, by definition, $\mathcal{E}_{10}$ implies $\All^t = C$.}\\
&\ge \left(1-ne^{-\gamma/2}\right) \cdot \left(1-(n+2)e^{-\gamma/2}\right) \\
&\ge 1-2(n+1)e^{-\gamma/2},
\end{align*}
which gives  the lemma.
\end{proof}

We next show that the stability inhibitor, with high probability, induces exactly the outputs that fired at the previous time step to fire in the next step. We show the lemma in fact for any  configuration in which exactly one inhibitor fires. Since $a_s$ and $a_c$ have identical outgoing edges, they have a symmetric effect on the firing probabilities of other neurons. 
\begin{tcolorbox}
\begin{lemma}[Stability Inhibitor Effect]\label{lem:stability1}
For any time $t$ and configuration $C$ of $\Netg$ with ($C(a_s)  = 1$ and $C(a_c)  = 0$) or ($C(a_s) = 0$ and $C(a_c)  = 1$) and $C(y_i) \le C(x_i)$ for all $i$,
$$\Pr[Y^{t+1} = Y^t | \All^t = C] \ge 1-ne^{-\gamma/2}.$$
\end{lemma}
\end{tcolorbox}
\begin{proof}

For any configuration $C$  with $C(a_c) = 1$  and $C(a_s) = 0$, let $\bar{C}$ denote the configuration with $\bar{C}(a_c) = 0$, $\bar{C}(a_s) = 1$, and $\bar{C}(u) = C(u)$ for all other $u \in \All \setminus \{a_c,a_s\}$.
Since $a_c$ and $a_s$ have  no self-loops and have identical outgoing connections, the distribution of $\All^{t+1}$  given $\All^t = C$ is identical to its distribution given $\All^t = \bar{C}$. Thus, we can assume without loss of generality in the proof of this lemma that $C(a_s) = 1$ and $C(a_c) = 0$.

Conditioned on $\All^t = C$, $y_i^{t} \le x_i^t$ by assumption. So for  any output with $y_i^t = 1$, we have $x_i^t= 1$. This gives:
\begin{align*}
\pot(y_i,t+1) &= w(x_i,y_i) x_i^t + w(y_i,y_i) y_i^t + w(a_s,y_i) a_s^t+ w(a_c,y_i) a_c^t - b(y_i) \\&= 3\gamma + 2\gamma - \gamma + 0 - 3\gamma\\&= \gamma.
\end{align*}
In contrast, for any output with $y_i^t = 0$:
\begin{align*}
\pot(y_i,t+1) &= w(x_i,y_i) x_i^t + w(y_i,y_i) y_i^t + w(a_s,y_i) a_s^t+ w(a_c,y_i) a_c^t - b(y_i) \\ &\le 3\gamma + 0 - \gamma + 0 - 3\gamma \\&= -\gamma.
\end{align*}
Thus, by Lemma \ref{lem:gap}, if $y_i^t = 1$, then $y_i^{t+1} = 1$ with probability $\ge 1-e^{-\gamma/2}$. If $y_i^t = 0$, then $y_i^{t+1} = 0$ with probability $\ge 1-e^{-\gamma/2}$. The lemma follows after  union bounding over all $n$  outputs. 
\end{proof}

Finally, we show that when both the stability and convergence inhibitors fire at time $t$, not only do outputs not firing at time $t$ not fire at  time $t+1$  with high probability, but also, all firing outputs at time $t$ stop firing with probability $1/2$  at time $t+1$. Conditioned on the configuration at time $t$, these outputs fire independently, a property which will be useful in our eventual proof of progress towards a valid WTA configuration in Lemma \ref{lem:progress}.
\begin{tcolorbox}
\begin{lemma}[Convergence Inhibitor Effect]\label{lem:convergence}
For any time $t$ and configuration $C$ of $\Netg$ with $C(a_s)  = C(a_c)  = 1$ and $C(y_i) \le C(x_i)$ for all $i$,
\begin{enumerate}
\item $\Pr[y_i^{t+1} \le y_i^t\text{ for  all }i | \All^t = C] \ge 1-ne^{-\gamma/2}$.
\item If $y_i^t = 1$, $\Pr[y_i^{t+1} = 1 | \All^t = C] = 1/2$.
\item For $i \neq j$, $y_i^{t+1}$ and $y_j^{t+1}$ are independent conditioned on $\All^t = C$.
\end{enumerate}
\end{lemma}
\end{tcolorbox}
\begin{proof}
Conditioned on $\All^t = C$, if $y_i^t = 1$, by assumption $x_i^t = 1$. We can thus compute:
\begin{align*}
\pot(y_i,t+1) &= w(x_i,y_i) x_i^t + w(y_i,y_i) y_i^t + w(a_s,y_i) a_s^t+ w(a_c,y_i) a_c^t - b(y_i) \\ &= 3\gamma + 2\gamma - \gamma -\gamma - 3\gamma \nonumber\\&= 0.
\end{align*}
We thus have  $\Pr[y_i^{t+1} = 1 | \All^t = C] = 1/2$ by Lemma \ref{lem:gap}. This gives conclusion (2). Conclusion (3) holds trivially since, with $\All^t$ fixed, $u^{t+1}$ is independent of $v^{t+1}$ for all $u \neq v$.

We can also bound if $y_i^t  =  0$: 
\begin{align*}
\pot(y_i,t+1) &= w(x_i,y_i) x_i^t + w(y_i,y_i) y_i^t + w(a_s,y_i) a_s^t+ w(a_c,y_i) a_c^t - b(y_i) \\ &\le 3\gamma + 0 - \gamma -\gamma - 3\gamma \\&= -2\gamma.
\end{align*}
Thus, by  Lemma \ref{lem:gap}, $\Pr[y_i^{t+1} = 1 | \All^t = C] \le e^{-\gamma/2}$. By a union bound over at most $n$ such outputs, we have, with probability $\ge 1-ne^{-\gamma/2}$, $y_i^{t+1}  \le y_i^t$  for all $i$  , completing the lemma.
\end{proof}

\subsection{Stability}\label{sec:stability}

In this section we extend our definition of a valid WTA output configuration (Definition \ref{def:output}), to give a more restrictive notion of a valid WTA configuration, which additionally requires that the auxiliary neurons $a_s,a_c$ are in a good state. We show that once the network is in such a state at time $t$, it remains there with high probability at time $t+1$.

\begin{definition}[Valid WTA Configuration] \label{def:valid}
A \emph{valid WTA configuration} of $\Netg$ is a configuration $C$ with
$C(y_i) \le C(x_i)$ for all $i \in \{1,...,n\}$ and $\norm{C(\Output)}_1 = \min(1,\norm{C(\Input)}_1)$ (i.e., the outputs satisfy Definition \ref{def:output}) and further, $C(a_c) = 0$ and $C(a_s) = \min(1,\norm{C(\Input)}_1)$.  
\end{definition}
In the above we require $C(a_s) = \min(\norm{C(\Input)}_1,1)$. That is, the stability  inhibitor fires in a valid WTA configuration, unless no inputs fire. If no inputs fire, a valid WTA configuration requires that neither $a_s$ nor $a_c$ fire and additionally, that no outputs fire.

\begin{tcolorbox}
\begin{lemma}[Stability of Valid Configurations]\label{lem:stability} Assume the input execution $\alpha_\Input$ of $\Netg$ has $X^t$ fixed for all $t$. For any time $t$ and valid WTA configuration $C$ with $C(X) = X^t$,
$$\Pr[\All^{t+1} = \All^t | \All^t = C] \ge 1-(n+2)e^{-\gamma/2}.$$
\end{lemma}
\end{tcolorbox}
\begin{proof}
By  Definition \ref{def:valid}, since $C$ is a valid WTA configuration, 
we have $$\norm{C(\Output)}_1 = \min(1,\norm{C(\Input)}_1) \in \{0,1\}.$$
We prove the lemma via a case analysis on $\norm{C(Y)}_1$.

\medskip
\spara{Case 1 : $\norm{C(\Output)}_1 = 0$}
\medskip

In this case, since $C$ is a valid WTA configuration, according to Definition \ref{def:valid}, conditioned on $\All^t = C$,  we must have $\norm{X^t}_1 = 0$ and $a_s^t = a_c^t = 0$. 
By  Corollary \ref{cor:mono}, since $\norm{X^t}_1 = 0$, 
$$\Pr[\norm{\Output^{t+1}}_1 = 0 | \All^t  = C] \ge 1-ne^{-\gamma/2}.$$
 By Lemma \ref{lem:reset} conclusion (1), since $\norm{C(\Output)}_1 = 0$,  $\Pr[a_s^{t+1} = a_c^{t+1} = 0 | \All^t  = C] \ge  1-2e^{-\gamma/2}$. By a union bound, recalling that $\Input^t$  is fixed for all $t$ by assumption, $\Pr[\All^{t+1} = \All^{t} | \All^t  = C] \ge  1-(n+2)e^{-\gamma/2}$.

\medskip
\spara{Case 2 : $\norm{C(\Output)}_1 = 1$}
\medskip

In this case by Definition \ref{def:valid}, we have $C(y_i)  \le C(x_i)$ for all $i$, $C(a_s) = 1$, and $C(a_c) = 0$. We can thus apply Lemma \ref{lem:stability1}, giving that $\Pr[Y^{t+1} = Y^t | \All^t = C] \ge  1-ne^{-\gamma/2}$.
Additionally, by Lemma \ref{lem:reset} conclusion (2), since $\norm{C(\Output)}_1 = 1$, $\Pr[a_s^{t+1} = 1\text{ and }a_c^{t+1} = 0 | \All^t  = C] \ge 1-2e^{-\gamma/2}$. So by a  union bound, 
$$\Pr[\All^{t+1} = \All^t | \All^t = C] \ge 1-(n+2)e^{-\gamma/2},$$ 
giving the result in this case and completing the lemma.
\end{proof}

Lemma \ref{lem:stability} immediately implies a bound on the probability that $\Netg$ remains in a valid WTA configuration for $t_s$ consecutive times. 
\begin{tcolorbox}
\begin{corollary}[Stability of Valid WTA Configurations]\label{cor:stability}
Assume  the input execution $\alpha_X$ of $\Netg$ has $X^t$ fixed for all $t$. For any time $t$ and valid WTA configuration $C$ with $C(X) = X^t$, 
\begin{align*}
\Pr[\All^t  = \All^{t+1} = ... = \All^{t+t_s} | \All^t = C] \ge 1-t_s (n+2)e^{-\gamma/2}.
\end{align*}
\end{corollary}
\end{tcolorbox}
\begin{proof}
Applying Lemma \ref{lem:stability} for each time $t+1,...,t+t_s$ in succession gives the result.
\end{proof}

\subsection{Convergence to Good Configurations}\label{sec:bad2good}

\begin{table}[ht]
\small
\centering
\begin{tabu}{|c|[1pt]c|c|c|c|} \hline
Configuration Type& $C(y_i) \le C(x_i)$ $\forall i$?  & $\norm{C(\Output)}_1$ & $C(a_s)$ & $C(a_c)$ \\
 \tabucline[1pt]{1-5}
\cellcolor{verylight}  Valid WTA (Def. \ref{def:valid}) &  $\checkmark$ & $\min(1,\norm{C(\Input)}_1)$ & $1$ & $0$ 	 \\
\hline
\cellcolor{verylight}  Near-Valid WTA (Def. \ref{def:nearvalid})&  $\checkmark$ & $\min(1,\norm{C(\Input)}_1)$ & $1$ & $1$ 	 \\
\hline
\cellcolor{verylight}  Valid $k$-WTA (Def. \ref{def:kwta})&  $\checkmark$ & $k \ge 2$ & $1$ & $1$ 	 \\
\hline
Reset (Def. \ref{def:reset})& -- & -- & $0$ & $0$ 	 \\
\hline
\end{tabu}
\caption{\label{tab:configs}  Summary of good configuration types (Definition \ref{def:good}), from which we show rapid convergence to a valid WTA configuration. We refer to the configuration types shaded in gray as \emph{active configurations} (Definition \ref{def:active}).}
\normalsize
\end{table}

With the stability  bound of Corollary \ref{cor:stability} in place, it remains to prove that $\Netg$ converges quickly to a valid WTA configuration. We do this in two main steps:
\begin{enumerate}
\item In this Section, we define three additional good configuration types and show that all other network configurations converge to a good configuration within just two time steps with high probability.
\item In Sections \ref{sec:good2WTA} and \ref{sec:convergence} we show that, in turn, each of these good configurations rapidly converges to a valid WTA configuration with constant probability. 
\end{enumerate}
A high level illustration of our proof is shown in Figure \ref{fig:2inhProof}.
\begin{figure}[h]
\centering
\includegraphics[width=0.6\textwidth]{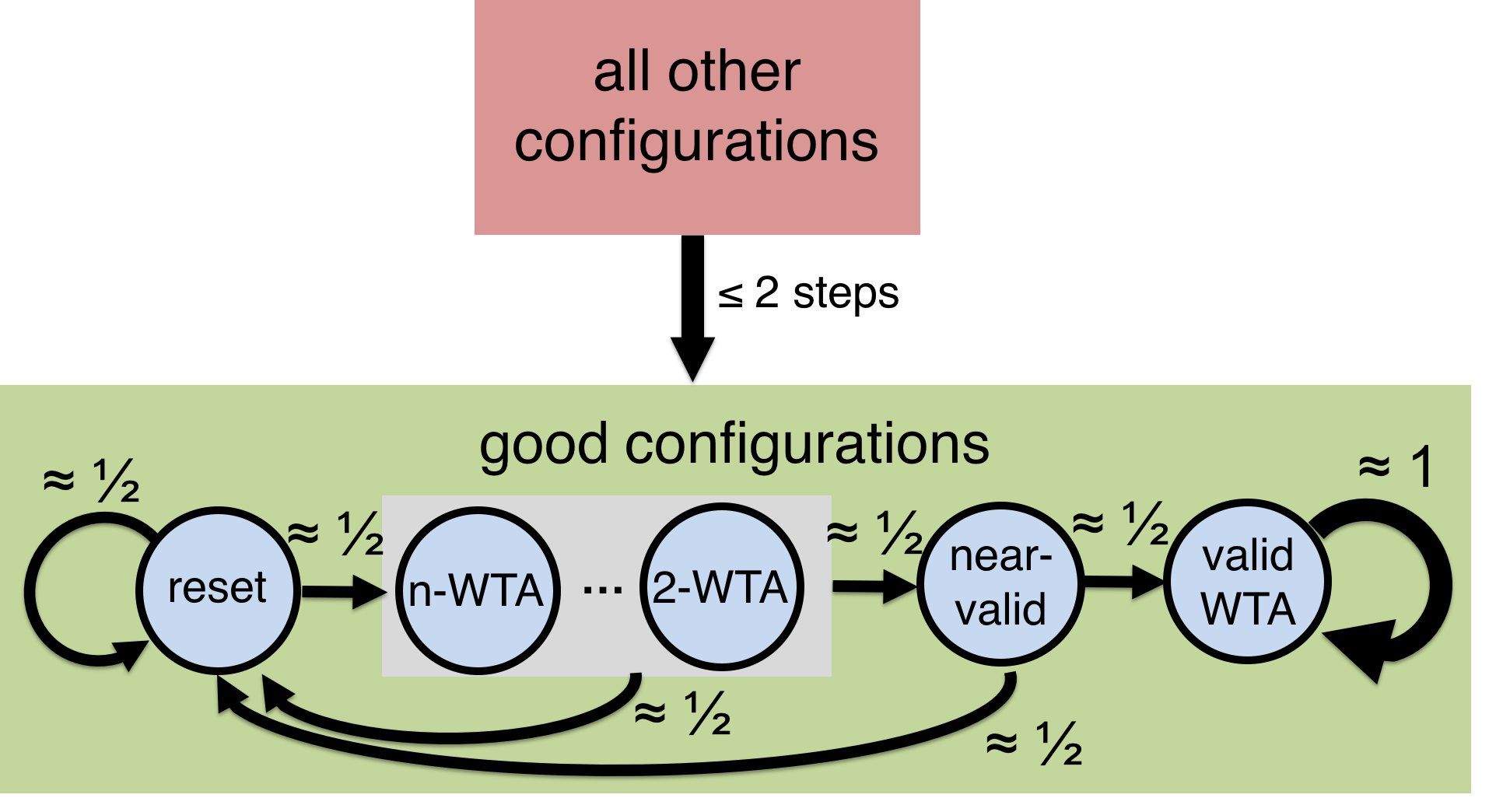}
\caption{A high level illustration of our proof that $\Netg$ solves the WTA problem. We show that all configurations converge to the set of good configurations. Once in a good configuration, $\Netg$ converges to a valid WTA state very rapidly, with constant probability. In the illustration, arrow size corresponds to relative probability.}
\label{fig:2inhProof}
\end{figure}


The first class of good configurations are \emph{valid $k$-WTA configurations}. In such configurations the network behaves as expected before  convergence. Multiple outputs corresponding to firing inputs fire and the inhibitors $a_s$ and $a_c$ fire, driving convergence towards a valid WTA configuration.

\begin{definition}[Valid $k$-WTA Configuration]\label{def:kwta}
For any $k \ge 2$, a \emph{valid $k$-WTA configuration} of $\Netg$ is a configuration $C$ with $C(y_i) \le C(x_i)$ for all $i \in \{1,...,n\}$, $\norm{C(\Output)}_1 = k$  and $C(a_s) = C(a_c) = 1$.
\end{definition}

We next define a class of \emph{near-valid WTA configurations}, each of which is a small perturbation of a valid WTA configuration, with correct output but incorrect inhibitor behavior. We will show in Section \ref{sec:convergence} that the network rapidly converges to a near-valid WTA configuration from any configuration. In turn, it transitions with probability $\approx 1/2$ from a near-valid WTA configuration to a valid WTA configuration.
 
\begin{definition}[Near-Valid WTA Configuration] \label{def:nearvalid}
A \emph{near-valid WTA configuration} of $\Netg$ is a configuration $C$ in which $C(Y)$ is a valid WTA output configuration for $C(X)$ (Definition \ref{def:output}) but $C(a_s) = C(a_c)  = 1$. 
\end{definition}

Finally, we define the class of \emph{reset configurations} with $C(a_s) = C(a_c) = 0$. Since there is no inhibition in such a configuration, each output corresponding to a firing input will fire with probability $\ge 1/2$ at the next time. With probability $\approx 1/2$, the network will transition to either a valid $k$-WTA, a near-valid WTA, or a valid WTA configuration within three steps (Lemma \ref{lem:reset2good}). 

\begin{definition}[Reset Configuration] \label{def:reset}
A \emph{reset configuration} of $\Netg$ is a configuration $C$ with $C(a_s) = C(a_c)  = 0$. 
\end{definition}

\begin{definition}[Good Configuration]\label{def:good} A \emph{good configuration} is any configuration that is either a valid WTA configuration (Definition \ref{def:valid}),  a valid $k$-WTA configuration (Definition \ref{def:kwta}), a reset configuration (Definition \ref{def:reset}), or a near-valid WTA configuration (Definition \ref{def:nearvalid}).
\end{definition}

For conciseness, we also give a name to the good configurations excluding reset  configurations:
\begin{definition}[Active Configuration]\label{def:active} An \emph{active configuration} is any configuration that is either a valid WTA configuration (Definition \ref{def:valid}),  a valid $k$-WTA configuration (Definition \ref{def:kwta}), or a near-valid WTA configuration (Definition \ref{def:nearvalid}).
\end{definition}

We first give a simple lemma building on Lemma \ref{lem:stability1}, which characterizes the network's behavior when a single inhibitor and at least one output corresponding to a firing input fire at time $t$:
\begin{tcolorbox}
\begin{lemma}\label{lem:single} Assume the input execution $\alpha_\Input$ of $\Netg$ has $X^t$ fixed for all $t$.
For any  time $t$ and any  configuration $C$ with $C(X) = X^t$, ($C(a_s) = 1$ and $C(a_c) = 0$) or ($C(a_s) = 0$  and $C(a_c) = 1$), $\norm{C(\Output)}_1 \ge 1$, and $C(y_i) \le C(x_i)$ for all $i$:
$$\Pr[\All^{t+1}\text{ is a valid WTA or valid $k$-WTA configuration } | \All^t = C] \ge 1-(n+2) e^{-\gamma/2}.$$
\end{lemma}
\end{tcolorbox}
\begin{proof}

If $\norm{C(\Output)}_1 = 1$, then $C$ is a valid WTA configuration.
Thus, by  Lemma \ref{lem:stability}, conditioned in $\All^t = C$, $\All^{t+1}$ is a valid WTA configuration with probability $\ge 1-(n+2)e^{-\gamma/2}$.

 If $\norm{C(\Output)}_1 \ge 2$, then by  Lemma \ref{lem:reset} conclusion (3),
 $$\Pr[a_s^{t+1} = a_c^{t+1} = 1 | \All^t = C] \ge 1-2e^{-\gamma/2}.$$
 Further, 
 by Lemma \ref{lem:stability1}, $\Pr[Y^{t+1} = Y^t | \All^t  = C]\ge 1-ne^{-\gamma/2}$, which gives that  $\All^{t+1}$ is a valid $k$-WTA configuration since, by conditioned on $\All^t = C$, $\norm{\Output^t}_1  \ge  2$ and $y_i^t \le x_i^t$ for all $i$. Thus, by a union bound, conditioned on $\All^t = C$, $\All^{t+1}$ is a valid $k$-WTA configuration with probability $\ge  1-(n+2) e^{-\gamma/2}$, giving  the lemma.

\end{proof}

\begin{tcolorbox}
\begin{theorem}[Convergence to a Good Configuration]\label{thm:good}
Assume that  the input execution $\alpha_\Input$ of $\Netg$ has $\Input^t$ fixed for all $t$. For any time $t$ and configuration $C$ with $C(X) = X^t$,
$$\Pr[\text{at least one of }\{\All^{t+1},\All^{t+2}\}\text{ is a good configuration } | \All^t = C] \ge 1-2(n+1)e^{-\gamma/2}.$$
\end{theorem}
\end{tcolorbox}

\begin{proof}
Let $\mathcal{E}$ be the event that at least one of $\{\All^{t+1},\All^{t+2}\}$ is a good configuration.
Let  $\mathcal{E}_1$ be the event that, $\All^t = C$ and for all $i$, $y_i^{t+1} \le x_i^{t+1}$. 
By  Corollary \ref{cor:mono} and the fact that $\Input^t$ is fixed
\begin{align}\label{lem:e111}
\Pr[\mathcal{E}_1 | \All^t = C] \ge  1-ne^{-\gamma/2}.
\end{align}
We now give a simple case analysis, considering all values of $a_c^{t+1}$ and $a_s^{t+1}$.

\medskip
\spara{Case  1: $a_c^{t+1} = a_s^{t+1} = 0$.}
\medskip

\noindent
In this case, $\All^{t+1}$ is a reset configuration (Definition \ref{def:reset}). So we have:
\begin{align}\label{case11}
\Pr \left [\mathcal{E} \big | \mathcal{E}_1,a_c^{t+1} = a_s^{t+1} = 0 \right  ] = 1.
\end{align}

\medskip
\spara{Case  2: $a_c^{t+1} = a_s^{t+1} = 1$.}
\medskip

\noindent
If $\norm{Y^{t+1}}_1 = 0$, then $\All^{t+2}$ is a reset  configuration with probability $\ge 1-2e^{-\gamma/2}$  by Lemma \ref{lem:reset} conclusion (1).
 If $\norm{Y^{t+1} }_1 = 1$ and $\mathcal{E}_1$  occurs, then $\All^{t+1}$ is a near-valid WTA configuration (Definition \ref{def:nearvalid}).
If $\norm{Y^{t+1}}_1 \ge 2$ and $\mathcal{E}_1$  occurs, then $\All^{t+1}$ is a valid $k$-WTA  configuration (Definition \ref{def:kwta}). Thus:
\begin{align}\label{case21}
\Pr \left [\mathcal{E} \big | \mathcal{E}_1,a_c^{t+1} = a_s^{t+1} = 1 \right  ] \ge 1-2e^{-\gamma/2}.
\end{align}

\medskip
\spara{Case  3: ($a_s^{t+1} = 1$ and $a_c^{t+1} = 0$) or ($a_s^{t+1} = 0$ and $a_c^{t+1} = 1$).}
\medskip

\noindent
Let $\mathcal{E}_2$ be the event that ($a_s^{t+1} = 1$ and $a_c^{t+1} = 0$) or ($a_s^{t+1} = 0$ and $a_c^{t+1} = 1$).
Again, if $\norm{Y^{t+1}}_1 = 0$, then $\All^{t+2}$ is a reset  configuration with probability $\ge 1-2e^{-\gamma/2}$  by Lemma \ref{lem:reset} conclusion (1). 

If $\mathcal{E}_2$ and $\mathcal{E}_1$ occur and $\norm{Y^{t+1}}_1 \ge 1$, we can apply Lemma \ref{lem:single}, giving that $\All^{t+2}$ is either a valid WTA or valid $k$-WTA configuration with probability $\ge  1-(n+2)e^{-\gamma/2}$.
We thus have:
\begin{align}\label{case31}
\Pr[\mathcal{E} | \mathcal{E}_1,\mathcal{E}_2] \ge 1-(n+2)e^{-\gamma/2}.
\end{align}

\medskip
\spara{Completing the lemma.}
\medskip

Combining \eqref{case11}, \eqref{case21}, and \eqref{case31}, by the law of total probability, $\Pr[\mathcal{E} | \mathcal{E}_1] \ge 1-(n+2)e^{-\gamma/2}$. We can then use that  $\Pr[\mathcal{E}_1 | \All^t  = C] \ge  1-ne^{-\gamma/2}$ by \eqref{lem:e111} to give:
\begin{align*}
\Pr[\mathcal{E} | \All^t = C] &\ge \Pr[\mathcal{E}_1 | \All^t = C] \cdot \Pr[\mathcal{E} | \mathcal{E}_1, \All^t = C]\\
&= \Pr[\mathcal{E}_1 | \All^t = C] \cdot \Pr[\mathcal{E} | \mathcal{E}_1]\tag{Since, by definition, $\mathcal{E}_1$ implies $\All^t = C$.}\\
&\ge 1-2(n+1)e^{-\gamma/2},
\end{align*}
completing the lemma.
\end{proof}

\subsection{Transition Lemmas for Good Configurations}\label{sec:good2WTA}

We now give a set of lemmas that characterize the transitions of $\Netg$ when starting from a good configuration. 
We first  show that  a near-valid WTA configuration transitions with  probability $\approx 1/2$ to the adjacent valid WTA configuration (i.e., the configuration with the same output behavior, but correct inhibitor behavior).

\begin{tcolorbox}
\begin{lemma}[From Near-Valid to Valid Configurations]\label{lem:nearvalid}
Assume that  the input execution $\alpha_\Input$ of $\Netg$ has $\Input^t$ fixed for all $t$. For any time $t$ and near-valid WTA configuration $C$ with $C(X) = \Input^t$,
\begin{align*}
\Pr [\All^{t+1}\text{ is a valid WTA configuration }|\All^t  = C] \ge 1/2-(n+2)e^{-\gamma/2}.
\end{align*}
\end{lemma} 
\end{tcolorbox}
\begin{proof}
We give a proof similar to that of Lemmas \ref{lem:stability1} and \ref{lem:stability}. We consider two cases:

\medskip
\spara{Case 1: $\norm{C(\Output)}_1 = 0$.}
\medskip

In this case, for $C$ to be a near-valid WTA  configuration, according to Definition \ref{def:nearvalid}, we must also have $\norm{C(\Input)}_1 = 0$. So by Corollary \ref{cor:mono}, 
$$\Pr[\norm{\Output^{t+1}}_1 = 0 | \All^t = C] \ge 1-ne^{-\gamma/2}.$$
Additionally, by Lemma \ref{lem:reset} conclusion (1), since $\norm{C(\Output)}_1 = 0$, $\Pr[a_s^{t+1} = a_c^{t+1} = 0|\All^t = C] \ge  1-2e^{-\gamma/2}$. By a union bound,
\begin{align*}
\Pr[\norm{\All^{t+1}}_1 = 0 | \All^t = C] \ge 1-(n+2)e^{-\gamma/2}.
\end{align*}
By Definition \ref{def:valid}, since $\norm{\Input^{t+1}}_1 = 0$, if $\norm{\All^{t+1}}_1 = 0$, then $\All^{t+1}$ is a valid WTA configuration. So, conditioned on $\All^t  = C$,
$\All^{t+1}$ is a valid WTA configuration with probability $\ge 1-(n+2)e^{-\gamma/2}$, giving the lemma in this case.

\medskip
\spara{Case 2: $\norm{C(\Output)}_1 = 1$.}
\medskip

In this case, for some $i$, $C(y_i) =1$ and $C(y_j) = 0$ for all $j \neq i$. Further, we must have $C(x_i) = 1$ and $C(a_s) = C(a_c) =  1$ by the requirements of Definition \ref{def:nearvalid}.
Define the event $\mathcal{E}_{1}$ by:
$$\mathcal{E}_{1} \eqdef \left  (y_i^{t+1} =  1,\ a_s^{t+1} = 1,\ a_c^{t+1} = 0\ \text{ and }y_j^{t+1} = 0\text{ for all }j \neq i  \right ).$$
By Lemma \ref{lem:convergence}, $\Pr[y_i^{t+1}  = 1 | \All^t = C] = 1/2$ and $\Pr[y_i^t \le y_i^{t+1}\text{ for all }i|\All^t = C] \ge  1-ne^{-\gamma/2}$.
By Lemma \ref{lem:reset} conclusion (2), since $\norm{C(\Output)}_1 = 1$, $\Pr[a_s^{t+1} = 1\text{ and }a_c^{t+1} = 0 | \All^t = C] \ge 1-2e^{-\gamma/2}$.
By a union bound this gives, 
\begin{align*}
\Pr[\mathcal{E}_{1} | \All^t = C] \ge 1/2 -(n+2)e^{-\gamma/2}.
\end{align*}
If $\mathcal{E}_{1}$ occurs, then $\All^{t+1}$ is a valid WTA configuration, giving the  lemma in this case.
\end{proof}

We next show that  a reset configuration transitions to some active configuration (Defintion \ref{def:active}) with probability $\approx 1/2$.

\begin{tcolorbox}
\begin{lemma}[From Reset to Active Configurations]\label{lem:reset2good}
Assume  the input execution $\alpha_X$ of $\Netg$ has $X^t$ fixed for all $t$. For any time $t$ and reset configuration $C$ with $C(X) = C^t$,
\begin{align*}
\Pr[\text{at least one of }\{\All^{t+1},\All^{t+2}, \All^{t+3} \}\text{ is active }|\All^t = C] \ge 1/2-3(n+2)e^{-\gamma/2}.
\end{align*}
\end{lemma}
\end{tcolorbox}
\begin{proof}
Let $\mathcal{E}$ be the event that at least one of $\{\All^{t+1},\All^{t+2}, \All^{t+3} \}$ is an active configuration.
We consider two cases:

\medskip
\spara{Case 1: $\norm{X^t}_1 = 0.$}
\medskip

Let $\mathcal{E}_{20}$ be the event that $\norm{\All^{t+2}}_1 = 0$. That is, that no neurons fire  at time $t+2$.
By  Lemma \ref{lem:quiet}, since $\norm{X^t}_1 = 0$, $\Pr[\mathcal{E}_{20} | \All^t = C] \ge 1-2(n+1)e^{-\gamma/2}$.

$\mathcal{E}_{20}$ requires that no neurons fire in $\All^{t+2}$, which makes this a valid WTA configuration since $\norm{X^t}_ 1 = 0$ and the input is fixed for all $t$. Thus, $\Pr[\mathcal{E}|\All^t = C] \ge \Pr[\mathcal{E}_{20} | \All^t = C] \ge  1-2(n+1)e^{-\gamma/2}$, completing the lemma in this case.

\medskip
\spara{Case  2: $\norm{X^t}_1 \ge 1.$}
\medskip

Let $\mathcal{E}_{11}$ be the event that $\All^t  = C$ and $\norm{Y^{t+1}} \ge 1$ (i.e., at least  one output fires at time $t+1$) and $y_i^{t+1} \le x_i^{t+1}$ for all $i$.
For any $y_i$ with $x_i^t = 1$, we have:
\begin{align*}
\pot(y_i,t+1) &= w(x_i,y_i) x_i^t + w(y_i,y_i) y_i^t + w(a_s,y_i) a_s^t+ w(a_c,y_i) a_c^t - b(y_i)\\
&\ge 3\gamma + 0 + 0 + 0 - 3\gamma = 0.
\end{align*}
So, by Lemma  \ref{lem:gap}, each output $y_i$ with $x_i^t = 1$ fires with probability at least $1/2$ at time $t+1$. Each with $x_i^t = 0$ does not fire with probability $\ge 1-e^{-\gamma/2}$ by Lemma \ref{lem:mono}. Since, by assumption, $\norm{X^t}_1 \ge 1$, by a union bound, with probability  $\ge 1/2-ne^{-\gamma/2}$, at least one output with $x_i^t = 1$ fires in $\All^{t+1}$, and no outputs with $x_i^t = 0$ fire. That is, 
\begin{align}\label{eq:beforeSubE11}
\Pr[\mathcal{E}_{11} | \All^t = C] \ge 1/2-ne^{-\gamma/2}. 
\end{align}
We now proceed with a case analysis on the inhibitor behavior at  time $t+1$:

\medskip
\spara{Sub-case  1: $a_s^{t+1} = a_c^{t+1} = 1$.}
\medskip

In this case, assuming $\mathcal{E}_{11}$ occurs, $\All^{t+1}$ is either a valid $k$-WTA configuration for some $k$ or a near-valid WTA configuration. We thus have:
\begin{align}\label{eq:sub11}
\Pr\left [\mathcal{E} \big| \mathcal{E}_{11},a_s^{t+1} = a_c^{t+1} = 1\right ] = 1.
\end{align}

\medskip
\spara{Sub-case 2: $a_s^{t+1} = a_c^{t+1} = 0$.}
\medskip

Let $\mathcal{E}_{21}$ be the event that $\norm{Y^{t+2}}_1 \ge 1$, $y_i^{t+2} \le x_i^{t+2}$  for all  $i$, and $a_s^{t+2} = 1$.
Assuming $\mathcal{E}_{11}$ occurs and $a_s^{t+1} = a_c^{t+1} = 0$, any output with $y_i^{t+1} = 1$ has
\begin{align*}
\pot(y_i,t+2) &= w(x_i,y_i) x_i^{t+1} + w(y_i,y_i) y_j^{t+1} + w(a_s,y_i) a_s^{t+1}+ w(a_c,y_i) a_c^{t+1} - b(y_i) \\
&= 3\gamma + 2\gamma + 0 + 0 - 3\gamma > 0.
\end{align*}
Thus, any such output has $y_i^{t+2} = 1$  with probability  $\ge 1-e^{-\gamma/2}$  by  Lemma \ref{lem:gap}.  Combined with Corollary \ref{lem:mono}, with probability $\ge 1-ne^{-\gamma/2}$, at least one output fires in $\All^{t+2}$, and no outputs with $x_i^{t+1} = 0$ fire. Further by  Lemma \ref{lem:reset} conclusions (2) and (3), since $\mathcal{E}_{11}$ requires that $\norm{Y^{t+1}} \ge 1$, with probability $\ge 1-2e^{-\gamma/2}$, $a_s^{t+2} = 1$. Thus,
\begin{align}\label{eq:e21givee11}
\Pr \left [\mathcal{E}_{21} \big| \mathcal{E}_{11}, a_s^{t+1} = a_c^{t+1} = 0\right] \ge 1-(n+2)e^{-\gamma/2}.
\end{align}

Assume that both $\mathcal{E}_{11}$ and $\mathcal{E}_{21}$ occur. $\mathcal{E}_{21}$ requires that $a_s^{t+2} = 1$ and $\norm{Y^{t+2}}_1 \ge  1$.
If we also have $a_c^{t+2}  = 1$, then $\All^{t+2}$ is a near-valid WTA or valid $k$-WTA configuration. If $a_c^{t+2} = 0$ and $\norm{Y^{t+2}}_1 = 1$, then $\All^{t+2}$ is a valid WTA configuration. If $a_c^{t+2} = 0$ and $\norm{Y^{t+2}}_1 \ge 2$, then by Lemma \ref{lem:reset} conclusion (3), with probability  $\ge 1-2e^{-\gamma/2}$, $a_s^{t+3} = a_c^{t+3} =  1$. Further, we can apply Lemma \ref{lem:stability1}, giving that $Y^{t+3}  = Y^{t+2}$ with probability $\ge 1-ne^{-\gamma/2}$. This ensures that  $\All^{t+3}$ is a valid $k$-WTA  configuration. So we have by a  union  bound:
\begin{align*}
\Pr  \left  [\mathcal{E} \big  | \mathcal{E}_{11}, \mathcal{E}_{21}, a_s^{t+1} = a_c^{t+1} = 0 \right  ] \ge 1-(n+2)e^{-\gamma/2}.
\end{align*}
Combined with \eqref{eq:e21givee11} the above gives:
\begin{align}\label{eq:sub12}
\Pr  \left  [\mathcal{E} \big  | \mathcal{E}_{11},  a_s^{t+1} = a_c^{t+1} = 0 \right  ] &\ge \Pr\left[\mathcal{E}_{21} \big  | \mathcal{E}_{11},  a_s^{t+1} = a_c^{t+1} = 0\right] \cdot \Pr  \left  [\mathcal{E} \big  | \mathcal{E}_{11}, \mathcal{E}_{21}, a_s^{t+1} = a_c^{t+1} = 0 \right  ] \nonumber\\
&\ge \left (1-(n+2)e^{-\gamma/2} \right ) \cdot \left (1-(n+2)e^{-\gamma/2} \right  )\nonumber\\
&\ge 1-2(n+2)e^{-\gamma/2}.
\end{align}

\medskip
\spara{Sub-case 3: ($a_s^{t+1} = 1$ and $a_c^{t+1} = 0$) or ($a_s^{t+1} = 0$ and $a_c^{t+1} = 1$).}
\medskip

Let $\mathcal{E}_{13}$ denote the event that  ($a_s^{t+1} = 1$ and $a_c^{t+1} = 0$) or ($a_s^{t+1} = 0$ and $a_c^{t+1} = 1$). Assuming $\mathcal{E}_{13}$ and $\mathcal{E}_{11}$, we can apply Lemma \ref{lem:single}, which gives that $\All^{t+2}$ is either a valid WTA or valid $k$-WTA configuration with probability $\ge  1 - (n+2)e^{-\gamma/2}$. I.e., 
\begin{align}\label{eq:sub13}
\Pr  \left  [\mathcal{E} \big  | \mathcal{E}_{11},  \mathcal{E}_{13} \right  ] \ge 1 - (n+2)e^{-\gamma/2}.
\end{align}

\medskip
\spara{Completing Case 2.}
\medskip

Overall,  combining \eqref{eq:sub11}, \eqref{eq:sub12}, and \eqref{eq:sub13}, by the law of total probability $\Pr[\mathcal{E} | \mathcal{E}_{11}] \ge 1-2(n+2)e^{-\gamma/2}$. Recalling that by \eqref{eq:beforeSubE11} $\Pr[\mathcal{E}_{11} | \All^t  = C] \ge 1/2-ne^{-\gamma/2}$:
\begin{align*}
\Pr  [\mathcal{E} | \All^t = C] &\ge \Pr[\mathcal{E}_{11} | \All^t = C] \cdot \Pr[\mathcal{E} | \mathcal{E}_{11}, \All^t = C]\\
&=\Pr[\mathcal{E}_{11} | \All^t = C] \cdot \Pr[\mathcal{E} | \mathcal{E}_{11}]\tag{Since, by assumption, $\mathcal{E}_{11}$ implies $\All^t  = C$.}\\
&\ge 1/2-3(n+2)e^{-\gamma/2}
\end{align*}
which completes the lemma in this case.
\end{proof}

We next show that, in each time step, with high probability, the number of firing outputs $k$ does not increase. Further, with probability $\approx 1/2$, $k$ is reduced by a factor of $1/2$. This ensures rapid convergence towards having just a single firing output (i.e., a near-valid WTA state). While there is some chance that the convergence will `overshoot' the target and zero outputs will fire at some time step, we show that the probability of this event is upper bounded by the probability of  the desired event -- i.e., reaching a near-valid WTA configuration.
\begin{tcolorbox}
\begin{lemma}[Progress from $k$-WTA Configurations]\label{lem:progress}
Assume the input execution $\alpha_X$ of $\Netg$ has $X^t$ fixed for all $t$. For any time $t$ and any valid $k$-WTA configuration $C$ with $C(X) = X^t$,
\begin{enumerate}
\item  Letting $\mathcal{E}$ be the event that $\All^{t+1}$ is either a near-valid WTA configuration, a valid $k$-WTA configuration with $\norm{\Output^{t+1}}_1 \le \norm{\Output^t}_1$, or has $\norm{Y^{t+1}}_1 = 0$. $$\Pr[\mathcal{E} | \All^t  = C] \ge 1-(n+2)e^{-\gamma/2}.$$
\item $\Pr \left[\norm{\Output^{t+1}}_1 \le \left\lceil \frac{\norm{\Output^t}_1}{2}\right \rceil \big| \All^t = C\right] \ge 1/2-(n+2)e^{-\gamma/2}$.
\item $\Pr[\norm{Y^{t+1}}_1 = 0 | \All^t = C] - (n+2)e^{-\gamma/2} \le \Pr[\All^{t+1} \text{ is a near-valid } | \All^t = C].$
\end{enumerate}
\end{lemma}
\end{tcolorbox}
\begin{proof}
Since $C$  is a valid $k$-WTA configuration, conditioned on $\All^t = C$, we have $\norm{Y^t}_1 \ge 2$ and $a_s^t  = a_c^t  =  1$. Event $\mathcal{E}$ in (1) above occurs if and only if $a_s^{t+1} = a_c^{t+1} = 1$, and $y_i^{t+1} \le y_i^t$ for all $i$.

By Lemma \ref{lem:reset} conclusion (3),  since $\norm{Y^t}_1 \ge 2$, both inhibitors  remain firing at time $t+1$ with high probability. That is, $$\Pr[a_s^{t+1} = a_c^{t+1} = 1|\All^t=C] \ge 1-2e^{-\gamma/2}.$$
Further, by Lemma \ref{lem:convergence},  $\Pr[y_i^{t+1}\le y_i^t\text{ for all } i|\All^t=C] \ge 1-ne^{-\gamma/2}$. By a union bound, this gives that  $\Pr[\mathcal{E}|\All^t=C] \ge 1-(n+2)e^{-\gamma/2}$, giving (1).

Let $\bar{Y} \subseteq  \Output$ denote the set of all output neurons with $y_i^t = 1$.
Let $k = \norm{Y^{t}}_1 = \norm{\bar{Y}^{t}}_1$ and $k' =  \norm{\bar{Y}^{t+1}}_1$. By Lemma \ref{lem:convergence} properties (2) and (3), conditioned on $\All^t=C$, $k'$ is distributed according to the binomial distribution $B(k,1/2)$. That is, it is the number of successes in $k$ independent trials each with success probability $1/2$. Since $B(k,1/2)$ is symmetric with mean $k/2$, its median is upper bounded by $\lceil k/2 \rceil$. Thus, $\Pr\left[k' \le \lceil k/2 \rceil | \All^t=C \right] \ge 1/2$. This gives by a union bound, 
$$\Pr\left[k' \le \lceil k/2 \rceil\text{ and } \mathcal{E} | \All^t=C \right] \ge 1/2-(n+2)e^{-\gamma/2}.$$ Note that if $\mathcal{E}$ holds, then $k' =  \norm{\bar{Y}^{t+1}}_1 =  \norm{{Y}^{t+1}}_1$, which thus gives conclusion (2).

Finally, we have $\Pr[k' = 1 | \All^t=C] = k \cdot \frac{1}{2^k}$ and $\Pr[k' = 0 | \All^t=C] = \frac{1}{2^k}$ and so 
$$\Pr[k' = 0 | \All^t=C]  \le \Pr[k' = 1 | \All^t=C].$$
Let $\mathcal{E}_1$ be the  event that $\All^{t+1}$ is a near-valid WTA configuration.
Assuming $\mathcal{E}$ occurs, $\All^{t+1}$ is a near-valid WTA configuration if and only if $k' = 1$. That is, 
$$\Pr[\mathcal{E}_1 | \mathcal{E},\All^t=C] = \Pr[k' = 1 | \mathcal{E},\All^t=C] = \Pr[k' = 1 | \All^t=C].$$ The second equality follows since $\mathcal{E}$ and $k'$ are independent conditioned on $\All^t$. $k'$ only depends on the firing of $y \in \bar{Y}$ while $\mathcal{E}$ only depends on the firing of $u \in \{a_s,a_c\} \cup (Y \setminus \bar{Y})$. Using the above:
\begin{align*}
\Pr[\mathcal{E}_1 | \All^t=C] &\ge  \Pr[\mathcal{E} | \All^t=C]\cdot \Pr[\mathcal{E}_1 | \mathcal{E},\All^t=C]\\
&\ge \left  ( 1- (n+2)e^{-\gamma/2} \right) \cdot \Pr[k'=1 | \All^t=C]\\
&\ge \Pr[k'=0 | \All^t=C] - (n+2)e^{-\gamma/2}\\
&\ge \Pr[\norm{Y^{t+1}}_1=0 | \All^t=C] - (n+2)e^{-\gamma/2}
\end{align*}
where the last bound follows since $\norm{Y^{t+1}}_1 \ge \norm{\bar{Y}_{t+1}}_1 = k'$ so $\norm{Y^{t+1}}_1=0$ at least requires $k'=0$. This gives conclusion (3), completing the lemma.
\end{proof}

\subsection{Convergence to WTA}\label{sec:convergence}

We now use the good configuration transition probabilities given in Section \ref{sec:good2WTA}, along with the results of Section \ref{sec:bad2good}, to show that, if sufficiently large $\gamma$, starting from any configuration $\Netg$ converges with probability $\ge 1/18$ to a valid WTA configuration within $O(\log n)$ steps (see Lemma \ref{thm:mainConvergence}).
 The proof is in four main parts, which we outline here. We first  define:
 \begin{definition}[Terminal Configuration]\label{def:term} For $\Netg$, a \emph{terminal configuration} is any configuration $C$  which is either a near-valid WTA configuration or has $\norm{C(\Output)}_1  = 0$ (i.e., no outputs fire).
 \end{definition}
 With this definition we can describe the general proof outline: 
  
\begin{enumerate}
\item \textbf{Monotonicity} (Lemma \ref{lem:mono1}). We prove that, starting from a $k$-WTA configuration, with high probability, $\Netg$ remains in a  $k$-WTA configuration, with the number of firing outputs consistently decreasing until it reaches a terminal configuration. 
\item \textbf{Convergence} (Lemma \ref{lem:kWTAconverge1}). We prove that the number of firing  outputs decreases rapidly. That is, starting from a $k$-WTA configuration, with high probability, a terminal configuration is reached within $O(\log n)$ steps.
\item \textbf{Probability of valid WTA} (Lemma \ref{lem:kWTAconverge}, Corollary \ref{cor:kWTAconverge}).  We show that, starting from a valid $k$-WTA configuration, with constant probability, the terminal configuration reached is in fact a  near-valid WTA configuration. By Lemma \ref{lem:nearvalid}, with constant probability, this configuration transitions to a valid WTA configuration.
\item \textbf{Convergence from any starting configuration} (Theorem \ref{thm:mainConvergence}). We show that, starting in any configuration, with constant probability, the network reaches either a valid WTA configuration or a $k$-WTA configuration in few steps. Combined with our convergence results  for $k$-WTA configurations, this proves fast convergence to a valid WTA  state from any  starting configuration.
\end{enumerate}

We begin with a few definitions which we use to formalize the high level description above.

\begin{definition}[Termination Step]\label{def:termt}
Given any infinite execution $\alpha = C^0C^1....$ let $\term(\alpha,t,\Delta)$ be the minimum value in $\{t+1,...,t+\Delta\}$ for which $C^{\term(\alpha,t,\Delta)}$ is a terminal configuration (Definition \ref{def:term}). If no such time exists let $\term(\alpha,t,\Delta) = t+\Delta$.
\end{definition}

\begin{definition}[Monotonicity Until Termination]\label{def:eterm}
Let $\Emo(t,\Delta)$ be the event that the execution of $\Netg$ is in set of executions $\alpha = C^0C^1...$ satisfying:
\begin{align*}\{\alpha |\text{ for all }t' \in \{t+1,...,\term(\alpha,t,\Delta)\},\ C^{t'}\text{ is a valid}&\text{ $k$-WTA  configuration}\\ &\text{ with }\norm{Y^{t'}}_1 \le \norm{Y^{t'-1}}_1\}.
\end{align*}
\end{definition}

We begin by showing that, starting from any $k$-WTA configuration, with high probability  $\Netg$ behaves monotonically as described above. 
\begin{tcolorbox}
\begin{lemma}[Monotonicity]\label{lem:mono1}
Assume the input execution $\alpha_\Input$ of $\Netg$ has $\Input^t$ fixed for all $t$. For any time $t$, any
valid $k$-WTA configuration $C$ with $C(X) = X^t$, and any $\Delta  \ge 1$,
$$ \Pr[\Emo(t,\Delta) | \All^t = C] \ge 1-\Delta(n+2)e^{-\gamma/2}.$$
\end{lemma}
\end{tcolorbox}
\begin{proof}
Consider any $\Delta \ge 2$.
If $\Emo(t,\Delta-1)$ occurs, then either $\All^{t+(\Delta-1)}$ is  a valid $k$-WTA configuration, or, for  some $t' \in \{t+1,...,t+\Delta-1\}$, $\All^{t'}$ is a terminal configuration. Thus,
by conclusion (1) of Lemma \ref{lem:progress} we have:
\begin{align}\label{eq:induct1}
\Pr[\Emo(t,\Delta)| \Emo(t,\Delta-1), \All^t = C] \ge 1-(n +2)e^{-\gamma/2}.
\end{align}
Using \eqref{eq:induct1} we can show by induction that for any $\Delta \ge 1$, 
\begin{align}\label{eq:inductConclusion}
\Pr[\Emo(t,\Delta) | \All^t  = C] \ge 1-\Delta(n +2)e^{-\gamma/2}.
\end{align}
For any  $\Delta \ge 2$, assume by way  of induction that \eqref{eq:inductConclusion} holds for all $\Delta' < \Delta$. The assumption holds in the base case when $\Delta=1$ again by conclusion (1) of Lemma \ref{lem:progress}, since $C$ is a valid $k$-WTA configuration so $\Pr[\Emo(t,1)| \All^t = C] \ge  1-(n +2)e^{-\gamma/2}$.
Applying \eqref{eq:induct1} and the inductive assumption:
\small
\begin{align*}
\Pr[\Emo(t,\Delta) | \All^t = C ] &\ge \Pr[\Emo(t,\Delta-1) | \All^t = C] \cdot \Pr[\Emo(t,\Delta)| \Emo(t,\Delta-1),\All^t = C]\\
&\ge \left (1-(\Delta-1)(n +2)e^{-\gamma/2} \right ) \cdot \left (1-(n +2)e^{-\gamma/2}\right)\\
&\ge 1-\Delta(n+2)e^{-\gamma/2}.
\end{align*}
\normalsize
which gives \eqref{eq:inductConclusion} for all $\Delta \ge 1$, and so the lemma.
\end{proof}

We next show that, starting from a $k$-WTA configuration, with high probability, $\Netg$ reaches a terminal configuration within $O(\log n)$ steps. This requires showing that for $\Delta = O(\log n)$ with high probability, $\All^{\term(\alpha,t,\Delta)}$ (where $\term(\alpha,t,\Delta)$ is defined in Definition \ref{def:termt})  is actually a terminal configuration. We note that if the network does not reach a terminal configuration within $\Delta$ steps after time $t$, then, by  definition, $\All^{\term(\alpha,t,\Delta)} = \All^{t+\Delta}$, which is some non-terminal configuration.

We first  define a termination event:
\begin{definition}[Termination by $\Delta$]\label{def:econv}
Let $\Ec(t,\Delta)$ be the intersection of $\Emo(t,\Delta)$ (Definition \ref{def:eterm}) and the event that $\All^{\term(\alpha,t,\Delta)}$  is a terminal configuration. 
\end{definition}

\begin{tcolorbox}
\begin{lemma}[Convergence from $k$-WTA Configurations]\label{lem:kWTAconverge1}
Assume the input execution $\alpha_\Input$ of $\Netg$ has $\Input^t$ fixed for all $t$ and that $\gamma \ge 4\ln(n+2)+10$. 
Let $\Delta = \kbound$. For any time $t$ and valid $k$-WTA configuration $C$  with $C(X) = X^t$, 
$$\Pr[\Ec(t,\Delta) | \All^t = C] \ge  1-\Delta(n+2)e^{-\gamma/2}-\frac{1}{7n} $$
\end{lemma}
\end{tcolorbox}
\begin{proof}
Let $\bEc(t,\Delta)$ be the event that $\Emo(t,\Delta)$ occurs but $\Ec(t,\Delta)$ does not.
 For any  $t' \in \{t+1,..., t +\Delta \}$, define the indicator $I_{t'} \in \{0,1\}$ with $I_{t'}= 1$ if and only  if either:
 \begin{itemize}
 \item $\All^{t'-1}$ is a valid $k$-WTA configuration and $\norm{\Output^{t'}}_1 \le \lceil k/2 \rceil$. 
 \item $\All^{t'-1}$ is not a valid $k$-WTA configuration for any $k \ge  2$.
  \end{itemize}  
  $\bEc(t,\Delta)$ requires that each of $\All^t,...,\All^{t+\Delta}$ is a valid $k$-WTA configuration and that 
  $$\norm{Y^t}_1 \ge \norm{Y^{t+1}}_1 \ge  ...\ge \norm{Y^{t+\Delta}}_1 \ge 2.$$
  Otherwise, a terminal configuration with $\norm{Y^{t'}}_1 = 0$  would be reached and $\Ec(t,\Delta)$ would occur.
  
Initially $\norm{\Output^t}_1 \le n$. Since each time ${I}_{t'} = 1$, either $\norm{\Output^t}_1$ is cut in half or a configuration other than a valid $k$-WTA configuration occurs, $\bEc(t)$
   can only occur if $\sum_{t'=t+1}^{t+\Delta} {I}_{t'} < \log_2 n + 1$. Thus we can bound:
\begin{align}\label{ecBound}
\Pr[\bEc(t,\Delta) | \All^t = C] \le \Pr \left [\sum_{t'=t+1}^{t+\Delta} {I}_{t'} < (\log_2 n +1) \big | \All^t = C \right ].
\end{align}

We will show that  this probability  is low since ${I}_{t'} = 1$ with good probability. 
Specifically, if $\All^{t'-1}$ is not a valid $k$-WTA configuration, then ${I}_{t'}  = 1$ deterministically. If $\All^{t'-1}$ is a valid $k$-WTA configuration, then by conclusion (2) of Lemma \ref{lem:progress}, ${I}_{t'} = 1$ with probability  $ \ge 1/2-(n+2)e^{-\gamma/2}$. Overall, we have: $\Pr[{I}_{t'} = 1 | \All^{t'-1} ]\ge 1/2-(n+2)e^{-\gamma/2}$.
In fact, by Lemma \ref{lem:independence}, we can also condition on all past configurations and have: 
\begin{align*}
\Pr[{I}_{t'} = 1 | \All^{t'-1}\All^{t'-2}...\All^t, \All^t = C ] \ge 1/2-(n+2)e^{-\gamma/2}.
\end{align*}

The above bound lets us use Lemma \ref{lem:coupling} to upper bound the probability that $\sum_{t'=t+1}^{t+\Delta} I_{t'}$ is below any value $d$ by the probability that a sum of $\Delta$ independent coin flips, each with success probability  $1/2-(n+2)e^{-\gamma/2}$, is below $d$. 
Specifically,
let  $Z_{t+1},...,Z_{t+\Delta}$ be i.i.d. random variables with $Z_{t'} = 1$ with probability  $1/2-(n+2)e^{-\gamma/2}$ and $Z_{t'} = 0$ otherwise.  Invoking \eqref{ecBound} and Lemma \ref{lem:coupling}, 
\begin{align}\label{ecBound2}
\Pr[\bEc(t,\Delta) | \All^t = C] &\le \Pr \left [\sum_{t'=t+1}^{t+\Delta} {I}_{t'} < (\log_2 n +1) \big | \All^t = C  \right ]\\
&\le \Pr \left [\sum_{t'=t+1}^{t+\Delta} Z_{t'} < (\log_2 n +1)  \right ].
\end{align}
By our assumption that $\gamma \ge 4\ln(n+2) + 10$ and our setting of $\Delta = \kbound \le 14n$:
\begin{align*}
\E \left  [\sum_{t'=t+1}^{t+\Delta} Z_{t'} \right ] = \Delta/2 - \Delta(n+2)e^{-\gamma/2} \ge \Delta/3 = 4(\log_2 n+2).
\end{align*}
By a standard Chernoff bound \cite{mitzenmacher2005probability}, 
\begin{align*}
\Pr \left [ \sum_{t'=t+1}^{t+\Delta} Z_{t'} \le (\log_2 n + 1) \right  ] \le e^{-\frac{(3/4)^2 \cdot 4(\log_2+2)}{2}} \le e^{-(\log_2n + 2)} \le \frac{1}{7n}.
\end{align*}
We thus have, by \eqref{ecBound2}, $\Pr[\bEc(t,\Delta) | \All^t = C] \le \frac{1}{7n}$. Combined with Lemma \ref{lem:mono1} this gives:
\begin{align*}
\Pr[\Ec(t,\Delta) | \All^t = C] &=  \Pr[\Emo(t,\Delta) | \All^t = C] - \Pr[\bEc(t,\Delta) | \All^t = C]\\
&\ge 1- \Delta(n+2)e^{-\gamma/2} - \frac{1}{7n}.
\end{align*}

\end{proof}

We next combine Lemma \ref{lem:kWTAconverge1} with conclusion (3) of Lemma \ref{lem:progress} and Lemma \ref{lem:nearvalid} to show that, starting from a valid $k$-WTA configuration, not only does $\Netg$ reach a terminal configuration quickly, but also, if $\gamma$ is large enough, this terminal configuration is a near-valid WTA configuration with probability $\approx 1/2$.
\begin{tcolorbox}
\begin{lemma}[Constant Probability of Near-Valid WTA, from $k$-WTA Configurations]\label{lem:kWTAconverge}
Assume the input execution $\alpha_\Input$ of $\Netg$ has $\Input^t$ fixed for all $t$ and that $\gamma \ge 4\ln(n+2)+10$. Let $\Delta = \kbound$ and $\mathcal{E}_1(t)$ be the event that there is some $t' \in\{t+1 ,..., t + \Delta\}$, such that $\All^{t'}$ is a near-valid WTA configuration. For any $t$ and valid $k$-WTA configuration $C$  with $C(X) = X^t$,
$$\Pr[\mathcal{E}_1(t)|\All^t = C] \ge \frac{1}{2} - \frac{\Delta+1}{2}(n+2)e^{-\gamma/2} - \frac{1}{14n}.$$
\end{lemma}
\end{tcolorbox}
\begin{proof}
$\mathcal{E}_1(t)$ is equivalent to the event that $\Ec(t,\Delta)$ occurs and $\All^{\term(\alpha,t,\Delta)}$ is a near-valid WTA configuration. Let $\mathcal{E}_{0}(t)$ be the event that $\Ec(t,\Delta)$ occurs and $\norm{\Output^{\term(\alpha,t,\Delta)}}_1 = 0$. $\mathcal{E}_{0}(t)$ and $\mathcal{E}_{1}(t)$ are disjoint with $\mathcal{E}_{0}(t) \cup \mathcal{E}_{1}(t) = \Ec(t,\Delta)$. So by Lemma \ref{lem:kWTAconverge1},
\begin{align}\label{eq:sum2one}
\Pr[\mathcal{E}_{0}(t)|\All^t  = C] + \Pr[\mathcal{E}_{1}(t) | \All^t  = C]  &= \Pr[\Ec(t,\Delta) | \All^t  = C]\nonumber\\
 &\ge 1- \Delta(n+2)e^{-\gamma/2} - \frac{1}{7n}.
\end{align}
We will use conclusion (3) of Lemma \ref{lem:progress} to show that 
\begin{align}\label{eq:half1}
\Pr[\mathcal{E}_{1}(t)|\All^t  = C] \ge \Pr[\mathcal{E}_{0}(t) |\All^t  = C] - (n+2)e^{-\gamma/2},
\end{align}
 which combined with \eqref{eq:sum2one} gives the conclusion of the lemma, that
 \begin{align}\label{eq:half}
\Pr[\mathcal{E}_{1}(t) |\All^t  = C] \ge \frac{1}{2} - \frac{\Delta+1}{2}(n+2)e^{-\gamma/2} - \frac{1}{14n}.
\end{align}

For each $t' \in \{t+1 ,..., t +\Delta\}$, let $\Ec(t,t',\Delta)$ be the event that $\Ec(t,\Delta)$ occurs and $\term(\alpha,t,\Delta) = t'$. Define $\mathcal{E}_{0}(t,t')$ and $\mathcal{E}_{1}(t,t')$ analogously. Let $\mathcal{E}_{2}(t,t')$ be the event that $\All^{t},...,\All^{t'-1}$ are all valid $k$-WTA configurations with $\norm{\Output^t}_1 \ge  ...  \ge \norm{\Output^{t'-1}}_1 \ge 2$.  Let $\mathcal{\bar  {E}}_{2}(t,t')$ be its complement.
\begin{align}\label{0case}
\Pr[\mathcal{E}_{1}(t,t') | \mathcal{\bar  {E}}_{2}(t,t')] = \Pr[\mathcal{E}_{0}(t,t') | \mathcal{\bar  {E}}_{2}(t,t')] =0
\end{align}
since both $\mathcal{E}_{1}(t,t')$ and $\mathcal{E}_{0}(t,t')$ require $\Emo(t,\Delta)$  to hold, which requires $\mathcal{E}_{2}(t,t')$ to hold if $\term(\alpha,t,\Delta) = t'$. Further, by conclusion (3) of Lemma \ref{lem:progress}, since $\mathcal{E}_{2}(t,t')$  requires that $\All^{t'-1}$ is a valid $k$-WTA  configuration,
\begin{align}\label{nextcase}
\Pr[\mathcal{E}_{1}(t,t') | \mathcal{  {E}}_{2}(t,t'),\All^t=C] \ge \Pr[\mathcal{E}_{0}(t,t') | \mathcal{  {E}}_{2}(t,t'),\All^t=C] - (n+2)e^{-\gamma/2}
\end{align}
By the law of total probability,  \eqref{0case} and \eqref{nextcase} give 
$$\Pr[\mathcal{E}_{1}(t,t')|\All^t = C] \ge \Pr[\mathcal{E}_{0}(t,t')|\All^t = C] - (n+2)e^{-\gamma/2}.$$
Again by the law of total probability, this gives 
$$\Pr[\mathcal{E}_{1}(t)|\All^t = C] \ge \Pr[\mathcal{E}_{0}(t)|\All^t = C] - (n+2)e^{-\gamma/2},$$
 yielding \eqref{eq:half1} and thus \eqref{eq:half} and the lemma.
 \end{proof}

We next combine Lemma \ref{lem:kWTAconverge} with Lemma \ref{lem:nearvalid}, which shows that any near-valid WTA configuration transitions with probability $\approx 1/2$ to a valid WTA configuration. This gives fast convergence to a valid WTA configuration starting from any  valid $k$-WTA configuration, with probability  $\ge 1/8$.
\begin{tcolorbox}
\begin{corollary}[Constant Probability of Success, from $k$-WTA Configurations]\label{cor:kWTAconverge}
Assume the input execution $\alpha_\Input$ of $\Netg$ has $\Input^t$ fixed for all $t$ and that $\gamma \ge 4\ln(n+2)+10$. Let $\mathcal{E}(t)$ be the event that there is some $t' \in\{t+1 ,..., t + 12\log_2 n + 25\}$, such that $\All^{t'}$ is a valid WTA configuration. For any $t$ and valid $k$-WTA configuration $C$  with $C(X) = X^t$,
$$\Pr[\mathcal{E}(t)|\All^t = C] \ge 1/8.$$
\end{corollary}
\end{tcolorbox}
\begin{proof}
As in Lemma \ref{lem:kWTAconverge}, let $\Delta = \kbound$ and $\mathcal{E}_{1}(t)$ be the event that there is some $t' \in\{t+1 ,..., t + \Delta\}$, such that $\All^{t'}$ is a near-valid WTA configuration. Let $\Ev(t)$ be the event that $\All^{t'+1}$ is a valid WTA  configuration. We have $\mathcal{E}(t) \subseteq \Ev(t)$ (since $t'+1 \in \{t+2,...,t+\Delta  +1\}$ where $t+\Delta  +1 = t+(12\log_2 n + 25)$). 
Thus it suffices to show that $\Pr [\Ev(t)|\All^t = C] \ge 1/8$.

By  Lemmas \ref{lem:kWTAconverge} and  \ref{lem:nearvalid}, for any configuration $C$:
\begin{align*}
\Pr [\Ev(t)|\All^t = C] &\ge \Pr [\mathcal{E}_{1}(t)|\All^t = C]  \cdot \Pr [\Ev(t) | \mathcal{E}_{1}(t),\All^t = C]\\
&\ge \left (\frac{1}{2} - \frac{\Delta+1}{2}(n+2)e^{-\gamma/2} - \frac{1}{14n}\right) \cdot \left (1/2-(n+2)e^{-\gamma/2} \right  )\\
&\ge \frac{1}{4} - \frac{\Delta + 3}{4}(n+2)e^{-\gamma/2} - \frac{1}{28n}.
\end{align*}

We can loosely  bound $\frac{\Delta+3}{4} = \frac{\kbound +3}{4} \le \frac{12(n+2n) + 3n}{4} \le 10n$. Further, by our assumption that $\gamma \ge 4\ln(n+2) + 10$ we have:
\begin{align*}
\Pr [\Ev(t)|\All^t = C] \ge \frac{1}{4} - \frac{10 n(n+2)}{(n+2)^2 \cdot e^5}- \frac{1}{28} \ge  \frac{1}{8}
\end{align*}
which gives the corollary.
\end{proof}

Finally, we show that, starting from \emph{any  configuration}, with constant probability, $\Netg$ converges  to a valid WTA configuration in $O(\log n)$ steps. Our proof combines Theorem \ref{thm:good} and Lemma \ref{lem:reset2good} which show that any configuration transitions to an active configuration (Definition \ref{def:active}) in few steps with constant probability. We then apply  Corollary \ref{cor:kWTAconverge} to show convergence from such a configuration.
\begin{tcolorbox}
\begin{theorem}\label{thm:mainConvergence}
Assume the input execution $\alpha_\Input$ of $\Netg$ has $\Input^t$ fixed for all $t$  and that $\gamma \ge 4\ln(n+2)+10$. Let $\mathcal{E}(t)$ be the event that there is some  $t' \in \{t+1,...,t + (12 \log_2  n+30)\}$ such that $\All^{t'}$ is a valid WTA configuration. For any time $t$ and configuration $C$ with $C(X) = X^t$,
$$\Pr[\mathcal{E}(t) | \All^t = C]\ge  1/18.$$
\end{theorem}
\end{tcolorbox}
\begin{proof}
Let $\Ea(t)$ be the event that $\All^t = C$ and that at least one of $\{\All^{t+1},...,\All^{t+5}\}$ is an active configuration (Definition \ref{def:active}).
Let $\All^{\hat{t}}$ be the first active configuration in this set, or $\All^{\hat{t}} =\All^{ t+5}$ if there is no such configuration. 

By Theorem \ref{thm:good}, conditioned on $\All^t = C$, with probability $\ge 1- 2(n+1) e^{-\gamma/2}$ one of $\{\All^{t+1},\All^{t+2}\}$ is a good configuration. Let  $
\All^{\tilde{t}}$ be the first  good configuration in this set or $\All^{\tilde{t}} = \All^{t+2}$ if neither are good. 
If $\All^{\tilde{t}}$ is also an active configuration then $\Ea(t)$ holds.

If not, then $\All^{\tilde{t}}$ is a reset configuration. Let $C'$ be any reset configuration. By Lemma \ref{lem:reset2good}, conditioned on $\All^{\tilde t} = C'$, with probability $\ge 1/2-3(n+2)e^{-\gamma/2}$ at least one of $\{\All^{\tilde{t}+1},\All^{\tilde{t}+2}, \All^{\tilde{t}+3}\}$, is an active configuration. Thus, overall we  have:

\begin{align}\label{eqeg}
\Pr[\Ea|\All^t =C] &\ge \left (1- 2(n+1) e^{-\gamma/2}\right) \cdot  \left (1/2-3(n+2)e^{-\gamma/2} \right )\nonumber\\&\ge \frac{1}{2}  - 5(n+2)e^{-\gamma/2}.
\end{align}
We can define three disjoint events:
\begin{align*}
\Ean{1}(t) &\eqdef (\All^{\hat t} \text{ is a valid $k$-WTA configuration })\\
\Ean{2}(t) &\eqdef (\All^{\hat t} \text{ is a near-valid WTA configuration })\\
\Ean{3}(t) &\eqdef (\All^{\hat t} \text{ is a valid WTA configuration })
\end{align*}
We have $\Ea(t) = \bigcup_{i=1}^3 \Ean{i}(t)$ and so by the law of total probability:
\begin{align}\label{eq:weightedAvg}
\Pr[\mathcal{E}(t) | \Ea(t)] &= \sum_{i=1}^3 \Pr[\mathcal{E}(t) | \Ean{i}(t)] \cdot \Pr[\Ean{i}(t)|\Ea(t)]\nonumber\\
&\ge \sum_{i=1}^3 \min_{i \in \{1,2,3\}} \Pr[\mathcal{E}(t) | \Ean{i}(t)] \cdot \Pr[\Ean{i}(t)|\Ea(t)] \nonumber\\
&\ge \min_{i \in \{1,2,3\}} \cdot \Pr[\mathcal{E}(t) | \Ean{i}(t)]
\end{align}
where the last bound follows since the $\Ean{i}(t)$ events are disjoint and $\sum_{i=1}^3 \Pr[\Ean{i}(t)| \Ea(t)] = 1$. We now bound this minimum via a case analysis:

\medskip 
\spara{Case 1: $\Pr[\mathcal{E}(t) | \Ean{1}(t)]$}
\medskip

In this case, $\All^{\hat t}$ is a valid $k$-WTA configuration, so applying Corollary \ref{cor:kWTAconverge}, conditioned on $\Ean{1}(t)$, with probability  $1/8$ there is some time $t' \in \{\hat{t}+1,..., \hat{t} + (12\log_2 n + 25)\}$ such that $\All^{t'}$ is a valid WTA  configuration. Note that $\hat{t} \le t+5$ giving $t' \le (12\log_2 n + 30)$. We thus have:
\begin{align}\label{eq:gcase1}
\Pr[\mathcal{E}(t) | \Ean{1}(t)] \ge 1/8.
\end{align}

\medskip 
\spara{Case 2: $\Pr[\mathcal{E}(t) | \Ean{2}(t)]$}
\medskip

In this case,  $\All^{\hat t}$ is a near-valid WTA configuration, so applying Lemma \ref{lem:nearvalid}, conditioned on $\Ean{2}(t)$, $\All^{\hat t +1}$ is a valid WTA configuration with probability $\ge 1/2-(n+2)e^{-\gamma/2}$. By our assumption that $\gamma \ge 4\ln(n+2)+10$:
\begin{align}\label{eq:gcase2}
\Pr[\mathcal{E}(t) | \Ean{2}(t)] \ge 1/2-(n+2)e^{-\gamma/2} \ge 1/3.
\end{align}

\medskip 
\spara{Case 3: $\Pr[\mathcal{E}(t) | \Ean{3}(t)]$}
\medskip

If $\Ean{3}(t)$ holds then $\All^{\hat t}$ is a valid WTA configuration by  definition, so trivially
\begin{align}\label{eq:gcase3}
\Pr[\mathcal{E}(t) | \Ean{3}(t)] = 1.
\end{align}

\medskip 
\spara{Completing the theorem:}
\medskip

Combining \eqref{eq:gcase1}, \eqref{eq:gcase2}, \eqref{eq:gcase3}, and \eqref{eq:weightedAvg} we have $\Pr[\mathcal{E}(t) | \Ea(t)] \ge 1/8$. Using \eqref{eqeg} we then have:
\begin{align*}
\Pr[\mathcal{E}(t) | \All^t = C] &\ge \Pr[\Ea(t) | \All^t = C]\cdot \Pr[\mathcal{E}(t) | \Ea(t),
\All^t  = C]\\
&= \Pr[\Ea(t) | \All^t = C]\cdot \Pr[\mathcal{E}(t) | \Ea(t)]\tag{Since $\Ea(t) \subseteq (\All^t =C)$ by definition.}\\
&\ge \frac{1}{8} \cdot \left ( \frac{1}{2}  - 5(n+2)e^{-\gamma/2}\right ) \ge  \frac{1}{18}
\end{align*}
where the last bound follows from our assumption that  $\gamma \ge 4\ln(n+2)+10$.
\end{proof}

\subsection{Completing the Bounds}\label{sec:completing}

Given Theorem \ref{thm:mainConvergence}, it is easy to show that, with $\gamma$ set large enough, $\Netg$ solves the WTA problem (Definitions \ref{high:wta} and \ref{exp:wta}), giving Theorem \ref{thm:high} and \ref{thm:exp}.
We start with the basic WTA problem of Definition \ref{high:wta}. By Theorem \ref{thm:mainConvergence}, starting from any configuration, the network converges to a valid WTA configuration in $O(\log n)$ steps. By applying this analysis in sequence to $O(\log 1/\delta)$ sets of $O(\log n)$ steps, we show that the network converges to a valid WTA state with probability $\ge 1-\delta$ within $O(\log n \cdot \log(1/\delta))$ steps. Further, if $\gamma$ is large enough, by Lemma \ref{cor:stability}, it remains in this state for $t_s$ steps with high probability.
\begin{tcolorbox}
\begin{reptheorem}{thm:high}[Two-Inhibitor WTA]
For $\gamma \ge 4\ln((n+2)t_s/\delta)+10$,
$\Netg$ solves $\wta(n,t_c,t_s,\delta)$ for any  
$
t_c \ge 72(\log_2 n+1)\cdot (\log_2(1/\delta)+1).
$
\end{reptheorem}
\end{tcolorbox}
\begin{proof}
Consider $\Netg$  starting from any initial configuration $\All^0$  and given an infinite input execution $\alpha_\Input$ with $\Input^t$  fixed for all $t$. Let $\Delta = (12 \log_2 n + 30)$ and $r = 6 (\log_2(1/\delta)+1)$.
Let $\mathcal{{E}}$ be the event that there is some time $t \le t_c$ where  $\All^t$  is a valid WTA configuration. 

For any $i \ge 0$, let $\mathcal{{E}}_i$ be the event that there is some time $t \in \{i\Delta+1,...,(i+1)\Delta\}$ where $\All^t$ is a valid WTA configuration.
By Theorem \ref{thm:mainConvergence} and Lemma \ref{lem:independence} we have: 
$$\Pr[\mathcal{ E}_i | \All^{i\Delta}] = \Pr[\mathcal{ E}_i | \All^{i\Delta},\All^{i\Delta-1},...,\All^{1}]  \ge 1/8.$$
Let $Z_0,...,Z_{r-1} \in \{0,1\}$ be independent coin flips, with $\Pr[Z_i = 1] = 1/8$. 
Applying Lemma \ref{lem:coupling}:
\begin{align*}
\Pr [ \mathcal{E} ] = \Pr \left [\bigcap_{i=0}^{r-1} \mathcal{E}_i \right ] \ge \Pr \left [\sum_{i=0}^{r-1} Z_i \ge 1 \right ] = 1-\left (\frac{7}{8}\right)^{r}.
\end{align*}
%
%
Using that $r = 6(\log_2(1/\delta)+1)$:
\begin{align*}
\Pr [\mathcal{{E}}] \ge 1- \left (\frac{7}{8}\right  )^{6 (\log_2(1/\delta)+1)} \ge 1-\frac{\delta}{2}.
\end{align*}
Thus, with probability  $\ge 1-\frac{\delta}{2}$ there is some time $t \le r\cdot \Delta \le 72(\log_2 n+1)\cdot (\log_2(1/\delta)+1) \le t_c$ in which $\All^t$ is a valid WTA configuration.
By  Corollary \ref{cor:stability}, if $C$ is a valid WTA configuration then
\begin{align*}
\Pr [\All^t = \All^{t+1} = ... = \All^{t+t_s} | \All^t = C ] \ge 1-t_s (n+2)e^{-\gamma/2} \ge 1-\frac{\delta}{e^5},
\end{align*}
 where the bound holds by  our assumption that $\gamma \ge 4\ln((n+2)t_s/\delta)+10$. We thus have that the network reaches a valid WTA configuration within time $t_c$ and remains in it for time $t_s$ with probability $\ge \left (1-\frac{\delta}{2} \right ) \cdot \left ( 1-\frac{\delta}{e^5} \right ) \ge 1-\delta$, yielding the theorem.
\end{proof}

We conclude by showing with what parameters $\Netg$ solves  the expected-time WTA problem of Definition \ref{exp:wta}.
\begin{tcolorbox}
\begin{reptheorem}{thm:exp}[Two-Inhibitor Expected-Time WTA]
For $\gamma \ge 4\ln((n+2)t_s)+10$,
$\Netg$ solves $\ewta(n,t_c,t_s)$ for any  $
t_c \ge 108(\log_2 n + 3).$
\end{reptheorem}
\end{tcolorbox}
\begin{proof}
Recall that in Definition \ref{exp:wta} we defined the convergence time for any infinite input execution $\alpha_X$ and output execution $\alpha_Y$:
\small
 $$t(\alpha_\Input,t_s,\alpha_\Output) = \min \left \{t : \Output^t\text{ is a valid WTA output configuration for }\Input^t\text{ and } \Output^t=...=\Output^{t+t_s}\right \}.$$
 \normalsize
Define the worst case expected convergence time of $\Netg$  on input $\alpha_X$ by:
\begin{align*}
t_{max}(\alpha_X) = \max_{\All^0} \left (\E_{\alpha_Y \sim \mathcal{D}_\Output(\Netg,\All^0,\alpha_\Input)} t(\alpha_\Input,t_s,\alpha_\Output) \right ).
\end{align*}
To prove the lemma we must prove that for any $\alpha_X$ with $X^t$ fixed for all $t$, $t_{max}(\alpha_X) \le 108(\log_2 n + 3).$ Fixing such an $\alpha_X$, for any starting configuration $\All^0$, let $\underset{\All_0}{\E}$ and $\underset{\All_0}{\Pr}$ denote the expectation and probability of an event taken over executions drawn from $\mathcal{D}(\Netg,\All^0,\alpha_\Input)$.

 Let $\Delta = (12 \log_2 n + 30)$ and let $\mathcal{E}_1$ be the event that there is some $t \in \{1,...,\Delta\}$ where $\All^t$ is a valid WTA configuration. Let $\mathcal{E}_{stab}$ be the event that there is some $t \in \{1,...,\Delta\}$ where $\All^t$ is a valid WTA configuration and additionally, where $\All^t = ... = \All^{t+t_s}$. Let $\mathcal{\bar E}_1$ and $\mathcal{\bar E}_{stab}$ be the complements of these two events.
By Theorem \ref{thm:mainConvergence}, for any initial configuration $\All^0$
\begin{align}\label{eiBound}
\Pr_{\All^0} [\mathcal{ E}_1] \ge 1/8.
\end{align}
Further, by Corollary \ref{cor:stability}, if $C$ is a valid WTA configuration then,
\begin{align}\label{eisBound}
\Pr_{\All^0} [\All^t = \All^{t+1} = ... = \All^{t+t_s} |\All^t = C] \ge 1-t_s (n+2)e^{-\gamma/2} \ge 1-\frac{1}{t_s \cdot e^5}
\end{align}
 where the bound holds since $\gamma \ge 4\ln((n+2)t_s)+10 \ge 2\ln((n+2)t_s^2)+10$ and so $e^{-\gamma/2} \le \frac{1}{(n+2)t_s^2 \cdot e^5}$. Together \eqref{eiBound} and \eqref{eisBound} give that:
\begin{align*}
\Pr_{\All^0}[\mathcal{E}_{stab}] \ge \Pr_{\All^0}[\mathcal{E}_{stab}  | \mathcal{ E}_{1}] \cdot  \Pr_{\All^0}[\mathcal{ E}_1 ]  \ge \frac{1}{8}\cdot \left (1-\frac{1}{t_s \cdot e^5}\right) \ge\frac{1}{8} - \frac{1}{t_s \cdot e^5}.
\end{align*}
We can write:
\begin{align}\label{secretary}
\E_{\All_0} [t(\alpha_\Input,t_s,\alpha_\Output) ]&= \E_{\All_0} [t(\alpha_\Input,t_s,\alpha_\Output) |\mathcal{E}_{stab}] \cdot \Pr_{\All_0}[\mathcal{E}_{stab}]\nonumber\\ &\hspace{2em}+ \E_{\All_0} [t(\alpha_\Input,t_s,\alpha_\Output) |\mathcal{E}_1, \mathcal{\bar E}_{stab}] \cdot \Pr_{\All_0}[\mathcal{E}_1, \mathcal{\bar E}_{stab}]\nonumber \\ &\hspace{2em}+ \E_{\All_0} [t(\alpha_\Input,t_s,\alpha_\Output) |\mathcal{\bar E}_1] \cdot \Pr_{\All_0}[\mathcal{\bar E}_1] 
\end{align}
Conditioned on $\mathcal{E}_{stab}$ (which also requires that $\mathcal{E}_1$ occurs), the network converges within $\Delta$ steps and stabilizes for $t_s$ steps. Thus, we have:
$$\E_{\All_0} [t(\alpha_\Input,t_s,\alpha_\Output) |\mathcal{E}_{stab}] \le \Delta.$$
Conditioned on $\mathcal{E}_1, \mathcal{\bar E}_{stab}$ the network converges, but does not stabilize. We can bound 
\begin{align*}
\E_{\All_0} [t(\alpha_\Input,t_s,\alpha_\Output) |\mathcal{E}_1, \mathcal{\bar E}_{stab}] &\le (\Delta + t_s) + \E_{\All^{\Delta+t_s}} [t(\alpha_\Input,t_s,\alpha_\Output)]\\
&\le \Delta + t_s + t_{max}(\alpha_X).
\end{align*}
Finally, conditioned on $\mathcal{\bar E}_1$, the network does not converge within $\Delta$ steps. We have:
\begin{align*}
\E_{\All_0} [t(\alpha_\Input,t_s,\alpha_\Output) |\mathcal{\bar E}_1] &\le \Delta + \E_{\All^{\Delta}} [t(\alpha_\Input,t_s,\alpha_\Output)]\\
&\le \Delta + t_{max}(\alpha_X).
\end{align*}

We can plug these bounds along with the probability bounds of \eqref{eiBound} and \eqref{eisBound} into \eqref{secretary} to obtain:
\small
\begin{align*}
\E_{\All_0} [t(\alpha_\Input,t_s,\alpha_\Output) ] &\le \Delta \cdot \left (\frac{1}{8} - \frac{1}{t_s \cdot e^5} \right) +  (\Delta + t_{max}(\alpha_X) + t_s) \cdot \frac{1}{t_s \cdot e^5} + (\Delta + t_{max}(\alpha_X)) \cdot \frac{7}{8}
\\& \le \Delta + t_{max} \left (\frac{7}{8} + \frac{1}{t_s \cdot e^5}\right ) + \frac{t_s}{t_s \cdot e^5}\\
&\le \Delta + t_{max}(\alpha_X)\cdot \frac{8}{9} + \frac{1}{e^5}. 
\end{align*}
\normalsize
Since this bound holds for all $\All^0$ we have:
\begin{align*}
t_{max}(\alpha_X) \le \Delta +t_{max}(\alpha_X)\cdot \frac{8}{9} + \frac{1}{e^5}\\
\end{align*}
which gives $t_{max}(\alpha_X) \le 9\Delta + \frac{9}{e^5} \le 9\Delta + 1$. This bound holds for all $\alpha_X$ and so gives the lemma, after recalling that $\Delta = (12 \log_2 n+30)$ so $9\Delta + 1 \le 108(\log_2 n +3)$.
\end{proof}

\section{WTA Lower Bounds}\label{sec:lb}

The simple family of  two-inhibitor networks presented in Section \ref{sec:wta2} gives convergence to a valid WTA output configuration in $ O(\log n \cdot \log(1/\delta))$ steps with probability  $\ge 1-\delta$, as long as the weight scaling parameter $\gamma$   is set large  enough. Specifically, by  Theorem \ref{thm:high}, these networks solve $\wta(n,t_c,t_s,\delta)$ with $t_c = O(\log n \cdot \log(1/\delta))$ and $t_s$ exponentially large in $\gamma$.
 In this section we ask whether this is optimal, considering two questions:
 \begin{enumerate}
 \item  Are there networks that achieve comparable convergence speed with just a single auxiliary  neuron?
 \item Are there networks using two auxiliary neurons that converge faster? 
 \end{enumerate}
 
 We answer these questions for somewhat restricted classes of \emph{simple SNNs} and \emph{symmetric SNNs}, described in Definitions \ref{def:simple} and \ref{def:sym} below. These classes of networks include, in particular, the construction studied in Section \ref{sec:wta2}.
 
 We show that a simple SNN with just  a single auxiliary neuron cannot solve $\wta(n,t_c,t_s,\delta)$  with $t_s = \tilde \Omega \left ( \frac{t_c}{\log n}\right)$. 
 That is, the network cannot effectively converge to a valid WTA output state and remain in this state for a significant time compared to its convergence time.
 Additionally, we show that no symmetric SNN with two auxiliary neurons can improve on the convergence time of the two-inhibitor network $\Netg$ proven in Theorem \ref{thm:high} by more than a $O(\log \log n)$ factor.
 
 We define the restricted network classes we consider in our lower bounds below. 
 \begin{definition}[Simple SNN]\label{def:simple}
 A spiking neural network $\Net = \langle \All,w,b,f\rangle$ is a \emph{simple SNN} if it contains $n$ input neurons labeled $x_1,...,x_n$ and $n$ output neurons labels $y_1,...y_n$ and satisfies:
 \begin{itemize}
 \item $w(x_i,y_j) = 0$ and $w(y_i,y_j) = 0$ for all $j \neq i$. I.e., each input does not connect to outputs, other than its corresponding output, and outputs do not connect to each other. 
 \end{itemize}
 \end{definition}
 Note that in a simple SNN, auxiliary neurons may connect to each other, may  have incoming edges from the input neurons, and may form unrestricted connections with the output neurons. In our two-auxiliary  neuron lower bound we consider a further restricted class of networks:
 \begin{definition}[Symmetric SNN]\label{def:sym}
 A simple SNN $\Net = \langle \All,w,b,f\rangle$ is a \emph{symmetric SNN} if it contains $n$ input neurons labeled $x_1,...,x_n$ and $n$ output neurons labels $y_1,...y_n$ and satisfies:
 \begin{itemize}
 \item For all $u,v \in \Inh$, $w(u,v) = w(v,u) = 0$. I.e., there are no connections between auxiliary neurons.
 \item For all $u \in \Inh$, $w(y_i,u) = w(y_j,u)$. I.e., each auxiliary neuron is affected in the same way by  each output.
 \item For all $i,j$, $w(x_i,y_i) = w(x_j,y_j)$, $w(y_i,y_i) = w(y_j,y_j)$, $w(u,y_i) = w(u,y_j)$ for all $u \in \Inh$, and $b(y_i) = b(y_j)$. I.e., all outputs have identical incoming connections from their corresponding inputs, themselves, and the auxiliary neurons, and have identical biases.
 \end{itemize}
 \end{definition}

 \subsection{Single Auxiliary Neuron Lower Bound}\label{sec:oneInh}
 
 We begin with our lower bound for simple SNNs with just a single auxiliary  neuron.

\begin{theorem}[One Neuron Lower Bound]\label{thm:lb1} For any $n \ge 20$, $\delta \le 1/2$, and any spike probability function $f:\R \rightarrow [0,1]$ (satisfying the restrictions in Section \ref{sec:structure}), there is no simple SNN $\Net = \langle \All,w,b,f\rangle$ with just a single auxiliary neuron that solves $\wta(n,t_c,t_s,\delta)$ with $t_s > 10t_c \cdot \frac{\ln (2t_c)}{\ln  n}$.
\end{theorem}

We note that Theorem  \ref{thm:lb1} applies to simple SNNs (Definition \ref{def:simple}). However, we conjecture that the result holds for any SNN, without structural restrictions. We leave proving a more general bound as an open question, for now.

\spara{Proof Outline.} 
We will prove Theorem \ref{thm:lb1} assuming that the single auxiliary neuron is an inhibitor. It will be easy to see that  a nearly identical  proof goes through when the auxiliary neuron is excitatory. At a high level, our proof shows that the convergence and stability  inhibitors employed by  the two-inhibitor networks described Section \ref{sec:wta2} are necessary. A single auxiliary  neuron is not able to both drive fast convergence to a valid WTA output configuration and to maintain stability once the network is in such a state.
In more detail, 
 our proof breaks into three steps:

\begin{enumerate}
\item
We show in Lemma \ref{lem:sactive} that the network must be relatively `active' if it solves $\wta(n,t_c,t_s,\delta)$. Specifically, if the single inhibitor $a$ does not fire, then any output corresponding to a firing input must fire with probability  $\ge 1 - \delta^{1/t_c}$. 
Otherwise, starting from a configuration in which no outputs fire, the network would take longer than $t_c$ steps to reach a valid WTA output configuration with probability  $1-\delta$.
\item Conversely, we show in Lemma \ref{lem:sconv} that if the inhibitor  $a$ does fire, then any output with a firing input must \emph{cease firing} at the next time with probability $\ge 1- n^{-\frac{1}{10t_c}}$. Otherwise, starting from a configuration in which $\Omega(n)$ outputs fire, with probability $\ge  1/2$, the network would take longer than $t_c$ steps to converge to a valid WTA output configuration (with a single firing output).
\item We combine these results a network with just one inhibitor cannot maintain a valid WTA  output configuration for $t_s = \Omega \left  (t_c \cdot \frac{\ln t_c}{\ln n}\right)$ consecutive steps with probability $\ge 1/2$ (i.e., the network cannot achieve sufficient stability). 

Consider any time $t$ in which $\Net$ is in a valid WTA output configuration. If $a$ does not fire at time $t$, then by (1), if there are $\Omega(n)$ active inputs, at least one output which did not fire at time $t$ fires with probability  $\ge 1 - \delta^{\Omega(n/t_c)}$ at time $t+1$.  If $a$ does fire, then by (2) the winning output stops firing at time $t+1$ with probability $\ge 1- n^{-\frac{1}{10t_c}}$. In any case,  if $\delta \le 1/2$, with probability $\ge 1- n^{-\frac{1}{10t_c}}$, the output configuration changes, and stability is broken. 

Since this relatively  high probability of breaking stability holds at any time in which $\Net$ is in a valid WTA output configuration, it is enough to show that stability cannot be maintained with probability $\ge \epsilon$ for  $\Omega \left  (t_c \cdot \frac{\ln 1/\epsilon}{\ln n}\right)$ steps. By setting $\epsilon = O(t_c)$ and applying a union bound, we can show that, with probability $\ge 1/2$, in the first $t_c$ time steps, $\Net$ never reaches a valid WTA output configuration and remains in this configuration for $t_s$ consecutive steps. Thus, $\Net$ does not solve $\wta(n,t_c,t_s,\delta)$ for $\delta \le 1/2$.
\end{enumerate}

We start by showing that if the inhibitor $a$ does not fire at time $t$, then any output corresponding to a firing input must fire with reasonably high probability at time $t+1$. We in fact prove a more general lemma, for networks containing any  number of inhibitors, since this result  will be useful in our lower bound for two-inhibitor networks presented in Section \ref{sec:lb2}. The proof is not  complicated by adding more inhibitors.

\begin{lemma}[Output Firing Probability When No Inhibitors Are Active]\label{lem:sactive} Let $\Net = \langle \All,w,b,f \rangle$ be any simple SNN which solves $\wta(n,t_c,t_s,\delta)$
and  whose auxiliary neurons $a_1,...,a_m$ are all inhibitory.
For any $i$, any configuration $C$ with $ C(x_i) = 1$ and $C(a_j) = 0$ for all $j$, and any $t$, $$\Pr [y_i^{t+1} =1 | \All^t = C] \ge 1-\delta^{1/t_c}.$$
\end{lemma}
\begin{proof}
Consider any  $i \in \{1,...,n\}$. Assume for the sake of contradiction that there is some configuration  $C$ with $C(x_i)  = 1$, $C(a_j) = 0$ for all $j$, and 
\begin{align*}
\Pr [y_i^{t+1} =1 | \All^t = C] < 1-\delta^{1/t_c}.
\end{align*}
This assumption additionally implies three claims. First we can see that, since $y_i$ is excitatory,
\begin{claim}\label{claim11}
There exists some configuration $C$ with $C(x_i)=1$, $C(a_j) = 0$ for all $j$ and $C(y_i) = 0$ with $\Pr [y_i^{t+1} =1 | \All^t = C] < 1-\delta^{1/t_c}.$
\end{claim}
From Claim \ref{claim11} we can additionally conclude that, since, by Definition \ref{def:simple},  $y_i$ can only have connections from itself, $x_i$, and $a_1,...,a_m$:
\begin{claim}\label{claim12}
For any configuration $C$ with $C(x_i)=1$, $C(a_j) = 0$ for all $j$ and $C(y_i) = 0$, $\Pr [y_i^{t+1} =1 | \All^t = C] < 1-\delta^{1/t_c}.$
\end{claim}
Finally, Claim \ref{claim12} implies that, since $a_1,...,a_m$ are all inhibitors, 
\begin{claim}\label{claim13}
For any configuration with $C(x_i)=1$ and $C(y_i) = 0$,  $\Pr [y_i^{t+1} =1 | \All^t = C] < 1-\delta^{1/t_c}.$
\end{claim}
Consider input execution
$\alpha_\Input$ with $X^t$ fixed for all $t$ (i.e., $X^t = X^{t'}$ for any  $t,t'$), $x_i^t = 1$, and $x_j^t = 0$ for all $j \neq i$. Let
 $\All^0$ be any initial configuration of $\Net$  consistent with $\alpha_\Input$ (i.e., $\All^0(X) = X^0$) and with $y_i^0 = 0$.
 
 Since $\Net$ solves $\wta(n,t_c,t_s,\delta)$, an infinite output execution drawn from $\mathcal{D}_\Output(\Net,\All^{0},\alpha_\Input)$ must, with probability $\ge 1-\delta$,  reach a valid WTA output configuration (Definition \ref{def:output}) for some $t \le t^c$. In particular, with probability $\ge 1-\delta$, there must be some $t \le t_c$ in which $y_{i}^t  =  1$. So, letting  $\mathcal{E}_0(t)$ be the event that $y_{i}^{t'} = 0$ for all $t' \le t$, we have $\Pr[\mathcal{E}_0(t_c)] \le \delta$.
 
 Additionally, 
  $\Pr[\mathcal{E}_0(0)] = 1$ and  
 using Claim \ref{claim13} above and inducting on $t$, for any $t \ge 1$:
 \begin{align*}
 \Pr[\mathcal{E}_0(t)] &=  \Pr[\mathcal{E}_0(t) | \mathcal{E}_0(t-1)] \cdot  \Pr[\mathcal{E}_0(t-1)]\tag{Since $\mathcal{E}_0(t) \subseteq \mathcal{E}_0(t-1)$.} \\
 &> (1-(1-\delta^{1/t_c}))^t = \delta^{t/t_c}.
 \end{align*}
 We thus have 
 \begin{align*}
  \Pr[\mathcal{E}_0(t_c)] > \delta^{t_c/t_c} = \delta
 \end{align*}
which contradicts that fact that $\Pr[\mathcal{E}_0(t_c)] \le \delta$, giving the lemma.

\end{proof}

We next show that if the inhibitor $a$ does fire at time $t$, then any output must cease firing at time $t+1$ with reasonably large  probability. Formally, we show this statement for roughly  half  of the $n$ outputs. It may be possible that some outputs fire with large probability at time $t+1$ whenever their corresponding inputs fire at time $t$. However, for convergence to occur rapidly give an input execution in which all inputs, this cannot be the case for most outputs.

Again, since it will be useful in the two-inhibitor lower bound proven in Section \ref{sec:lb2}, we show a more general result  which pertains to  networks  with any  number of auxiliary inhibitors.

\begin{lemma}[Output Firing Probability When All Inhibitors Are Active]\label{lem:sconv} Let $\Net = \langle \All,w,b,f\rangle$ be any simple SNN which solves $\wta(n,t_c,t_s,\delta)$ for $n \ge 20$ and $\delta \le 1/2$
and whose auxiliary neurons $a_1,...,a_m$ are all inhibitory. There is some set $\mathcal{S} \subseteq \{1,...,n\}$ with $|\mathcal{S}| \ge \lceil n/2 \rceil$ such that,
for any $i \in \mathcal{S}$, any configuration $C$ with $ C(x_i) = 1$ and $C(a_j)  = 1$ for  all $j$, and any $t$, $$\Pr [y_i^{t+1} =0 | \All^t = C] \ge 1- n^{-\frac{1}{10t_c}}.$$
\end{lemma}

\begin{proof}
Assume for the sake of contradiction that there is some set $\mathcal{R} \subseteq \{1,...,n\}$ with $|\mathcal{R}| = \lfloor n/2 \rfloor + 1$ such that, for each $i \in \mathcal{R}$, there exists some configuration $C$ with $C(x_i) = 1$, $C(a_j) = 1$ for all $j$ and
\begin{align*}
\Pr [y_i^{t+1} =1 | \All^t = C] > n^{-\frac{1}{10t_c}}.
\end{align*}
From this assumption we can deduce three claims. First, since each $y_i$ is excitatory,
\begin{claim}\label{claim21}
For each $i \in \mathcal{R}$ there exists some configuration $C$ with $C(x_i) = 1$, $C(a_j) = 1$ for all $j$ and $C(y_i) = 1$ with 
$
\Pr [y_i^{t+1} =1 | \All^t = C] > n^{-\frac{1}{10t_c}}.$
\end{claim}
Claim \ref{claim21} further implies, since $y_i$ can only have connections from itself, $x_i$, and $a_1,...,a_m$ (see Definition \ref{def:simple}):
\begin{claim}\label{claim22}
For any configuration $C$ with  $C(x_i) = 1$, $C(a_j) = 1$ for all $j$ and $C(y_i) = 1$, $\Pr [y_i^{t+1} =1 | \All^t = C] > n^{-\frac{1}{10t_c}}.$
\end{claim}
Finally, from Claim \ref{claim22} and the fact that $a_1,...,a_m$ are all inhibitors we can conclude:
\begin{claim}\label{claim23}
For any configuration with $C(x_i)=1$ and $C(y_i) = 1$, $\Pr [y_i^{t+1} =1 | \All^t = C] > n^{-\frac{1}{10t_c}}.$
\end{claim}

Let
$\alpha_\Input$ be any infinite input execution with $X^t$ fixed for all $t$ and $X^t(x_{i}) =  1$ for all $i$.
Let
 $\All^0$ be any initial configuration of $\Net$  consistent with $\alpha_\Input$ (i.e., $\All^0(X) = X^0$) and with $Y^0(y_{i}) = 1$ for all $i$.
 
 Since $\Net$ solves $\wta(n,t_c,t_s,\delta)$, an infinite output execution drawn from $\mathcal{D}_\Output(\Net,\All^{0},\alpha_\Input)$ must, with probability $\ge 1-\delta$,  reach a valid WTA output configuration (Definition \ref{def:output}) for some $t \le t^c$. In particular, with probability $\ge 1-\delta$, there must be some $t \le t_c$ in which $y_{i}^t  =  0$ for at most one  $i \in \mathcal{R}$. 
 
 Let  $\Ifo(t,i) \in \{0,1\}$ be an indicator of the event that $y_{i}^{t'} = 1$ for all $t' \le t$.
 If $\Ifo(t_c,i) = 1$ for more than one $i \in \mathcal{R}$, then the network has not reached a valid WTA output state  within time $t_c$.  Thus, since $\Net$ solves $\wta(n,t_c,t_s,\delta)$ we have:
 \begin{align}\label{ifailBound}
 \Pr \left [\sum_{i  \in \mathcal{R}} \Ifo(t_c,i) \ge 2 \right] \le \delta.
 \end{align}

We can bound $ \Pr \left [\sum_{i  \in \mathcal{R}} \Ifo(t_c,i) \ge 2 \right]$ from below using Claim \ref{claim23} above and Lemma \ref{lem:coupling2}.
  Define for all $i$ and $t$ a random variable $Z_{i,t} \in \{0,1\}$ which is set to $1$ independently with probability $n^{-\frac{1}{10t_c}}$. Let $\mathcal{\bar I}(t_c,i)$ be an indicator of the event that $Z_{i,1} = ... = Z_{i,t_c} = 1$. Clearly the $\mathcal{\bar I}(t_c,i)$ variables are independent and $\Pr[\mathcal{\bar I}(t_c,i) = 1] = \left(n^{-\frac{1}{10t_c}}\right)^{t_c} = n^{-\frac{1}{10}}$.
  
By \eqref{ifailBound} and Lemma \ref{lem:coupling2} we thus have:

 \begin{align*}
\delta \ge \Pr \left [\sum_{i  \in \mathcal{R}} \Ifo(t_c,i) \ge 2 \right  ] &\ge   \Pr \left [\sum_{i  \in \mathcal{R}}\mathcal{\bar I}(t_c,i) \ge 2 \right  ]\\
 &= 1 - \left (1-n^{-\frac{1}{10}} \right)^{|\mathcal{R}|} - |\mathcal{R}|\left (1-n^{-\frac{1}{10}}\right)^{|\mathcal{R}|-1}\\
 &\ge 1- n \left (1-n^{-\frac{1}{10}} \right)^{n/4}
 \end{align*}
 where in the last step we bound $|\mathcal{R}|-1 \ge \lfloor n/2 \rfloor \ge n/4$.
Rearranging:
 \begin{align}\label{whatWellCont}
 \frac{1-\delta}{n} \le \left (1-n^{-\frac{1}{10}}\right )^{n/4} \le e^{-\frac{n^{9/10}}{4}}.
 \end{align} 
%
 %
 We can check that whenever $n \ge 20$, $e^{-\frac{n^{9/10}}{4}} < \frac{1}{2n} < \frac{1-\delta}{n}$, by our assumption that $\delta \le 1/2$. This  contradicts \eqref{whatWellCont}, thus giving the lemma.
\end{proof}

We conclude by combining Lemmas \ref{lem:sactive} and \ref{lem:sconv} to prove Theorem \ref{thm:lb1}. We first show a simple auxiliary  lemma which lower bounds the probability that  a valid WTA output configuration at time $t$ remains fixed at time $t+1$ in terms of the network convergence time $t_c$ and stability time $t_s$. The smaller $t_c$ and the larger $t_s$, the larger the lower bound on this probability is.

Note that the theorem only lower bounds the probability  that the output $Y^t$ stays fixed for \emph{some configuration} $C$ where $C(Y)$ is a  valid WTA output. It does not bound this probability  for all such configurations, and in fact  such a bound cannot be  shown for all such configurations. Recall, for example,  that for our two-inhibitor networks in Section \ref{sec:wta2}, both near-valid WTA  configurations (Definition \ref{def:nearvalid}) and valid WTA configurations (Definition \ref{def:valid})  have valid WTA output configurations.  However, the output  of a  near-valid configuration only remains fixed with constant probability  (Lemma \ref{lem:nearvalid}).
 
\begin{lemma}[Single Step Stability Probability]\label{lem:oneStepStab}
If $\Net = \langle \All,w,b,f \rangle$ is a  simple SNN which solves $WTA(n,t_c,t_s,\delta)$ for $\delta \le 1/2$, then there must exist some configuration $C$ with $C(Y)$ a valid WTA output configuration such that:
\begin{align}\label{asumpAbove}
\Pr[\Output^{t+1} =  \Output^{t} | \All^t = C] \ge \frac{1}{(2t_c)^{1/t_s}} \ge 1 - \frac{\ln 2t_c}{t_s}.
\end{align}
\end{lemma}
\begin{proof}
The second inequality simply follows since $e^{-x} \ge 1-x$ for all $x$. So we focus on proving  the first inequality.
Assume for the sake of contradiction that for every configuration $C$ with $C(Y)$ a valid WTA output configuration, 
\begin{align*}
\Pr[ \Output^{t+1}= \Output^{t} | \All^t = C] < \frac{1}{(2t_c)^{1/t_s}} .
\end{align*}
For any $t,i$, let $\Ess(t,i)$ be  the event that $\Output^t = \Output^{t+1} =...=\Output^{t+i}$. Trivially $\Pr[\Ess(t,0)]=1$ for all $t$.
Using the assumption of \eqref{asumpAbove} and induction we  can bound for any $i \ge 1$:
\begin{align*}
\Pr[\Ess(t,i) | \All^t = C] &= \Pr[\Ess(t,i) | \Ess(t,i-1), \All^t = C] \cdot \Pr[ \Ess(t,i-1) | \All^t = C]\tag{Since $\Ess(t,i-1) \subseteq \Ess(t,i)$.}\\
&< \frac{1}{(2t_c)^{1/t_s}}\cdot \Pr[ \Ess(t,i-1) | \All^t = C]\\
&\le \left ( \frac{1}{(2t_c)^{1/t_s}} \right )^{i}.
\end{align*}
Thus, for any  $C$ with $C(Y)$ a valid WTA output configuration,
\begin{align}\label{4421}
\Pr [\Ess(t,t_s) | \All^t = C]  &< \left(\frac{1}{(2t_c)^{1/t_s}}\right)^{t_s}
<  \frac{1}{2t_c}.
\end{align}

 Let $\mathcal{E}$ be the event that $\Net$ reaches a valid WTA output configuration within $t \le t_c$ steps and remains in this output configuration for $t_s$ consecutive steps.
Let  $Z_{1},...,Z_{t_c}$ be i.i.d. random variables with $Z_{t} = 1$ with probability  $\frac{1}{2t_c}$ and $Z_{t} = 0$ otherwise.  Invoking \eqref{4421} and Lemma \ref{lem:coupling}, 
\begin{align}\label{eq:lose}
\Pr[\mathcal{E}] < \Pr \left [\sum_{t=1}^{t_c} Z_{t} \ge 1  \right ] &= 1 - \left (1-\frac{1}{2t_c}\right)^{t_c}\nonumber \\
&< 1/2\nonumber\\
&< 1-\delta
\end{align}
for $\delta \le 1/2$. This contradicts the fact that $\Net$ solves $\wta(n,t_c,t_s,\delta)$, giving the lemma.
\end{proof}

We can now combine Lemmas \ref{lem:sactive} and \ref{lem:sconv} to show that in a single-inhibitor network, for any configuration $C$ with $C(Y)$ a valid WTA output configuration $\Pr[\Output^{t+1} =  \Output^{t} | \All^t = C]$ cannot be too large. This contradicts Lemma \ref{lem:oneStepStab}, giving our lower bound.

\medskip
\begin{proof}[\spara{Proof of Theorem \ref{thm:lb1}}]
$ $
\medskip

Consider any simple SNN $\Net = \langle \All,w,b,f \rangle$ with just a single auxiliary neuron $a$, which solves $\wta(n,t_c,t_s,\delta)$ for $n \ge 20$, $\delta \le 1/2$, $t_s > 10t_c \cdot \frac{\ln (2 t_c)}{\ln n}$. Let $\mathcal{S}$ be  the set  of indices shown to exist in Lemma  \ref{lem:sconv},  which do not  fire  with too high probability  when all  inhibitors in the network fire.

Let
   $\alpha_\Input$ be  the infinite input execution with $X^t$  fixed for all $t$, $X^t(x_i) =1$  for all $i \in \mathcal{S}$ and $X^t(x_i)=0$ for all $i \notin  \mathcal{S}$.
Let
 $ \All^0$ be any initial configuration of $\Net$  consistent with $\alpha_\Input$.  
 
 Since $\Net$ solves $\wta(n,t_c,t_s,\delta)$, an infinite output execution drawn from $\mathcal{D}_\Output(\Net,\All^{0},\alpha_\Input)$ must, with probability  $\ge 1-\delta$,  reach a valid WTA output state (Definition \ref{def:output}) for some $t \le t_c$ and remain in this state for $t_s$ consecutive steps. Let $C$ be any configuration where $C(Y)$ is a valid WTA output configuration for $\alpha_X$. We must have for exactly one $i \in \mathcal{S}$, $C(y_i) = 1$. Let  $\Ef(t)$ be the event that $\Output^{t+1} \neq \Output^t$. We consider two cases:
 
 \medskip
 \spara{Case 1: $C(a) = 0$.}
 \medskip
 
In this case, by Lemma \ref{lem:sactive}, for all $j \in \mathcal{S}$, $y_j$  fires with probability  $\ge 1-\delta^{1/t_c}$ at time $t+1$. So, with probability $\ge 1 -\delta^{\frac{|\mathcal{S}|-1}{t_c}} \ge 1-\delta^{\frac{n}{4t_c}}$ as least one output other than the winner fires at time $t+1$. This gives:
 \begin{align}\label{fail1}
 \Pr[\Ef(t) | \All^t = C] \ge 1-\delta^{\frac{n}{4t_c}} \ge  1-n^{-\frac{1}{4t_c}},
 \end{align}
 where the second inequality follows from the assumption that $n \ge 20$ and $\delta \le 1/2$.
 
  \medskip
 \spara{Case 2: $C(a) = 1$.}
 \medskip
 
In this case, Lemma \ref{lem:sconv} gives that, conditioned on $\All^t = C$, $y_i$ does not fire at time $t+1$ with probability  $\ge 1- n^{-\frac{1}{10t_c}}$. This gives:
  \begin{align}\label{fail2}
 \Pr[\Ef(t) | \All^t = C] \ge 1- n^{-\frac{1}{10t_c}}.
 \end{align}
 So overall, combining \eqref{fail1} and \eqref{fail2} we have for $C$ with $C(Y)$ a valid WTA output configuration,  
 \begin{align*}
 \Pr[\Ef(t) | \All^t = C] \ge 1- n^{-\frac{1}{10t_c}}
 \end{align*}
 and so 
 \begin{align*}
 \Pr[\Output^{t+1} = \Output^t |  \All^t = C] \le n^{-\frac{1}{10t_c}}.
 \end{align*}
 For $t_s > 10t_c \cdot \frac{\ln (2t_c)}{\ln n}$ this gives 
 \begin{align}\label{stabFail}
 \Pr[\Output^{t+1} = \Output^t |  \All^t = C] < n^{-\frac{\ln 2 t_c}{t_s \cdot \ln n}} < \frac{1}{(2t_c)^{1/{t_s}}}.
 \end{align}
This contradicts Lemma \ref{lem:oneStepStab}, giving the theorem.
\end{proof}

\spara{Remark on the Tightness of Theorem \ref{thm:lb1}}.
We note that our proof of  Theorem \ref{thm:lb1} is loose. In bounding $\Pr[\mathcal{E}]$ in \eqref{eq:lose}, we do not consider the time required to converge to a valid WTA output state, or the time spent in this converged state before convergence is broken. We conjecture that if this time were taken into account, it would be possible improve the $\ln (2t_c)$ term, as well as add a dependence on the failure probability $\delta$, giving a lower  bound of $t_s = O \left (\frac{t_c}{\log n \cdot \log(1/\delta)} \right ).$

We note that by simply removing the stability inhibitor from our two-inhibitor network family presented in Section \ref{sec:wta2}, we obtain a family of single inhibitor WTA networks with $t_s = 1$ and $t_c = O \left (\log n \cdot \log(1/\delta)\right )$. This matches the conjectured stronger lower bound above up to a constant factor since for $t_c = O \left (\log n \cdot \log(1/\delta)\right )$ we have $O \left (\frac{t_c}{\log n \cdot \log(1/\delta)} \right ) = O(1)$.

\begin{theorem}[Single-Inhibitor WTA Network]\label{thm:ub1} There exists a simple SNN with a single inhibitory auxiliary  neuron which solves $WTA(n,t_c,t_s,\delta)$ for $t_s = 1$ and  $t_c = O(\log n \cdot \log(1/\delta))$.
\end{theorem}
\begin{proof}[Proof Sketch]
Let $\Netg'$ be identical to $\Netg$ as described in Section \ref{sec:2inhDef}, but with $a_s$ removed from the network and with $w(a_c,y_i) = -2\gamma$ for all $i$. In $\Netg$, $w(a_c,y_i) = w(a_s,y_i)  = -\gamma$ so $w(a_c,y_i) + w(a_s,y_i) = -2\gamma$.
Thus, in $\Netg'$, when $a_c^t = 1$, the firing probabilities of all neurons at time $t+1$ are identical to what they would be in $\Netg$ assuming the same configuration at time $t$ but with $a_c^t = a_s^t = 1$. 

Using this equivalence, it is tedious but easy to check that an analogous result to Theorem \ref{thm:mainConvergence} holds, where a valid WTA configuration is redefined to be any configuration $C$ where $C(Y)$ is a valid WTA output configuration and $C(a_c) = 1$. By a similar to result to Lemma \ref{lem:nearvalid}, with probability $1/2-ne^{-\gamma/2}$, if $\Netg'$ is in such a configuration at time $t$, it is also in a valid WTA output configuration at time $t+1$. Thus, with $\Theta(1)$ probability, a valid WTA output configuration is reached and maintained for $t_s = 1$ steps within $O(\log n)$ steps. We can then use an argument similar to that of Theorem \ref{thm:high} to argue that in $O(\log n \cdot \log(1/\delta))$ steps a valid WTA output configuration is reached and maintained for $1$ step with probability $\ge 1-\delta$.
\end{proof}

  \subsection{Two Auxiliary Neuron Lower Bound}\label{sec:lb2}
 
 We next give a convergence time lower bound for SNNs with two auxiliary neurons, showing that  the rate given by the family of two-inhibitor networks $\Netg$ presented in Section \ref{sec:wta2} is optimal up to a $O(\log \log n)$ factor. To simplify  our argument, we focus the further restricted class of symmetric SNNs described in Definition \ref{def:sym}, proving:
 \begin{theorem}[Two Neuron Lower Bound]\label{thm:lb2} For any $n \ge 341$, $\delta \le 1/2$, and any spike probability function $f:\R \rightarrow [0,1]$ (satisfying the restrictions in Section \ref{sec:structure}), there is no symmetric SNN $\Net = \langle \All,w,b,f\rangle$ using two auxiliary neurons that solves $\wta(n,t_c,t_s,\delta)$ with $t_c \le \frac{\ln n }{30\ln \ln  n}$ and $t_s \ge 32 \ln n \cdot  \ln 2t_c$.
  \end{theorem}
  
  As with Theorem \ref{thm:lb1}, we conjecture that this result holds more generally, for any  SNN with two auxiliary neurons (i.e., without making the assumptions of Definition \ref{def:sym}). In the theorem we require $t_s \ge 32 \ln n \cdot  \ln 2t_c$. However, we conjecture that the result holds even for $t_s = O(1)$.
 
  \spara{Proof Outline.} Our proof of Theorem \ref{thm:lb2} is similar in spirit to that  of Theorem \ref{thm:lb1}. We consider the case when both auxiliary neurons are inhibitors and note that a similar proof applies when one or both of the neurons are excitatory.
  
  With two inhibitors we have to not only consider the cases when neither inhibitor fires (analyzed in Lemma \ref{lem:sactive}) and when both inhibitors fire (analyzed in Lemma \ref{lem:sconv}), but also the cases in which one of the inhibitors fires. Our analysis breaks down as follows:
\begin{enumerate}
\item By Lemma \ref{lem:sactive}, if neither inhibitor fires at time $t$, then at time $t+1$
any output corresponding to a firing input must fire with probability  $\ge 1 - \delta^{1/t_c} = \Omega( \frac{1}{\ln n})$ when $\delta \le 1/2$ and $t_c \le \frac{\ln n }{30\ln \ln  n}$. By Lemma \ref{lem:sconv}, if both inhibitors fire at time $t$, then at time $t+1$, any output with a firing input must \emph{not fire} at time $t+1$ with probability $\ge 1 - n^{\frac{1}{10t_c}} \ge 1-\frac{1}{\ln^3 n}$ when $t_c \le \frac{\ln n }{30\ln \ln  n}$ and $n$ is sufficiently large.
\item In Lemma \ref{lem:02bad} we show that, by the above bounds, if we consider an input with $\Theta(\ln^2 n)$ firing inputs, then if either both inhibitors fire or neither inhibitor fires at time $t$, except with probability $\le \frac{1}{\ln n}$, at time $t+1$, either $0$ or $\ge 2$ outputs will fire, and so $\Net$ will not be in a valid WTA output configuration.
\item In Lemma \ref{lem:1needed} we show that, in order for $\Net$ to stabilize to a valid WTA output configuration for $t_s$ steps, if $Y^t$ is a valid WTA output configuration, then \emph{exactly one inhibitor} must fire with good probability at time $t+1$. 
Otherwise, if neither or both inhibitors fire at time $t+1$, by Lemma \ref{lem:02bad}, $Y^{t+2}$ is unlikely  to be a valid WTA output configuration. So if $t_s$ is large, over $t_s$ steps, stability is likely  to be broken at some point.
\item In Lemma \ref{lem:1orOther} we prove that, since the inhibitors fire independently at time $t+1$ conditioned on $\All^t$, Lemma \ref{lem:1needed} in fact requires that one of the inhibitors fires with high probability at time $t+1$ when $Y^t$ is a valid WTA output configuration and that the other inhibitor is silent with high probability.

By our symmetry assumption (Definition \ref{def:sym}), given a fixed input firing pattern, the probability that an inhibitor fires at time $t+1$ depends only on the number of outputs that fire at time $t$. So we can in fact show a stronger result: 
 one inhibitor (assume without loss of generality $a_1$) fires with high probability at time $t+1$ when $\norm{Y^t}_1 = 1$ while the other (assume without loss of generality $a_2$) remains silent with high probability at time $t+1$ when $\norm{Y^t}_1 = 1$. Since all outputs are excitatory, this implies that $a_1$ fires with high probability whenever $\norm{Y^t}_1 \ge 1$, while $a_2$ does not fire with high probability  whenever $\norm{Y^t}_1 \le 1$. 

This result shows that, in any  two-inhibitor network with fast convergence and a reasonably long stability period $t_s$, the inhibitors must exhibit separate behaviors. 
Neuron $a_1$ fires with high probability at time $t+1$  whenever one output fires at time $t$, maintaining stability of valid WTA configurations. Neuron $a_2$ is silent  with high probability at time $t+1$, except  possibly when   $\norm{Y^t}_1 \ge 2$. We can see this separation, for example, in the behavior of $a_s$ and $a_c$ in the two-inhibitor $\Netg$ networks analyzed in Section \ref{sec:wta2} (see Lemma \ref{lem:reset}). The stability inhibitor, $a_s$, fires with high probability at time $t+1$ when at least one output fires at time $t$ (e.g., when the network is in a valid WTA configuration). The convergence inhibitor, $a_c$, only fires with high probability at time $t+1$ when at least two outputs fire at time $t$.

\item In Corollary \ref{cor:1over2} we use Lemma \ref{lem:1needed} to show that it is unlikely that $a_2$ ever fires at a time in which $a_1$ does not, which will be useful in our eventual case analysis (see Step 7).

\item In Lemma \ref{lem:oneMustGiveStability} we show that, since when $a_1$ fires alone it must maintain stability of a valid WTA output configuration, when $a_1$ fires at time $t$, any output with an active input which fired at time $t$ is likely to continue firing at time $t+1$. Any output that did not fire at time $t$ is unlikely to fire at  time $t+1$. This result shows that $a_1$  \emph{must act as a stability inhibitor}, reflecting the role of $a_s$ in the two-inhibitor network family  presented in Section \ref{sec:wta2}.

In Corollary \ref{cor:oneMustGiveStability} we show that Lemma \ref{lem:oneMustGiveStability} implies that, if $\Output^t$ is \emph{not a valid WTA output configuration} and just $a_1$  fires at time $t$, since the output firing states are maintained with high probability at time $t+1$, $\Output^{t+1}$ is unlikely to be a valid WTA output configuration.
\item We finally prove Theorem \ref{thm:lb2} via a case analysis. By Lemma \ref{lem:02bad}, if $0$ or $2$ inhibitors fire at time $t$, $\Output^{t+1}$ is unlikely to be a valid WTA output configuration (see Step 2 above). By Corollary \ref{cor:oneMustGiveStability},  if $a_1^t  = 1$, $a_2^t = 0$, and $\Output^t$ is not a valid WTA output configuration, then $\Output^{t+1}$ is unlikely to be a valid WTA output configuration (see Step 6 above). Finally, it is unlikely that we ever have $a_1^t  = 0$, $a_2^t = 1$ (see Step 5 above). Thus, convergence to a valid WTA output configuration is relatively unlikely at all times, letting us prove a lower bound on convergence time.
\end{enumerate}

We first  define a \emph{hard input  execution}, based  on Lemma \ref{lem:sconv}, which shows that at  least half of the output neurons must not  fire  with too high probability  when all  inhibitors in the network fire.
Specifically, Lemma \ref{lem:sconv} guarantees the existence of some set $\mathcal{S}$ with $|\mathcal{S}| \ge \lceil n/2 \rceil$ such that for any $i \in \mathcal{S}$, any configuration $C$ of $\Net$ with $C(x_i) = 1$ and $C(a_1) = C(a_2) = 1$, and any  time $t$,
$$\Pr [y_i^{t+1} =0 | \All^t = C] \ge 1- n^{-\frac{1}{10t_c}}.$$
Since in Theorem \ref{thm:lb2} we require $n \ge 341$ we easily  have $\lceil n/2 \rceil \ge \lceil \ln^2 n \rceil$ and so can construct  a hard input execution as follows:
%
 
\begin{definition}[Hard Input Execution]\label{def:hardInput}
Let $\Net = \langle \All,w,b,f \rangle$ be any simple SNN with two auxiliary inhibitory neurons $a_1,a_2$ which solves $\wta(n,t_c,t_s,\delta)$ for $n \ge 341$ and $\delta \le 1/2$. 
Fix any set $\mathcal{R} \subseteq \{1,...,n\}$ with 
$|\mathcal{R} | = \lfloor \ln^2 n \rfloor$ such that for any $i \in \mathcal{R}$, any $C$ with $C(x_i) = 1$ and $C(a_1) = C(a_2) = 1$, and any  time $t$,
$$\Pr [y_i^{t+1} =0 | \All^t = C] \ge 1- n^{-\frac{1}{10t_c}}.$$
Let  $\Xh$ be the input configuration with $\Xh(x_i) = 1$ for  all $i \in \mathcal{R}$, and $\Xh(x_i) =0$  for all $i  \notin \mathcal{R}$.
Let  $\ah$ be the infinite input execution with $\Input^t = \Xh$ for all $t$.
\end{definition}

\begin{lemma}[Valid WTA is Unlikely After Both or Neither Inhibitors Fire]\label{lem:02bad} Let  $\Net = \langle \All,w,b,f\rangle$ be any  simple SNN with two auxiliary inhibitory neurons $a_1,a_2$ which solves $WTA(n,t_c,t_s,\delta)$ for $n \ge 341$, $\delta \le 1/2$, and $t_c \le \frac{\ln  n }{30\ln \ln n}.$ For  any  configuration $C$ with $C(\Input) = \Xh$ (Definition \ref{def:hardInput}) and with $C(a_1)= C(a_2) = 1$ or $C(a_1)= C(a_2) = 0$, 
\begin{align*}
\Pr[\norm{Y^{t+1}}  =1 | \All^t = C]  \le \frac{1}{\ln  n}.
\end{align*}
\end{lemma}
\begin{proof}
We prove  the lemma in two cases depending on the inhibitor behavior.

\medskip
\spara{Case 1: $C(a_1) = C(a_2) = 0$.}
\medskip

In this case, by Lemma \ref{lem:sactive}, since $X^t(x_i) =\Xh(x_i)=1$ for any $i \in \mathcal{R}$ by Definition \ref{def:hardInput},
\begin{align*}
\Pr[y_i^{t+1} = 1 | \All^t = C] \ge 1-\delta^{1/t_c} \ge 1- \frac{1}{2^{\frac{\ln \ln n}{\ln n}}} 
\end{align*}
where the second inequality follows from our assumption that $\delta \le  1/2$ and $t_c \le \frac{\ln  n }{30  \ln \ln n} \le \frac{\ln n}{\ln \ln n}.$ For any $x \in [1,2]$, $(1-1/x) \ge \frac{\log_2 x}{2}$. Since $2^{\frac{\ln \ln n}{\ln n}} \in [1,2]$ this gives:
\begin{align}\label{atLeastlglg}
\Pr[y_i^{t+1} = 1 | \All^t = C]  \ge \frac{ \ln \ln n}{2\ln n}.
\end{align}
Using \eqref{atLeastlglg} we can bound the probability  that $\norm{\Output^{t+1}}_1  \neq 1$ by:
\begin{align*}
 \Pr[\norm{\Output^{t+1}}_1  \neq1  | \All^t = C] &\ge 1 - |\mathcal{R}| \left(1-\frac{\ln \ln n}{2\ln n}\right  )^{|\mathcal{R}| -1}. \end{align*}
 We can check numerically that since $|\mathcal{R}| = \lfloor \ln^2 n \rfloor$ and by assumption $n  \ge  341$, the above can be  lower  bounded to give:
 \begin{align*}
 \Pr[\norm{\Output^{t+1}}_1  \neq 1 | \All^t = C] &\ge 1 - \frac{1}{\ln  n}.
 \end{align*}  
 This gives $\Pr[\norm{\Output^{t+1}}_1  = 1 | \All^t = C] \le  \frac{1}{\ln  n}$ and thus the lemma in this case.

\medskip
\spara{Case 2: $C(a_1) = C(a_2) = 1$.}
\medskip

In this case, by Lemma \ref{lem:sconv} and the assumption that $t_c \le \frac{\ln n}{30 \ln \ln n}$, for any  $i \in \mathcal{R}$:
\begin{align}\label{Sfires}
\Pr[y_i^{t+1} = 1 | \All^t = C] \le n^{-\frac{1}{10t_c}} \le e^{-\frac{30\ln \ln n}{10}} \le \frac{1}{\ln^3 n}.
\end{align}
Using \eqref{Sfires} we can bound:
\begin{align*}
\Pr[\norm{Y^{t+1}}_1  =1 | \All^t = C]
&\le \frac{|\mathcal{R}| }{\ln^3 n}\\
&= \frac{\lfloor \ln^2 n\rfloor}{\ln^3 n}\\
&\le \frac{1}{\ln n},
\end{align*}
which gives the lemma in this case.

\end{proof}

By Lemma \ref{lem:02bad}, if $a_1^t = a_2^t  = 0$ or $a_1^t  = a_2^t = 1$, then $\Output^{t+1}$ is unlikely  to have  $\norm{Y^{t+1}}_1  =1$ and thus is unlikely to be a valid WTA output configuration. Thus,
to stabilize to any valid WTA output configuration $C$, if $\All^t = C$ exactly one of $a_1$ or $a_2$ must fire at  time $t+1$. Otherwise convergence will likely be broken at time $t+2$. To prove this, we first show an auxiliary lemma, very similar to Lemma \ref{lem:oneStepStab}, which bounds the probability of maintaining stability over two time steps in terms of the convergence and stability times $t_c$ and $t_s$. 
\begin{lemma}[Lower Bound on the Two-Step Stability Probability]\label{lem:4lnn}
If $\Net = \langle \All,w,b,f \rangle$ is a  simple SNN which solves $WTA(n,t_c,t_s,\delta)$ for $\delta \le 1/2$, then there exists some configuration $C$ with $C(\Output)$ a valid WTA output configuration such that:
\begin{align}\label{asumpAbove2}
\Pr[\Output^{t+2} = \Output^{t+1} =  \Output^{t} | \All^t = C] \ge 1 - \frac{4 \ln 2t_c}{t_s}.
\end{align}
\end{lemma}
\begin{proof}
Our proof mirrors that of Lemma \ref{lem:oneStepStab}.
Assume for the sake of contradiction that for any $C$ with $C(Y)$ a valid WTA output configuration, 
\begin{align*}
\Pr[\Output^{t+2} = \Output^{t+1}= \Output^{t} | \All^t = C] < 1 - \frac{4 \ln 2t_c}{t_s}.
\end{align*}
For any $t,i$, let $\Ess(t,i)$ be  the event that $\Output^t = \Output^{t+1} =...=\Output^{t+i}$. Trivially, $\Pr[\Ess(t,0)] = 1$ for all $t$.
Using the assumption if \eqref{asumpAbove2} and induction we  can bound for any $i \ge 2$:
\begin{align*}
\Pr[\Ess(t,i) | \All^t = C] &= \Pr[\Ess(t,i) | \Ess(t,t-2), \All^t = C] \cdot \Pr[ \Ess(t,i-2) | \All^t = C]\tag{Since $\Ess(t,i-2) \subseteq \Ess(t,i)$.}\nonumber\\
&< \left ( 1 - \frac{4 \ln 2t_c}{t_s} \right ) \cdot \Pr[ \Ess(t,i-2) | \All^t = C]\nonumber\\
&\le \left ( 1 - \frac{4 \ln 2t_c}{t_s} \right )^{\lfloor i/2\rfloor}.
\end{align*}
Thus, for any $C$ with $C(Y)$ a valid WTA output configuration,
\begin{align}\label{4422}
\Pr [\Ess(t,t_s) | \All^t = C]  &< \left(1 - \frac{4 \cdot \ln 2t_c}{t_s}\right)^{\lfloor t_s/2 \rfloor}\nonumber\\
&< \left(1 - \frac{4 \cdot \ln 2t_c}{t_s}\right)^{t_s/4}\nonumber\\
&<  \frac{1}{2t_c}.
\end{align}

 Let $\mathcal{E}$ be the event that $\Net$ reaches a valid WTA output configuration within $t \le t_c$ steps and remains in this output configuration for $t_s$ consecutive steps.
Let  $Z_{1},...,Z_{t_c}$ be i.i.d. random variables with $Z_{t} = 1$ with probability  $\frac{1}{2t_c}$ and $Z_{t} = 0$ otherwise.  Invoking \eqref{4422} and Lemma \ref{lem:coupling}, 
\begin{align}\label{eq:lose2}
\Pr[\mathcal{E}] < \Pr \left [\sum_{t=1}^{t_c} Z_{t} \ge 1  \right ] &= 1 - \left (1-\frac{1}{2t_c}\right)^{t_c}\nonumber \\
&< 1/2\nonumber\\
&< 1-\delta
\end{align}
for $\delta \le 1/2$. This contradicts the fact that $\Net$ solves $\wta(n,t_c,t_s,\delta)$, giving the lemma.
\end{proof}

We can now use Lemmas \ref{lem:02bad} and \ref{lem:4lnn} to show that, for $\Net$ to stabilize to a valid WTA output configuration for $t_s$ steps with good probability, if $\Output^t$ is a valid WTA output configuration, then, with high probability, exactly one inhibitor must fire at time $t+1$.

\begin{lemma}[A Single Inhibitor is Likely to Fire After a Valid WTA  Configuration]\label{lem:1needed} Let  $\Net = \langle \All,w,b,f \rangle$ be any  simple SNN with two auxiliary inhibitory neurons $A = \{a_1,a_2\}$ which solves $WTA(n,t_c,t_s,\delta)$ for $n \ge 341$, $\delta \le 1/2$, and $t_c \le \frac{\ln  n }{30\ln \ln n}.$ For  any configuration $C$ with $C(\Input) = \Xh$ (Definition \ref{def:hardInput}) and with $C(Y)$ a valid WTA output configuration, 
\begin{align*}
\Pr[\norm{A^{t+1}}_1 =1 | \All^t = C]  \ge 1-\frac{8 \cdot \ln 2t_c}{t_s}.
\end{align*}
\end{lemma}
\begin{proof}
Assume for  the sake of contradiction that there exists $C$ with $C(Y)$ a valid WTA output configuration and
$$\Pr[\norm{A^{t+1}}_1 =1 | \All^t = C]  < 1-\frac{8 \cdot  \ln 2t_c}{t_s}.$$
From this assumption, since in a symmetric SNN (Definition \ref{def:sym}), $w(y_i,u)=w(y_j,u)$  for all $i,j$ and $u \in A$, and since each auxiliary neuron may only have connections from the inputs and outputs (not to the other auxiliary neurons), we can deduce the stronger claim:
\begin{claim}\label{clmL59}
For \emph{every} $C$ with $C(Y)$ a valid WTA output configuration, $$\Pr[\norm{A^{t+1}}_1 =1 | \All^t = C]  < 1-\frac{8 \cdot  \ln 2t_c}{t_s}.$$
\end{claim}

Using Claim \ref{clmL59} we can bound the probability  that $\Net$ remains in a valid WTA output configuration at time $t+2$ if it is in a valid configuration at time $t$. Specifically, for any $C$ with $C(Y)$ a valid WTA output configuration:
\small
\begin{align}\label{firstPassBreak}
\Pr[\Output^{t+2} = \Output^t | \All^t = C] &= \Pr[\Output^{t+2} = \Output^{t} | \norm{A^{t+1}}_1 =1, \All^t = C] \cdot \Pr[\norm{A^{t+1}}_1 =1 | \All^t = C]\nonumber\\
&\hspace{1.4em}+ \Pr[\Output^{t+2} = \Output^{t} | \norm{A^{t+1}}_1 \neq 1, \All^t = C] \cdot \Pr[\norm{A^{t+1}}_1 \neq 1 | \All^t = C]\nonumber\\
&\hspace{-2em}< \left (1-\frac{8 \cdot \ln 2t_c}{t_s}\right ) + \left (\frac{8 \cdot \ln 2t_c}{t_s}\right ) \cdot \Pr[\Output^{t+2} = \Output^{t} | \norm{A^{t+1}}_1 \neq 1, \All^t = C].
\end{align}
\normalsize
Since $C(Y)$ is a valid WTA output configuration, having $\Output^{t+2} = \Output^t$ requires that $\norm{Y^{t+2}}_1 = 1$. Thus by Lemma  \ref{lem:02bad} we can bound:
\begin{align*}
\Pr[\Output^{t+2} = \Output^{t} | \norm{A^{t+1}}_1 \neq 1, \All^t = C] \le \frac{1}{\ln n}.
\end{align*}
Plugging back into \eqref{firstPassBreak}, for any $C$ with $C(Y)$ a valid WTA output configuration:
\begin{align}\label{willBreak}
\Pr[\Output^{t+2} = \Output^t | \All^t = C] &< \left (1-\frac{8 \cdot \ln 2t_c}{t_s}\right ) + \left (\frac{8 \cdot \ln 2t_c}{t_s}\right ) \cdot \left  (\frac{1}{\ln n}  \right )\nonumber\\
&< 1 - \frac{4 \cdot \ln 2t_c}{t_s}
\end{align}
where the last inequality holds easily since we require $n \ge 341$ and so $\ln n \ge  2$. 
Using \eqref{willBreak} we can bound:
$$\Pr[\Output^{t+2} = \Output^{t+1}= \Output^{t} | \All^t = C] \le \Pr[\Output^{t+2} = \Output^{t} | \All^t = C] < 1 - \frac{4 \cdot \ln 2t_c}{t_s}.$$
This contradicts Lemma \ref{lem:4lnn}, 
giving the lemma.
\end{proof}

Since at time $t+1$, conditioned on the configuration at time $t$, $a_1$ and $a_2$ fire independently,
we can in fact  strengthen Lemma \ref{lem:1needed} to show that one of $a_1,a_2$ fires with high probability at time $t+1$  when $\norm{Y^t}_1 \ge  1$ and that the other  remains silent  with high probability when $\norm{Y^t}_1 \le 1$.

\begin{lemma}[Separation of Inhibitor Behaviors]\label{lem:1orOther} Let  $\Net = \langle \All,w,b,f \rangle$ be any  symmetric SNN with two auxiliary inhibitory neurons $A = \{a_1,a_2\}$ which solves $WTA(n,t_c,t_s,\delta)$ for $n \ge 341$, $\delta \le 1/2$, $t_c \le \frac{\ln  n }{30\ln \ln n}$, and $\frac{8 \ln 2t_c}{t_s} \le \frac{1}{4}$. There exists some $i \in \{1,2\}$ such that, for any configuration $C$ with $C(\Input) = \Xh$ (Definition \ref{def:hardInput}),
\begin{enumerate}
\item If $\norm{C(Y)}_1  \ge 1$, then
\begin{align*}
\Pr[a_i^{t+1} =1 | \All^t = C]  \ge 1-\frac{8 \cdot \ln 2t_c}{t_s}.
\end{align*}
\item If $\norm{C(Y)}_1  \le 1$, then for $j \neq i$,
\begin{align*}
\Pr[a_j^{t+1} =1 | \All^t = C]  \le \frac{12 \cdot \ln 2t_c}{t_s}.
\end{align*}
\end{enumerate}
We can assume without loss of generality that $i = 1$.
\end{lemma}

\begin{proof}
We prove the two conclusions in sequence. We first show that there exists some $i$ such that conclusion (1) holds. Fixing $i$, we then show that for $j \neq  i$, conclusion (2) holds.

\medskip
\spara{Conclusion 1:}
\medskip

Assume for the sake of contradiction that for both $i \in \{1,2\}$ there exists  some configuration $C$ with $C(\Input) = \Xh$, with $\norm{C(Y)}_1 \ge 1$, and with
\begin{align*}
\Pr[a_i^{t+1} =1 | \All^t = C]  < 1-\frac{8 \cdot \ln 2t_c}{t_s}.
\end{align*}
From this assumption we can deduce two claims. First, since all $y_i$ are excitatory:
\begin{claim} For both $i \in \{1,2\}$, there exists some configuration $C$ with $C(\Input) = \Xh$, with $\norm{C(Y)}_1 = 1$, and with $\Pr[a_i^{t+1} =1 | \All^t = C]  < 1-\frac{8 \cdot \ln 2t_c}{t_s}$.
\end{claim}
Further, since by Definition \ref{def:sym}, a symmetric SNN has $w(y_i,u) = w(y_j,u)$ for all $i,j$ and $u \in A$, and since each inhibitor can only have incoming connections from the inputs and outputs,
\begin{claim}\label{clm:allBoth} For both $i \in \{1,2\}$, \emph{for every $C$} with $C(\Input) = \Xh$ and with $C(Y)$ a valid WTA output configuration, 
\begin{align*}\Pr[a_i^{t+1} =1 | \All^t = C]  < 1-\frac{8 \cdot \ln 2t_c}{t_s}.
\end{align*}
\end{claim}
Since conditioned on $\All^t$, $a_1^{t+1}$ and $a_2^{t+2}$ are independent,  we can compute, for any $C$ with $C(Y)$ a valid WTA output configuration,
\begin{align}
\Pr[\norm{A^{t+1}}_1 = 1 | \All^t = C] &= \Pr[a_1^{t+1} =1 | \All^t = C]\cdot \Pr[a_2^{t+1} =0 | \All^t = C] \nonumber\\ &\hspace{1.5em}+ \Pr[a_1^{t+1} =0 | \All^t = C]\cdot \Pr[a_2^{t+1} =1 | \All^t = C]\nonumber\\
&= \Pr[a_1^{t+1} =1 | \All^t = C]\cdot \left (1- \Pr[a_2^{t+1} =1 | \All^t = C]\right ) \nonumber\\ &\hspace{1.5em}+ \left (1- \Pr[a_1^{t+1} =1 | \All^t = C]\right )\cdot \Pr[a_2^{t+1} =1 | \All^t = C]\label{opt12}.
\end{align}
If $ \Pr[a_1^{t+1} =1 | \All^t = C] \ge \frac{1}{2}$, then \eqref{opt12} is maximized by setting $\Pr[a_2^{t+1} =1 | \All^t = C] = 0$, giving via Claim \ref{clm:allBoth}:
\begin{align*}
\Pr[\norm{A^{t+1}}_1 = 1 | \All^t = C] &= \Pr[a_1^{t+1} =1 | \All^t = C] < 1-\frac{8 \ln 2t_c}{t_s}.
\end{align*}
If $ \Pr[a_1^{t+1} =1 | \All^t = C] \le \frac{1}{2}$, \eqref{opt12} is maximized by setting $Pr[a_2^{t+1} =1 | \All^t = C] = 1-\frac{8 \ln 2t_c}{t_s}$, giving:
\begin{align}\label{opt122}
\Pr[\norm{A^{t+1}}_1 = 1 | \All^t = C] &< \Pr[a_1^{t+1} =1 | \All^t = C] \cdot \frac{8 \ln 2t_c}{t_s}\nonumber\\ &\hspace{2em}+ \left (1- \Pr[a_1^{t+1} =1 | \All^t = C]  \right) \cdot \left (1-\frac{8 \ln 2t_c}{t_s} \right).
\end{align}
Using our requirement that $\frac{8 \ln 2t_c}{t_s} \le \frac{1}{4}$, \eqref{opt122} is maximized by setting  $\Pr[a_1^{t+1} =1 | \All^t = C] = 0$, again giving:
\begin{align*}
\Pr[\norm{A^{t+1}}_1 = 1 | \All^t = C] &= \Pr[a_1^{t+1} =1 | \All^t = C] < 1-\frac{8 \ln 2t_c}{t_s}.
\end{align*}
In either case, we have a contradiction of Lemma \ref{lem:1needed}, giving the result.

\medskip
\spara{Conclusion 2:}
\medskip

From conclusion (1) proven above, we can
assume without loss of generality that, for any configuration $C$ with $C(\Input) = \Xh$ and with $\norm{C(Y)}_1 \ge 1$,
\begin{align}\label{eq:wlog}
\Pr[a_1^{t+1} =1 | \All^t = C]  \ge 1-\frac{8 \cdot \ln 2t_c}{t_s}.
\end{align}
I.e., we assume that the index $i$ in the lemma statement satisfies $i =1$. To prove conclusion (2)
Assume for the sake  of contradiction that there exists some configuration $C$ with $C(\Input) = \Xh$ and with $\norm{C(Y)} \le 1$ such that $\Pr[a_2^{t+1} =1 | \All^t = C]  > \frac{12 \cdot \ln 2t_c}{t_s}$. Since all outputs are excitatory, from this assumption we can conclude:
\begin{claim}
There exists some configuration $C$ with $C(\Input) = \Xh$ and with $\norm{C(Y)}_1 = 1$ such that $\Pr[a_2^{t+1} =1 | \All^t = C]  > \frac{12 \cdot \ln 2t_c}{t_s}$.
\end{claim}
Further, since, by Definition \ref{def:sym}, a symmetric SNN has $w(y_i,u) = w(y_j,u)$ for all $i,j$ and $u \in A$, and since each inhibitor can only have incoming connections from the inputs and outputs,
\begin{claim}\label{clm:63}
For \emph{every $C$ with $C(\Input) = \Xh$ and with $C(Y)$ a valid WTA output configuration}, $$\Pr[a_2^{t+1} =1 | \All^t = C]  > \frac{12 \cdot \ln 2t_c}{t_s}.$$
\end{claim}
Using Claim \ref{clm:63} we can prove conclusion (2) by considering two cases. 

\medskip
\spara{Case 1: $\Pr[a_2^{t+1} =1 | \All^t = C]  > 1/2$.}
\medskip

In this case, by \eqref{opt12} and \eqref{eq:wlog}, we have for $C$ with $C(Y)$ a valid WTA output configuration:
\begin{align}
\Pr[\norm{A^{t+1}}_1 = 1 | \All^t = C] &=  \Pr[a_1^{t+1} =1 | \All^t = C]\cdot \left (1- \Pr[a_2^{t+1} =1 | \All^t = C]\right ) \nonumber\\ &\hspace{1.5em}+ \left (1- \Pr[a_1^{t+1} =1 | \All^t = C]\right )\cdot \Pr[a_2^{t+1} =1 | \All^t = C]\nonumber\\
&< \frac{1}{2} + \frac{8 \cdot \ln 2t_c}{t_s}\nonumber\\
&< 1 - \frac{8 \cdot \ln 2t_c}{t_s}\label{contra1}
\end{align}
where the last  inequality  follows from our requirement that $\frac{8 \cdot \ln 2t_c}{t_s} \le \frac{1}{4}$. \eqref{contra1} contradicts Lemma \ref{lem:1needed}, giving  the result in this case.

\medskip
\spara{Case 2: $\Pr[a_2^{t+1} =1 | \All^t = C]  \le 1/2$.}
\medskip

In this case, by \eqref{eq:wlog}, Claim \ref{clm:63}, and \eqref{opt12} we have for any $C$ with $C(Y)$ a valid WTA output configuration:
\begin{align}
\Pr[\norm{A^{t+1}}_1 = 1 | \All^t = C] &=  \Pr[a_1^{t+1} =1 | \All^t = C]\cdot \left (1- \Pr[a_2^{t+1} =1 | \All^t = C]\right ) \nonumber\\ &\hspace{1.5em}+ \left (1- \Pr[a_1^{t+1} =1 | \All^t = C]\right )\cdot \Pr[a_2^{t+1} =1 | \All^t = C]\nonumber\\
&< \left (1 - \frac{12 \ln 2 t_c}{t_s} \right ) + \frac{4 \cdot \ln 2t_c}{t_s}\nonumber\\
&< 1 - \frac{8 \cdot \ln 2t_c}{t_s}.\label{contra2}
\end{align}
Again, \eqref{contra2} contradicts Lemma \ref{lem:1needed}, giving the result in this case.
\end{proof}

From Lemma \ref{lem:1orOther} we can easily show that it is unlikely  that  $a_2$ ever fires when $a_1$ does not.
\begin{corollary}[Neuron $a_2$ Rarely Fires Alone]\label{cor:1over2} Let  $\Net = \langle \All,w,b,f \rangle$ be any  symmetric SNN with two auxiliary inhibitory neurons $a_1,a_2$ which solves $WTA(n,t_c,t_s,\delta)$ for $n \ge 341$, $\delta \le 1/2$, $t_c \le \frac{\ln  n }{30\ln \ln n}$, and $\frac{8\ln 2t_c}{t_s} \le \frac{1}{4}$. For any configuration $C$ with $C(\Input) = \Xh$ (Definition \ref{def:hardInput}):
\begin{align*}
\Pr[a_1^{t+1} = 0\text{ and }a_2^{t+1} =1  | \All^t = C]  \le \frac{12 \cdot \ln 2t_c}{t_s}.
\end{align*}
\end{corollary}
\begin{proof}
We prove this result  in two cases.

\medskip
\spara{Case 1: $\norm{C(Y)}_1 \le 1$.}
\medskip

In this case, using Lemma \ref{lem:1orOther} conclusion (2): 
\begin{align*}
\Pr[a_1^{t+1} = 0\text{ and }a_2^{t+1} =1  | \All^t = C]  \le \Pr[a_2^{t+1} =1  | \All^t = C] \le \frac{12 \cdot \ln 2t_c}{t_s}.
\end{align*}

\medskip
\spara{Case 2: $\norm{C(Y)}_1 > 1$.}
\medskip

In this case, applying Lemma \ref{lem:1orOther} conclusion (1):
\begin{align*}
\Pr[a_1^{t+1} = 0\text{ and }a_2^{t+1} =1  | \All^t = C]  &\le \Pr[a_1^{t+1} = 0  | \All^t = C]\\
&= 1- \Pr[a_1^{t+1} = 0  | \All^t = C]\\
&\le \frac{8 \cdot \ln 2t_c}{t_s}.
\end{align*}
\end{proof}


In combination, Lemma \ref{lem:02bad} and Corollary  \ref{cor:1over2} show that, in order for a valid WTA output state to be maintained with high probability, some inhibitor, (which we assume without loss of generality is $a_1$) must fire alone. Using this fact, we can characterize the firing behavior of the outputs when $a_1$ fires alone. Since $a_1$ maintains stability, if it 
 fires alone at time $t$, all outputs maintain the same firing state at time $t+1$ as at time $t$ with high probability.
\begin{lemma}[Neuron $a_1$ Enforces Stability]\label{lem:oneMustGiveStability}
 Let  $\Net = \langle \All,w,b,f \rangle$ be any  symmetric SNN with two auxiliary inhibitory neurons $a_1,a_2$ which solves $WTA(n,t_c,t_s,\delta)$ for $n \ge 341$, $\delta \le 1/2$, $t_c \le \frac{\ln  n }{30\ln \ln n}$, and $\frac{8\ln 2t_c}{t_s} \le \frac{1}{4}$. For any configuration $C$ with $C(\Input) = \Xh$ (Definition \ref{def:hardInput}) and with $C(a_1) = 1$ and $C(a_2) = 0$:
\begin{enumerate}
\item For all $i \in \mathcal{R}$, if $C(y_i) = 1$ then $\Pr[y_i^{t+1} =  1 | \All^t  = C] \ge 1-\frac{16 \ln 2 t_c}{t_s}.$
\item For all $i \in \mathcal{R}$, if $C(y_i) = 0$ then $\Pr[y_i^{t+1} =  0 | \All^t  = C] \ge \left (1-\frac{16 \ln 2 t_c}{t_s}\right)^{1/(|\mathcal{R}|-1)}$
\end{enumerate}
\end{lemma}

\begin{proof}
In a symmetric SNN (Definition \ref{def:sym}) all outputs have identical incoming connections and biases. Thus, to prove the lemma it suffices to show that \emph{there exists a configuration $C$} with $C(\Input) = \Xh$ and with $C(a_1) = 1$ and $C(a_2) = 0$, such that:
\begin{enumerate}
\item There exists $i \in \mathcal{R}$ with $C(y_i) = 1$ and $\Pr[y_i^{t+1} =  1 | \All^t  = C] \ge 1-\frac{16 \ln 2 t_c}{t_s}.$
\item There exists $i \in \mathcal{R}$ with $C(y_i) = 0$ and $\Pr[y_i^{t+1} =  0 | \All^t  = C] \ge \left (1-\frac{16 \ln 2 t_c}{t_s}\right)^{1/(|\mathcal{R}|-1)}.$
\end{enumerate}

We will show that these bounds must  hold for at least  one configuration in order for the network to reach a valid WTA output configuration in $t_c$ steps and remain in this configuration for $t_s$ consecutive steps. Let $\mathcal{E}(t)$ denote the event that $\Output^{t} = \Output^{t+1} =  \Output^{t+2}$.
For any $C$ we have:
\begin{align}
 \Pr[\mathcal{E}(t) | \All^t = C] &= \Pr [\mathcal{E}(t) | \All^{t}=C, a_1^{t+1}=a_2^{t+1} ] \cdot \Pr [a_1^{t+1}=a_2^{t+1} | \All^t = C]\nonumber\\
&\hspace{1em}+ \Pr[\mathcal{E}(t) | \All^t = C, a_1^{t+1} = 0, a_2^{t+1} = 1]  \cdot \Pr [a_1^{t+1}=0,a_2^{t+1} = 1 | \All^t = C]
\nonumber\\
&\hspace{1em}+ \Pr[\mathcal{E}(t) | \All^t = C, a_1^{t+1} = 1, a_2^{t+1} = 0]  \cdot \Pr [a_1^{t+1}=1,a_2^{t+1} = 0 | \All^t = C]\nonumber\\
&\hspace{-3em}\le \Pr [a_1^{t+1}=0,a_2^{t+1} = 1 | \All^t = C] \nonumber\\&\hspace{-2em}+ \max \left (\Pr [\mathcal{E}(t) | \All^{t}=C, a_1^{t+1}=a_2^{t+1} ], \Pr[\mathcal{E}(t) | \Output^t = C, a_1^{t+1} = 1, a_2^{t+1} = 0] \right )
\label{3termsNow}.
\end{align}
We can upper bound the first term of \eqref{3termsNow} by  $\frac{12 \ln 2t_c}{t_s}$ using Corollary  \ref{cor:1over2}. Additionally, using Lemma \ref{lem:02bad}, we can bound $\Pr [\mathcal{E}(t) | \All^{t}=C, a_1^{t+1}=a_2^{t+1} ] \le \frac{1}{\ln n}$.
This gives:
\begin{align}\label{gap2nn} 
\Pr[\mathcal{E}(t) | \All^t = C]  \le \frac{12 \ln 2t_c}{t_s} + \max \left (\frac{1}{\ln n}, \Pr[\mathcal{E}(t) | \All^t = C, a_1^{t+1} = 1, a_2^{t+1} = 0] \right).
\end{align}
Applying Lemma \ref{lem:4lnn} gives that there must exist some $C$ with $C(Y)$ a valid WTA output configuration such that $\Pr[\mathcal{E}(t) | \All^t = C]  \ge 1 -\frac{4 \ln 2t_c}{t_s}$. So combining with \eqref{gap2nn}:
\begin{align*}
\max \left ( \frac{1}{\ln n}, \Pr[\mathcal{E}(t) | \All^t = C, a_1^{t+1} = 1, a_2^{t+1} = 0] \right )&\ge 1 - \frac{16 \ln 2t_c}{t_s}
\end{align*}
By  our assumption that $\frac{8\ln 2t_c}{t_s} \le \frac{1}{4}$  and the fact that $n \ge 341$ we easily have that $ \frac{1}{\ln n} \le \frac{1}{2} \le 1 - \frac{16 \ln 2t_c}{t_s}$, meaning that in fact, there must exist some $C$ with $C(Y)$ a valid WTA output configuration such that: 
\begin{align}\label{eq:3lnn}
\Pr[\mathcal{E}(t) | \All^t = C, a_1^{t+1} = 1, a_2^{t+1} = 0] &\ge 1 - \frac{16 \ln 2t_c}{t_s}.
\end{align}
We can now easily prove the two conclusions of the lemma using similar arguments.

\medskip
\spara{Conclusion 1:}
\medskip

Assume for the sake of contradiction that for every configuration $C$ with $C(\Input) = \Xh$ and with $C(a_1) = 1$ and $C(a_2) = 0$, for all $i \in \mathcal{R}$ with $C(y_i) = 1$,
\begin{align*}
 \Pr[y_i^{t+1} = 1 | \All^t = C] < 1-\frac{16 \ln 2 t_c}{t_s}.
 \end{align*}
 Then, for every $C$ with $C(Y)$ a valid WTA output configuration we have:
\begin{align*}
\Pr[\mathcal{E}(t) | \All^t = C, a_1^{t+1} = 1, a_2^{t+1} = 0] &\le \Pr[\Output^{t+2} = \Output^{t+1} | \All^{t+1} = C, a_1^{t+1} = 1, a_2^{t+1} = 0] \\
&< 1 - \frac{16 \ln 2t_c}{t_s}.
\end{align*}
contradicting \eqref{eq:3lnn} and giving the lemma.

 \medskip
\spara{Conclusion 2:}
\medskip

Assume for the sake of contradiction that for every configuration $C$ with $C(\Input) = \Xh$  and with $C(a_1) = 1$ and $C(a_2) = 0$, for all $i \in \mathcal{R}$ with $C(y_i) = 0$,
\begin{align*}
 \Pr[y_i^{t+1} = 0 | \All^t = C] < \left (1-\frac{16 \ln 2 t_c}{t_s}\right)^{1/(|\mathcal{R}|-1)}.
 \end{align*}
  Then, again for every $C$ with $C(Y)$ a valid WTA output configuration we have:
\begin{align*}
\Pr[\mathcal{E}(t) | \All^t = C, a_1^{t+1} = 1, a_2^{t+1} = 0] &\le \Pr[\Output^{t+2} = \Output^{t+1} | \All^{t+1} = C, a_1^{t+1} = 1, a_2^{t+1} = 0] \\
&< \left (\left (1-\frac{16 \ln 2 t_c}{t_s}\right)^{1/(|\mathcal{R}|-1)}\right )^{|\mathcal{R}|-1}\\
&< 1-\frac{16 \ln 2 t_c}{t_s}
\end{align*}
contradicting \eqref{eq:3lnn} and giving the lemma.
\end{proof}

Lemma \ref{lem:oneMustGiveStability} implies that when just $a_1$ fires, since this leads to stability, it is unlikely to lead to a valid WTA output configuration when starting from an invalid output configuration. That is, $a_1$ does not drive convergence to a valid WTA output configuration.
\begin{corollary}[Neuron $a_1$ Does Not Drive Convergence]\label{cor:oneMustGiveStability}
 Let  $\Net = \langle \All,w,b,f \rangle$ be any  symmetric SNN with two auxiliary inhibitory neurons $a_1,a_2$ which solves $WTA(n,t_c,t_s,\delta)$ for $n \ge 341$, $\delta \le 1/2$, $t_c \le \frac{\ln  n }{30\ln \ln n}$, and $\frac{8\ln 2t_c}{t_s} \le \frac{1}{4}$. For any configuration $C$ with $C(\Input) = \Xh$  (Definition \ref{def:hardInput}), with $C(a_1) = 1$ and $C(a_2) = 0$, and in which $C(Y)$ is \emph{not a valid WTA output configuration}:
\begin{align*}
 \Pr[\Output^{t+1}\text{ is a valid WTA output configuration } | \All^t = C] \le \frac{32 \ln 2t_c}{t_s}.
 \end{align*}
\end{corollary}
\begin{proof}
The proof can be split into two cases, depending on if $C(Y)$ is not a valid WTA output configuration because too many valid outputs are firing or because no outputs are firing. Let $\mathcal{R} \subseteq \{1,...,\}$ be  the set of indices with $\Xh(x_i) = 1$. 

\medskip
\spara{Case 1: $| \{ i \in \mathcal{R} | C(y_i) = 1\} | \ge 2$.}
\medskip

Let $Y(\mathcal{R})$ denote the set of outputs restricted to the indices in $\mathcal{R}$. In order for $\Output^{t+1}$ to be a valid WTA output configuration for $\Xh$, we must have $\norm{Y(\mathcal{R})^t}_1 = 1$.
We can apply Lemma \ref{lem:oneMustGiveStability} conclusion (1) to bound:
\begin{align*}
\Pr [\norm{Y(\mathcal{R})^{t+1}}_1 \ge 2 | \All^t = C] \ge \left (1 - \frac{16 \ln 2t_c}{t_s} \right )^2 \ge 1 - \frac{32 \ln 2t_c}{t_s}
\end{align*}
and thus $ \Pr[\Output^{t+1}\text{ is a valid WTA output configuration } | \All^t = C] \le \frac{32 \ln 2t_c}{t_s}.$

\medskip
\spara{Case 2: $| \{ i \in \mathcal{R} | C(y_i) = 1\} | = 0$.}
\medskip

In this case, we can apply Lemma \ref{lem:oneMustGiveStability} conclusion (2) to bound:
\begin{align*}
\Pr [\norm{Y(\mathcal{R})^{t+1}}_1 = 0 | \All^t = C] \ge \left ( \left (1 - \frac{16 \ln 2t_c}{t_s} \right )^{1/(|\mathcal{R}|-1)} \right )^{|\mathcal{R}|}
\end{align*}
By our assumption that $n \ge 341$, $\frac{|\mathcal{R}|}{|\mathcal{R}|-1} = \frac{\lfloor \ln^2 n \rfloor}{\lfloor \ln^2 n \rfloor-1} \le \frac{34}{33} < 2$. 
So we have:
\begin{align*}
\Pr [\norm{Y(\mathcal{R})^{t+1}}_1 = 0 | \All^t = C] \ge \left (1 - \frac{16 \ln 2t_c}{t_s} \right )^2 \ge 1 - \frac{32 \ln 2t_c}{t_s}.
\end{align*}
Thus, $ \Pr[\Output^{t+1}\text{ is a valid WTA output configuration } | \All^t = C] \le \frac{32 \ln 2t_c}{t_s}$, giving the lemma.

\end{proof}

We can now prove the main two-inhibitor lower bound Theorem \ref{thm:lb2} via a case analysis which combines Lemma \ref{lem:02bad}, Corollary  \ref{cor:1over2}, and Corollary \ref{cor:oneMustGiveStability}.

\medskip
\begin{proof}[\spara{Proof of Theorem \ref{thm:lb2}}]
$ $
\medskip

Assume that $\Net = \langle \All,w,b,f \rangle$ is a  symmetric SNN with two auxiliary inhibitory neurons $a_1,a_2$ which solves $WTA(n,t_c,t_s,\delta)$ for $n \ge 341$, $\delta \le 1/2$, $t_c \le \frac{\ln  n }{30\ln \ln n}$, and $\frac{\ln 2t_c}{t_s} \le \frac{1}{32 \ln n}.$ Assume that the network is given the hard input execution $\ah$ (Definition \ref{def:hardInput}) and 
starts with initial state $\All^0$ in which $Y^0$ is not a valid WTA output configuration and $a_1^0 = a_2^0 = 0$.

Let $\Ef(t)$ be the event that for every $t' \le t$, $\Output^{t'}$ is not a valid WTA output configuration. Let $\mathcal{E}_{01}(t)$ be the event that for every  $t' \le t$, we \emph{do not have} $a_1^{t'} = 0$ and $a_2^{t'} = 1$.

We can bound $\Pr[\Ef(t)] \ge \Pr[\Ef(t),\mathcal{E}_{01}(t)]$ and write, using that $\Ef(t-1) \subseteq \Ef(t)$ and $\mathcal{E}_{01}(t-1) \subseteq \mathcal{E}_{01}(t)$,
\begin{align}
\Pr[\Ef(t),\mathcal{E}_{01}(t)]
&= \Pr[\Ef(t),\mathcal{E}_{01}(t) | \Ef(t-1), \mathcal{E}_{01}(t-1)] \cdot \Pr[\Ef(t-1), \mathcal{E}_{01}(t-1)]\label{eq:efailRecursion}.
\end{align}
Conditioned on $\Ef(t-1)$ and $\mathcal{E}_{01}(t-1)$, $\Output^{t-1}$ is not a valid WTA output configuration and $(a_1^{t-1}, a_2^{t-1})  \neq (0,1)$. Let $\mathcal{C}$ be the set of all configurations  which are not valid WTA output configurations and which have $(C(a_1),C(a_2))  \neq (0,1)$.
We can expand \eqref{eq:efailRecursion} as:
\begin{align}\label{eq:efailRecursion2}
\Pr[\Ef(t),\mathcal{E}_{01}(t)]
&= \Pr[\Ef(t),\mathcal{E}_{01}(t) | \Ef(t-1), \mathcal{E}_{01}(t-1)] \cdot \Pr[\Ef(t-1), \mathcal{E}_{01}(t-1)]\nonumber\\
&= \sum_{C \in\mathcal{C}} \Pr[\Ef(t),\mathcal{E}_{01}(t) | \All^{t-1} = C] \cdot \Pr[\All^{t-1} = C, \Ef(t-1), \mathcal{E}_{01}(t-1)].
\end{align}
Note that by  the law of total probability since $\Ef(t-1)\text{ and }\mathcal{E}_{01}(t-1)$ requires that $C \in \mathcal{C}$:
\begin{align}\label{eq:efailRecursion3}
\sum_{C \in\mathcal{C}} \Pr[\All^{t-1} = C, \Ef(t-1), \mathcal{E}_{01}(t-1)] = \Pr[\Ef(t-1), \mathcal{E}_{01}(t-1)] .
\end{align}
We now bound $\Pr[\Ef(t),\mathcal{E}_{01}(t) | \All^{t-1} = C]$ in \eqref{eq:efailRecursion2} in two cases:

\medskip
\spara{Case 1: $C(a_1) = C(a_2) = 1$ or $C(a_1) = C(a_2) = 0$.}
\medskip

In this case, for any  $C \in \mathcal{C}$ with $C(a_1) = C(a_2)$, by Lemma \ref{lem:02bad}, 
$$\Pr[\Ef(t) | \All^{t-1} = C] \ge 1-\frac{1}{\ln n}.$$
By  Corollary  \ref{cor:1over2}, 
$$\Pr[\mathcal{E}_{01}(t)| \All^{t-1} = C] \ge 1-\frac{12\cdot \ln 2 t_c}{t_s} \ge 1-\frac{1}{\ln n},$$ 
where the second inequality  follows from our requirement that $\frac{\ln 2 t_c}{t_s} \le \frac{1}{32 \ln n}$. Additionally, conditioned on $\All^{t-1}$, $\Ef(t)$ and $\mathcal{E}_{01}(t)$ are independent  since they involve disjoint sets of neurons. Thus we can bound:
\begin{align}\label{lnnCase1}
\Pr[\Ef(t),\mathcal{E}_{01}(t) | \All^{t-1} = C] &\ge \left (1-\frac{1}{\ln n}\right )^2.
\end{align}

\medskip
\spara{Case 2: $C(a_1) = 1$ and $C(a_2) = 0$.}
\medskip

In this case, for any  $C \in \mathcal{C}$, since $C(Y)$ is not a valid WTA output configuration, by Corollary \ref{cor:oneMustGiveStability},
$$\Pr[\Ef(t) | \All^{t-1} = C] \ge 1-\frac{32\cdot \ln 2t_c}{t_s} \ge 1- \frac{1}{\ln n}$$
where the second inequality follows from our requirement that  $\frac{\ln 2t_c}{t_s} \le \frac{1}{32 \ln n}$. As above we also have 
$\Pr[\mathcal{E}_{01}(t)| \All^{t-1} = C] \ge 1-\frac{1}{\ln n}$ and since $\Ef(t)$ and $\mathcal{E}_{01}(t)$ are independent conditioned on $\All^{t-1}$:
\begin{align}\label{lnnCase2}
\Pr[\Ef(t),\mathcal{E}_{01}(t) | \All^{t-1} = C] &\ge \left (1-\frac{1}{\ln n}\right )^2.
\end{align}
Together \eqref{lnnCase1} and \eqref{lnnCase2} and the definition of $\mathcal{C}$ give us:
\begin{claim}\label{expClaim}
For any $C$ which is not a valid WTA configuration and with $(C(a_1),C(a_2))  \neq (0,1)$,
$$\Pr[\Ef(t),\mathcal{E}_{01}(t) | \All^{t-1} = C] \ge \left (1-\frac{1}{\ln n}\right )^2.$$
 \end{claim}

\medskip
\spara{Completing the Theorem.}
\medskip

We conclude by using Claim \ref{expClaim} to lower bound the probability of converging to a valid WTA configuration within $t_c$ steps.
Substituting the bound of Claim \ref{expClaim} into \eqref{eq:efailRecursion2} and \eqref{eq:efailRecursion3} we have:
\begin{align*}
\Pr[\Ef(t),\mathcal{E}_{01}(t)] &\ge \sum_{C \in\mathcal{C}}  \left (1-\frac{1}{\ln n}\right )^2 \cdot \Pr[\All^{t-1} = C, \Ef(t-1), \mathcal{E}_{01}(t-1)]\\
&\ge \left (1-\frac{1}{\ln n}\right )^2 \cdot \Pr[\Ef(t-1), \mathcal{E}_{01}(t-1)]
\end{align*}
Since we choose $\All^0$ with $\Output^0$ not a valid WTA output configuration and $a_1^0 = a_2^0 = 0$, $\Pr[\Ef(0), \mathcal{E}_{01}(0)] = 1$. Thus, via induction, we have for any $t \ge 0$,
\begin{align*}
\Pr[\Ef(t),\mathcal{E}_{01}(t)] \ge  \left (1-\frac{1}{\ln n}\right )^{2t}.
\end{align*}
Plugging in $t_c \le \frac{\ln n}{30 \ln \ln n}$ we have: 
\begin{align*}
\Pr[\Ef(t_c)] \ge \Pr[\Ef(t_c),\mathcal{E}_{01}(t_c)] &\ge  \left (1-\frac{1}{\ln n}\right )^{\frac{\ln n}{15 \ln \ln n}}\\
&\ge \frac{1}{e^{1/15} } > 1/2.
\end{align*}
This contradicts the fact that $\Net$ solves $WTA(n,t_c,t_s,\delta)$ for $\delta \ge 1/2$, giving  the lower bound.
\end{proof}

\section{Faster Convergence With More Inhibitors}\label{sec:multi}

In this section we show how to speed up the two-inhibitor construction of Section \ref{sec:wta2} by using $\alpha > 2$ inhibitors. We give a formal convergence proof for a construction which uses $
\lognc + 1$ inhibitors and converges with constant  probability  in $O(1)$ time (and with probability $\ge 1-\delta$ in $O(\log 1/\delta)$ time). We then describe the high level idea behind two constructions that give a tradeoff between the number of inhibitors used and the convergence time.
\subsection{Use of History  Period}\label{sec:history}
Our $\lognc + 1$-inhibitor construction (as well as the sketched constructions which give an inhibitor-convergence time tradeoff) requires using a history period of $h =2$, as suggested in Section \ref{sec:futureDirs}. At a high level, to achieve fast convergence, our networks use the larger number of inhibitors available to create higher levels of inhibition at time $t$ corresponding to higher number of firing outputs at time $t-1$. 
This strategy however, leads to `race conditions'. If many outputs fire at time $t$ and exactly  one output fires at time $t+1$, there will still be a high level of inhibition at time $t+1$. Thus, at time $t+2$ it is likely that the single firing output at time $t+1$ will stop firing and so the network will not stabilize to a valid WTA output state.

To avoid this race condition, we use history so that an output's self-loop excites the output for two time steps. The stability inhibitor $a_s$ will also be excited for two time steps by the outputs. In this way, if a single output $y_i$ fires at time $t$, while at time $t+1$ no outputs may fire due to high levels of inhibition, at time $t+2$ the history will cause both $a_s$ and $y_i$ to fire, and the network to stabilize  to a valid WTA state.

\subsubsection{Generalized Model with History Period} 
Following the basic SNN model of Section \ref{sec:modelW} we define an SNN model with history  period $h$  for any $h \ge 1$. The changes to the definition are highlighted in gray.

An SNN $\Net =\langle \All, w,b,f, h \rangle$ with  history period $h$ consists of:
\begin{itemize}
\item $\All$, a set of neurons, partitioned into a set of input neurons $\Input$, a set of output neurons $\Output$, and a set of auxiliary  neurons $\Inh$. $\All$ is also partitioned into a set of excitatory and inhibitory neurons $\Ex$ and $\Ih$. All input and output neurons are excitatory. 
\begin{tcolorbox}
\item $h \in \mathbb{Z}^{\ge 1}$, a positive integer indicating the neural response  history period.
\end{tcolorbox}
\begin{tcolorbox}
\item $w: \All \times \All \times \{1,...,h\} \rightarrow \R$, a weight function describing the weighted synaptic connections between the neurons in the network. $w$ is restricted in a few notable ways: 
\begin{itemize}
\item $w(u,x,l) = 0$ for all $u \in \All$, $x \in \Input$, and $l \in \{1,...,h\}$. 
\item Each excitatory neuron $v \in \Ex$ has $w(v,u,l)\ge 0$ for every $u$ and $l \in \{1,...,h\}$. Each inhibitory  neuron $v \in \Ih$ has $w(v,u,l)\le 0$ for every $u$ and $l \in \{1,...,h\}$.
\end{itemize}
\end{tcolorbox}
\item $b: \All \rightarrow \mathbb{R}$, a bias function, assigning an activation bias to each neuron.
\item $f: \R \rightarrow [0,1]$, a spike probability function, satisfying a few restrictions:
\begin{itemize}
\item $f$ is continuous and monotonically increasing.
\item $\lim_{x\rightarrow \infty} f(x) = 1$ and $\lim_{x \rightarrow -\infty} f(x) = 0$.
\end{itemize}
\end{itemize}

\spara{Remark on the Time Dependent  Weight Function}: 
The only difference between the above model and our basic SNN model of Section \ref{sec:structure} is in the specification of the weight function $w$. In the model with history period $h$, $w$ describes the strength of the synaptic connections between neurons in $\All$, as a function of the time difference between a spike and the current time (from $1$ up to $h$). $w(u,v,1)$ is the weight corresponding to the most recent time and $w(u,v,h)$ corresponds to the most distant time within the history  period. The weight function, for example, can be used to model the decaying effect of a spike over time, if we set $|w(u,v,1)| \ge |w(u,v,2)| \ge ... \ge |w(u,v,h)|$.

\subsubsection{Network Dynamics With History  Period}\label{sec:dynamicsH}

In our SNN with history period model, configurations and executions are defined as in our basic model (see Section \ref{sec:dynamics}). The behavior of the SNN is determined as follows:
\begin{itemize}
\item \textbf{Input Neurons}: As in our basic model, we specify how the infinite input execution $X^0X^1....$ is determined. Through this work, we will fix the input so that for each $u \in \Input$, $u^t$ is constant for all $t \ge 0$.
\item \textbf{Initial Firing States}: For each non-input $u \in \All \setminus \Input$, the firing states in the first $h$ time slots $u^0,u^1,...,u^{h-1}$ are arbitrary, where $h$ is this history  length.

\item \textbf{Firing Dynamics}:
For each non-input neuron $u \in \All \setminus \Input $ and every time $t \ge h$, let $\pot(u,t)$ denote the membrane potential at time $t$ and $p(u,t)$ denote
the corresponding firing probability. These values are calculated as:
\begin{align}
\label{eq:potentialOutH}
\pot(u,t) = \left (\sum_{i = 1}^h \sum_{v \in \All}w(v,u,i) \cdot v^{t-i} \right) -b(u) 
\text{ and } p(u,t)= f(\pot(u,t))
\end{align}
where $f$  is the spike probability  function.
At time $t$, each non-input neuron $u$ fires independently with probability $p(u,t)$. Note that equation \eqref{eq:potentialOutH} is defined only for $t \ge h$, in which case $t-i \ge 0$  for all $i \in \{1,...,h\}$. It is analogous to the potential calculation \eqref{eq:potentialOut}, used for our basic model, except that the summation of spikes is over $h$ time steps.
\end{itemize}


An SSN $\Net =\langle \All, w,b,f, h \rangle$ with history  period $h$, length $h$ execution $\alpha^{h} = \All^0\All^1...\All^{h-1}$, and infinite input execution $\alpha_\Input$ define a probability distribution over infinite executions, $\mathcal{D}(\Net,\alpha^{h},\alpha_\Input)$. This distribution is the natural distribution that follows from 
applying the stochastic firing dynamics of \eqref{eq:potentialOutH}. As with our basic model (see Section \ref{sec:dynamics}), we can also define a corresponding distribution $\mathcal{D}_\Output(\Net,\alpha^{h},\alpha_\Input)$ on infinite output executions.

\subsubsection{Solving Problems in Networks with History}
As in our basic model, a problem $P$ is a mapping from 
an infinite input execution $\alpha_\Input$ to a set of output distributions. A network $\Net$ is said to \emph{solve problem $P$ on input $\alpha_\Input$} if, for any length $h$ initial execution $\alpha^{h}$, the output distribution $\mathcal{D}_\Output(\Net,\alpha^{h},\alpha_\Input)$  is an element of $P(\alpha_\Input)$. A network $\Net$ is said to \emph{solve problem P} if it solves $P$ on every infinite input execution $\alpha_\Input$.

\subsection{$O(1)$ Convergence Time with $O(\log n)$ Inhibitors}\label{sec:logn}

In this section we describe and analyze a family  of networks that converge to a valid WTA state with constant  probability  in constant time and uses $O(\log n)$ inhibitors. 

We have one stability inhibitor $a_s$ that functions similarly to the stability inhibitor in our two-inhibitor construction (Definition \ref{def:wtaNet}), ensuring that, once the network reaches a valid WTA output configuration, it remains in this configuration for $t_s$ consecutive time steps with high probability (see Corollary \ref{cor:stabilityA}). This stability inhibitor employs a history period of length $2$. We prove in Lemma \ref{lem:resetA1} that it fires with high probability at time $t+1$ whenever at least one output fires at time $t$ or $t-1$.

We additionally have $\lognc$ convergence inhibitors, labeled $a_1,...,a_{\lognc}$. For each $i$, $a_i$ fires with high probability at time $t+1$ whenever $k \ge 2^i$ outputs fire at time $t$. In this way, with high probability, $a_1,...,a_i$ fire and all other inhibitors do not fires at time $t+1$ whenever the number of firing outputs $k$ is in the range $\left [2^{i},2^{i+1} \right )$ (see Lemma \ref{lem:resetA2}). Note that we define our network so $a_1,...,a_{\lognc}$ have no incoming connections that use the length $2$ history period (i.e., $w(u,a_i,2) = 0$ for all $u$ and all $i \in \{1,...,\lognc\}$). Thus, the firing probabilities of these inhibitors at time $t+1$  depend only  on the firing pattern at time $t$.

We set the inhibitory weights such that 
when $a_i$ fires at time $t$ (along with $a_j$ for all $j \le i$), each firing output fires at time $t+1$ with probability $p_i \approx \frac{1}{2^i}$ (see Lemma \ref{lem:convergenceA}). For $k \in \left [2^{i},2^{i+1} \right )$ we thus have $p_i \cdot k \in [1,2)$. We will show that this ensures that, with constant probability, exactly one output fires at time $t+1$ (see Corollary \ref{cor:convergenceA}). Once a single output fires, using the length-two history mechanism described in Section \ref{sec:history}, the network stabilizes to a valid WTA state with high probability.

We begin with a formal definition of the network below:

\begin{definition}[Constant Time WTA Network]\label{def:wtac} For any $n \in \Z^{\ge 2}$ and $\gamma \in\R^{+}$, let $\Netc= \langle \All,w,b,f,2\rangle$ where the spike probability, weight, and bias functions are defined as follows:
\begin{itemize}
\item The spike probability  function $f$ is defined to be the basic sigmoid function:
\begin{align*}
f(x) \eqdef \frac{1}{1+e^{-x}}.
\end{align*}
\item The set  of neurons $\All$ consists of a set of $n$ input neurons $\Input$, labeled $x_1,...,x_n$, a set  of $n$ corresponding outputs $\Output$, labeled $y_1,...,y_n$, and $\lognc+1$ auxiliary inhibitory neurons labeled $a_s$ and $a_1,...,a_{\lognc}$.
\item The weight function $w$ is given by:
\begin{itemize}
\item $w(x_i,y_i,1) = 6\gamma$, for all $i$.
\item $w(y_i,y_i,1) = w(y_i,y_i,2)  = 2\gamma$, for all $i$.
\item $w(a_s,y_i,1) = - \gamma$, for all $i$
\item $w(a_1,y_i,1) = - 7\gamma/2-\ln(2)$, for all $i$
\item $w(a_{j},y_i,1) = -\ln(2)$, for all $i$ and $j \in \{2,...,\lognc\}$.
\item $w(y_i,a_s,1) = w(y_i,a_s,2) = \gamma$, for all $i$.
\item $w(y_i,a_j,1) = \gamma$, for all $i$ and $j \in \{1,...,\lognc\}$.
\item $w(u,v,i) = 0$ for any $u,v$ and $i \in \{1,2\}$ whose connection is not specified above.  
\end{itemize}
\item The bias function $b$ is given by:
\begin{itemize}
\item $b(y_i) = 11\gamma/2$ for all $i$.
\item $b(a_s) = \gamma/2$.
\item $b(a_{i}) = 2^{i}\cdot \gamma -\gamma/2$ for all $i \in \{1,...,\lognc\}$.
\end{itemize}
\end{itemize}
\end{definition}
\spara{Use of History.} Note that in our network definition above, we use the history  mechanism only in two places. We set $w(y_i,y_i,1) = w(y_i,y_i,2)  = 2\gamma$, for all $i$, meaning that each output's self-loop affects  its potential for two time steps. We also set  $w(y_i,a_s,1) = w(y_i,a_s,2) = \gamma$ for all $i$, meaning that the stability  inhibitor is affected by the outputs for two steps. 

Due to the high level of inhibition inducted by the convergence inhibitors $a_1,...,a_{\lognc}$, after the network reaches a configuration with just a single firing output, it will  likely transition to a state with no firing outputs, since the number of firing inhibitors will still reflect the number of firing outputs in the previous time step. The length-two output self-loop and output to inhibitor connections allow the network to recover from this state. Specifically, in Lemma \ref{lem:stabilityA} and Corollary  \ref{cor:stabilityA}, we show that if the network reaches a valid WTA output state at time $t$, with good probability it will return to this state at time $t+2$ and remain  in this state for $t_s$ consecutive steps.

In our two-inhibitor construction, no history was necessary. The inhibition in the network was always low enough such that, after reaching a near-valid WTA configuration with a single firing output, with constant probability the network would transition to a valid-WTA configuration and stabilize in this configuration (see Lemma \ref{lem:nearvalid}).

\medskip
We prove the following theorems on the performance of $\Netc$:

\begin{tcolorbox}
\begin{theorem}[$O(\log n)$-Inhibitor WTA]\label{thm:logn}
For $\gamma \ge \ln((n+2)t_s/\delta)$,
$\Netc$ solves $\wta(n,t_c,t_s,\delta)$ for any  
$
t_c \ge \log_2(1/\delta)+1.
$ $\Netc$ contains $\lognc+1$ auxiliary inhibitors.
\end{theorem}
\end{tcolorbox}

\begin{tcolorbox}
\begin{theorem}[$O(\log n)$-Inhibitor Expected-Time WTA]\label{thm:lognExp}
For $\gamma \ge \ln((n+2)t_s/\delta)$,
$\Netc$ solves $\ewta(n,t_c,t_s)$ for any  $
t_c \ge  4.$ $\Netc$ contains $\lognc+1$ auxiliary inhibitors.
\end{theorem}
\end{tcolorbox}

\spara{Proof Roadmap.}
We prove Theorems \ref{thm:logn} and \ref{thm:lognExp}
in Section \ref{sec:convA}. The analysis is broken down as follows:
\begin{itemize}
\item[] Section \ref{sec:basicA}: Prove basic \emph{two-step lemmas} which characterize  single time step transitions of $\Netc$, showing that  the neurons behave as described in the above high-level description.
\item [] Section  \ref{sec:stabilityA}: Prove that, once in a valid WTA output configuration, as long as certain other stability conditions are satisfied, $\Netc$ stays in this configuration with high probability.
\item [] Section \ref{sec:convA}: Show that all configurations of $\Netc$ transition with constant probability to a stable and valid WTA configuration within $O(1)$ time steps.
\item [] Section \ref{sec:completingA}: Complete the analysis, demonstrating with what parameter values $\Netc$ solves the winner-take-all problem (Definitions \ref{high:wta} and \ref{exp:wta}).
\end{itemize}

\subsection{Two-Step Lemmas}\label{sec:basicA}

As in our analysis of our two-inhibitor network family in Section \ref{sec:basic}, we begin with a series of lemmas which characterize the basic transition probabilities of $\Netc$. Since we employ history period $h=2$, these lemmas consider the behavior of the network at time $t+1$ conditioned on its configuration at times $t$ and $t-1$. 

We first note that an analog to Lemma \ref{lem:gap} still holds for all inhibitory neurons.
\begin{lemma}[Characterization of Firing Probabilities]\label{lem:gapA}
For any time $t\ge 1$ and any $a \in \Inh$:
\begin{align*}
&\text{If }\pot(u,t) = 0, \text{ then }  p(u,t) = 1/2.\\
&\text{If }\pot(u,t) < 0, \text{ then } p(u,t) \le e^{-\gamma/2}.\\
&\text{If }\pot(u,t) > 0, \text{ then } p(u,t) \ge 1-e^{-\gamma/2}.
 \end{align*}
\end{lemma}
\begin{proof}
The proof is 
essentially identical to that of Lemma \ref{lem:gap}. We just use that for each inhibitor $a \in \Inh$ and each $i$ and $h \in \{1,2\}$, $w(y_i,a,h)$ and $b(a)$ are all multiples of $\gamma/2$.
\end{proof}

Analogs to Lemma \ref{lem:mono} and Corollary  \ref{cor:mono} also still hold, ensuring  that, with high probability, outputs that do not correspond to firing inputs do not fire.
\begin{tcolorbox}
\begin{lemma}[Correct Output Behavior]\label{lem:monoA} For any time $t\ge  1$, any configurations $C,C'$ of $\Netc$, and any $i$ with $C(x_i) = 0$, 
$$\Pr[y_i^{t+1} = 1 | \All^t = C,\All^{t-1} = C'] \le  e^{-3\gamma/2}.$$
\end{lemma}
\end{tcolorbox}
\begin{proof}
If $N^t = C$ then $x_i^t = C(x_i)$.
We can compute $y_i$'s potential at time $t+1$, assuming $x_i^{t} = 0$:
\begin{align*}
\pot(y_i,t+1) &= w(x_i,y_i,1) x_i^{t} + w(y_i,y_i,1) y_i^{t} + w(y_i,y_i,2) y_i^{t-1}  + w(a_s,y_i,1) a_s^{t}\\ &\hspace{3em}+\sum_{j=1}^{\lognc} w(a_c,y_i,1) a_c^{t} - b(y_i) \\ &\le 0 + 2\gamma + 2\gamma + 0 + 0 - 11\gamma/2 = -3\gamma/2.
\end{align*}
We thus have $p(y_i,t+1) \le \frac{1}{1+e^{3\gamma/2}} \le e^{-3\gamma/2}$.
\end{proof}

\begin{tcolorbox}
\begin{corollary}[Correct Output Behavior, All Neurons]\label{cor:monoA} For any time $t \ge 1$ and configurations $C,C'$ of $\Netc$,  
$$\Pr[y_i^{t+1} \le x_i^t\text{ for all }i | \All^t = C, \All^{t-1} = C' ] \ge 1-ne^{-3\gamma/2}.$$
\end{corollary}
\end{tcolorbox}
\begin{proof}
The proof is essentially identical to that of Corollary \ref{cor:mono}.
If $C(x_i) = 1$ then conditioned on $\All^t  = C$, $x_i^t = 1$ and so $y_i^{t+1} \le x_i^t$ always. Otherwise, by Lemma \ref{lem:monoA}, if $C(x_i) = 0$, then $\Pr[y_i^{t+1} = 0 | \All^t = C,\All^{t-1}=C'] \ge  1-e^{-3\gamma/2}$. Union bounding over  all such inputs (of which there are at most $n$) gives the corollary. 
\end{proof}

We next show that the inhibitors $a_s$, $a_1,...,a_{\lognc}$ behave as expected. The following lemmas can be viewed as a generalization of Lemma \ref{lem:reset}. We first show that, due to our use of history, $a_s$ fires with high probability  at time $t+1$ whenever \emph{at least one output fires at  time $t$ or $t-1$}.

\begin{tcolorbox}
\begin{lemma}[Correct Stability Inhibitor Behavior]\label{lem:resetA1}
For any time $t \ge 1$ and configurations $C,C'$  of $\Netc$,
\begin{enumerate}
\item If $\norm{C(\Output)}_1 = \norm{C'(\Output)}_1 =0$, then $\Pr[a_s^{t+1} = 0 | \All^t = C, \All^{t-1} =C'] \ge 1-e^{-\gamma/2}$.
\item If $\norm{C(\Output)}_1 \ge 1$ or $\norm{C'(\Output)}_1 \ge 1$, then $\Pr[a_s^{t+1} = 1 | \All^t = C,\All^{t=1}=C'] \ge 1-e^{-\gamma/2}$.
\end{enumerate}
\end{lemma}
\end{tcolorbox}
\begin{proof}  We prove the two conclusions separately.
 
\medskip
\spara{Conclusion  1: $\norm{C(\Output)}_1 = \norm{C'(\Output)}_1 = 0$.}
\medskip

In this case, $a_s$ receives no excitatory signal from the outputs so, 
$$\pot(a_s,t+1) = -b(a_s) < 0.$$
Thus by Lemma \ref{lem:gapA}, 
$$\Pr[a_s^{t+1} = 0 | \All^t  = C] \ge 1-e^{-\gamma/2}.$$

\medskip
\spara{Conclusion  2: $\norm{C(\Output)}_1 \ge 1$ or $\norm{C'(\Output)}_1 \ge 1$.}
\medskip

In this case we have:
\begin{align*}
\pot(a_s,t+1) &= \sum_{j=1}^n [w(y_j,a_s,1 ) y_j^t + w(y_j,a_s,2) y_j^{t-1}] - b(a_s) \\ &\ge \gamma - \gamma/2 = \gamma/2.
\end{align*}
We thus have by Lemma \ref{lem:gapA}:
$$\Pr[a_s^{t+1} = 1 | \All^t  = C] \ge 1-e^{-\gamma/2},$$
which gives the lemma.
\end{proof}

As described in Section \ref{sec:logn},
the convergence inhibitors $a_1,...,a_{\lognc}$ fire at time $t+1$ depending on the number of firing outputs at time $t$. They have no incoming connections which affect them for two rounds, and thus their firing probabilities at time $t+1$ do not depend on the firing pattern at time $t-1$. We prove that $a_j$ for all $j \le i$ fire with high probability at time $t+1$ whenever  the number of firing outputs at time $t$ falls in the range $[2^i,2^{i+1})$.  Further, all $a_j$ for $j > i$ do not fire with high probability.
\begin{tcolorbox}
\begin{lemma}[Correct Convergence Inhibitor Behaviors]\label{lem:resetA2}
For any time $t\ge 1$ and configurations $C,C'$  of $\Netc$,
\begin{enumerate}
\item If $\norm{C(\Output)}_1 \le 1$, then $$\Pr \left [a_i^{t+1} = 0\text{ for all } i \in \{1,...,\lognc\} \big | \All^t = C,\All^{t-1}=C'\right ] \ge 1-\lognc \cdot e^{-\gamma/2}.$$
\item For any  $i \in \{1,...,\lognc\}$, if $\norm{C(\Output)}_1 = k$ for $k \in [2^i,2^{i+1})$, then:
\begin{align*}\Pr[a_1^{t+1} = ... = a_i^{t+1} = 1\text{ and } a_{i+1}^{t+1} = ... = a_{\lognc}^{t+1} = 0 | \All^t& = C,\All^{t-1}=C']\\
&\ge 1-\lognc \cdot e^{-\gamma/2}.
\end{align*}
\end{enumerate}
\end{lemma}
\end{tcolorbox}
\begin{proof}
We prove the two conclusions of the lemma separately.

\medskip
\spara{Conclusion  1: $\norm{C(\Output)}_1 \le 1$.}
\medskip

In this case we have for all $i$:
\begin{align*}
\pot(a_i,t+1) &= \sum_{j=1}^n w(y_j,a_i,1) y_j^t - b(a_i) \\ &\le \gamma - 3\gamma/2 = -\gamma/2.
\end{align*}
Again by  Lemma \ref{lem:gapA} and a union bound,
$$\Pr \left [\sum_{i=1}^{\lognc} a_i^{t+1} = 0 \big | \All^t = C,\All^{t-1}=C'\right ] \ge 1-\lognc \cdot e^{-\gamma/2}.$$

\medskip
\spara{Conclusion  2: $\norm{C(\Output)}_1 = k$ for $k \in [2^i,2^{i+1})$.}
\medskip

In this case, for any $j \le i$ we have:
\begin{align*}
\pot(a_j,t+1) &= \sum_{l=1}^n w(y_l,a_s,1) y_l^t - b(a_j) \\ &= k \cdot \gamma - 2^j \gamma + \gamma/2 \ge\gamma/2\end{align*}
where the last inequality follows since $k \ge 2^{j} \ge 2^j$. In contrast for $j > i$: 
\begin{align*}
\pot(a_j,t+1) &= \sum_{l=1}^n w(y_l,a_s,1) y_l^t - b(a_j) \\ &= k \cdot \gamma - 2^j \gamma + \gamma/2 \le -\gamma/2\end{align*}
where the last inequality follows from the fact that $k < 2^{i+1} \le 2^j$.
Overall, by  Lemma \ref{lem:gapA} and a union bound, 
\begin{align*}
\Pr[a_1^{t+1} = ... = a_i^{t+1} = 1\text{ and } a_{i+1}^{t+1} = ... = a_{\lognc}^{t+1} = 0 | \All^t &= C,\All^{t-1}=C']\\
&\ge 1-\lognc \cdot e^{-\gamma/2},
\end{align*}
which gives the lemma
\end{proof}

It will be useful in our future bounds to consider the class of configurations in which all outputs and inhibitors behave as expected. Such a configuration is analogous to the good configurations of our two-inhibitor networks (Definition \ref{def:good}), from which we were able to show convergence.
\begin{definition}[Typical Configuration]\label{def:typical} A \emph{typical configuration} is any configuration $C$ with $C(y_i) \le C(x_i)$ for all $i$ and $C(a_s) \ge C(a_1) \ge ... \ge C(a_{\lognc})$.
\end{definition}

In combination, Corollary \ref{cor:monoA}, Lemma \ref{lem:resetA1} and Lemma \ref{lem:resetA2} give:
\begin{tcolorbox}
\begin{corollary}[Correct Behavior, All Neurons]\label{cor:typicalA} Assume the input execution $\alpha_\Input$ of $\Netc$ has $\Input^t$ fixed for all $t$ and consider configurations $C,C'$ with $C(X) = C'(X) = X^t$. For any time $t \ge 1$:  
$$\Pr[\All^{t+1}\text{ is a typical configuration }| \All^t = C, \All^{t-1} = C' ] \ge 1-(n+\lognc+1)\cdot e^{-\gamma/2}.$$
\end{corollary}
\end{tcolorbox}
\begin{proof}
Since the input is fixed, by Corollary \ref{cor:monoA},
\begin{align}\label{dog}
\Pr[y_i^{t+1}\le x_i^{t+1}| \All^t = C, \All^{t-1} = C' ] \ge 1-n\cdot e^{-3\gamma/2}.\end{align}
If $\norm{C(\Output)}_1 \le 1$ then by Lemma \ref{lem:resetA2} conclusion (2),
$$\Pr[a_1^{t+1} = .... = a_{\lognc}^{t+1} = 0| \All^t = C, \All^{t-1} = C' ] \ge 1-\lognc \cdot e^{-\gamma/2}.$$
If $\norm{C(\Output)}_1 > 1$ then by Lemma \ref{lem:resetA1} conclusion (2), Lemma \ref{lem:resetA2} conclusion (2), and a union bound, 
\begin{align*}
\Pr[a_s^{t+1} = a_1^{t+1} = ... = a_i^{t+1} = 1\text{ and } a_{i+1}^{t+1} = ... = a_{\lognc}^{t+1} = 0 | \All^t& = C,\All^{t-1}=C']\\
&\ge 1-(\lognc +1)\cdot e^{-\gamma/2}.
\end{align*}
Combining these two cases, we have
$$\Pr[a_s^{t+1} \ge a_1^{t+1} \ge ... \ge a_{\lognc}^{t+1} | \All^t = C,\All^{t-1}=C'] \ge 1-(\lognc +1)\cdot e^{-\gamma/2},$$
which gives the lemma after a union bound with \eqref{dog}
\end{proof}

We next show that the stability inhibitor  firing alone, with high probability,  induces exactly the outputs that fired at one of the previous two time steps to fire in the next step. This Lemma is analogous to Lemma \ref{lem:stability1} for our two-inhibitor networks.

\begin{tcolorbox}
\begin{lemma}[Stability Inhibitor Effect]\label{lem:stability1A}
Assume the input execution $\alpha_\Input$ of $\Netc$ has $\Input^t$ fixed for all $t$ and consider configurations $C,C'$ with $C(X) = C'(X) = X^t$, $C(a_s)  = 1$, $C(a_i) = 0$ for all $i \in \{1,...,\lognc\}$, and $C(y_i) \le C(x_i)$, $C'(y_i) \le C'(x_i)$ for all $i$. For any time $t\ge 1$,
$$\Pr[y_i^{t+1} = \max(y_i^t,y_i^{t-1})\text{ for all }i | \All^t = C, \All^{t-1}=C'] \ge 1-ne^{-\gamma/2}$$
\end{lemma}
\end{tcolorbox}
\begin{proof}
Conditioned on $\All^t = C,\All^{t-1}=C'$, $y_i^{t} \le x_i^t$ and $y_i^{t-1} \le x_i^{t-1}$ by assumption. So for  any output with $\max(y_i^t,y_i^{t-1}) = 1$ we must have $x_i^t= 1$. This gives:
\begin{align*}
\pot(y_i,t+1) &= w(x_i,y_i,1) x_i^t + w(y_i,y_i,1) y_i^t + w(y_i,y_i,2) y_i^{t-1}  + w(a_s,y_i,1) a_s^t \\
&\hspace{3em}+ \sum_{j=1}^{\lognc} w(a_j,y_i,1) a_j^t - b(y_i) \\&\ge 6\gamma + 2\gamma \cdot \max(y_i^t,y_i^{t-1}) - \gamma + 0 - 11\gamma/2s\\&\ge 3\gamma/2.
\end{align*}
In contrast, for any output with $\max(y_i^t, y_i^{t-1}) = 0$:
\begin{align*}
\pot(y_i,t+1) &= w(x_i,y_i,1) x_i^t + w(y_i,y_i,1) y_i^t + w(y_i,y_i,2) y_i^{t-1}  + w(a_s,y_i,1) a_s^t\\
&\hspace{3em}+ \sum_{j=1}^{\lognc} w(a_j,y_i,1) a_j^t - b(y_i) \\ &\le 6\gamma + 0 + 0 - \gamma + 0 - 11\gamma/2 \\&= -\gamma/2.
\end{align*}
So, if $\max(y_i^t, y_i^{t-1})= 1$, then $y_i^{t+1} = 1$ with probability $\ge 1-e^{-3\gamma/2}$. If $\max(y_i^t, y_i^{t-1})= 0$, then $y_i^{t+1} = 0$ with probability $\ge 1-e^{-\gamma/2}$. The lemma follows after  union bounding over all $n$  outputs. 
\end{proof}

We next characterize the effect of the convergence inhibitors. We show that when $l$  inhibitors fire, any firing output  that also fired in the previous time step, fires with probability $\Theta(1/l)$. We will show in Corollary \ref{cor:convergenceA} that this implies that in the next  step, with constant  probability, exactly one output fires, and in fact the configuration is a valid WTA output configuration. 

Lemma  \ref{lem:convergenceA} below is analogous to Lemma \ref{lem:convergence} for our two-inhibitor networks, except that the firing probability is $\Theta(1/l)$ rather than $1/2$. Since this firing probability is smaller when a larger number of outputs fire at time $t$, convergence to a single firing output with constant probability occurs in $1$ step, rather than $O(\log n)$ steps.
\begin{tcolorbox}
\begin{lemma}[Convergence Inhibitor Effect]\label{lem:convergenceA}
Assume the input execution $\alpha_\Input$ of $\Netc$ has $\Input^t$ fixed for all $t$ and consider configurations $C,C'$ with $C(X) = C'(X) = X^t$, $C(a_s)  = C(a_1)  = ... C(a_l) = 1$, $C(a_{l+1}) =...=C(a_{\lognc}) = 0$ for some $l \ge 1$,  and $C(y_i) \le C(x_i)$, $C'(y_i) \le C'(x_i)$ for all $i$. For any time $t \ge 1$,
\begin{enumerate}
\item $\Pr[y_i^{t+1} \le \min(y_i^t,y_i^{t-1})\text{ for  all }i | \All^t = C,\All^{t-1} = C'] \ge 1-ne^{-2\gamma}$.
\item If $\min(y_i^t,y_i^{t-1}) = 1$, $\Pr[y_i^{t+1} = 1 | \All^t = C,\All^{t-1} = C'] = \frac{1}{1+2^\ell}$.
\item For $i \neq j$, $y_i^{t+1}$ and $y_j^{t+1}$ are independent conditioned on $\All^t = C, \All^{t-1} = C'$.
\end{enumerate}
\end{lemma}
\end{tcolorbox}
\begin{proof}
Conditioned on $\All^t = C,\All^{t-1} = C'$, if $\min(y_i^t,y_i^{t-1}) = 1$, then $y_i^t = y_i^{t-1} = 1$ and by assumption $x_i^t = 1$. We can thus compute:
\begin{align*}
\pot(y_i,t+1) &= w(x_i,y_i,1) x_i^t + w(y_i,y_i,1) y_i^t + w(y_i,y_i,2) y_i^{t-1} + w(a_s,y_i,1) a_s^t + w(a_1,y_i,1) a_1^t \\&\hspace{3em}+ \sum_{j=2}^{\lognc} w(a_j,y_i,1) a_j^t - b(y_i) \\ &= 6\gamma + 4\gamma - \gamma -7\gamma/2 - \ln 2 - (l-1)\cdot \ln 2 - 11\gamma/2 \\&= -l\cdot \ln 2.
\end{align*}
We thus have  
$$\Pr[y_i^{t+1} = 1 | \All^t = C] = f(-l \cdot \ln 2) = \frac{1}{1+2^{l}}.$$
This gives conclusion (2). Conclusion (3) holds since, with $\All^t$ and $\All^{t-1}$ fixed with history length $2$, $u^{t+1}$ is independent of $v^{t+1}$ for all $u \neq v$.
We can also bound if $\min(y_i^t,y_i^{t-1})  =  0$: 
\begin{align*}
\pot(y_i,t+1) &= w(x_i,y_i,1) x_i^t + w(y_i,y_i,1) y_i^t + w(y_i,y_i,2) y_i^{t-1} + w(a_s,y_i,1) a_s^t + w(a_1,y_i,1) a_1^t \\&\hspace{3em}+ \sum_{j=2}^{\lognc} w(a_j,y_i,1) a_j^t - b(y_i) \\ &\le 6\gamma + 2\gamma - \gamma -7\gamma/2-\ln 2 - (l-1) \cdot \ln 2 - 11\gamma/2 \\&\le -2\gamma.
\end{align*}
Thus, $\Pr[y_i^{t+1} = 1 | \All^t = C,\All^{t-1}=C'] \le e^{-2\gamma}$. By a union bound over at most $n$ such outputs, we have, with probability $\ge 1-ne^{-2\gamma}$, $y_i^{t+1}  \le \min(y_i^t,y_i^{t-1})$  for all $i$, giving conclusion (1) and completing the lemma.
\end{proof}

We now formalize the fact that  the network converges to a valid WTA output configuration in just a single step with constant probability, as long as the number of inhibitors matches the minimum number of firing outputs in the preceding two steps. Corollary \ref{cor:convergenceA} can be viewed as an analog to Lemma \ref{lem:progress} for our two inhibitor networks, except that the number of outputs is reduced to $1$, rather than just cut in half, with constant probability.
%
\begin{tcolorbox}
\begin{corollary}[Constant Probability of a Valid WTA Configuration]\label{cor:convergenceA}
Assume the input execution $\alpha_\Input$ of $\Netc$ has $\Input^t$ fixed for all $t$ and consider configurations $C,C'$ with $C(X) = C'(X) = X^t$, $C(a_s)  = C(a_1)  = ... C(a_l) = 1$, $C(a_{l+1}) =...=C(a_{\lognc}) = 0$ for some $l \ge 1$,  and $C(y_i) \le C(x_i)$, $C'(y_i) \le C'(x_i)$ for all $i$. For any time $t \ge 1$, if $\norm{\min(C',C)}_1 \in \left [ 2^l,2^{l+1}\right)$, then
$$\Pr[Y^{t+1} \text{ is a valid WTA output configuration } | \All^t = C,\All^{t-1}=C'] \ge \frac{1}{16}-ne^{-2\gamma}.$$
\end{corollary}
\end{tcolorbox}
\begin{proof}
Let  $\bar Y$ be the set of outputs who fire in both in $C,C'$. So $\norm{\bar Y}_1 = \norm{\min(C',C)}_1$.
By conclusions (2) and (3) of Lemma \ref{lem:convergenceA}, and the assumption that $\norm{\min(C',C)}_1 \in \left [ 2^l,2^{l+1}\right)$:
\begin{align*}
\Pr [\norm{\bar Y^{t+1}}_1 = 1 | \All^t = C, \All^{t-1}=C'] &= \frac{1}{1+2^l} \cdot \left (1 - \frac{1}{1+2^l}\right)^{\norm{\bar Y^{t}}_1} \cdot \norm{\bar Y^{t}}_1\\
&\ge \frac{1}{1+2^l} \left (1 - \frac{1}{1+2^l}\right)^{2^{l+1}} \cdot 2^l\\
&\ge \frac{2^l}{1+2^l} \cdot \frac{1}{8} \ge \frac{1}{16}.
\end{align*}
Further, by conclusion (1) of Lemma \ref{lem:convergenceA}, no output outside if $\bar Y$ fires at time $t+1$ with probability $\ge 1-ne^{-2\gamma}$. Additionally, by assumption, for all $y_i \in \bar Y$, $x_i^t  = 1$. So since if exactly one output in $\bar Y$ fires, $\All^{t+1}$ is a valid WTA output configuration.
This gives the corollary by a union bound.
\end{proof}

We also show a related corollary -- if the number of firing inhibitors \emph{exceeds} the appropriate amount for the number of firing outputs, then with good probability, no outputs fire in the next time step.
\begin{tcolorbox}
\begin{corollary}[Constant Probability of Zero Firing Outputs]\label{cor:convergenceA0}
Assume the input execution $\alpha_\Input$ of $\Netc$ has $\Input^t$ fixed for all $t$ and consider configurations $C,C'$ with $C(X) = C'(X) = X^t$, $C(a_s)  = C(a_1)  = ... C(a_l) = 1$, $C(a_{l+1}) =...=C(a_{\lognc}) = 0$ for some $l \ge 1$,  and $C(y_i) \le C(x_i)$, $C'(y_i) \le C'(x_i)$ for all $i$. For any time $t \ge 1$, if $\norm{\min(C',C)}_1 \in \left [ 0,2^{l+1}\right)$, then
$$\Pr[\norm{Y^{t+1}}_1 = 0 | \All^t = C,\All^{t-1}=C'] \ge \frac{1}{8} - ne^{-2\gamma}$$
\end{corollary}
\end{tcolorbox}
\begin{proof}
Let  $\bar Y$ be the set of outputs who fire in both in $C,C'$. So $\norm{\bar Y}_1 = \norm{\min(C',C)}_1 \in  \left [ 0, 2^{l+1}\right)$.
By conclusions (2) and (3) of Lemma \ref{lem:convergenceA}:
\begin{align*}
\Pr [\norm{\bar Y^{t+1}}_1 = 0 | \All^t = C, \All^{t-1}=C'] &= \left (1 - \frac{1}{1+2^l}\right)^{\norm{\bar Y^{t}}_1}\\
&\ge \left (1 - \frac{1}{1+2^l}\right)^{2^{l+1}}\\
&\ge \frac{1}{8}.
\end{align*}
Further, by conclusion (1) of Lemma \ref{lem:convergenceA}, no output outside if $\bar Y$ fires at time $t+1$ with probability $\ge 1-ne^{-2\gamma}$. 
This gives the corollary by a union bound.
\end{proof} 


Finally, we show that if there is no inhibition in the network, all outputs corresponding to firing inputs are likely  to fire at the next time step.
\begin{tcolorbox}
\begin{lemma}[No Inhibitor Effect]\label{lem:quietA}
For any time $t \ge 1$ and configurations $C,C'$ of $\Netc$, if $\norm{C(\Inh)}_1 = 0$, then 
$$\Pr [y_i^{t+1} = x_i^t \text{ for all } i| \All^t = C,\All^{t-1}=C'] \ge 1-ne^{-\gamma/2}.$$
\end{lemma}
\end{tcolorbox}
\begin{proof}
We consider two cases:

\medskip
\spara{Case  1: $x_i^t = 0$.}
\medskip

In this case:
\begin{align*}
\pot(y_i,t+1) &= w(x_i,y_i,1) x_i^t + w(y_i,y_i,1) y_i^t + w(y_i,y_i,2) y_i^{t-1} + w(a_s,y_i,1) a_s^t + \\&\hspace{3em}+ \sum_{j=1}^{\lognc} w(a_j,y_i,1) a_j^t - b(y_i) \\ &\le 0 + 4\gamma + 0 + 0 - 11\gamma/2 \\&\le -3\gamma/2.
\end{align*}
This gives 
$$\Pr[y_i^{t+1} = 0 = x_i^t | \All^t = C,\All^{t-1}=C'] \ge 1-e^{-2\gamma}.$$

\medskip
\spara{Case  2: $x_i^t = 1$.}
\medskip

In this case:
\begin{align*}
\pot(y_i,t+1) &= w(x_i,y_i,1) x_i^t + w(y_i,y_i,1) y_i^t + w(y_i,y_i,2) y_i^{t-1} + w(a_s,y_i,1) a_s^t + \\&\hspace{3em}+ \sum_{j=1}^{\lognc} w(a_j,y_i,1) a_j^t - b(y_i) \\ &\ge 6\gamma + 0 + 0 + 0 - 11\gamma/2 \\&\ge \gamma/2.
\end{align*}
This gives $$\Pr[y_i^{t+1} = 1= x_i^t | \All^t = C,\All^{t-1}=C'] \ge 1-e^{-\gamma/2}.$$
The lemma then follows after union bounding over all $n$ outputs.
\end{proof}

\subsection{Stability}\label{sec:stabilityA}

In this section we show that once in a valid WTA output configuration (Definition \ref{def:output}), the network remains in this configuration with high probability. Due to our use of a length-two history period, our stability proof requires certain conditions on the firing states at both times $t$ and times $t-1$. We will focus on the case when there is at least one firing input (i.e., when $\norm{\Input^t} \ge 1$.) In the case $\norm{\Input^t} = 0$, convergence to a valid WTA output configuration and stability  of this configuration follow easily from Lemma \ref{lem:monoA}.

Definition \ref{def:nearStableA} below can be viewed as a two-step generalization of a near-valid WTA configuration of our two-inhibitor networks (Definition \ref{def:nearvalid}).

\begin{definition}[Near-Stable Pair of Configurations]\label{def:nearStableA}
Assume the input execution $\alpha_\Input$ of $\Netc$ has $X^t$ fixed for all $t$ and $\norm{X^t}_1 \ge 1$. Consider configurations $C,C'$ with $C(X) = C'(X) = X^t$. The ordered pair $(C',C)$ is \emph{near-stable} if:
\begin{enumerate}
\item $\norm{\max(C'(\Output),C(\Output))}_1 =1$, where $\max(C'(\Output),C(\Output))$ is the entrywise maximum of $C'(\Output),C(\Output)$.
 This condition requires that exactly  one output fires in configurations $C',C$. It may fire in one or both configurations.
\item $C(a_s) = C'(a_s) = 1$.
\item $C(a_i) = 0$ for  all $i \in \{1,\lognc\}$.
\item $C(y_i) \le C(x_i)$, $C'(y_i) \le C'(x_i)$ for all $i$.
\end{enumerate}
\end{definition}
Note that by conditions (1) and (4), at  least  one of $C(Y),C'(Y)$ is a valid WTA output configuration (Definition \ref{def:output}).
In our proofs, it will be useful to refer to the output whose existence is guaranteed by  condition (1). 
Thus we define:
\begin{definition}\label{def:nearStableIndex}
Assume the input execution $\alpha_\Input$ of $\Netc$ has $X^t$ fixed for all $t$ and $\norm{X^t}_1 \ge 1$. For any near-stable pair of configurations $(C',C)$ with $C(X)=C'(X) = X^t$, let  $out(C',C) \in \{1,...,n\}$ be equal to the index of the unique output that fires in $C',C$ (whose existence is guaranteed by  condition (1) of Definition \ref{def:nearStableA}).
\end{definition}

We next show that  if the configurations $(\All^{t-1}, \All^t)$ are near-stable, then with high probability, $\All^{t+1}$ will be a valid WTA output  configuration. Further, the network will stabilize for $t_s$ steps. That  is, with high probability, we will have $\All^{t+1} = ... = \All^{t+t_s+1}$. Lemma \ref{def:nearStableA} is analogous to Lemma \ref{lem:nearvalid} for our two-inhibitor networks.

\begin{tcolorbox}
\begin{lemma}[Reaching Stability From Near-Stable Configurations]\label{lem:stabilityA} Assume the input execution $\alpha_\Input$ of $\Netc$ has $X^t$ fixed for all $t$ and that  $\norm{X^t}_1 \ge 1$. Consider any near-stable pair of configurations $(C',C)$ with $C(X)=C'(X) = X^t$.
For any time $t \ge 1$, conditioned on $\All^t = C,\All^{t-1}=C'$, with probability $\ge 1-(n+\lognc+1) \cdot e^{-\gamma/2}$, 
\begin{enumerate}
\item $\All^{t+1}$ is a valid WTA output configuration (Definition \ref{def:output}).
\item $y_{out(\All^{t-1},\All^t)}^{t+1} = 1$. That is, the winner at time $t+1$  is the output firing in $\All^{t-1},$ and/or $\All^t$. 
\item  $(\All^t,\All^{t+1})$ is also a near-stable pair of configurations.
\end{enumerate}
\end{lemma}
\end{tcolorbox}
\begin{proof}
By condition (1) of Definition \ref{def:nearStableA}, for all $j \neq out(\All^{t-1},\All^t)$,  $y_j^{t}=y_j^{t-1} = 0$. Additionally, by conditions (2) and (3), $a_s$  is  the only inhibitor that fires at time $t$. So by Lemma \ref{lem:stability1A},
\begin{align}
\Pr[Y^{t+1} = \max(Y^t,Y^{t-1}) | \All^t = C] \ge 1-ne^{-\gamma/2}.\label{eq:Fail1A}
\end{align}
This gives that $y_{out(\All^{t-1},\All^t)}^{t+1} = 1$ while $y_j^{t+1} = 0$ for all $j \neq out(\All^{t-1},\All^t)$. By  condition (4) of Definition \ref{def:nearStableA} we must have also have $x_{out(\All^{t-1},\All^t)}^{t+1} = 1$.
 This implies conclusions (1) and (2) of the lemma. 
 It remains to show conclusion (3). 

Condition (1) of Definition \ref{def:nearStableA} holds if $Y^{t+1} = \max(Y^t,Y^{t-1})$ (see \eqref{eq:Fail1A}) since $y_{out(\All^{t-1},\All^t)}^{t+1} = 1$ and further, $y_{out(\All^{t-1},\All^t)}$ is the only output that may fire at time $t$. Conditions (2) and (3) with probability $\ge 1-(\lognc+1) \cdot e^{-\gamma/2}$ conditioned on $\All^t = C,\All^{t-1}=C'$ by Lemmas \ref{lem:resetA1} and \ref{lem:resetA2} and a union bound. Finally, condition (4) holds if $Y^{t+1} = \max(Y^t,Y^{t-1})$ (see \eqref{eq:Fail1A}). Overall, by a union bound, all three conclusions hold with probability $\ge 1-(n+\lognc+1) \cdot e^{-\gamma/2}$, giving the lemma.

\end{proof}

We can use
Lemma \ref{lem:stabilityA} to show that $\Netc$ remains in a valid WTA configuration for $t_s$ consecutive time steps with good probability. 
\begin{tcolorbox}
\begin{corollary}[Stability of Valid WTA Configurations]\label{cor:stabilityA}
Assume  the input execution $\alpha_X$ of $\Netc$ has $X^t$ fixed for all $t$ and that $\norm{X^t}_1 \ge 1$. Consider any  near-stable pair of configurations $(C',C)$ with $C(X)=C'(X) = X^t$. For any  time $t \ge 1$, . For any time $t$ conditioned on $\All^t = C,\All^{t-1}=C'$, with probability $\ge 1-3t_s n\cdot e^{-\gamma/2}$, $Y^{t+1}$ is a valid WTA output configuration and further,
\begin{align*}
Y^{t+1}  = Y^{t+2} = ... = Y^{t+t_s+1}.
\end{align*}
\end{corollary}
\end{tcolorbox}
\begin{proof}
We apply Lemma \ref{lem:stabilityA} for each time $t+1,...,t+t_s+1$ in succession. This is possible since by conclusion (3), if $(\All^{t-1},\All^{t-1})$ is a near-stable pair of configurations, then with high probability $(\All^{t},\All^{t+1})$ is as well. Conclusion (1) gives that $\All^{t+1},...,\All^{t+t_s+1}$ are all valid WTA configurations and conclusion (2) gives that $Y^{t+1}  = Y^{t+2}= ... = Y^{t+t_s+1}.$ By a union bound over these $t_s$ steps, the conclusion holds with probability $\ge 1- t_s(n+\lognc+1) \cdot e^{-\gamma/2} \ge  1-3t_sn \cdot e^{-\gamma/2}$.
\end{proof}

\subsection{Convergence in $O(1)$ Steps}\label{sec:convA}
We now use the transition lemmas of Section \ref{sec:basicA} to show that, starting from any configuration, the network converges to a near-stable pair  of configurations (Definition \ref{def:nearStableA}) with constant probability  in $O(1)$ time steps. Combined with Corollary \ref{cor:stabilityA} this shows convergence to a valid WTA configuration (and stability within this configuration for $t_s$ steps) in $O(1)$ steps with constant probability. 

Our analysis is tedious by straightforward. It
breaks down into nine cases, based on the initial output and inhibitor behavior. These cases are summarized in Table \ref{tab:logCases}. Since some cases depend on our bounds for others, we do not prove them in the order listed.

\begin{table}[H]
\centering
\begin{tabu}{|c|c|c|} \hline
Output Count $\norm{Y^t}_1$ & Inhibitor Count $\norm{A^t}_1$ & Lemma \\
\hline
0 & $0$ & Lemma \ref{lem:converge1A}\\
0 & $1$ & Lemma \ref{lem:converge01A}\\
0 & any $ a> 1$ & Lemma \ref{lem:converge3A}\\
\hline 
1 & $0$ & Lemma \ref{lem:converge10A}\\
1 & $1$ & Lemma \ref{lem:converge11A}\\
1 & any  $ a > 1$ & Lemma \ref{lem:converge6A}\\
\hline
any  $k > 1$ & $0$ & Lemma \ref{lem:converge9A}\\
any $k > 1$ & $1$ & Lemma \ref{lem:converge8A}\\
any $k > 1$ & any $a  > 1$ & Lemma \ref{lem:converge7A}\\
\hline
\end{tabu}
\caption{\label{tab:logCases}  Summary of cases from which we show convergence in $O(1)$ steps to a near-stable pair of configurations (Definition \ref{def:nearStableA}) with constant probability.}
\end{table}

\begin{tcolorbox}
\begin{lemma}[$\norm{Y^t}_1=\norm{A^t}_1 = 0$]\label{lem:converge1A}
Assume  the input execution $\alpha_X$ of $\Netc$ has $X^t$ fixed for all $t$ and that $\norm{X^t}_1 \ge 1$. Consider any pair of configurations $C',C$ with $C(X)  = C'(X) = X^t$ and $\norm{C(\Output)}_1 = \norm{C(\Inh)}_1 = 0$. For any  time $t \ge 1$,
$$\Pr[(\All^{t+3},\All^{t+4})\text{ is near-stable } | \All^t = C,\All^{t-1}=C'] \ge \frac{1}{16} - 12n\cdot e^{-\gamma/2}.$$
\end{lemma}
\end{tcolorbox}
\begin{proof}
The proof follows from a series of four steps, arguing about the state of $\Netc$ at times $t+1,t+2,t+3,t+4$.

\medskip
\spara{Step 1:}
\medskip

Let $\mathcal{E}_1$ be the event that $y_i^{t+1} = x_i^{t+1} \text{ for all } i$ and that $a_i^{t+1} = 0$ for all $i \in \{1,...,\lognc \}$. By Lemma \ref{lem:quietA}, since $\norm{C(A)}_1  = 0$, 
$$\Pr [y_i^{t+1} = x_i^{t+1} \text{ for all } i| \All^t = C,\All^{t-1}=C'] \ge 1-ne^{-\gamma/2}.$$
Additionally, since $\norm{C(Y)}_1 = 0$, by Lemma \ref{lem:resetA2} conclusion (1), 
$$\Pr [a_i^{t+1} = 0\text{ for all }i \in \{1,...,\lognc \}| \All^t = C,\All^{t-1}=C'] \ge 1-\lognc \cdot e^{-\gamma/2}.$$
Thus, by a union bound we have:
\begin{align}\label{e1A}
\Pr [\mathcal{E}_1 | \All^t = C,\All^{t-1}=C'] \ge 1-(n+\lognc)\cdot e^{-\gamma/2}.
\end{align}

\medskip
\spara{Step 2:}
\medskip

Let $\mathcal{E}_2$ be the event that $y_i^{t+2} = x_i^{t+2} \text{ for all } i$ and that for $l =\left \lfloor \log_2 \left (\norm{X^t}_1 \right) \right\rfloor$, $a_s^{t+2} = a_1^{t+2} = ... = a_l^{t+2} = 1\text{ and } a_{l+1}^{t+2} = ... = a_{\lognc}^{l+2} = 0$ (if $l=0$, just $a_s^{t+2} = 1$). Conditioned on $\mathcal{E}_1$, the only inhibitor that possibly fires at time $t+1$ is $a_s$. We can separately  consider the cases when $a_s^{t+1} =0$ and when $a_s^{t+1} = 1$. By Lemma \ref{lem:stability1A},
\begin{align*}
\Pr \left [y_i^{t+2} = x_i^{t+2} \text{ for all } i | \mathcal{E}_1,a_s^{t+1}=1, \All^t = C,\All^{t-1}=C'\right] \ge 1-ne^{-\gamma/2}.
\end{align*}
By Lemma \ref{lem:quietA} we also have
\begin{align*}
\Pr \left [y_i^{t+2} = x_i^{t+2} \text{ for all } i | \mathcal{E}_1,a_s^{t+1}=0, \All^t = C,\All^{t-1}=C'\right] \ge 1-ne^{-\gamma/2}.
\end{align*}
By the law of total probability this gives:
\begin{align}\label{e2First}
\Pr \left [y_i^{t+2} = x_i^{t+2} \text{ for all } i | \mathcal{E}_1, \All^t = C,\All^{t-1}=C'\right] \ge 1-ne^{-\gamma/2}.
\end{align}
We also apply Lemma \ref{lem:resetA2}. Conditioned on $\mathcal{E}_1$, for $l =\left \lfloor \log_2 \left (\norm{X^t}_1 \right) \right\rfloor$, we have $\norm{\Output^{t+1}}_1 = \norm{\Input^t}_1 \in [2^{l}, 2^{l+1})$, which gives that,
\begin{align}\label{e2Second}
\Pr[a_1^{t+2} = ... = a_i^{t+2} = 1\text{ and } a_{i+1}^{t+2} = ... = a_{\lognc}^{t+2} = 0 | &\mathcal{E}_1, \All^t = C,\All^{t-1}=C']\nonumber \\&\ge 1-\lognc \cdot e^{-\gamma/2}.
\end{align}
Similarly, applying Lemma \ref{lem:resetA1}, since conditioned on $\mathcal{E}_1$, $\norm{Y^{t+1}}_1 = \norm{X^{t+1}}_1 \ge 1$:
\begin{align}\label{e2Second2}
\Pr[a_s^{t+2} = 1 | \mathcal{E}_1, \All^t = C,\All^{t-1}=C'] \ge 1-e^{-\gamma/2}.
\end{align}
Combining \eqref{e2First}, \eqref{e2Second}, and \eqref{e2Second2} we have:
\begin{align}\label{e2A}
\Pr [\mathcal{E}_2 | \mathcal{E}_1, \All^t = C,\All^{t-1}=C'] 
&\ge 1-(n+\lognc + 1)\cdot e^{-\gamma/2}.
\end{align}

\medskip
\spara{Step 3:}
\medskip

Let  $\mathcal{E}_3$ be the event that $Y^{t+3}$ is a valid WTA configuration, and that for $l =\left \lfloor \log_2 \left (\norm{X^t}_1 \right) \right\rfloor$, $a_s^{t+3} = a_1^{t+3} = ... = a_l^{t+3} = 1\text{ and } a_{l+1}^{t+3} = ... = a_{\lognc}^{l+3} = 0$. 

If  $l=0$, conditioned on $\mathcal{E}_1,\mathcal{E}_2$, we have $Y^{t+1} = Y^{t+2} = X^{t}$ and so $\norm{Y^{t+1}}_1 = \norm{Y^{t+2}}_1 = \norm{X^t}_1 = 1$. By the stability  property of Lemma \ref{lem:stability1A} we thus have:
\begin{align*}
\Pr[Y^{t+3}\text{ is a valid WTA output configuration } |\mathcal{E}_1,\mathcal{E}_2,\All^t=C,\All^{t-1}=C'] \ge  1-ne^{-\gamma/2}.
\end{align*}
Oftherwise, for $l \ge  1$,
by Corollary \ref{cor:convergenceA}, since conditioned on $\mathcal{E}_1$ and $\mathcal{E}_2$, $$\norm{\min(Y^{t+1},Y^{t+2})}_1 = \norm{X^{t}}_1 \in [2^l,2^{l+1})$$ and $a_s^{t+2} = a_1^{t+2} = ... = a_l^{t+2} = 1\text{ and } a_{l+1}^{t+2} = ... = a_{\lognc}^{t+2} = 0$, 
\begin{align*}
\Pr[Y^{t+3}\text{ is a valid WTA output configuration } |\mathcal{E}_1,\mathcal{E}_2,\All^t=C,\All^{t-1}=C'] \ge  \frac{1}{16}-ne^{-2\gamma}.
\end{align*}
We can easily bound the probability of $a_s^{t+3} = a_1^{t+3} = ... = a_l^{t+3} = 1\text{ and } a_{l+1}^{t+3} = ... = a_{\lognc}^{l+3} = 0$ using the same arguments as in \eqref{e2Second} and \eqref{e2Second2}, giving, via a union bound:
\begin{align}\label{e3A}
\Pr [\mathcal{E}_3 | \mathcal{E}_1,\mathcal{E}_2, \All^t = C,\All^{t-1}=C']
&\ge \frac{1}{16}-(n+\lognc +1)\cdot e^{-\gamma/2}.
\end{align}

\medskip
\spara{Step 4:}
\medskip

Finally, let $\mathcal{E}_4$ be the event that $\max(Y^{t+3},Y^{t+4}) = 1$, $a_s^{t+4} = 1$, $\sum_{j=1}^{\lognc} a_j^{t+4} = 0$ and $y_i^{t+4} \le x_i^{t+4}$ for all $i$. We can check via Definition \ref{def:nearStableA}  that if $\mathcal{E}_3$ and $\mathcal{E}_4$ occur, then $(\All^{t+3},\All^{t+4})$ is a near-stable pair.

Since conditioned on $\mathcal{E}_3$, $a_s^{t+3} = a_1^{t+3} = ... = a_l^{t+3} = 1\text{ and } a_{l+1}^{t+3} = ... = a_{\lognc}^{l+3} = 0$, if $l \ge 1$,  by  Lemma \ref{lem:convergenceA} conclusion (1),
\begin{align}\label{t41} 
\Pr [\norm{Y^{t+4}}_1 \le \norm{Y^{t+3}}_1 | \mathcal{E}_1,\mathcal{E}_2,\mathcal{E}_3,\All^t=C,\All^{t-1}=C'] \ge 1-ne^{-\gamma/2}.
\end{align}
Since conditioned on $\mathcal{E}_3$, $\norm{Y^{t+3}}_1 = 1$, this gives $\max(Y^{t+3},Y^{t+4}) = 1$. If $l = 0$,  then we have an identical bound via the stability property of Lemma \ref{lem:stability1A}. 

Again, since conditioned on $\mathcal{E}_3$, $\norm{Y^{t+3}}_1 = 1$, by Lemma \ref{lem:resetA1}, 
\begin{align}\label{t41A}
\Pr [a_s^{t+4} =1 | 
\mathcal{E}_1, \mathcal{E}_2,\mathcal{E}_3, \All^t = C,\All^{t-1}=C'] \ge 1-e^{-\gamma/2}.
\end{align}
By Lemma \ref{lem:resetA2}, this also gives 
\begin{align}\label{t42A}
\Pr \left  [\sum_{j=1}^{\lognc} a_j^{t+4} = 0 \big | 
\mathcal{E}_1, \mathcal{E}_2,\mathcal{E}_3, \All^t = C,\All^{t-1}=C'\right] \ge 1-\lognc \cdot e^{-\gamma/2}.
\end{align}
By a union bound using \eqref{t41},\eqref{t41A}, and \eqref{t42A},
\begin{align}\label{e4A}
\Pr [\mathcal{E}_4 | \mathcal{E}_1,\mathcal{E}_2,\mathcal{E}_3, \All^t = C,\All^{t-1}=C']
&\ge 1-(n+\lognc+1) \cdot e^{-\gamma/2}.
\end{align}

\spara{Completing the proof:}

Let $\mathcal{E}$ be the event that $(\All^{t+3},\All^{t+4})\text{ is near-stable }$. We can 
complete the proof by bounding:
\begin{align*}
\Pr[\mathcal{E} | \All^t = C,\All^{t-1}=C'] &\ge \Pr[\mathcal{E}_3,\mathcal{E}_4 | \All^t = C,\All^{t-1}=C']\\
&\ge \Pr[\mathcal{E}_4  | \mathcal{E}_1,\mathcal{E}_2, \mathcal{E}_3, \All^t = C,\All^{t-1}=C']\\ &\hspace{2em}\cdot \Pr[ \mathcal{E}_3 | \mathcal{E}_1,\mathcal{E}_2,\All^t = C,\All^{t-1}=C']\\ &\hspace{2em}\cdot \Pr[ \mathcal{E}_2 | \mathcal{E}_1,\All^t = C,\All^{t-1}=C'] \\&\hspace{2em}\cdot \Pr[ \mathcal{E}_1 | \All^t = C,\All^{t-1}=C'].
\end{align*}
We can bound the above terms using \eqref{e1A}, \eqref{e2A},\eqref{e3A}, and \eqref{e4A} giving:
\begin{align*}
\Pr[\mathcal{E} | \All^t = C,\All^{t-1}=C'] &\ge \left (1-(n+\lognc) \cdot e^{-\gamma/2}\right ) \cdot \left (1-(n+\lognc+1) \cdot e^{-\gamma/2}\right )\\ &\cdot \left (\frac{1}{16}-(n+\lognc+1) \cdot e^{-\gamma/2}\right ) \cdot \left (1-(n+\lognc+1) \cdot e^{-\gamma/2}\right ) \\
&\ge \frac{1}{16} - 4(n + \lognc +1) \cdot e^{-\gamma/2}\\
&\ge  \frac{1}{16} - 12n\cdot e^{-\gamma/2}.
\end{align*}

\end{proof}

\begin{tcolorbox}
\begin{lemma}[$\norm{Y^t}_1 > 1$, $\norm{A^t}_1 > 1$]\label{lem:converge7A}
Assume  the input execution $\alpha_X$ of $\Netc$ has $X^t$ fixed for all $t$ and that $\norm{X^t}_1 \ge 1$. Consider any pair of configurations $C',C$ with $C(X) = C'(X) = X^t$, $\norm{C(\Output)}_1 > 1$, $\norm{C(A)}_1 > 1$, $C(a_s)=1$, and $C(y_i) \le C(x_i)$, $C'(y_i)\le C'(x_i)$ for all $i$. For any  time $t \ge 1$,
$$\Pr[(\All^{t+i},\All^{t+i+1})\text{ is near-stable for some }i\le 7 | \All^t = C,\All^{t-1}=C'] \ge \frac{1}{128} - 24n\cdot e^{-\gamma/2}.$$
\end{lemma}
\end{tcolorbox}
\begin{proof}
Again the proof follows from a series of steps, arguing about the state of $\Netc$ at times $t+1,t+2,...$. Let  $\mathcal{E}$ be  the event  that  $(\All^{t+i},\All^{t+i+1})\text{ is near-stable}$ for some $i \le 7$. 

\medskip
\spara{Step 1:}
\medskip

Let $\mathcal{E}_1$ be the event that $y_i^{t+1} \le y_i^t$ for all $i$ and that for $l =\left \lfloor \log_2 \left (\norm{Y^t}_1 \right) \right\rfloor$, $a_s^{t+2} = a_1^{t+2} = ... = a_l^{t+2} = 1\text{ and } a_{l+1}^{t+2} = ... = a_{\lognc}^{l+2} = 0$ (note that $l \ge 1$ since $\norm{Y^t}_1 > 1$).
Let $\mathcal{E}_{1,0}$ be the event that $\mathcal{E}_1$  holds and $\norm{Y^{t+1}}_1 = 0$. Let $\mathcal{E}_{1,1}$ be the event that $\mathcal{E}_1$  holds and $\norm{Y^{t+1}}_1 = 1$. Finally, let  $\mathcal{E}_{1,> 1}$ be the event that $\mathcal{E}_1$  holds and $\norm{Y^{t+1}}_1 > 1$.
Applying Lemmas \ref{lem:resetA1},  \ref{lem:resetA2}, and \ref{lem:convergenceA} and a union bound:
\begin{align}\label{eq:case2e1}
\Pr[\mathcal{E}_1 | \All^t = C,\All^{t-1}=C'] \ge 1 - (n+\lognc + 1)\cdot e^{-\gamma/2}.
\end{align}

\medskip
\spara{Step 2:}
\medskip

Let  $\mathcal{E}_2$ be the event that $\norm{Y^{t+2}}_1 = 0$, that $a_s^{t+2} = 1$ if $\norm{Y^{t+1}}_1 > 0$, and that for $l' =\left \lfloor \log_2 \left (\norm{Y^{t+1}}_1 \right) \right\rfloor$, $a_1^{t+2} = ... = a_{l'}^{t+2} = 1\text{ and } a_{l'+1}^{t+2} = ... = a_{\lognc}^{l+2} = 0$. Since, conditioned on $\mathcal{E}_1$, $\norm{Y^{t+1}}_1  \le \norm{Y^t}_1 \le 2^{l+1}$, $a_s^{t+2} = a_1^{t+2} = ... = a_l^{t+2} = 1\text{ and } a_{l+1}^{t+2} = ... = a_{\lognc}^{l+2} = 0$, by Corollary  \ref{cor:convergenceA0} combined with Lemmas \ref{lem:resetA1},  \ref{lem:resetA2}:
\begin{align}\label{eq:case2e2}
\Pr[\mathcal{E}_2 | \mathcal{E}_1,\All^t = C,\All^{t-1}=C'] \ge \frac{1}{8} - (n+\lognc + 1)\cdot e^{-\gamma/2}.
\end{align}
We now write:
\begin{align*}
\Pr[\mathcal{E} |\All^t = C,\All^{t-1}=C']
&\ge \min \big  (\Pr[\mathcal{E} | \mathcal{E}_2,\mathcal{E}_{1,0},\All^t = C,\All^{t-1}=C'],\nonumber\\
&\hspace{3.5em}\Pr[\mathcal{E} | \mathcal{E}_2, \mathcal{E}_{1,1},\All^t = C,\All^{t-1}=C'],\nonumber\\
&\hspace{3.5em} \Pr[\mathcal{E} | \mathcal{E}_2,\mathcal{E}_{1,>1} ,\All^t = C,\All^{t-1}=C']\big )\nonumber\\
&\cdot \Pr[\mathcal{E}_2 | \mathcal{E}_1,\All^t = C,\All^{t-1}=C'] \cdot \Pr[\mathcal{E}_1 | \All^t = C,\All^{t-1}=C']
\end{align*}
Using \eqref{eq:case2e1} and \eqref{eq:case2e2} we can bound the above by:
\begin{align}\label{dunkinSplit}
\Pr[\mathcal{E} |\All^t = C,\All^{t-1}=C']
&\ge \min \big  (\Pr[\mathcal{E} | \mathcal{E}_2,\mathcal{E}_{1,0},\All^t = C,\All^{t-1}=C'],\nonumber\\
&\hspace{3.5em}\Pr[\mathcal{E} | \mathcal{E}_2, \mathcal{E}_{1,1},\All^t = C,\All^{t-1}=C'],\nonumber\\
&\hspace{3.5em} \Pr[\mathcal{E} | \mathcal{E}_2,\mathcal{E}_{1,>1} ,\All^t = C,\All^{t-1}=C']\big ) \cdot \left  ( \frac{1}{8} - 6n \cdot e^{-\gamma/2}\right).
\end{align}
We bound the minimum above by considering each of the three cases separately.

\medskip
\spara{Case 1: $\Pr[\mathcal{E} | \mathcal{E}_2,\mathcal{E}_{1,0},\All^t = C,\All^{t-1}=C']$.}
\medskip

Conditioned on $\mathcal{E}_2$, $\norm{Y^{t+2}}_1 = 0$. So by Lemma \ref{lem:converge1A}, 
\begin{align}\label{dunkin1}
\Pr[(\All^{t+5},\All^{t+6})\text{ is near-stable } | \mathcal{E}_2,\mathcal{E}_{1,0}, \norm{A^{t+2}}_1 = 0,\All^t = C,\All^{t-1}=C'] \ge \frac{1}{16} - 12n\cdot e^{-\gamma/2}.
\end{align}
If $\norm{A^{t+2}} \ge 1$ then, conditioned on $\mathcal{E}_2$, we  must have $a_s^{t+2} = 1$. Let $\mathcal{E}_3$ be  the event  that  $\norm{Y^{t+3}}_1 = \norm{A^{t+3}}_1 = 0$. By Lemmas \ref{lem:resetA1},  \ref{lem:resetA2}, and \ref{lem:stability1A} and the fact that conditioned on $\mathcal{E}_{1,0}$, $\norm{Y^{t+1}}_1  = 0$, 
\begin{align}\label{eqE3dunkin}
\Pr[\mathcal{E}_3 | \mathcal{E}_2,\mathcal{E}_{1,0}, \norm{A^{t+2}}_1 \ge 1,\All^t = C,\All^{t-1}=C'] \ge 1-(n + \lognc + 1)\cdot e^{-\gamma/2}
\end{align}
Again by Lemma \ref{lem:converge1A}, 
$$\Pr[(\All^{t+6},\All^{t+7})\text{ is near-stable } | \mathcal{E}_3,\mathcal{E}_2,\mathcal{E}_{1,0}, \norm{A^{t+2}}_1 = 0,\All^t = C,\All^{t-1}=C'] \ge \frac{1}{16} - 12n\cdot e^{-\gamma/2}.$$
Combined with \eqref{eqE3dunkin} this gives
\begin{align*}
\Pr[(\All^{t+6},\All^{t+7})\text{ is near-stable } | \mathcal{E}_2,\mathcal{E}_{1,0}, \norm{A^{t+2}}_1 \ge 1,\All^t = C,\All^{t-1}=C'] \ge \frac{1}{16}-15n\cdot e^{-\gamma/2}.
\end{align*}
By  the law of total probability, combined with \eqref{dunkin1} we have 
\begin{align}\label{dunkinMin1}
\Pr[\mathcal{E} | \mathcal{E}_2,\mathcal{E}_{1,0}, \All^t = C,\All^{t-1}=C'] \ge \frac{1}{16} - 15n\cdot e^{-\gamma/2}
\end{align}
which completes this case.

\medskip
\spara{Case 2: $\Pr[\mathcal{E} | \mathcal{E}_2,\mathcal{E}_{1,1},\All^t = C,\All^{t-1}=C']$.}
\medskip

In this case, conditioned on $\mathcal{E}_{1,1}$  and $\mathcal{E}_2$, $a_s^{t+1} = a_s^{t+2} = 1$, $a_s$ is the only inhibitor that fires at time $t+2$, $\norm{Y^{t+1}}_1 = 1$, and $\norm{Y^{t+2}}_1 = 0$. Thus, $(\All^{t+1},\All^{t+2})$ is a near-stable pair of configurations, and so vacuously,
\begin{align}\label{dunkinMin2}
\Pr[\mathcal{E} | \mathcal{E}_2,\mathcal{E}_{1,0}, \All^t = C,\All^{t-1}=C'] = 1
\end{align}
which completes this case.

\medskip
\spara{Case 3: $\Pr[\mathcal{E} | \mathcal{E}_2,\mathcal{E}_{1,>1},\All^t = C,\All^{t-1}=C']$.}
\medskip

In this case, conditioned on $\mathcal{E}_{1,>1}$ and $\mathcal{E}_2$, $\norm{A^{t+2}}_1  > 1$  and $a_s^{t+2} = 1$. Let $\mathcal{E}_3$ be the event that $\norm{Y^{t+3}}_1 = 0$ and that if $\norm{A^{t+3}}_1 \ge 1$, $a_s^{t+3} = 1$. By Lemmas \ref{lem:resetA1},  \ref{lem:resetA2}, and \ref{lem:convergenceA}, 
\begin{align}\label{eqE3dunkin11}
\Pr[\mathcal{E}_3 | \mathcal{E}_2,\mathcal{E}_{1,> 1},\All^t = C,\All^{t-1}=C'] \ge 1-(n +\lognc+1) e^{-\gamma/2}.
\end{align}
In the case that $\norm{A^{t+3}}_1 = 0$, we can again apply  Lemma \ref{lem:converge1A} to give:
\begin{align}\label{dunkinTired}\Pr[(\All^{t+6},\All^{t+7})\text{ is near-stable } | \mathcal{E}_3,\norm{A^{t+3}}_1 = 0,\mathcal{E}_2,\mathcal{E}_{1,> 1},&\All^t = C,\All^{t-1}=C']\nonumber\\
&\ge \frac{1}{16} - 12n\cdot e^{-\gamma/2}.
\end{align}
In the case that $\norm{A^{t+3}}_1 \ge 1$, let $\mathcal{E}_4$ be the event that $\norm{Y^{t+4}}_1  = \norm{A^{t+4}} = 0$. We have by Lemmas  \ref{lem:resetA1},  \ref{lem:resetA2} and \ref{lem:stability1A},
\begin{align}\label{eqE3dunkin1}
\Pr[\mathcal{E}_4 | \mathcal{E}_3,\norm{A^{t+3}}_1 \ge 1, \mathcal{E}_2,\mathcal{E}_{1,> 1},\All^t = C,\All^{t-1}=C'] \ge 1-(n +\lognc+1) e^{-\gamma/2}.
\end{align}
Further, again by  Lemma \ref{lem:converge1A} we have:
\begin{align*}
\Pr[(\All^{t+7},\All^{t+8})\text{ is near-stable } | \mathcal{E}_4,\mathcal{E}_3,\norm{A^{t+3}}_1 = 0,\mathcal{E}_2,\mathcal{E}_{1,> 1},&\All^t = C,\All^{t-1}=C']\\
&\ge \frac{1}{16} - 12n\cdot e^{-\gamma/2}.
\end{align*}
Combined with \eqref{eqE3dunkin1} this gives:
\begin{align*}
\Pr[(\All^{t+7},\All^{t+8})\text{ is near-stable } | \mathcal{E}_3,\norm{A^{t+3}}_1 = 0,\mathcal{E}_2,\mathcal{E}_{1,> 1},&\All^t = C,\All^{t-1}=C']\\ 
&\ge \frac{1}{16} - 15n\cdot e^{-\gamma/2}.
\end{align*}
Further, combined with \eqref{dunkinTired}, by the law of total probability,
\begin{align}\label{dunkinDone}
\Pr[\mathcal{E} | \mathcal{E}_3, \mathcal{E}_2,\mathcal{E}_{1,>1}, \All^t = C,\All^{t-1}=C'] \ge \frac{1}{16} - 15n\cdot e^{-\gamma/2}.
\end{align}
Finally, combining \eqref{dunkinDone} with \eqref{eqE3dunkin11} we have:
\begin{align}\label{dunkinMin3}
\Pr[\mathcal{E} | \mathcal{E}_2,\mathcal{E}_{1,>1}, \All^t = C,\All^{t-1}=C'] \ge \frac{1}{16} - 18 n\cdot e^{-\gamma/2}
\end{align}
completing this case.

\medskip
\spara{Completing the proof.}
\medskip

Using  equations \eqref{dunkinMin1}, \eqref{dunkinMin2}, and \eqref{dunkinMin3} in the three cases above along \eqref{dunkinSplit}:
\begin{align*}
\Pr[\mathcal{E} |\All^t = C,\All^{t-1}=C']
&\ge \min \big  (\Pr[\mathcal{E} | \mathcal{E}_2,\mathcal{E}_{1,0},\All^t = C,\All^{t-1}=C'],\nonumber\\
&\hspace{3.5em}\Pr[\mathcal{E} | \mathcal{E}_2, \mathcal{E}_{1,1},\All^t = C,\All^{t-1}=C'],\nonumber\\
&\hspace{3.5em} \Pr[\mathcal{E} | \mathcal{E}_2,\mathcal{E}_{1,>1} ,\All^t = C,\All^{t-1}=C']\big ) \cdot \left  ( \frac{1}{8} - 6n \cdot e^{-\gamma/2}\right)\\
&\ge  \left  ( \frac{1}{16} - 18n \cdot e^{-\gamma/2}\right) \cdot  \left  ( \frac{1}{8} - 6n \cdot e^{-\gamma/2}\right)\\
&\ge \frac{1}{128} - 24n\cdot e^{-\gamma/2},
\end{align*}
completing  the lemma.
\end{proof}

Using Lemma \ref{lem:converge7A} it is not hard to complete the cases when $\norm{Y^t}_1 > 1$ and $\norm{A^t}_1 \le 1$.

\begin{tcolorbox}
\begin{lemma}[$\norm{Y^t}_1 > 1$, $\norm{A^t}_1 = 1$]\label{lem:converge8A}
Assume  the input execution $\alpha_X$ of $\Netc$ has $X^t$ fixed for all $t$ and that $\norm{X^t}_1 \ge 1$. Consider any pair of configurations $C',C$ with $C(X)=C'(X)=X^t$, $\norm{C(\Output)}_1 > 1$, $\norm{C(A)}_1 = 1$, $C(a_s)=1$, and $C(y_i) \le C(x_i)$, $C'(y_i)\le C'(x_i)$ for all $i$. For any  time $t \ge 1$,
$$\Pr[(\All^{t+i},\All^{t+i+1})\text{ is near-stable for some }i\le 8 | \All^t = C,\All^{t-1}=C'] \ge \frac{1}{128} - 28n\cdot e^{-\gamma/2}.$$
\end{lemma}
\end{tcolorbox}
\begin{proof}
Let $\mathcal{E}$ be the event that $(\All^{t+i},\All^{t+i+1})\text{ is near-stable for some }i\le 8$.
Let $\mathcal{E}_1$ be the event that $y_i^{t+1} \ge y_i^t$ and $y_i^{t+1} \le x_i^{t+1}$ for all $i$, that $a_s^{t+1} = 1$, and that $\norm{A^{t+1}}_1 > 1$.
By Corollary \ref{cor:monoA} and Lemmas \ref{lem:resetA1},  \ref{lem:resetA2}, and \ref{lem:stability1A}, we have:
\begin{align*}
\Pr[\mathcal{E}_1 | \All^t = C,\All^{t-1}=C'] \ge 1 - (2n + \lognc + 1) e^{-\gamma/2}.
\end{align*}
Further, by Lemma \ref{lem:converge8A} since conditioned on $\mathcal{E}_1$,  $\norm{Y^{t+1}}_1 >1$  and  $\norm{A^{t+1}}_1 >1$ with $a_s^{t+1}  = 1$:
$$\Pr[\mathcal{E} | \mathcal{E}_1, \All^t = C,\All^{t-1}=C'] \ge \frac{1}{128} - 24n\cdot e^{-\gamma/2}.$$
This gives the lemma since:
\begin{align*}
\Pr[\mathcal{E} | \All^t = C,\All^{t-1}=C']  &\ge\Pr[\mathcal{E} | \mathcal{E}_1, \All^t = C,\All^{t-1}=C'] \cdot  \Pr[\mathcal{E}_1 | \All^t = C,\All^{t-1}=C']\\
&\ge \frac{1}{128} - 28n\cdot e^{-\gamma/2}.
\end{align*}
\end{proof}

\begin{tcolorbox}
\begin{lemma}[$\norm{Y^t}_1 > 1$, $\norm{A^t}_1 = 0$]\label{lem:converge9A}
Assume  the input execution $\alpha_X$ of $\Netc$ has $X^t$ fixed for all $t$ and that $\norm{X^t}_1 \ge 1$. Consider any pair of configurations $C',C$ with $C(X)=C'(X)=X^t$, $\norm{C(\Output)}_1 > 1$, $\norm{C(A)}_1 = 0$, and $C(y_i) \le C(x_i)$, $C'(y_i)\le C'(x_i)$ for all $i$. For any  time $t \ge 1$,
$$\Pr[(\All^{t+i},\All^{t+i+1})\text{ is near-stable for some }i\le 8 | \All^t = C,\All^{t-1}=C'] \ge \frac{1}{128} - 27n\cdot e^{-\gamma/2}.$$
\end{lemma}
\end{tcolorbox}
\begin{proof}
Let $\mathcal{E}$ be the event that $(\All^{t+i},\All^{t+i+1})\text{ is near-stable for some }i\le 8$.
Let $\mathcal{E}_1$ be the event that $y_i^{t+1} = x_i^{t+1}$ for all $i$, that $a_s^{t+1} = 1$, and that $\norm{A^{t+1}}_1 > 1$. Note that since $\norm{C(Y)}_1 > 1$ and $C(y_i)\le C(x_i)$ for all $i$, $\mathcal{E}_1$ implies that $\norm{Y^{t+1}}_1 = \norm{X^{t+1}}_1 > 1$.
By Lemmas \ref{lem:resetA1},  \ref{lem:resetA2}, and \ref{lem:quietA}:
\begin{align*}
\Pr[\mathcal{E}_1 | \All^t = C,\All^{t-1}=C'] \ge 1 - (n + \lognc + 1) e^{-\gamma/2}.
\end{align*}
Further, by Lemma \ref{lem:converge8A}  since conditioned on $\mathcal{E}_1$,  $\norm{Y^{t+1}}_1 >1$  and  $\norm{A^{t+1}}_1 >1$ with $a_s^{t+1}  = 1$:
$$\Pr[\mathcal{E} | \mathcal{E}_1, \All^t = C,\All^{t-1}=C'] \ge \frac{1}{128} - 24n\cdot e^{-\gamma/2}.$$
This gives the lemma since:
\begin{align*}
\Pr[\mathcal{E} | \All^t = C,\All^{t-1}=C']  &\ge\Pr[\mathcal{E} | \mathcal{E}_1, \All^t = C,\All^{t-1}=C'] \cdot  \Pr[\mathcal{E}_1 | \All^t = C,\All^{t-1}=C']\\
&\ge \frac{1}{128} - 27n\cdot e^{-\gamma/2}.
\end{align*}
\end{proof}
We next complete the remaining cases when $\norm{Y^t}_1 \le 1$.

\begin{tcolorbox}
\begin{lemma}[$\norm{Y^t}_1= 1$, $\norm{A^t}_1 = 0$]\label{lem:converge10A}
Assume  the input execution $\alpha_X$ of $\Netc$ has $X^t$ fixed for all $t$ and that $\norm{X^t}_1 \ge 1$. Consider any pair of configurations $C',C$ with $C(X)=C'(X)=X^t$, $\norm{C(\Output)}_1 = 1$, $\norm{C(A)}_1 = 0$, and $C(y_i) \le C(x_i)$, $C'(y_i)\le C'(x_i)$ for all $i$. For any  time $t \ge 1$,
$$\Pr[(\All^{t+3},\All^{t+4})\text{ is near-stable } | \All^t = C,\All^{t-1}=C'] \ge \frac{1}{16} - 12n\cdot e^{-\gamma/2}.$$
\end{lemma}
\end{tcolorbox}
\begin{proof}

The proof follows from a series of four steps, arguing about the state of $\Netc$ at times $t+1,t+2,t+3,t+4$. The analysis closely mirrors that of Lemma \ref{lem:converge1A}, for the case when $\norm{Y^t}_1  = 0$ and $\norm{A^t}_1  = 0$.

\medskip
\spara{Step 1:}
\medskip

Let $\mathcal{E}_1$ be the event that $y_i^{t+1} = x_i^{t+1} \text{ for all } i$, that $a_i^{t+1} = 0$ for all $i \in \{1,...,\lognc \}$, and that $a_s^{t+1} = 1$. By Lemma \ref{lem:quietA}, since $\norm{C(A)}_1  = 0$, 
$$\Pr [y_i^{t+1} = x_i^{t+1} \text{ for all } i| \All^t = C,\All^{t-1}=C'] \ge 1-ne^{-\gamma/2}.$$
Additionally, since $\norm{C(Y)}_1 = 1$, by Lemma \ref{lem:resetA2} conclusion (1), 
$$\Pr [a_i^{t+1} = 0\text{ for all }i \in \{1,...,\lognc \}| \All^t = C,\All^{t-1}=C'] \ge 1-\lognc \cdot e^{-\gamma/2}.$$
Finally, by Lemma \ref{lem:resetA1}, 
$$\Pr [a_s^{t+1} = 1| \All^t = C,\All^{t-1}=C'] \ge 1- e^{-\gamma/2}.$$
Thus, by a union bound we have:
\begin{align}\label{e1A4}
\Pr [\mathcal{E}_1 | \All^t = C,\All^{t-1}=C'] \ge 1-(n+\lognc+1)\cdot e^{-\gamma/2}.
\end{align}

\medskip
\spara{Step 2:}
\medskip

Let $\mathcal{E}_2$ be the event that $y_i^{t+2} = x_i^{t+2} \text{ for all } i$ and that for $l =\left \lfloor \log_2 \left (\norm{X^t}_1 \right) \right\rfloor$, $a_s^{t+2} = a_1^{t+2} = ... = a_l^{t+2} = 1\text{ and } a_{l+1}^{t+2} = ... = a_{\lognc}^{l+2} = 0$ (if $l=0$, just $a_s^{t+2} = 1$). Conditioned on $\mathcal{E}_1$, the only inhibitor that fires at time $t+1$ is $a_s$. We can separately  consider the cases when $a_s^{t+1} =0$ and when $a_s^{t+1} = 1$. By Lemma \ref{lem:stability1A} and the fact that $y_i^{t+1} = x_i^{t+1}$ for all $i$,
\begin{align}\label{e2First4}
\Pr \left [y_i^{t+2} = x_i^{t+2} \text{ for all } i | \mathcal{E}_1, \All^t = C,\All^{t-1}=C'\right] \ge 1-ne^{-\gamma/2}.
\end{align}
We also apply Lemma \ref{lem:resetA2}. Conditioned on $\mathcal{E}_1$, for $l =\left \lfloor \log_2 \left (\norm{X^t}_1 \right) \right\rfloor$, we have $\norm{\Output^{t+1}}_1 = \norm{\Input^t}_1 \in [2^{l}, 2^{l+1})$, which gives that,
\begin{align}\label{e2Second4}
\Pr[a_1^{t+2} = ... = a_i^{t+2} = 1\text{ and } a_{i+1}^{t+2} = ... = a_{\lognc}^{t+2} = 0 | \mathcal{E}_1, &\All^t = C,\All^{t-1}=C']\nonumber \\
&\ge 1-\lognc \cdot e^{-\gamma/2}.
\end{align}
Similarly, applying Lemma \ref{lem:resetA1}, since conditioned on $\mathcal{E}_1$, $\norm{Y^{t+1}}_1 = \norm{X^{t+1}}_1 \ge 1$:
\begin{align}\label{e2Second24}
\Pr[a_s^{t+2} = 1 | \mathcal{E}_1, \All^t = C,\All^{t-1}=C'] \ge 1-e^{-\gamma/2}.
\end{align}
Combining \eqref{e2First4}, \eqref{e2Second4}, and \eqref{e2Second24} we have:
\begin{align}\label{e2A4}
\Pr [\mathcal{E}_2 | \mathcal{E}_1, \All^t = C,\All^{t-1}=C'] 
&\ge 1-(n+\lognc + 1)\cdot e^{-\gamma/2}.
\end{align}

\medskip
\spara{Step 3:}
\medskip

Let  $\mathcal{E}_3$ be the event that $Y^{t+3}$ is a valid WTA configuration, and that for $l =\left \lfloor \log_2 \left (\norm{X^t}_1 \right) \right\rfloor$, $a_s^{t+3} = a_1^{t+3} = ... = a_l^{t+3} = 1\text{ and } a_{l+1}^{t+3} = ... = a_{\lognc}^{l+3} = 0$. 

If  $l=0$, conditioned on $\mathcal{E}_1,\mathcal{E}_2$, we have $\norm{Y^{t+1}}_1 = \norm{Y^{t+2}}_1 = 1$ and $Y^{t+1} = Y^{t+2} = X^{t}$. By the stability  property of Lemma \ref{lem:stability1A} we thus have:
\begin{align*}
\Pr[Y^{t+3}\text{ is a valid WTA output configuration } |\mathcal{E}_1,\mathcal{E}_2,\All^t=C,\All^{t-1}=C'] \ge  1-ne^{-\gamma/2}.
\end{align*}
Oftherwise, for $l \ge  1$,
by Corollary \ref{cor:convergenceA}, since conditioned on $\mathcal{E}_1$ and $\mathcal{E}_2$, 
$$\norm{\min(Y^{t+1},Y^{t+2})}_1 = \norm{X^{t}}_1 \in [2^l,2^{l+1})$$
and $a_s^{t+2} = a_1^{t+2} = ... = a_l^{t+2} = 1\text{ and } a_{l+1}^{t+2} = ... = a_{\lognc}^{t+2} = 0$, 
\begin{align*}
\Pr[Y^{t+3}\text{ is a valid WTA output configuration } |\mathcal{E}_1,\mathcal{E}_2,\All^t=C,\All^{t-1}=C'] \ge  \frac{1}{16}-ne^{-2\gamma}.
\end{align*}
We can easily bound the probability of $a_s^{t+3} = a_1^{t+3} = ... = a_l^{t+3} = 1\text{ and } a_{l+1}^{t+3} = ... = a_{\lognc}^{l+3} = 0$ using the same arguments as in \eqref{e2Second4} and \eqref{e2Second24}, giving, via a union bound:
\begin{align}\label{e3A4}
\Pr [\mathcal{E}_3 | \mathcal{E}_1,\mathcal{E}_2, \All^t = C,\All^{t-1}=C']
&\ge \frac{1}{16}-(n+\lognc +1)\cdot e^{-\gamma/2}.
\end{align}

\medskip
\spara{Step 4:}
\medskip

Finally, let $\mathcal{E}_4$ be the event that $\max(Y^{t+3},Y^{t+4}) = 1$, $a_s^{t+4} = 1$, $\sum_{j=1}^{\lognc} a_j^{t+4} = 0$ and $y_i^{t+4} \le x_i^{t+4}$ for all $i$. We can check via Definition \ref{def:nearStableA}  that if $\mathcal{E}_3$ and $\mathcal{E}_4$ occur, then $(\All^{t+3},\All^{t+4})$ is a near-stable pair.

Since conditioned on $\mathcal{E}_3$, $a_s^{t+3} = a_1^{t+3} = ... = a_l^{t+3} = 1\text{ and } a_{l+1}^{t+3} = ... = a_{\lognc}^{l+3} = 0$, if $l \ge 1$,  by  Lemma \ref{lem:convergenceA} conclusion (1),
\begin{align}\label{t414} 
\Pr [\norm{Y^{t+4}}_1 \le \norm{Y^{t+3}}_1 | \mathcal{E}_1,\mathcal{E}_2,\mathcal{E}_3,\All^t=C,\All^{t-1}=C'] \ge 1-ne^{-\gamma/2}.
\end{align}
Since conditioned on $\mathcal{E}_3$, $\norm{Y^{t+3}}_1 = 1$, this gives $\max(Y^{t+3},Y^{t+4}) = 1$. If $l = 0$,  then we have an identical bound via the stability property of Lemma \ref{lem:stability1A}. 

Again, since conditioned on $\mathcal{E}_3$, $\norm{Y^{t+3}}_1 = 1$, by Lemma \ref{lem:resetA1}, 
\begin{align}\label{t41A4}
\Pr [a_s^{t+4} =1 | 
\mathcal{E}_1, \mathcal{E}_2,\mathcal{E}_3, \All^t = C,\All^{t-1}=C'] \ge 1-e^{-\gamma/2}.
\end{align}
By Lemma \ref{lem:resetA2}, this also gives 
\begin{align}\label{t42A4}
\Pr \left  [\sum_{j=1}^{\lognc} a_j^{t+4} = 0 \big | 
\mathcal{E}_1, \mathcal{E}_2,\mathcal{E}_3, \All^t = C,\All^{t-1}=C'\right] \ge 1-\lognc \cdot e^{-\gamma/2}.
\end{align}
By a union bound using \eqref{t414},\eqref{t41A4}, and \eqref{t42A4},
\begin{align}\label{e4A4}
\Pr [\mathcal{E}_4 | \mathcal{E}_1,\mathcal{E}_2,\mathcal{E}_3, \All^t = C,\All^{t-1}=C']
&\ge 1-(n+\lognc+1) \cdot e^{-\gamma/2}.
\end{align}

\spara{Completing the proof:}

Let $\mathcal{E}$ be the event that $(\All^{t+3},\All^{t+4})\text{ is near-stable }$. We can 
complete the proof by bounding:
\begin{align*}
\Pr[\mathcal{E} | \All^t = C,\All^{t-1}=C'] &\ge \Pr[\mathcal{E}_3,\mathcal{E}_4 | \All^t = C,\All^{t-1}=C']\\
&\ge \Pr[\mathcal{E}_4  | \mathcal{E}_1,\mathcal{E}_2, \mathcal{E}_3, \All^t = C,\All^{t-1}=C']\\ &\hspace{2em}\cdot \Pr[ \mathcal{E}_3 | \mathcal{E}_1,\mathcal{E}_2,\All^t = C,\All^{t-1}=C']\\ &\hspace{2em}\cdot \Pr[ \mathcal{E}_2 | \mathcal{E}_1,\All^t = C,\All^{t-1}=C'] \\&\hspace{2em}\cdot \Pr[ \mathcal{E}_1 | \All^t = C,\All^{t-1}=C']
\end{align*}
We can bound the above terms using \eqref{e1A4}, \eqref{e2A4},\eqref{e3A4}, and \eqref{e4A4} giving:
\begin{align*}
\Pr[\mathcal{E} | \All^t = C,\All^{t-1}=C'] &\ge \left (1-(n+\lognc+1) \cdot e^{-\gamma/2}\right ) \cdot \left (1-(n+\lognc+1) \cdot e^{-\gamma/2}\right )\\ &\hspace{.5em}\cdot \left (\frac{1}{16}-(n+\lognc+1) \cdot e^{-\gamma/2}\right ) \cdot \left (1-(n+\lognc+1) \cdot e^{-\gamma/2}\right ) \\
&\ge \frac{1}{16} - 4(n + \lognc +1) \cdot e^{-\gamma/2}\\
&\ge  \frac{1}{16} - 12n\cdot e^{-\gamma/2}.
\end{align*}

\end{proof}

The remaining cases follow relatively straightforwardly from the previous lemmas. 
\begin{tcolorbox}
\begin{lemma}[$\norm{Y^t}_1= 1$, $\norm{A^t}_1 = 1$]\label{lem:converge11A}
Assume  the input execution $\alpha_X$ of $\Netc$ has $X^t$ fixed for all $t$ and that $\norm{X^t}_1 \ge 1$. Consider any pair of configurations $C',C$ with $C(X)=C'(X)=X^t$, $\norm{C(\Output)}_1 = 1$, $\norm{C(A)}_1 = 1$, $C(a_s)=1$, and $C(y_i) \le C(x_i)$, $C'(y_i)\le C'(x_i)$ for all $i$. For any  time $t \ge 1$,
$$\Pr[(\All^{t+i},\All^{t+i+1})\text{ is near-stable for some }i\le 9 | \All^t = C,\All^{t-1}=C'] \ge \frac{1}{128} - 28n\cdot e^{-\gamma/2}.$$
\end{lemma}
\end{tcolorbox}
\begin{proof}
Let $\mathcal{E}$ be the event that $(\All^{t+i},\All^{t+i+1})\text{ is near-stable for some }i\le 9$.
Let $\mathcal{E}_1$ be the event that $y_i^{t+1} \ge y_i^t$ and $y_i^{t+1} \le x_i^{t+1}$ for all $i$, that $a_s^{t+1} = 1$, and that $\norm{A^{t+1}}_1 > 1$.
By Corollary \ref{cor:monoA} and Lemmas \ref{lem:resetA1},  \ref{lem:resetA2}, and \ref{lem:stability1A}, we have:
\begin{align}\label{eq:tortilla}
\Pr[\mathcal{E}_1 | \All^t = C,\All^{t-1}=C'] \ge 1 - (2n + \lognc + 1) e^{-\gamma/2}.
\end{align}
We consider two cases, dependent on the number of firing outputs at time $t+1$. Conditioned on $\mathcal{E}_1$, $\norm{Y^{t+1}}_1 \ge \norm{Y^{t}}_1 = 1$.
By Lemma \ref{lem:converge8A} since conditioned on $\mathcal{E}_1$, $\norm{A^{t+1}}_1 > 1$ with $a_s^{t+1} = 1$:
$$\Pr[\mathcal{E} | \mathcal{E}_1, \norm{Y^{t+1}} > 1, \All^t = C,\All^{t-1}=C'] \ge \frac{1}{128} - 24n\cdot e^{-\gamma/2}.$$
Alternatively, if $\norm{Y^{t+1}} = 1$, then $(\All^t,\All^{t+1})$ is already a near-stable pair of configurations (Definition \ref{def:nearStableA}) so we vacuously have 
$$\Pr[\mathcal{E} | \mathcal{E}_1, \norm{Y^{t+1}} = 1, \All^t = C,\All^{t-1}=C'] =1.$$
By  the law of total probability  this gives:
$$\Pr[\mathcal{E} | \mathcal{E}_1 \All^t = C,\All^{t-1}=C'] \ge \frac{1}{128} - 24n\cdot e^{-\gamma/2}.$$
Finally, we obtain the lemma by combining this bound with \eqref{eq:tortilla}.
\end{proof}

\begin{tcolorbox}
\begin{lemma}[$\norm{Y^t}_1= 1$, $\norm{A^t}_1 > 1$]\label{lem:converge6A}
Assume  the input execution $\alpha_X$ of $\Netc$ has $X^t$ fixed for all $t$ and that $\norm{X^t}_1 \ge 1$. Consider any pair of configurations $C',C$ with $C(X)=C'(X) = X^t$, $\norm{C(\Output)}_1 = 1$, $\norm{C(A)}_1 > 1$, $C(a_s)=1$, and $C(y_i) \le C(x_i)$, $C'(y_i)\le C'(x_i)$ for all $i$. For any  time $t \ge 1$,
$$\Pr[(\All^{t},\All^{t+1})\text{ is near-stable } | \All^t = C,\All^{t-1}=C'] \ge 1 - 3n\cdot e^{-\gamma/2}.$$
\end{lemma}
\end{tcolorbox}
\begin{proof}
Let $\mathcal{E}_1$ be the event that $\norm{A^{t+1}}_1 = 1\text{ and }a_s^{t+1} = 1$. Again Lemmas \ref{lem:resetA1} and \ref{lem:resetA2} we have:
\begin{align*}
\Pr[\mathcal{E}_1 | \All^t = C,\All^{t-1}=C'] \ge 1-(\lognc + 1)\cdot e^{-\gamma/2}.
\end{align*}
Let $\mathcal{E}_2$ be the event that $\norm{Y^{t+1}} \le 1$ and that $y_i^{t+1}\le x_i^{t+1}$ for all $i$. By  Lemma \ref{lem:convergenceA}, 
\begin{align*}
\Pr[\mathcal{E}_2 | \All^t = C,\All^{t-1}=C'] \ge 1-n e^{-\gamma/2}.
\end{align*}
If $\mathcal{E}_1$ and $\mathcal{E}_2$ both occur, $(\All^t,\All^t+1)$ is a near-stable pair of configurations, giving the lemma by a union bound and the fact that $(n+\lognc + 1) \le 3n$.

\end{proof}

\begin{tcolorbox}
\begin{lemma}[$\norm{Y^t}_1= 0$, $\norm{A^t}_1 = 1$]\label{lem:converge01A}
Assume  the input execution $\alpha_X$ of $\Netc$ has $X^t$ fixed for all $t$ and that $\norm{X^t}_1 \ge 1$. Consider any pair of configurations $C',C$ with $C(X)=C'(X)=X^t$, $\norm{C(\Output)}_1 = 0$, $\norm{C(A)}_1 = 1$, $C(a_s) = 1$, and $C(y_i) \le C(x_i)$, $C'(y_i)\le C'(x_i)$ for all $i$. For any  time $t \ge 1$,
$$\hspace{-.2em}\Pr[(\All^{t+i},\All^{t+i+1})\text{ is near-stable for some }i\le 10 | \All^t = C,\All^{t-1}=C'] \ge \frac{1}{128} - 31n e^{-\gamma/2}.$$
\end{lemma}
\end{tcolorbox}
\begin{proof}
Let $\mathcal{E}$ be the event that $(\All^{t+i},\All^{t+i+1})\text{ is near-stable for some }i\le 10$.
Let $\mathcal{E}_1$ be the event that $a_s^{t+1} = 1$ if $\norm{Y^{t-1}}_1 \ge 1$ and $0$ otherwise, that $a_i^{t+1} =0\text{ for all }i \in \{1,...,\lognc\}$, and that $y_i^{t+1} = \max( y_i^{t-1},y_i^t)$ for all $i$. 
By Lemmas  \ref{lem:resetA1}, \ref{lem:resetA2}, and \ref{lem:stability1A}
\begin{align}\label{noNormalFire}
\Pr[\mathcal{E}_1|\All^t = C,\All^{t-1}=C'] \ge 1-(n + \lognc+1) \cdot e^{-\gamma/2}
\end{align}
We next consider two cases, based off the number of firing outputs at time $t-1$.

\medskip
\spara{Case 1: $\norm{C'(Y)}_1 = 0$.}
\medskip

In this case, conditioned on $\mathcal{E}_1$, $\norm{Y^{t+1}}_1 = \norm{A^{t+1}}_1 = 0$. Thus, by Lemma \ref{lem:converge1A}, 
\begin{align}\label{tort1}
\Pr[(\All^{t+4},\All^{t+5})\text{ is near-stable }| \mathcal{E}_1, \All^t = C,\All^{t-1}=C'] \ge \frac{1}{16} - 12n\cdot e^{-\gamma/2}
\end{align}

\medskip
\spara{Case 2: $\norm{C'(Y)}_1 \ge 1$.}
\medskip

In this case, conditioned on $\mathcal{E}_1$, $a_s^{t+1} = 1$, $\norm{A^{t+1}}_1 = 1$ and $\norm{Y^{t+1}}_1 \ge 1$. Thus by  Lemmas \ref{lem:converge8A} and \ref{lem:converge11A},
\begin{align}\label{tort2}
\Pr[\mathcal{E}| \mathcal{E}_1, \All^t = C,\All^{t-1}=C'] \ge \frac{1}{128} - 28n\cdot e^{-\gamma/2}
\end{align}
Overall by  \eqref{tort1} and \eqref{tort2} we have 
\begin{align*}
\Pr[\mathcal{E}| \mathcal{E}_1, \All^t = C,\All^{t-1}=C'] \ge \frac{1}{128} - 28n\cdot e^{-\gamma/2},
\end{align*}
which gives the lemma when combined with \eqref{noNormalFire}.
\end{proof}

\begin{tcolorbox}
\begin{lemma}[$\norm{Y^t}_1= 0$, $\norm{A^t}_1 > 1$]\label{lem:converge3A}
Assume  the input execution $\alpha_X$ of $\Netc$ has $X^t$ fixed for all $t$ and that $\norm{X^t}_1 \ge 1$. Consider any pair of configurations $C',C$ with $C(X)=C'(X)=X^t$, $\norm{C(\Output)}_1 = 0$, $\norm{C(A)}_1 > 1$, $C(a_s) = 1$, and $C(y_i) \le C(x_i)$, $C'(y_i)\le C'(x_i)$ for all $i$. For any  time $t \ge 1$,
$$\hspace{-.2em}\Pr[(\All^{t+i},\All^{t+i+1})\text{ is near-stable for some }i\le 11 | \All^t = C,\All^{t-1}=C'] \ge \frac{1}{128} - 33n e^{-\gamma/2}.$$
\end{lemma}
\end{tcolorbox}
\begin{proof}
Let $\mathcal{E}$ be the event that $(\All^{t+i},\All^{t+i+1})\text{ is near-stable for some }i\le 11$.
Let $\mathcal{E}_1$ be the event that $a_i^{t+1} =0\text{ for all }i \in \{1,...,\lognc\}$, and that $\norm{Y^{t+1}}_1 = 0$. By  Lemmas \ref{lem:resetA2} and \ref{lem:convergenceA}, 
\begin{align}\label{noNormalFire2}
\Pr[\mathcal{E}_1|\All^t = C,\All^{t-1}=C'] \ge 1-(n + \lognc) \cdot e^{-\gamma/2}.
\end{align}
Further, by  Lemmas \ref{lem:converge1A} and \ref{lem:converge01A},
\begin{align*}
\Pr[\mathcal{E}|\mathcal{E}_1,\All^t = C,\All^{t-1}=C'] \ge \frac{1}{128}-31n \cdot e^{-\gamma/2}.
\end{align*}
This gives the lemma when combined with \eqref{noNormalFire2}.
\end{proof}

\subsection{Completing the Analysis}\label{sec:completingA}

By combining the nine cases of Lemmas \ref{lem:converge1A}, \ref{lem:converge7A}, \ref{lem:converge8A}, \ref{lem:converge9A}, \ref{lem:converge10A}, \ref{lem:converge11A}, \ref{lem:converge6A}, \ref{lem:converge01A}, and \ref{lem:converge3A} we can conclude the following general statement:
\begin{tcolorbox}
\begin{lemma}[$O(1)$ Step Convergence From Typical Configurations]\label{lem:convergeGeneralA}
Assume  the input execution $\alpha_X$ of $\Netc$ has $X^t$ fixed for all $t$ and that $\norm{X^t}_1 \ge 1$. Consider any pair of typical configurations (Definition \ref{def:typical}) $C',C$ with $C(X)=C'(X)=X^t$. For any  time $t \ge 1$,
$$\hspace{-.2em}\Pr[(\All^{t+i},\All^{t+i+1})\text{ is near-stable for some }i\le 11 | \All^t = C,\All^{t-1}=C'] \ge \frac{1}{128} - 33n e^{-\gamma/2}.$$
\end{lemma}
\end{tcolorbox}
\begin{proof} The nine cases of the above listed lemmas cover all possible pairs of typical configurations $C,C'$. So the lemma follows immediately.
\end{proof}
From Lemma \ref{lem:convergeGeneralA} we can conclude an even more general result:
\begin{tcolorbox}
\begin{lemma}[General $O(1)$ Step Convergence]\label{lem:convergeGeneralA1}
Assume  the input execution $\alpha_X$ of $\Netc$ has $X^t$ fixed for all $t$ and that $\norm{X^t}_1 \ge 1$. Consider any pair of configurations  $C',C$ with $C(X)=C'(X)=X^t$. For any  time $t \ge 1$,
$$\hspace{-.2em}\Pr[(\All^{t+i},\All^{t+i+1})\text{ is near-stable for some }i\le 13 | \All^t = C,\All^{t-1}=C'] \ge \frac{1}{128} - 39n e^{-\gamma/2}.$$
\end{lemma}
\end{tcolorbox}
\begin{proof}
By Corollary \ref{cor:typicalA}, for any  $C',C$, 
\begin{align*}
\Pr[\All^{t+1},\All^{t+2}\text{ are typical } | \All^t = C,\All^{t-1}=C'] &\ge 1-2(n+\lognc+1)\cdot e^{-\gamma/2}\\
&\ge 1 - 6n \cdot e^{-\gamma/2}.
\end{align*}
The lemma then follows by  combining this bound with Lemma \ref{lem:convergeGeneralA}.
\end{proof}

Finally, using Lemma \ref{lem:convergeGeneralA1}, we can prove the two main theorems of this section. The proofs are similar to those of the main theorems for the two-inhibitor network in Section \ref{sec:completing}.

\begin{tcolorbox}
\begin{reptheorem}{thm:logn}[$O(\log n)$-Inhibitor WTA]
For $\gamma \ge 12 \ln(39 t_s n/\delta)$,
$\Netc$ solves $\wta(n,t_c,t_s,\delta)$ for any  
$
t_c \ge 2086(\log_2(1/\delta)+1).
$ $\Netc$ contains $\lognc+1$ auxiliary inhibitors.
\end{reptheorem}
\end{tcolorbox}
%
\begin{proof}
Consider $\Netc$  starting with any initial configurations $\All^0,\All^1$  and given an infinite input execution $\alpha_\Input$ with $\Input^t$  fixed for all $t$. We consider two cases:

\medskip
\spara{Case 1: $\norm{X^t}_1 \ge  1$.}
\medskip

Let $\Delta = 15$ and $r = 139 (\log_2(1/\delta)+1)$.
Let $\mathcal{{E}}$ be the event that there is some time $t \le t_c$ where  $(\All^t,\All^{t+1})$ is a near-stable pair of configurations.

For any $i \ge 0$, let $\mathcal{{E}}_i$ be the event that there is some time $t \in \{i\Delta+1,...,(i+1)\Delta-1\}$ where  $(\All^t,\All^{t+1})$ is a near-stable pair of configurations.
By Lemma \ref{lem:convergeGeneralA1} we have: 
\begin{align*}
\Pr[\mathcal{ E}_i | \All^{i\Delta}] \ge \frac{1}{128} - 39n\cdot e^{-\gamma/2} \ge \frac{1}{200}
\end{align*}
by  the requirement that $\gamma > 12\ln(39 t_s n/\delta)$. 
Let $Z_0,...,Z_{r-1} \in \{0,1\}$ be independent coin flips, with $\Pr[Z_i = 1] = 1/200$. 
Applying Lemma \ref{lem:coupling}:
\begin{align*}
\Pr [ \mathcal{E} ] = \Pr \left [\bigcap_{i=0}^{r-1} \mathcal{E}_i \right ] \ge \Pr \left [\sum_{i=0}^{r-1} Z_i \ge 1 \right ] = 1-\left (\frac{199}{200}\right)^{r}.
\end{align*}
Using that $r = 139(\log_2(1/\delta)+1)$ and that $\left ( \frac{199}{200}\right)^{139}  \le \frac{1}{2}$:
\begin{align}\label{firstFailurecam}
\Pr [\mathcal{{E}}] \ge 1- \left (\frac{199}{200}\right  )^{139 (\log_2(1/\delta)+1)} \ge 1-\frac{\delta}{2}.
\end{align}
Thus, with probability  $\ge 1-\frac{\delta}{2}$ there is some time $t \le r\cdot \Delta-1 \le 2085 \cdot (\log_2(1/\delta)+1)-1 \le t_c-2$ in which  $(\All^t,\All^{t+1})$ is a near-stable pair of configurations.
By  Corollary \ref{cor:stabilityA}, if $(\All^t,\All^{t+1})$ is a near-stable pair, with probability $\ge  1-3t_s n \cdot e^{-\gamma/2}$, $\All^{t+2}$ is a valid WTA configuration and:
\begin{align*}
Y^{t+2} = Y^{t+3} = ... = Y^{t+2 +t_s}.
\end{align*}
By  the requirement that $\gamma \ge 12 \ln(39 t_s n/\delta)$ and \eqref{firstFailurecam} we thus have that the network reaches a valid WTA configuration within time $t_c$ and remains in it for time $t_s$ with probability at  least:
$$\left (1-\frac{\delta}{2}\right) \cdot \left (1-\frac{\delta}{e^6}\right) \ge 1-\delta,$$ yielding the theorem in this case.

\medskip
\spara{Case 0: $\norm{X^t}_1 =0$}
\medskip

In this case, by Corollary \ref{cor:monoA} and a union bound, for any configurations $C,C'$ with $C(X) = C'(X) = X^t$:
\begin{align*}
\Pr[\norm{Y^{2}}_1=...=\norm{Y^{2+t_s}}_1=0 |\All^1 = C,\All^{0}=C'] &\ge 1 - 3t_s n \cdot e^{-\gamma/2}\\
&\ge 1-\delta
\end{align*}
by  the requirement that $\gamma \ge 12 \ln(39 t_s n/\delta)$. This easily gives the theorem in this case since when $\norm{X^t}_1 = 0$, $\norm{Y^t}_1 =  0$ is a valid WTA output configuration.

\end{proof}

We now use a similar  argument to show with what parameters $\Netc$ solves  the expected-time WTA problem of Definition \ref{exp:wta}.
\begin{tcolorbox}
\begin{reptheorem}{thm:lognExp}[$O(\log n)$-Inhibitor Expected-Time WTA]
For $\gamma \ge 12 \ln(39 t_s n)$,
$\Netc$ solves $\ewta(n,t_c,t_s)$ for any  $
t_c \ge  4001.$ $\Netc$ contains $\lognc+1$ auxiliary inhibitors.
\end{reptheorem}
\end{tcolorbox}
\begin{proof}
Our proof closely follows that of Theorem \ref{thm:exp}.
Recall that in Definition \ref{exp:wta} we defined the convergence time for any infinite input execution $\alpha_X$ and output execution $\alpha_Y$:
\small
 $$t(\alpha_\Input,t_s,\alpha_\Output) = \min \left \{t : \Output^t\text{ is a valid WTA output configuration for }\Input^t\text{ and } \Output^t=...=\Output^{t+t_s}\right \}.$$\normalsize
Define the worst case expected convergence time of $\Netc$  on input $\alpha_X$ by:
\begin{align*}
t_{max}(\alpha_X) = \max_{\alpha^1 = \All^0\All^1} \left (\E_{\alpha_Y \sim \mathcal{D}_\Output(\Netc,\alpha^1,\alpha_\Input)} t(\alpha_\Input,t_s,\alpha_\Output) \right ).
\end{align*}
To prove the lemma we must prove that for any $\alpha_X$ with $X^t$ fixed for all $t$, $t_{max}(\alpha_X) \le 16$. Fixing such an $\alpha_X$, for any initial configuration $\alpha^1 = \All^0\All^1$, let $\underset{\alpha^1}{\E}$ and $\underset{\alpha^1}{\Pr}$ denote the expectation and probability of an event taken over executions drawn from $\mathcal{D}(\Netc,\alpha^1,\alpha_\Input)$.

We consider the case when $\alpha_X$ has $\norm{X^t}_1 \ge 1$. The case when $\norm{X^t}_1 = 0$ follows easily from a similar proof using Corollary \ref{cor:monoA}.

 Let $\Delta = 15$ and let $\mathcal{E}_1$ be the event that there is some $t \in \{1,...,\Delta-1\}$ where $(\All^t,\All^{t+1})$ is a near-stable pair of configurations. Let $\mathcal{E}_{stab}$ be the event that there is some $t \in \{1,...,\Delta-1\}$ where $(\All^t,\All^{t+1})$ is a near-stable pair of configurations and additionally, where $\All^{t+2}$ is a valid WTA output configuration with $\All^{t+2} = ... = \All^{t+2+t_s}$. Let $\mathcal{\bar E}_1$ and $\mathcal{\bar E}_{stab}$ be the complements of these two events.
By Lemma \ref{lem:convergeGeneralA1}, for any initial $\alpha^1 = \All^0\All^1$:
\begin{align}\label{eiBoundA}
\Pr_{\alpha^1} [\mathcal{ E}_1] \ge \frac{1}{128} - 39 n e^{-\gamma/2} \ge \frac{1}{200}.
\end{align}
by the requirement that $\gamma \ge 12 \ln(39 t_s n)$.
Further, by Corollary \ref{cor:stabilityA}, if $(C',C)$ is a near-stable pair of configurations, then
\begin{align}\label{eisBoundA}
\Pr_{\alpha^1} [\All^{t+1}\text{ is a valid WTA output config. and }\All^{t+2} = ... = \All^{t+2+t_s} |&\All^{t} = C',\All^{t+1} = C]\nonumber\\
&\ge 1-3t_s n e^{-\gamma/2}\nonumber\\ &\ge 1-\frac{1}{1000 t_s}
\end{align}
 where the bound holds since $\gamma \ge 12 \ln(39 t_s n) \ge 6\ln(39t_s^2 n)$ and so $e^{-\gamma/2} \le \frac{1}{(39 n t_s^2)^6} \le \frac{1}{1000 n t_s^2}$. Together \eqref{eiBoundA} and \eqref{eisBoundA} give that:
\begin{align*}
\Pr_{\alpha^1}[\mathcal{E}_{stab}] \ge \Pr_{\alpha^1}[\mathcal{E}_{stab}  | \mathcal{ E}_{1}] \cdot  \Pr_{\alpha^1}[\mathcal{ E}_1 ]  \ge \frac{1}{200}\cdot \left (1-\frac{1}{1000 t_s}\right) \ge\frac{1}{200} - \frac{1}{1000t_s}.
\end{align*}
We can write:
\begin{align}\label{secretaryA}
\E_{\alpha^1} [t(\alpha_\Input,t_s,\alpha_\Output) ]&= \E_{\alpha^1} [t(\alpha_\Input,t_s,\alpha_\Output) |\mathcal{E}_{stab}] \cdot \Pr_{\alpha^1}[\mathcal{E}_{stab}]\nonumber\\ &\hspace{2em}+ \E_{\alpha^1} [t(\alpha_\Input,t_s,\alpha_\Output) |\mathcal{E}_1, \mathcal{\bar E}_{stab}] \cdot \Pr_{\alpha^1}[\mathcal{E}_1, \mathcal{\bar E}_{stab}]\nonumber \\ &\hspace{2em}+ \E_{\alpha^1} [t(\alpha_\Input,t_s,\alpha_\Output) |\mathcal{\bar E}_1] \cdot \Pr_{\alpha^1}[\mathcal{\bar E}_1] 
\end{align}
Conditioned on $\mathcal{E}_{stab}$ (which also requires that $\mathcal{E}_1$ occurs), the network reaches a near-stable pair of configuration within $\Delta$ steps, reaches a valid WTA output configuration within $\Delta+1$ steps, and stabilizes for $t_s$ steps. Thus, we have:
$$\E_{\alpha^1} [t(\alpha_\Input,t_s,\alpha_\Output) |\mathcal{E}_{stab}] \le \Delta+1.$$
Conditioned on $\mathcal{E}_1, \mathcal{\bar E}_{stab}$ the network reaches a near-stable pair of configurations, but does not stabilize. We can bound 
\begin{align*}
\E_{\alpha^1} [t(\alpha_\Input,t_s,\alpha_\Output) |\mathcal{E}_1, \mathcal{\bar E}_{stab}] &\le (\Delta + 1+t_s) + \E_{\All^{\Delta+ t_s}\All^{\Delta+1+ t_s}} [t(\alpha_\Input,t_s,\alpha_\Output)]\\
&\le \Delta + 1 + t_s + t_{max}(\alpha_X).
\end{align*}
Finally, conditioned on $\mathcal{\bar E}_1$, the network does not reach a pair of near-stable configurations within $\Delta$ steps. We have:
\begin{align*}
\E_{\alpha^1} [t(\alpha_\Input,t_s,\alpha_\Output) |\mathcal{\bar E}_1] &\le \Delta + \E_{\All^{\Delta-1}\All^{\Delta}} [t(\alpha_\Input,t_s,\alpha_\Output)]\\
&\le \Delta + t_{max}(\alpha_X).
\end{align*}

We can plug these bounds along with the probability bounds of \eqref{eiBoundA} and \eqref{eisBoundA} into \eqref{secretaryA} to obtain:
\begin{align*}
\E_{\alpha^1} [t(\alpha_\Input,t_s,\alpha_\Output) ] &\le (\Delta+1) \cdot \left (\frac{1}{200} - \frac{1}{1000t_s} \right) +  (\Delta + 1 + t_{max}(\alpha_X) + t_s) \cdot \frac{1}{1000 t_s} \\
&\hspace{7em}+ (\Delta + t_{max}(\alpha_X)) \cdot \frac{199}{200}
\\& \le \Delta + 1 + t_{max} \left (\frac{199}{200} + \frac{1}{1000 t_s}\right ) + \frac{t_s}{1000 t_s}\\
&\le \Delta + 1 + t_{max}(\alpha_X)\cdot \frac{249}{250} + \frac{1}{1000}. 
\end{align*}
Since this bound holds for all $\alpha^1$ we have:
\begin{align*}
t_{max}(\alpha_X) \le \Delta + 1 + t_{max}(\alpha_X)\cdot \frac{249}{250} + \frac{1}{100}\\
\end{align*}
which gives $t_{max}(\alpha_X) \le 250(\Delta + 1) +\frac{250}{1000} \le 250\Delta + 251$. This bound holds for all $\alpha_X$ and so gives the lemma, after recalling that $\Delta = 15$.
\end{proof}

\subsection{Constructions With Runtime Tradeoffs}\label{sec:tradeoffs}

The family of two-inhibitor networks of Section \ref{sec:wta2} and the family of $O(\log n)$-inhibitor networks of Section \ref{sec:logn} represent two extremes of a tradeoff between number of inhibitors and runtime. As shown in Theorem \ref{thm:lb1}, the two-inhibitor network uses the minimum number of neurons required to solve WTA with a reasonably long stability period. At least when considering a constant probability of success, the $O(\log n)$-inhibitor networks  gives optimal convergence time of $O(1)$ steps. In this section we outline, at a high level, two families of networks that allow a tradeoff between these two extremes.

\subsubsection{Fixed Convergence Time Construction}

The first construction lets us to achieve any desired convergence time $\theta$ if sufficiently many inhibitors are used. Specifically, we describe a
family of networks which converge to a valid WTA output configuration with constant probability in $O(\theta)$ steps (and with probability $\ge 1-\delta$ in $O(\theta \cdot \log1/\delta)$ steps), using 
$O(\theta \log^{1/\theta} n)$ inhibitors. Note that setting $\theta  = 1$ recovers the runtime-inhibitor tradeoff of our $O(\log n)$-inhibitor networks.

We have one stability inhibitor $a_s$ that functions in the same way as the stability inhibitor in our $O(\log n)$-inhibitor construction (Definition \ref{def:wtac}), ensuring that, once the network reaches a valid WTA output configuration, it remains in this configuration for $t_s$ consecutive time steps with high probability (as long as some other stability  conditions, similar to the near-stable condition of Definition \ref{def:nearStableA} are satisfied).

The remaining inhibitors are split into $\theta$ groups each containing $O(\log^{1/\theta} n)$ inhibitors. The $i^{th}$  group for $i \in \{2,...,\theta\}$ is responsible for reducing the number of competing outputs from any number 
$$k \in \left (2^{\log^{(i-1)/\theta} n},2^{\log^{i/\theta} n}\right]$$
 to some $k' \le 2^{\log^{(i-1)/\theta} n}$. With $O(\log^{1/\theta} n))$ inhibitors this can be done with high probability  in $O(1)$ time steps via a method described below. 
Thus, in $O(\theta)$ time steps, the number of firing outputs (corresponding to firing inputs will reduce from at most $2^{\log^{\theta/\theta} n} = n$ down to $2^{\log^{1/\theta} n}$. Once the number of computing outputs is this low, the final group of $O(\log^{1/\theta} n)$ inhibitors can drive convergence to WTA with constant probability  in $O(1)$ steps using an identical strategy to that  employed in our $O(\log n)$-inhibitor network. $n$ is just replaced by $2^{\log^{1/\theta} n}$.

It remains to explain how the number of competing outputs is reduced
from  $k \in [2^{\log^{(i-1)/\theta} n},2^{\log^{i/\theta} n}]$ to $k' \le 2^{\log^{(i-1)/\theta} n}$ in a single time step using $O(\log^{1/\theta} n)$ inhibitors. Again, we use a strategy similar to our $O(\log n)$-inhibitor construction.
depending on the number of firing outputs at time $t$, 
the $O(\log^{1/\theta} n)$ inhibitors in group $i$ induce $O(\log^{1/\theta} n)$ different firing probabilities. Letting $\Delta(j) = j \cdot \log^{(i-1)/\theta} n$, these firing probabilities are:
 $$\frac{1}{2^{\Delta(1)}},\frac{1}{2^{\Delta(2)}},...,\frac{1}{2^{\Delta((\log^{1/\theta} n - 1))}}, \frac{1}{2^{\Delta(\log^{1/\theta} n)}}.$$
 Note that $\Delta(1)  = \log^{(i-1)/\theta} n$ and $\Delta(\log^{1/\theta} n) = \log^{i/\theta} n$. Additionally,
 the ratio between adjacent probabilities is $2^{\log^{(i-1)/\theta} n }$. Thus, for any number $k \in \left (2^{\log^{(i-1)/\theta} n},2^{\log^{i/\theta} n}\right ]$ of firing outputs at time $t$, the inhibitors can induce a firing probability between $\frac{1}{k}$ and $\frac{2^{\log^{(i-1)/\theta} n }}{k}$. Thus, with good probability, the number of firing outputs will reduce to some value $k' \le 2^{\log^{(i-1)/\theta} n}$.

%

While we do not analyze this construction  in detail, using similar proof techniques to those used for our $O(\log n)$-inhibitor construction, it  is possible to show:
\begin{tcolorbox}
\begin{theorem}[$\theta$-Step WTA]\label{thm:theta} For any $n \in \mathbb{Z}^{\ge 2}$ and $\theta \in \mathbb{Z}^{\ge 1}$ there is an a family  of SNNs containing $O(\theta \cdot \log^{1/\theta} n)$ auxiliary inhibitory  neurons which solve $\wta(n,t_c,t_s,\delta)$ for any  $t_s,\delta$ and any
$
t_c \ge c_1\cdot \theta (\log_2(1/\delta)+1),
$
where $c_1$ is some fixed constant.
\end{theorem}
\end{tcolorbox}

\subsubsection{Fixed Inhibitor Budget Construction}

Our second construction is  a family of networks using $\alpha$ inhibitors for any $\alpha \ge 2$, which converges to a valid WTA state  with constant probability in $O(\alpha \cdot (\log n)^{1/(\alpha-1)})$ steps (and with probability $\ge 1-\delta$ in $O(\alpha\cdot ( \log n)^{1/(\alpha-1)} \cdot \log 1/\delta)$ steps. Note  that for $\alpha = 2$, this gives convergence in $O(\log n)$ steps, matching the performance of the two-inhibitor network construction of Section \ref{sec:wta2}. 

As in our two-inhibitor construction (Definition \ref{def:wtaNet}) and $O(\log n)$-inhibitor construction (Definition \ref{def:wtac}) we
employ one stability inhibitor, which ensures that, once the network reaches a valid WTA output configuration, it remains in such a configuration with good probability  for $t_s$ consecutive steps (again, as long as some other stability  conditions, similar to the near-stable condition of Definition \ref{def:nearStableA} are satisfied).

We label the remaining $\alpha-1$ inhibitors $a_1,...,a_{\alpha-1}$. For each $i$, $a_i$ fires with high probability at time $t+1$ whenever $k \ge 2^{(\log n)^{(i-1)/(\alpha-1)}}$ outputs fire at time $t$. In this way, with high probability, $a_1,...,a_i$ fire (and all other inhibitors do not fire) at time $t+1$ whenever the number of firing outputs $k$ is in the range: $$R_i \eqdef  \left [2^{(\log n)^{(i-1)/(\alpha-1)}},2^{(\log n)^{i/(\alpha-1)}} \right ).$$
We set the inhibitory weights such that 
when $a_i$ fires (along with $a_j$ for all $j \le i$), each firing output fires in the next step with probability $p_i = \frac{1}{2^{(\log n)^{(i-1)/(\alpha-1)}}}$. For $k \in R_i$ we have: 
$$p_i \cdot k \in \left [1,2^{(\log n)^{i/(\alpha-1)}\cdot \left (1-1/(\log n)^{1/(\alpha-1)}\right)}\right).$$
With each output continuing to fire with probability  $\frac{1}{2^{(\log n)^{(i-1)/(\alpha-1)}}}$ in each step, starting from at most $2^{(\log n)^{i/(\alpha-1)}}$ competing outputs, we thus reach a configuration with $k \le 2^{(\log n)^{(i-1)/(\alpha-1)}}$ firing outputs with good probability within $O((\log n)^{1/(\alpha-1)})$ steps. Overall, over $\alpha$ such levels, the network reaches a valid WTA configuration with constant probability in $O(\alpha \cdot (\log n)^{1/(\alpha-1)})$ steps.

%

Again, while we do not analyze  this construction here, it is possible to prove the following:


\begin{tcolorbox}
\begin{theorem}[$\alpha$-Inhibitor WTA]\label{thm:alpha} For any $n,\alpha \in \mathbb{Z}^{\ge 2}$ there is an a family  of SNNs containing $\alpha$ auxiliary inhibitory  neurons which solve $\wta(n,t_c,t_s,\delta)$ for any  $t_s,\delta$ and any
$
t_c \ge c_1\cdot \alpha (\log n)^{1/(\alpha-1)} \cdot (\log_2(1/\delta)+1),
$
where $c_1$ is some fixed constant.
\end{theorem}
\end{tcolorbox}

\section{Discussion and Future Work}\label{sec:wtadiscuss}

We have presented an exploration of the WTA problem in stochastic spiking neural networks, giving network constructions, runtime analysis, and lower bounds. Our work leaves open a number of questions, both in regards to the WTA problem, and more broadly, in exploring spiking neural networks from an algorithmic perspective. We discuss some of these directions below. See also Section \ref{sec:futureDirs} in which we discuss possible extensions two and modifications of our basic spiking neural network model.

\subsection{Winner-Take-All Extensions and Open Questions}

We first overview future work directly  related  to the WTA problem studied in this work.

\subsubsection{Lower Bounds}

In Section \ref{sec:lb} we present lower bounds for one and two-inhibitor WTA networks (Theorems \ref{thm:lb1} and \ref{thm:lb2}) . In \cite{lynch2016computational} we also give nearly tight lower bounds on the convergence time for networks using $\alpha > 2$ inhibitors, in a similar model to the one presented here. A few interesting open questions remain:

\begin{itemize}
\item Can our lower bound for two-inhibitor networks (Theorem \ref{thm:lb2}) be tightened to match our upper bound (Theorem \ref{thm:high}) up to a constant factor, rather than a $O(\log \log n)$ factor?
\item Our lower bounds apply to somewhat restricted classes of simple and symmetric SNNs (Definitions \ref{def:simple} and \ref{def:sym} respectively.) We conjecture, however, that these lower bounds can be shown for general SNNs, with no restrictions on the network structure.
\item It would be interesting to obtain lower bounds which help explain the relationship between history and convergence time. Our $O(\log n)$-inhibitor network of Section \ref{sec:logn} achieves constant probability, $O(1)$ convergence time. However, it requires a history period of two steps. Is it possible to show that, with no history period (i.e., in the basic SNN model of Section \ref{sec:modelW}), it is not possible to achieve this convergence time? Is the $O(\log n)$ time for our two-inhibitor network optimal for historyless networks, regardless of the number of inhibitors used? In a network with history, is it possible to solve WTA with a single auxiliary neuron, or can the lower bound of Theorem \ref{thm:lb1} be extended to this case?
\end{itemize}

\subsubsection{Variations on the WTA Problem}

Our work focuses on a simplified WTA problem in which all inputs either fire at  every  time step or never fire (see Definition \ref{high:wta}). The challenge in solving this problem is in breaking symmetry between the firing inputs. In future work we plan to extend this work to consider more general variants of WTA, discussed below. 

\begin{itemize}
\item \textbf{Non-Binary WTA:}
We plan to extend our WTA work to a \emph{non-binary} setting, in which inputs have different firing rates and selection is based on those rates.
The most common requirement is for the network to select an input neuron with the highest or near-highest firing rate \cite{yuille1989winner,coultrip1992cortical,oster2006spiking}.

In the most basic setting, we can consider $n$ inputs $x_1,...,x_n$ with firing rates $r_1,...,r_n \in [0,1]$ as neurons that fire independently at random  at each time step, each time with probability $r_i$. 
%
%
In~\cite{LynchMP-bda17}, we presented an initial exploration of the non-binary WTA problem with inputs of this type, giving a solution using $O(n \log n)$ auxiliary neurons to select an input with firing rate within a constant factor of the maximum. It seems that in our basic SNN model, $\Omega(n \log n)$ auxiliary neurons may in fact be necessary:
essentially, for each input neuron, $\Omega(\log n)$ auxiliary neurons are needed to record a short firing history, from which the firing probability $r_i$ can be (implicitly) estimated to within a constant factor, with high probability. 
Thus, networks using a sublinear number of auxiliary neurons, like those we gave for binary WTA, may be ruled out. 

However, more efficient solutions may be possible
if we use an SNN model with a history period.
%
Even without history, efficient solutions may be possible if we relax or modify the WTA problem.
For example, we may just require selecting each neuron with probability $\frac{r_i}{\sum_{i=1}^n r_i}$, or some other function of the firing probabilities. 

\item \textbf{WTA with Different Output Conditions:} 
It would be interesting to consider variations on the basic output condition of the WTA problem, both  in the binary and non-binary input setting. For example, $k$-WTA is a common variant of WTA, in which the goal is to select $k$ inputs with the highest or near-highest firing rates~\cite{majani1989k,maass1996computational,wang2003k,hu2008improved}, rather than just as single input. $k$-WTA, for example, has been used in recent work in modeling sparse coding in the fly olfactory via random projection methods \cite{dasgupta2017neural}. It would  be interesting to understand the fundamental complexity (in terms of network complexity and convergence time) of implementing this primitive in spiking neural networks. 

\end{itemize}

\subsubsection{Applications of WTA}

Finally, it would be  interesting to explore how our WTA constructions can be integrated into solutions to higher-level neural tasks. Some problems which may be worth considering are:
\begin{itemize}
\item  Attention in sensory processing systems, which is thought to be implemented via WTA competition~\cite{lee1999attention,itti2001computational}. In this setting, a specific set of neurons is activated while others are suppressed, allowing computation to happen on the selected neurons without interference. Formulating and studying simple tasks which model the use of WTA as an attention mechanism would be  an interesting direction.
\item The use of WTA in neural sparse coding algorithms~\cite{dasgupta2017neural}. In this setting, a stimulus, such as a specific odor, triggers the firing of a large set of neurons, which is then ``sparsified'' via non-binary $k$-WTA competition to $k$ neurons that encode the odor -- those with the strongest response to the stimulus. 
\item Clustering problems, in a which the network is set up such that, when given any input, the neuron which corresponds to the best cluster assignment for this input fires at the highest rate. Non-binary WTA can be used to select this neuron~\cite{abbas1994neural,waterhouse1994classification}.
\end{itemize}

\subsection{Other Neural Computational Primitives}\label{sec:otherPrim}

WTA is a basic primitive that seems to be useful in a broad range of neural tasks, discussed above.  It is also useful in exploring the computational power of SNNs and studying tradeoffs among such costs as runtime, network size, and the use of randomness. In future work, we hope to to identify other such basic primitives, as a means of building a general algorithmic theory for spiking neural networks. We discuss some possible directions below:

\subsubsection{Neural Indexing}

In~\cite{lynch2017neuro} we define and study the \emph{Neuro-RAM} problem: the network is given a set of $n$ binary inputs $X_1$, along with a smaller set of $\log n$ inputs $X_2$ whose firing pattern represents an index into $X_1$. The goal is the for the output to fire if and only if the selected neuron in $X_1$ is firing. In \cite{lynch2017neuro}, we applied our Neuro-RAM primitive to the similarity testing problem:  given two binary inputs $X_1,X_2$, determine with high probability whether their Hamming distance is greater than some threshold.  We solved this problem efficiently by using a Neuro-RAM model to select and compare random positions in $X_1$ and $X_2$.

It would be interesting to formulate other basic neural tasks that may employ a Neuro-RAM primitive. For example, one can consider estimating the differences between firing patterns (e.g., representing sensory inputs at successive times), estimating the average firing activity of a group of neurons responding to some stimulus, or randomly ``exploring'' different sets of neurons whose firing may trigger a desired response (such as the recall of a memory). 


In our work, aside from efficient Neuro-RAM constructions, we give lower bounds demonstrating that these constructions give a near-optimal tradeoff between network size and runtime in our basic SNN model with a sigmoidal spike probability function.  Our proofs are based on reducing an SNN to a distribution over deterministic networks, and then bounding the VC dimension \cite{vapnik1998statistical} of these networks.

Generally, while VC dimension has been studied for spiking networks~\cite{zador1996vc,maass1999complexity}, we are the first to apply it to prove computational lower bounds. We hope to continue our work, greatly expanding the toolkit for proving lower bounds in biologically plausible networks, including those with different spike probability functions, with history, and even with refractory periods or spike propagation delays. 

Further, Neuro-RAM lower bounds let us compare our SNN models with artificial neural network models that use continuous-output gates, e.g., sigmoidal gates~\cite{barron1993universal,bartlett1997valid} or rectified linear units (ReLUs)~\cite{jarrett2009best,nair2010rectified}. 
The Neuro-RAM problem can be solved very efficiently in such networks~\cite{koiran1996vc}, contrasting with its apparent difficulty in SNNs. 
%
%
However, the continuous-output solutions do not seem very robust:  they do not appear to tolerate small variations in edge weights, and they  seem to rely on precise (infinite precision) gate outputs.
We hope that comparing these models with our SNNs could lead to a better understanding of continuous-output circuits in noisy settings.  We conjecture that in these settings, these networks behave more like their spiking network counterparts than like ideal continuous-output networks.

\subsubsection{Binary Vector Problems}
Beyond the WTA and Neuro-RAM problems it would also be interesting to consider other basic problems with fixed input firing patterns, that is, binary vector problems.
For example, the problem of estimating the total firing strength of a population of neurons
corresponds to determining the Hamming weight of a binary input vector. 
%
More complex problems include string matching, in which we test whether the input firing pattern contains a copy of a specified substring.  
In the brain, this may be used for pattern recognition, or for alignment and comparison of patterns triggered by multiple stimuli. 

\subsubsection{Problems with Non-Binary Inputs}
It would also be valuable to consider problems in which the input is non-binary, for example, problems that require computing some function of the firing rates of the input neurons. 
For example, we may consider the problem of approximating the sum of input firing probabilities, or more generally, some norm of the firing rate vector $r_1,...,r_n$.  
We may also consider similarity testing under different norms (generalizing the Hamming distance similarity testing problem in~\cite{lynch2017neuro}).

Many algorithms proposed in computational neuroscience and computer science apply operations such as addition and multiplication to continuous neuron output values ~\cite{durbin1989product,allen2014sparse,arora2015simple,dasgupta2017neural}. 
These continuous output values can be thought of as abstractions of firing rates in SNNs. 
We would like to understand how such basic operations may be emulated in more biologically plausible SNNs.

\subsubsection{Synchronization Problems}
Finally, it would be interesting to study algorithmic primitives that are not expressed in terms of simple input-output mappings.  
Notably, we would like to consider synchronization problems in which neurons, starting from arbitrary firing states, cooperate to align their firing in some way. 
Synchronization of spiking patterns and the emergence of neural rhythms are widely studied in both empirical and computational neuroscience~\cite{van1994inhibition,buzsaki2006rhythms,riehle1997spike,womelsdorf2007modulation}. 
Synchronization is also an important technique used in distributed algorithms~\cite{burns1985byzantine,kopetz1987clock,cristian1989probabilistic,lenzen2009optimal},
and thus is a natural problem to include in an algorithmic theory of neural computation. 
%

\subsection{Learning Problems and Dynamic Networks}\label{sec:learning}

The work presented in this work, as well as the work proposed in Section \ref{sec:otherPrim} focuses on 
 computation in \emph{static} networks, with fixed edge weights. 
An important direction is studying algorithms for learning in \emph{dynamic spiking networks}, in which synapse weights change throughout the computation via, for example, a Hebbian update rule.

In classical Hebbian learning~\cite{hebb2005organization,caporale2008spike}, the weight of a synapse is continuously updated by a factor that depends on the product of the firing strengths of its two endpoints.  
The more the firing of the endpoints correlates, the stronger the synapse becomes.
We hope to define a Hebbian-style rule for our synchronous SNN models similar to, for example, the simple rule used in \cite{legenstein2018long}. When both endpoints fire at the same time,
the synapse weight should increase by a small factor, and when just one fires, the weight may decrease. 
We then hope to use our dynamic network model to study the costs of many learning problems such as memory formation and concept  association \cite{amari1977neural,valiant2005memorization,legenstein2018long}, linear classification \cite{huerta2006learning}, principal component analysis \cite{sanger1989optimal,hyvarinen1997one,hyvarinen2000independent}, and sparse coding \cite{olshausen2004sparse,arora2015simple}.

\subsection{Neural Linear Algebraic Computation}

Finally, we are interested in understanding how basic linear algebraic computation may be performed in spiking neural networks, possibly with the use of randomization.
Randomized algorithms have recently led to a number of breakthroughs in fundamental linear-algebraic and geometric problems such as linear regression~\cite{clarkson2013low,uniform}, low-rank approximation~\cite{sarlos2006improved,nelson2013osnap}, clustering~\cite{boutsidis2010random}, and locality-sensitive hashing~\cite{datar2004locality,andoni2006near}.
Many mechanisms used in neural computation have close analogs to randomized linear algebraic techniques; e.g., Oja's Hebbian-style
rule for neural principal component analysis (PCA)~\cite{hyvarinen1997one,hyvarinen2000independent} is a common technique for low-memory PCA and eigenvector approximation~\cite{shamir2015stochastic,jain2016streaming}. 
Fast linear-algebraic methods based on \emph{random projections}~\cite{achlioptas2003database,sarlos2006improved} are also conjectured to play a role in neural computation~\cite{ganguli2012compressed,allen2014sparse}, with random synaptic connectivity providing a natural implementation of randomized dimensionality-reduction. 
We hope to study the application of these techniques in stochastic spiking neural networks, forging connections with work on new algorithms for fast linear algebraic computation.

\bibliographystyle{alpha}
\bibliography{main}	

\newcommand{\etalchar}[1]{$^{#1}$}
\begin{thebibliography}{HNGS{\etalchar{+}}06}

\bibitem[Ach03]{achlioptas2003database}
Dimitris Achlioptas.
\newblock Database-friendly random projections: {J}ohnson-{L}indenstrauss with
  binary coins.
\newblock {\em Journal of Computer and System Sciences}, 66(4):671--687, 2003.

\bibitem[AF94]{abbas1994neural}
Hazem~M Abbas and Moustafa~M Fahmy.
\newblock Neural networks for maximum likelihood clustering.
\newblock {\em Signal Processing}, 36(1):111--126, 1994.

\bibitem[AGMM15]{arora2015simple}
Sanjeev Arora, Rong Ge, Tengyu Ma, and Ankur Moitra.
\newblock Simple, efficient, and neural algorithms for sparse coding.
\newblock In {\em \COLT{2015}}, pages 113--149, 2015.

\bibitem[AHS85]{ackley1985learning}
David~H Ackley, Geoffrey~E Hinton, and Terrence~J Sejnowski.
\newblock A learning algorithm for {B}oltzmann machines.
\newblock {\em Cognitive Science}, 9(1):147--169, 1985.

\bibitem[AI06]{andoni2006near}
Alexandr Andoni and Piotr Indyk.
\newblock Near-optimal hashing algorithms for approximate nearest neighbor in
  high dimensions.
\newblock In {\em \FOCS{2006}}, pages 459--468, 2006.

\bibitem[Ama77]{amari1977neural}
S-I Amari.
\newblock Neural theory of association and concept-formation.
\newblock {\em Biological Cybernetics}, 26(3):175--185, 1977.

\bibitem[And03]{andrew2003spiking}
Alex~M Andrew.
\newblock Spiking neuron models: Single neurons, populations, plasticity.
\newblock {\em Kybernetes}, 32(7/8), 2003.

\bibitem[AS94]{allen1994evaluation}
Christina Allen and Charles~F Stevens.
\newblock An evaluation of causes for unreliability of synaptic transmission.
\newblock {\em Proceedings of the National Academy of Sciences},
  91(22):10380--10383, 1994.

\bibitem[ASNN{\etalchar{+}}15]{al2015inherently}
Maruan Al-Shedivat, Rawan Naous, Emre Neftci, Gert Cauwenberghs, and Khaled~N
  Salama.
\newblock Inherently stochastic spiking neurons for probabilistic neural
  computation.
\newblock In {\em 2015 7th International IEEE/EMBS Conference on Neural
  Engineering (NER)}, pages 356--359. IEEE, 2015.

\bibitem[AZGMS14]{allen2014sparse}
Zeyuan Allen-Zhu, Rati Gelashvili, Silvio Micali, and Nir Shavit.
\newblock Sparse sign-consistent {J}ohnson--{L}indenstrauss matrices:
  Compression with neuroscience-based constraints.
\newblock {\em Proceedings of the National Academy of the Sciences}, 2014.

\bibitem[Bar93]{barron1993universal}
Andrew~R Barron.
\newblock Universal approximation bounds for superpositions of a sigmoidal
  function.
\newblock {\em IEEE Transactions on Information Theory}, 39(3):930--945, 1993.

\bibitem[Bar97]{bartlett1997valid}
Peter~L Bartlett.
\newblock For valid generalization the size of the weights is more important
  than the size of the network.
\newblock In {\em \NIPS{1997}}, pages 134--140, 1997.

\bibitem[BL85]{burns1985byzantine}
James~E Burns and Nancy~A Lynch.
\newblock The {B}yzantine firing squad problem.
\newblock Technical report, Massachusetts Institute of Technology, 1985.

\bibitem[Buz06]{buzsaki2006rhythms}
Gyorgy Buzsaki.
\newblock {\em Rhythms of the Brain}.
\newblock Oxford University Press, 2006.

\bibitem[BZD10]{boutsidis2010random}
Christos Boutsidis, Anastasios Zouzias, and Petros Drineas.
\newblock Random projections for $k$-means clustering.
\newblock In {\em \NIPS{2010}}, 2010.

\bibitem[CD08]{caporale2008spike}
Natalia Caporale and Yang Dan.
\newblock Spike timing--dependent plasticity: a {H}ebbian learning rule.
\newblock {\em Annual Review of Neuroscience}, 31:25--46, 2008.

\bibitem[CGL92]{coultrip1992cortical}
Robert Coultrip, Richard Granger, and Gary Lynch.
\newblock A cortical model of winner-take-all competition via lateral
  inhibition.
\newblock {\em Neural Networks}, 5(1):47--54, 1992.

\bibitem[CLM{\etalchar{+}}15]{uniform}
Michael~B. Cohen, Yin~Tat Lee, Cameron Musco, Christopher Musco, Richard Peng,
  and Aaron Sidford.
\newblock Uniform sampling for matrix approximation.
\newblock In {\em \ITCS{2015}}, 2015.

\bibitem[Cri89]{cristian1989probabilistic}
Flaviu Cristian.
\newblock Probabilistic clock synchronization.
\newblock {\em Distributed Computing}, 3(3):146--158, 1989.

\bibitem[CW13]{clarkson2013low}
Kenneth~L. Clarkson and David~P. Woodruff.
\newblock Low rank approximation and regression in input sparsity time.
\newblock In {\em \STOC{2013}}, 2013.

\bibitem[DIIM04]{datar2004locality}
Mayur Datar, Nicole Immorlica, Piotr Indyk, and Vahab~S Mirrokni.
\newblock Locality-sensitive hashing scheme based on p-stable distributions.
\newblock In {\em Proceedings of the Twentieth Annual Symposium on
  Computational Geometry}, pages 253--262. ACM, 2004.

\bibitem[DR89]{durbin1989product}
Richard Durbin and David~E Rumelhart.
\newblock Product units: A computationally powerful and biologically plausible
  extension to backpropagation networks.
\newblock {\em Neural Computation}, 1(1):133--142, 1989.

\bibitem[DSN17]{dasgupta2017neural}
Sanjoy Dasgupta, Charles~F Stevens, and Saket Navlakha.
\newblock A neural algorithm for a fundamental computing problem.
\newblock {\em Science}, 358(6364):793--796, 2017.

\bibitem[FSW08]{faisal2008noise}
A~Aldo Faisal, Luc~PJ Selen, and Daniel~M Wolpert.
\newblock Noise in the nervous system.
\newblock {\em Nature Reviews Neuroscience}, 9(4):292--303, 2008.

\bibitem[GK02]{gerstner2002spiking}
Wulfram Gerstner and Werner~M Kistler.
\newblock {\em Spiking Neuron Models: Single Neurons, Populations, Plasticity}.
\newblock Cambridge University Press, 2002.

\bibitem[GL09]{gupta2009hebbian}
Ankur Gupta and Lyle~N Long.
\newblock Hebbian learning with winner take all for spiking neural networks.
\newblock In {\em 2009 International Joint Conference on Neural Networks},
  pages 1054--1060. IEEE, 2009.

\bibitem[GS12]{ganguli2012compressed}
Surya Ganguli and Haim Sompolinsky.
\newblock Compressed sensing, sparsity, and dimensionality in neuronal
  information processing and data analysis.
\newblock {\em Annual Review of Neuroscience}, 35:485--508, 2012.

\bibitem[Hay09]{haykin2009neural}
Simon~S Haykin.
\newblock {\em Neural Networks and Learning Machines}, volume~3.
\newblock Pearson, 2009.

\bibitem[Heb05]{hebb2005organization}
Donald~Olding Hebb.
\newblock {\em The Organization of Behavior: A Neuropsychological Theory}.
\newblock Psychology Press, 2005.

\bibitem[HJM13]{habenschuss2013stochastic}
Stefan Habenschuss, Zeno Jonke, and Wolfgang Maass.
\newblock Stochastic computations in cortical microcircuit models.
\newblock {\em PLoS Computational Biology}, 9(11):e1003311, 2013.

\bibitem[HNGS{\etalchar{+}}06]{huerta2006learning}
Ram{\'o}n Huerta, Thomas Nowotny, Marta Garc{\'\i}a-Sanchez, Henry~DI
  Abarbanel, and Mikhail~I Rabinovich.
\newblock Learning classification in the olfactory system of insects.
\newblock {\em Learning}, 16(8), 2006.

\bibitem[HO97]{hyvarinen1997one}
Aapo Hyv{\"a}rinen and Erkki Oja.
\newblock One-unit learning rules for independent component analysis.
\newblock In {\em \NIPS{1997}}, pages 480--486, 1997.

\bibitem[HO00]{hyvarinen2000independent}
Aapo Hyv{\"a}rinen and Erkki Oja.
\newblock Independent component analysis: algorithms and applications.
\newblock {\em Neural Networks}, 13(4):411--430, 2000.

\bibitem[HSW89]{hornik1989multilayer}
Kurt Hornik, Maxwell Stinchcombe, and Halbert White.
\newblock Multilayer feedforward networks are universal approximators.
\newblock {\em Neural Networks}, 2(5):359--366, 1989.

\bibitem[HT{\etalchar{+}}86]{hopfield1986computing}
John~J Hopfield, David~W Tank, et~al.
\newblock Computing with neural circuits -- a model.
\newblock {\em Science}, 233(4764):625--633, 1986.

\bibitem[HW08]{hu2008improved}
Xiaolin Hu and Jun Wang.
\newblock An improved dual neural network for solving a class of quadratic
  programming problems and its $k$-winners-take-all application.
\newblock {\em IEEE Transactions on Neural Networks}, 19(12):2022--2031, 2008.

\bibitem[IK01]{itti2001computational}
Laurent Itti and Christof Koch.
\newblock Computational modeling of visual attention.
\newblock {\em Nature Reviews Neuroscience}, 2(3):194--203, 2001.

\bibitem[Izh04]{izhikevich2004model}
Eugene~M Izhikevich.
\newblock Which model to use for cortical spiking neurons?
\newblock {\em IEEE Transactions on Neural Networks}, 15(5):1063--1070, 2004.

\bibitem[JJK{\etalchar{+}}16]{jain2016streaming}
Prateek Jain, Chi Jin, Sham~M Kakade, Praneeth Netrapalli, and Aaron Sidford.
\newblock Streaming {PCA}: Matching matrix bernstein and near-optimal finite
  sample guarantees for {O}ja\'s algorithm.
\newblock In {\em \COLT{2016}}, pages 1147--1164, 2016.

\bibitem[JKLR09]{jarrett2009best}
Kevin Jarrett, Koray Kavukcuoglu, Yann LeCun, and Marc'Aurelio Ranzato.
\newblock What is the best multi-stage architecture for object recognition?
\newblock In {\em 12th International Conference on Computer Vision}, pages
  2146--2153. IEEE, 2009.

\bibitem[KK94]{kaski1994winner}
Samuel Kaski and Teuvo Kohonen.
\newblock Winner-take-all networks for physiological models of competitive
  learning.
\newblock {\em Neural Networks}, 7(6-7):973--984, 1994.

\bibitem[KKKL11]{khabbazian2011decomposing}
Majid Khabbazian, Dariusz Kowalski, Fabian Kuhn, and Nancy Lynch.
\newblock Decomposing broadcast algorithms using abstract {MAC} layers, 2011.
\newblock
  \url{http://groups.csail.mit.edu/tds/papers/Lynch/MIT-CSAIL-TR-2011-010.pdf}.

\bibitem[KO87]{kopetz1987clock}
Hermann Kopetz and Wilhelm Ochsenreiter.
\newblock Clock synchronization in distributed real-time systems.
\newblock {\em IEEE Transactions on Computers}, 100(8):933--940, 1987.

\bibitem[Koi96]{koiran1996vc}
Pascal Koiran.
\newblock {VC} dimension in circuit complexity.
\newblock In {\em Proceedings of the Eleventh Annual IEEE Conference on
  Computational Complexity}, pages 81--85. IEEE, 1996.

\bibitem[KU87]{koch1987shifts}
Christof Koch and Shimon Ullman.
\newblock Shifts in selective visual attention: towards the underlying neural
  circuitry.
\newblock In {\em Matters of Intelligence}, pages 115--141. Springer, 1987.

\bibitem[LBH15]{lecun2015deep}
Yann LeCun, Yoshua Bengio, and Geoffrey Hinton.
\newblock Deep learning.
\newblock {\em Nature}, 521(7553):436--444, 2015.

\bibitem[LIKB99]{lee1999attention}
Dale~K Lee, Laurent Itti, Christof Koch, and Jochen Braun.
\newblock Attention activates winner-take-all competition among visual filters.
\newblock {\em Nature Neuroscience}, 2(4):375--381, 1999.

\bibitem[LMP17a]{lynch2016computational}
Nancy Lynch, Cameron Musco, and Merav Parter.
\newblock Computational tradeoffs in biological neural networks:
  Self-stabilizing winner-take-all networks.
\newblock In {\em \ITCS{2017}}, 2017.

\bibitem[LMP17b]{lynch2017neuro}
Nancy Lynch, Cameron Musco, and Merav Parter.
\newblock Neuro-{RAM} unit with applications to similarity testing and
  compression in spiking neural networks.
\newblock In {\em \DISC{2017}}, 2017.
\newblock Full version available at \url{https://arxiv.org/abs/1706.01382}.

\bibitem[LMP17c]{LynchMP-bda17}
Nancy Lynch, Cameron Musco, and Merav Parter.
\newblock Spiking neural networks: An algorithmic perspective.
\newblock In {\em 5th Workshop on Biological Distributed Algorithms (BDA
  2017)}, July 2017.

\bibitem[LMPV18]{legenstein2018long}
Robert Legenstein, Wolfgang Maass, Christos~H Papadimitriou, and Santosh~S
  Vempala.
\newblock Long term memory and the densest k-subgraph problem.
\newblock In {\em \ITCS{2018}}, 2018.

\bibitem[LRMM88]{lazzaro1988winner}
John Lazzaro, Sylvie Ryckebusch, Misha~Anne Mahowald, and Caver~A Mead.
\newblock Winner-take-all networks of o(n) complexity.
\newblock Technical report, DTIC Document, 1988.

\bibitem[LSW09]{lenzen2009optimal}
Christoph Lenzen, Philipp Sommer, and Roger Wattenhofer.
\newblock Optimal clock synchronization in networks.
\newblock In {\em Proceedings of the 7th ACM Conference on Embedded Networked
  Sensor Systems}, pages 225--238. ACM, 2009.

\bibitem[Maa96]{maass1996computational}
Wolfgang Maass.
\newblock On the computational power of noisy spiking neurons.
\newblock In {\em \NIPS{1996}}, pages 211--217, 1996.

\bibitem[Maa97]{maass1997networks}
Wolfgang Maass.
\newblock Networks of spiking neurons: the third generation of neural network
  models.
\newblock {\em Neural Networks}, 10(9):1659--1671, 1997.

\bibitem[Maa99]{maass1999neural}
Wolfgang Maass.
\newblock Neural computation with winner-take-all as the only nonlinear
  operation.
\newblock In {\em \NIPS{1999}}, pages 293--299, 1999.

\bibitem[Maa00]{maass2000computational}
Wolfgang Maass.
\newblock On the computational power of winner-take-all.
\newblock {\em Neural Computation}, 12(11):2519--2535, 2000.

\bibitem[Maa14]{maass2014noise}
Wolfgang Maass.
\newblock Noise as a resource for computation and learning in networks of
  spiking neurons.
\newblock {\em Proceedings of the IEEE}, 102(5):860--880, 2014.

\bibitem[MEAM89]{majani1989k}
E~Majani, Ruth Erlanson, and Yaser~S Abu-Mostafa.
\newblock On the $k$-winners-take-all network.
\newblock In {\em \NIPS{1989}}, pages 634--642, 1989.

\bibitem[MP69]{minsky1969perceptrons}
Marvin Minsky and Seymour Papert.
\newblock {\em Perceptrons}.
\newblock MIT Press, 1969.

\bibitem[MS99]{maass1999complexity}
Wolfgang Maass and Michael Schmitt.
\newblock On the complexity of learning for spiking neurons with temporal
  coding.
\newblock {\em Information and Computation}, 153(1):26--46, 1999.

\bibitem[MSS91]{maass1991computational}
Wolfgang Maass, Georg Schnitger, and Eduardo~D Sontag.
\newblock On the computational power of sigmoid versus boolean threshold
  circuits.
\newblock In {\em \FOCS{1991}}, pages 767--776, 1991.

\bibitem[MU05]{mitzenmacher2005probability}
Michael Mitzenmacher and Eli Upfal.
\newblock {\em Probability and Computing: Randomized Algorithms and
  Probabilistic Analysis}.
\newblock Cambridge University Press, 2005.

\bibitem[Mus18]{cam18}
Cameron Musco.
\newblock {\em The Power of Randomized Algorithms: From Numerical Linear
  Algebra to Biological Systems}.
\newblock PhD thesis, Massachusetts Institute of Technology, 2018.

\bibitem[NH10]{nair2010rectified}
Vinod Nair and Geoffrey~E Hinton.
\newblock Rectified linear units improve restricted {B}oltzmann machines.
\newblock In {\em \ICML{2010}}, pages 807--814, 2010.

\bibitem[NN13]{nelson2013osnap}
Jelani Nelson and Huy~L Nguy{\^e}n.
\newblock {OSNAP}: Faster numerical linear algebra algorithms via sparser
  subspace embeddings.
\newblock In {\em \FOCS{2013}}, 2013.

\bibitem[Now89]{nowlan1989maximum}
Steven~J Nowlan.
\newblock Maximum likelihood competitive learning.
\newblock In {\em \NIPS{1989}}, pages 574--582, 1989.

\bibitem[ODL09]{oster2009computation}
Matthias Oster, Rodney Douglas, and Shih-Chii Liu.
\newblock Computation with spikes in a winner-take-all network.
\newblock {\em Neural Computation}, 21(9):2437--2465, 2009.

\bibitem[OF04]{olshausen2004sparse}
Bruno~A Olshausen and David~J Field.
\newblock Sparse coding of sensory inputs.
\newblock {\em Current Opinion in Neurobiology}, 14(4):481--487, 2004.

\bibitem[OL06]{oster2006spiking}
Matthias Oster and Shih-Chii Liu.
\newblock Spiking inputs to a winner-take-all network.
\newblock {\em \NIPS{2006}}, 2006.

\bibitem[Osb79]{osborne1979dale}
Neville~N Osborne.
\newblock Is {D}ale's principle valid?
\newblock {\em Trends in Neurosciences}, 2:73--75, 1979.

\bibitem[Osb13]{osborne2013dale}
Neville~N Osborne.
\newblock {\em Dale's Principle and Communication Between Neurons}.
\newblock Elsevier, 2013.
\newblock Based on a Colloquium of the Neurochemical Group of the Biochemical
  Society, Held at Oxford University, July 1982.

\bibitem[RB15]{roux2015tasks}
Lisa Roux and Gy{\"o}rgy Buzs{\'a}ki.
\newblock Tasks for inhibitory interneurons in intact brain circuits.
\newblock {\em Neuropharmacology}, 88:10--23, 2015.

\bibitem[RFLHL11]{rudy2011three}
Bernardo Rudy, Gordon Fishell, SooHyun Lee, and Jens Hjerling-Leffler.
\newblock Three groups of interneurons account for nearly 100\% of neocortical
  {GABA}ergic neurons.
\newblock {\em Developmental Neurobiology}, 71(1):45--61, 2011.

\bibitem[RGDA97]{riehle1997spike}
Alexa Riehle, Sonja Gr{\"u}n, Markus Diesmann, and Ad~Aertsen.
\newblock Spike synchronization and rate modulation differentially involved in
  motor cortical function.
\newblock {\em Science}, 278(5345):1950--1953, 1997.

\bibitem[San89]{sanger1989optimal}
Terence~D Sanger.
\newblock Optimal unsupervised learning in a single-layer linear feedforward
  neural network.
\newblock {\em Neural Networks}, 2(6):459--473, 1989.

\bibitem[Sar06]{sarlos2006improved}
T\'{a}mas Sarl\'{o}s.
\newblock Improved approximation algorithms for large matrices via random
  projections.
\newblock In {\em \FOCS{2006}}, pages 143--152, 2006.

\bibitem[Sha15]{shamir2015stochastic}
Ohad Shamir.
\newblock A stochastic {PCA} and {SVD} algorithm with an exponential
  convergence rate.
\newblock In {\em \ICML{2015}}, pages 144--152, 2015.

\bibitem[SN94]{shadlen1994noise}
Michael~N Shadlen and William~T Newsome.
\newblock Noise, neural codes and cortical organization.
\newblock {\em Current Opinion in Neurobiology}, 4(4):569--579, 1994.

\bibitem[SS95]{siegelmann1992computational}
Hava~T Siegelmann and Eduardo~D Sontag.
\newblock On the computational power of neural nets.
\newblock {\em Journal of Computer and System Sciences}, 50(1):132--150, 1995.

\bibitem[Stu]{neuronBasic}
Robert Stufflebeam.
\newblock Introduction to neurons, synapses, action potentials, and
  neurotransmission.
\newblock In {\em The Mind Project}.
  \url{http://www.mind.ilstu.edu/curriculum/modOverview.php?modGUI=232}.

\bibitem[Tho90]{thorpe1990spike}
Simon~J Thorpe.
\newblock Spike arrival times: A highly efficient coding scheme for neural
  networks.
\newblock {\em Parallel Processing in Neural Systems}, pages 91--94, 1990.

\bibitem[Tra09]{trappenberg2009fundamentals}
Thomas Trappenberg.
\newblock {\em Fundamentals of Computational Neuroscience}.
\newblock OUP Oxford, 2009.

\bibitem[Val00a]{valiant2000circuits}
Leslie~G Valiant.
\newblock {\em Circuits of the Mind}.
\newblock Oxford University Press, 2000.

\bibitem[Val00b]{valiant2000neuroidal}
Leslie~G Valiant.
\newblock A neuroidal architecture for cognitive computation.
\newblock {\em Journal of the ACM (JACM)}, 47(5):854--882, 2000.

\bibitem[Val05]{valiant2005memorization}
Leslie~G Valiant.
\newblock Memorization and association on a realistic neural model.
\newblock {\em Neural cCmputation}, 17(3):527--555, 2005.

\bibitem[Vap98]{vapnik1998statistical}
Vladimir Vapnik.
\newblock {\em Statistical Learning Theory}, volume~1.
\newblock Wiley New York, 1998.

\bibitem[VVAE94]{van1994inhibition}
Carl Van~Vreeswijk, LF~Abbott, and G~Bard Ermentrout.
\newblock When inhibition not excitation synchronizes neural firing.
\newblock {\em Journal of Computational Neuroscience}, 1(4):313--321, 1994.

\bibitem[WMK{\etalchar{+}}02]{watanabe2002gaba}
Masahito Watanabe, Kentaro Maemura, Kiyoto Kanbara, Takumi Tamayama, and Hana
  Hayasaki.
\newblock {GABA} and {GABA} receptors in the central nervous system and other
  organs.
\newblock In {\em International Review of Cytology}, volume 213, pages 1--47.
  Elsevier, 2002.

\bibitem[WR94]{waterhouse1994classification}
Steve~R Waterhouse and Anthony~J Robinson.
\newblock Classification using hierarchical mixtures of experts.
\newblock In {\em Proceedings of the IEEE Workshop Neural Networks for Signal
  Processing}, pages 177--186, 1994.

\bibitem[WS03]{wang2003k}
Wei Wang and Jean-Jacques~E Slotine.
\newblock $k$-winners-take-all computation with neural oscillators.
\newblock {\em \arXiv{q-bio/0401001}}, 2003.

\bibitem[WSO{\etalchar{+}}07]{womelsdorf2007modulation}
Thilo Womelsdorf, Jan-Mathijs Schoffelen, Robert Oostenveld, Wolf Singer,
  Robert Desimone, Andreas~K Engel, and Pascal Fries.
\newblock Modulation of neuronal interactions through neuronal synchronization.
\newblock {\em Science}, 316(5831):1609--1612, 2007.

\bibitem[YG89]{yuille1989winner}
Alan~L Yuille and Norberto~M Grzywacz.
\newblock A winner-take-all mechanism based on presynaptic inhibition feedback.
\newblock {\em Neural Computation}, 1(3):334--347, 1989.

\bibitem[ZP96]{zador1996vc}
Anthony~M Zador and Barak~A Pearlmutter.
\newblock {VC} dimension of an integrate-and-fire neuron model.
\newblock In {\em \COLT{1996}}, pages 10--18. ACM, 1996.

\end{thebibliography}

\end{document}